%% file: main.tex
\documentclass[review=false]{jfp-epi}

\usepackage[T1]{fontenc}

\usepackage{graphicx}

\usepackage{amsmath}

\makeatletter
\let\langle\undefined \let\rangle\undefined
\DeclareSymbolFont{yhlargesymbols}{OMX}{yhex}{m}{n}
\DeclareMathDelimiter{\rangle}
  {\mathclose}{symbols}{"69}{yhlargesymbols}{"0B}
\DeclareMathDelimiter{\langle}
  {\mathopen}{symbols}{"68}{yhlargesymbols}{"0A}

\makeatother

\usepackage{mathtools}
\usepackage{mathrsfs}
\usepackage{mathpartir}
\usepackage{tabularx}
\usepackage{tikz}
\usepackage{tikz-cd}
\tikzcdset{every label/.append style = {font = \Large}}
\usepackage{dashbox}
\usepackage{subfiles}
\allowdisplaybreaks
\usepackage{enumitem}

\usepackage{mleftright}
\mleftright

\mathtoolsset{showonlyrefs=true}

\usepackage{ifthen}
\newboolean{isFinal}
\setboolean{isFinal}{false}
\newcommand{\version}[2]{\ifbool{isFinal}{#1}{#2}}

\usepackage{xcolor}

\input{macros.tex}

\newenvironment{addmargin}[2][\empty]{\par
  \rightskip=#2\relax
  \ifx\empty#1\relax \leftskip=\rightskip
  \else \leftskip=#1\relax
  \fi}{\par}

\usepackage{hyperref}
\usepackage{color}

\AddToHook{cmd/appendix/before}{}

\begin{document}
\title{Expressive power of one-shot control operators and coroutines}

\author{Kentaro Kobayashi}
\orcid{0009-0004-3505-7151}
\affiliation{%
  \institution{University of Tsukuba}
  \city{Tsukuba}
  \state{Ibaraki}
  \country{Japan}
  \authoremail{kentaro.kobayashi@acm.org}
}

\author{Yukiyoshi Kameyama}
\orcid{0000-0002-2693-5133}
\affiliation{%
  \institution{University of Tsukuba}
  \city{Tsukuba}
  \state{Ibaraki}
  \country{Japan}
  \authoremail{kameyama@acm.org}
}

\begin{abstract}
  Control operators, such as exceptions and effect handlers, provide a means of
  representing computational effects in programs abstractly and modularly.
  While most theoretical studies have focused on multi-shot control operators,
  one-shot control operators---which restrict the use of captured continuations
  to at most once---are gaining attention for their balance between
  expressiveness and efficiency.
  This study aims to fill the gap.
  We present a mathematically rigorous comparison of the expressive power among
  one-shot control operators, including effect handlers, delimited
  continuations, and even asymmetric coroutines.
  Following previous studies on multi-shot control operators, we adopt
  Felleisen's macro-expressiveness as our measure of expressiveness.
  We verify the folklore that one-shot effect handlers and one-shot
  delimited-control operators can be macro-expressed by asymmetric coroutines,
  but not vice versa.
  We explain why a previous informal argument fails, and how to revise it to
  make a valid macro-translation.
\end{abstract}

\maketitle

\subfile{sec1-introduction.tex}

\subfile{sec2-framework.tex}

\subfile{sec3-deltoac.tex}

\subfile{sec4-efftodel.tex}

\subfile{sec5-nonexistence.tex}

\subfile{sec6-related-work.tex}

\subfile{sec7-conclusion.tex}

\bibliographystyle{ACM-Reference-Format}
\bibliography{refs}

\version{}
{\appendix

\subfile{appA-deltoac.tex}

\subfile{appB-efftodel.tex}

\subfile{appC-non-reftodel.tex}

\subfile{appD-non-deltoref.tex}}

\end{document}

%% file: macros.tex
\newenvironment{caseindent}{\begin{addmargin}[\parindent]{0pt}}{\end{addmargin}}

\newcommand{\mam}{\mathbf{MAM}}
\newcommand{\eff}{\mathbf{EFF}_{\mathrm{one}}}
\newcommand{\del}{\mathbf{DEL}_{\mathrm{one}}}
\newcommand{\ac}{\mathbf{AC}}
\newcommand{\reflang}{\mathbf{REF}}

\newcommand{\syntacticset}[1]{\mathscr{#1}}
\newcommand{\dom}[1]{\mathrm{Dom}(#1)}
\newcommand{\im}[1]{\mathrm{Im}(#1)}

\newcommand{\lb}[1]{\mathrm{LB}\left(#1\right)}

\newcommand{\var}[1]{\mathit{#1}}
\newcommand{\unit}{\mathtt{()}}
\newcommand{\vpair}[2]{(#1,#2)}
\newcommand{\inj}[2]{\mathbf{inj}_{\mathrm{#1}}\,#2}
\newcommand{\thunk}[1]{\left\{#1\right\}}
\newcommand{\force}[1]{#1!}
\newcommand{\return}[1]{\mathbf{return}\,#1}
\newcommand{\seq}[3]{\mathbf{let}\;\var{#1} = #2\;\mathbf{in}\;#3}
\newcommand{\letin}[3]{\mathbf{let}\;\var{#1} = #2\;\mathbf{in}\;#3} %
\newcommand{\abs}[2]{\lambda \var{#1}.\;#2}

\newcommand{\app}[2]{#1\,#2}

\newcommand{\cpair}[2]{\left\langle #1,#2\right\rangle}
\newcommand{\prj}[2]{\mathbf{prj}_{#1}\,#2}
\newcommand{\pcase}[4]{\mathbf{case}\;#1\;\mathbf{of}\;(#2,#3) \mapsto #4}
\newcommand{\scase}[4]{\mathbf{case}\;#1\;\mathbf{of}\;\{(\inj{#2}{#3} \mapsto #4)_i\}}
\newcommand{\underscore}{\texttt{\_}}
\newcommand{\paren}[1]{\left(#1\right)}

\newcommand{\nil}{\mathbf{nil}}

\newcommand{\shift}[2]{\mathbf{S_0} \var{#1}.\;#2}
\newcommand{\dollar}[3]{\left\langle #1 \middle\vert\, \var{#2}. #3 \right\rangle}
\newcommand{\dollart}[2]{\left\langle #1 \middle\vert\, #2 \right\rangle}
\newcommand{\throw}[2]{\mathbf{throw}\,#1\,#2}
\newcommand{\plug}[2]{\mathcal{#1}\left[#2\right]}

\newcommand{\op}[1]{\mathtt{#1}}
\newcommand{\opcall}[2]{\app{\op{#1}}{#2}}
\newcommand{\handle}[2]{\mathbf{with}\;#1\;\mathbf{handle}\;#2}
\newcommand{\handler}[3]{\{\mathbf{return}\;\var{#1} \mapsto #2,\;#3\}}

\newcommand{\labeledc}[2]{\var{#1} : #2}
\newcommand{\labeledcp}[2]{\var{#1} : \paren{#2}}
\newcommand{\create}[1]{\mathbf{create}\,#1}
\newcommand{\resume}[2]{\mathbf{resume}\,#1\,#2}
\newcommand{\yield}[1]{\mathbf{yield}\,#1}
\newcommand{\refcell}[1]{\mathit{RefCell}(#1)}
\newcommand{\activelabels}[1]{\mathit{AL}(#1)}
\newcommand{\wellformed}[1]{\mathsf{WF}_{\ac}(#1)}

\newcommand{\set}[2]{\mathbf{set}\;#1\;#2}
\newcommand{\get}[1]{\mathbf{get}\;#1}

\newcommand{\config}[2]{\left\langle #1; #2\right\rangle}

\newcommand{\hole}{[\ ]}
\newcommand{\context}[1]{\mathcal{#1}}
\newcommand{\Eval}[1]{\mathrm{Eval}_{#1}}
\newcommand{\eval}[2]{\mathrm{Eval}_{#1}\left(#2\right)}

\newcommand{\redbetaM}[2]{#1 \rightarrow_{\mathbf{M}}^{\beta} #2}
\newcommand{\arrbetaM}{\rightarrow_{\mathbf{M}}^{\beta}}

\newcommand{\redbetaE}[2]{{#1} \rightarrow_{\mathbf{E}}^{\beta} {#2}}
\newcommand{\redbetaEp}[4]{\config{#1}{#2} \rightarrow_{\mathbf{E}}^{\beta} \config{#3}{#4}}

\newcommand{\redbetaD}[2]{{#1} \rightarrow_{\mathbf{D}}^{\beta} {#2}}
\newcommand{\redbetaDp}[4]{\config{#1}{#2} \rightarrow_{\mathbf{D}}^{\beta} \config{#3}{#4}}
\newcommand{\arrbetaD}{\rightarrow_{\mathbf{D}}^{\beta}}

\newcommand{\redbetaAC}[2]{{#1} \rightarrow_{\mathbf{AC}}^{\beta} {#2}}
\newcommand{\redbetaACp}[4]{\config{#1}{#2} \rightarrow_{\mathbf{AC}}^{\beta} \config{#3}{#4}}

\newcommand{\redbetaR}[2]{{#1} \rightarrow_{\mathbf{R}}^{\beta} {#2}}
\newcommand{\redbetaRp}[4]{\config{#1}{#2} \rightarrow_{\mathbf{R}}^{\beta} \config{#3}{#4}}

\newcommand{\redM}[2]{#1 \rightarrow_{\mathbf{M}} #2}
\newcommand{\arrM}{\rightarrow_{\mathbf{M}}}
\newcommand{\redMclos}[2]{#1 \rightarrow^{*}_{\mathbf{M}} #2}

\newcommand{\redE}[2]{#1 \rightarrow_{\mathbf{E}} #2}
\newcommand{\arrE}{\rightarrow_{\mathbf{E}}}
\newcommand{\redEp}[4]{\config{#1}{#2} \rightarrow_{\mathbf{E}} \config{#3}{#4}}
\newcommand{\redEclos}[2]{#1 \rightarrow^{*}_{\mathbf{E}} #2}
\newcommand{\redEplus}[2]{#1 \rightarrow^{+}_{\mathbf{E}} #2}

\newcommand{\redD}[2]{#1 \rightarrow_{\mathbf{D}} #2}
\newcommand{\arrD}{\rightarrow_{\mathbf{D}}}
\newcommand{\redDp}[4]{\config{#1}{#2} \rightarrow_{\mathbf{D}} \config{#3}{#4}}
\newcommand{\redDclos}[2]{#1 \rightarrow^{*}_{\mathbf{D}} #2}
\newcommand{\redDplus}[2]{#1 \rightarrow^{+}_{\mathbf{D}} #2}

\newcommand{\redAC}[2]{#1 \rightarrow_{\mathbf{AC}} #2}
\newcommand{\arrAC}{\rightarrow_{\mathbf{AC}}}
\newcommand{\redACp}[4]{\config{#1}{#2} \rightarrow_{\mathbf{AC}} \config{#3}{#4}}
\newcommand{\redACclos}[2]{#1 \rightarrow^{*}_{\mathbf{AC}} #2}
\newcommand{\redACplus}[2]{#1 \rightarrow^{+}_{\mathbf{AC}} #2}

\newcommand{\redR}[2]{#1 \rightarrow_{\mathbf{R}} #2}
\newcommand{\arrR}{\rightarrow_{\mathbf{R}}}
\newcommand{\redRp}[4]{\config{#1}{#2} \rightarrow_{\mathbf{R}} \config{#3}{#4}}
\newcommand{\redRclos}[2]{#1 \rightarrow^{*}_{\mathbf{R}} #2}
\newcommand{\redRplus}[2]{#1 \rightarrow^{+}_{\mathbf{R}} #2}

\newcommand{\mt}[1]{\underline{#1}}
\newcommand{\mte}[2]{\underline{#1}_{#2}}
\newcommand{\defeq}{\overset{\mathrm{def}}{=}}

\newcommand{\mtempty}{\mt{\makebox[2mm][c]{$\cdot$}}}

\newcommand{\simu}[2]{#1 \sim #2}
\newcommand{\simue}[3]{#1\;\overset{#3}{\sim}\;#2}
\newcommand{\simues}[3]{#1\;\overset{#3}{\sim}_{\sharp}\;#2}
\newcommand{\simuec}[3]{#1\;\overset{#3}{\sim}_{\mathsf{c}}\;#2}

\newcommand{\rsimue}[1]{\overset{#1}{\sim}}
\newcommand{\rsimuec}[1]{\overset{#1}{\sim}_{\mathsf{c}}}

\newcommand{\rsimues}[1]{\overset{#1}{\sim}_{\sharp}}

\newcommand{\headmatch}[2]{#1 \mathrel{\vartriangleright} #2}

\newcommand{\Cont}{\mathrm{Cont}}
\newcommand{\Act}{\mathrm{Act}}
\newcommand{\ActCtx}{\mathrm{ActCtx}}
\newcommand{\Inv}{\mathsf{Inv}}

\newcommand{\valid}[1]{\inj{Valid}{#1}}
\newcommand{\invalid}{\inj{Invalid}{\unit}}

\newcommand{\WF}{\mathsf{WF}}
\newcommand{\Coh}{\mathsf{Coh}}

\newcommand{\cdollar}[1]{\mathcal{A}_{\mathrm{dollar}, #1}}
\newcommand{\cshift}[1]{\mathcal{A}_{\mathrm{shift0}, #1}}
\newcommand{\ccreate}{\mathcal{A}_{\mathrm{create}}}
\newcommand{\cget}{\mathcal{A}_{\mathrm{get}}}
\newcommand{\cset}{\mathcal{A}_{\mathrm{set}}}

%% file: sec1-introduction.tex
\section{Introduction}

Control operators are powerful tools for representing computational effects.
Exceptions and coroutines are classic control operators, implemented in many
languages.
Delimited-control operators (e.g., shift/reset) have been extensively studied in
the literature.
The last decade has seen growing interest in effect handlers, which support
modular abstraction of computational
effects~\cite{plotkin2003algebraic,plotkin2009handlers}.

This paper studies the theoretical foundation of \emph{one-shot} variants of
control operators, where captured continuations are restricted to at most one
use.
While most studies on control operators have focused on unrestricted (i.e.,
\emph{multi-shot}) control operators,\footnote{A notable exception is Berdine et
  al., who proposed linearly-used continuations~\cite{berdine2002linear}.}
one-shot variants have been recently gaining attention.
There are several reasons to consider one-shot variants: First, they can be
implemented more efficiently than multi-shot ones, as they avoid stack
copying~\cite{bruggeman1996oneshot}.
Second, they may alleviate the verification burden by reducing the complexity of
reasoning about sensitive resources~\cite{vilhena2021separation}.
Third, one-shotness is key to relating control operators based on continuations
to classic ones found in many dynamic languages.
For instance, a recent example in the former category is the one-shot effect
handlers of OCaml Version 5.x,\footnote{\url{https://ocaml.org}} while the
latter category includes coroutines and the yield operator, both of which are
intrinsically one-shot.\footnote{James and Sabry argued that a multi-shot
  variant of the yield operator is as expressive as delimited-control
  operators~\cite{yield2011}.}

Despite these advantages, few authors have studied the theoretical foundation of
one-shot control operators.
Indeed, it was folklore that results for multi-shot control operators carry over
easily to their one-shot counterparts; however, this is not the case.
Since one-shotness is a dynamic property, a formal calculus for one-shot control
operators should track the validity of each continuation, complicating both the
semantics and precise reasoning about them.
Among the few authors, de Moura and Ierusalimschy demonstrated a connection
between one-shot delimited-control operators and
coroutines~\cite{moura2009revisiting}.
However, this correspondence relies heavily on mutable state that can store
higher-order functions, which, in our view, obscures the raw expressive power of
these control operators.

We note that comparing the expressiveness of control operators is surprisingly
difficult.
For the case of one-shot control operators, the expressiveness results in the
literature often lacked correctness proofs (e.g.,~\cite{kawahara2020one}), or
were incorrect.
In this paper, we explain why a simple and seemingly correct translation from
delimited-control operators to coroutines fails to preserve semantics.

This paper studies the expressive power of three one-shot control facilities:
one-shot effect handlers, one-shot delimited-control operators, and
coroutines~\cite{moura2009revisiting}.
We adopt macro-expressibility~\cite{felleisen1991expressive} as the basis for
our comparison, since it is well-established in the literature, as in the study
of control operators by Forster, Kammar, Lindley, and
Pretnar~\cite{forster2019expressive}.
To our knowledge, this is the first systematic study to rigorously analyze the
expressiveness of one-shot control operators.

Our contributions are threefold:
\begin{enumerate}
\item We prove that one-shot delimited-control operators can be macro-expressed
  by asymmetric coroutines.
\item We also prove that one-shot effect handlers can be macro-expressed by
  asymmetric coroutines.
\item We show that the converse direction does not hold: asymmetric coroutines
  cannot be macro-expressed by either one-shot delimited-control operators or
  one-shot effect handlers.
\end{enumerate}
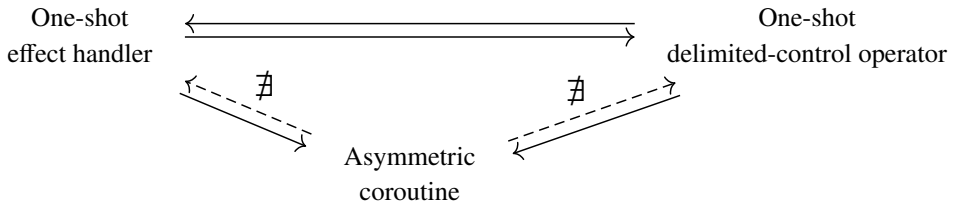
\begin{figure}[h]
  \centering
  \begin{tikzcd}
    \text{
      \begin{tabular}{c}
        One-shot \\
        effect handler
      \end{tabular}
    } \arrow[rrrr] \arrow[rrd, shift right=2] &  &                                                                                                &  & \text{
      \begin{tabular}{c}
        One-shot \\
        delimited-control operator
      \end{tabular}
    } \arrow[llll, shift right=2] \arrow[lld, shift left=2] \\
    &  & \text{
      \begin{tabular}{c}
        Asymmetric\\
        coroutine
      \end{tabular}
    } \arrow[rru, "\nexists", dashed] \arrow[llu, "\nexists"', dashed] &  &
  \end{tikzcd}
  \caption{Macro-expressibility among control operators.
    Solid lines indicate the existence of a macro-translation; dashed lines
    indicate its non-existence.}
  \label{fig:relations}
\end{figure}
Figure~\ref{fig:relations} illustrates the macro-expressibility results established in this paper.

This is the extended version of our paper presented at APLAS
2025~\cite{DBLP:conf/aplas/KobayashiK25}.
It differs from the conference version as follows:
\begin{enumerate}
\item We add all key definitions and several important proof cases that could
  not be included for lack of space.
\item We give a new translation from one-shot delimited-control operators to
  coroutines, which is much simpler than that presented in the conference
  version, and allows a more concise proof of correctness.
\item We present a new inexpressibility result in
  Section~\ref{sec:nonexistent:deltoref}, which completes the comparison of all
  the calculi introduced in this paper.
\end{enumerate}

The remainder of this paper is organized as follows.
In Section~\ref{sec:framework}, we introduce the core calculus and define
macro-expressibility.
In Section~\ref{sec:DELoneToAC}, we introduce two calculi for one-shot delimited
continuations and coroutines.
We first demonstrate that a na\"{i}ve translation from one-shot delimited
continuations to coroutines fails to preserve semantics, analyzing the cause of
failure.
We then give a refined translation and prove its correctness.
In Section~\ref{sec:EFFoneToDELone}, we introduce the calculus for one-shot
effect handlers, and establish the equi-expressivity of one-shot effect handlers
and one-shot delimited continuations.
In Section~\ref{sec:nonexistent}, we describe the inexpressibility results.
Finally, in Section~\ref{sec:conclusion}, we conclude the paper.
\version{%

Proofs omitted from the text appear in the full version.
}{%
Proofs omitted from the main text appear in the appendices.
}

%% file: sec2-framework.tex
\section{Core calculus and macro-expressibility}\label{sec:framework}

\subsection{Core calculus \texorpdfstring{$\mam$}{MAM}}

First, we present the calculus $\mam$ (multi-adjunctive metalanguage) by Forster
et al.~\cite{forster2019expressive}, which is designed after Levy's
call-by-push-value~\cite{levy2004cbpv}.
This framework subsumes both call-by-value and call-by-name evaluation
strategies and provides a uniform foundation to study various computational
effects.
In this paper, \(\mam\) serves as the common core calculus for our extensions.

\begin{figure}[tb]
  \setlength{\parindent}{0cm}
  \setlength{\arraycolsep}{4pt}%
  \renewcommand{\arraystretch}{0.9}%
  \noindent
  \begin{minipage}[t]{0.62\columnwidth}
    \begin{tabular}[t]{@{}l@{\hspace{0.1em}}l@{\hspace{0.5em}}l@{}}
      $V,W$ & $::=$ & \textsf{value} \\
            & \textbar\ $x \in \syntacticset{V}$ & variable \\
            & \textbar\ $\unit$ & unit \\[1em]
      $M,N$ & $::=$ & \textsf{computation} \\
            & \textbar\ $\return{V}$ & returner \\
            & \textbar\ $\seq{x}{M}{N}$ & sequencing \\
            & \textbar\ $\pcase{V}{x_1}{x_2}{M}$ & product matching \\      
            & \textbar\ $\scase{V}{L_{\mathnormal{i}}}{x_i}{M_i}$ & variant matching
    \end{tabular}
  \end{minipage}%
  \begin{minipage}[t]{0.35\columnwidth}
    \begin{tabular}[t]{@{}l@{\hspace{0.5em}}l@{}}  
      \textbar\ $\vpair{V}{W}$ & pairing \\
      \textbar\ $\inj{L}{V}\quad(\mathrm{L} \in \syntacticset{C})$ & variant \\
      \textbar\ $\thunk{M}$ & thunk \\[1em]
      \textbar\ $\force{V}$ & force  \\
      \textbar\ $\abs{x}{M}$ & abstraction \\
      \textbar\ $\app{M}{V}$ & application \\
      \textbar\ $\cpair{M}{N}$ & pairing \\
      \textbar\ $\prj{i}{M}$ & projection
    \end{tabular}
  \end{minipage}%
  \caption{Syntax of $\mam$}
  \label{fig:mam_syntax}
\end{figure}

Figure~\ref{fig:mam_syntax} gives the syntax of $\mam$.
In this language, values and computations are clearly separated into two
different syntactic categories.

Values, ranged over by the metavariables \(V\) and \(W\), are the data consumed
by computations.
As usual, this category includes variables (drawn from a countable set
\(\syntacticset{V}\)), the unit value, pairings \(\vpair{V}{W}\) (products), and
variants \(\inj{L}{V}\) (sums, whose tags \(\mathrm{L}\) range over a finite set
\(\syntacticset{C}\)).
In addition, a thunk \(\thunk{M}\) is a value that suspends a computation \(M\)
and reifies it as data.

Computations, ranged over by the metavariables \(M\) and \(N\), consume such
data and, when they terminate, return a value.
A computation returns a value \(V\) via the returner \(\return{V}\), and two
computations can be sequenced by \(\letin{x}{M}{N}\): once \(M\) returns a
value, that value is substituted for \(x\) in \(N\).
The pattern-matching constructs \(\pcase{V}{x_1}{x_2}{M}\) and
\(\scase{V}{L_{\mathnormal{i}}}{x_i}{M_i}\) match on a product and a variant
value, respectively.
A force term \(\force{V}\) runs the computation represented by \(V\); in
particular, if \(V\) is a thunk \(\thunk{M}\), \(\force{\thunk{M}}\) reduces to
the computation \(M\).
An abstraction \(\abs{x}{M}\) binds a variable \(x\) in computation \(M\), and
an application \(\app{M}{V}\) supplies a value \(V\) to such an abstraction.
Since an abstraction is classified as a computation, to pass it as a value to
other computations, it must be suspended as a thunk.
Finally, \(\cpair{M}{N}\) pairs two computations \(M\) and \(N\), and
\(\prj{i}{M}\) takes the \(i\)-th component of such a pairing, reducing to that
component.

Throughout the paper, we freely use nested patterns and a catch-all clause in
the pattern-matching constructs; e.g., we write a term such as
\(\mathbf{case}\;V\;\mathbf{of}\;\{\inj{L}{\vpair{x_1}{x_2}} \mapsto M,\;
\underscore \mapsto N\}\).
This is only a superficial extension: such a term can be desugared into a
combination of pattern-matching constructs.
We also write \(\underscore\) for a variable that does not occur free in the
body, as in \(\letin{\underscore}{M}{N}\).

We consider terms modulo renaming of bound variables as usual.
We define the set of \(\mam\) \emph{programs} to be the set of all computations.

We give the operational semantics of $\mam$ using evaluation
contexts~\cite{felleisen1987reduction}: we decompose a computation into an
evaluation context and a redex, and reduce the latter using beta reduction
rules.
Figure~\ref{fig:mam_frame} defines frames and contexts.
Evaluation contexts consist of computational frames, while pure contexts
comprise only pure frames.
We use an evaluation context to locate the redex as usual; a pure context, in
turn, identifies the part of an evaluation context that has no (immediate)
computational effects.
Although these two notions of contexts agree in \(\mam\) (since it has no
computational effects), in each extension of \(\mam\) introduced below, we add
computational frames according to the control abstraction it offers.

Figure~\ref{fig:mam_beta} defines the beta reduction rules $\arrbetaM$ on
computations.
\(M[V_1/x_1, \ldots, V_n/x_n]\) denotes the simultaneous, capture-avoiding
substitution of each \(x_i\) with \(V_i\) in \(M\).
\begin{figure}[tb]
  \setlength{\arraycolsep}{4pt}%
  \renewcommand{\arraystretch}{0.9}%
  \begin{tabular}{rrcl}
    \textsf{pure frame} & $\mathcal{P}$ & ::= & $\seq{x}{\hole}{N}$ \textbar\ $\app{\hole}{V}$ \textbar\ $\prj{i}{\hole}$ \\
    \textsf{computational frame} & $\mathcal{F}$ & ::= & $\mathcal{P}$ \\
    \textsf{pure context} & $\mathcal{H}$ & ::= & $\hole$ \textbar\ $\mathcal{P}[\mathcal{H}\hole]$ \\
    \textsf{evaluation context} & $\mathcal{C}$ & ::= & $\hole$ \textbar\ $\mathcal{F}[\mathcal{C}\hole]$
  \end{tabular}%
  \caption{Frames and contexts of $\mam$}
  \label{fig:mam_frame}
\end{figure}
\begin{figure}[tb]
  \centering
  \begin{minipage}[t]{0.49\linewidth}
    \begin{tabular}{l}
    $(\times)\quad\begin{array}{l}
      \redbetaM{\pcase{(V_1,V_2)}{x_1}{x_2}{M}\\}{M[V_1/x_1,V_2/x_2]}
                 \end{array}$ \\
    $(+)\quad\begin{array}{l}
              \redbetaM{\mathbf{case}\;\mathbf{inj}_{\mathrm{L}_k}\,{V}\;\mathbf{of}\;\{(\mathbf{inj}_{\mathrm{L}_i}\,{x_i} \mapsto M_i)_i\}\\}{M_k[V/x_k]}
            \end{array}$
  \end{tabular}
  \end{minipage}%
  \begin{minipage}[t]{0.51\linewidth}
    \begin{tabular}{l}
    $(F)\quad\redbetaM{\seq{x}{\return{V}}{M}}{M[V/x]}$ \\      
    $(U)\quad\redbetaM{\force{\thunk{M}}}{M}$ \\
    $(\to)\quad\redbetaM{\app{(\abs{x}{M})}{V}}{M[V/x]}$ \\
    $(\&)\quad\redbetaM{\prj{i}{\cpair{M_1}{M_2}}}{M_i}$
  \end{tabular}
  \end{minipage}%
  \caption{Beta reduction rules of $\mam$}
  \label{fig:mam_beta}
\end{figure}
We define the transition relation $\arrM$ on computations by
using evaluation contexts as follows:
\[
  \inferrule
    {\redbetaM{M}{M'}}
    {\redM{\mathcal{C}[M]}{\mathcal{C}[M']}}.
\]
We remark that the decomposition of a computation into an evaluation context and
a redex is unique.

Evaluation of a program is defined by:
$\eval{\mam}{M} \defeq V \quad\text{if}\quad
\redMclos{M}{\return{V}}$.\footnote{We write \(\to^*\) for the reflexive
  transitive closure of a binary relation \(\to\), and \(\to^+\) for its
  transitive closure.}
$\Eval{\mam}$ is a well-defined partial function, since $\arrM$ is
deterministic: there exists a unique maximal reduction sequence starting from
\(M\).
It is defined at \(M\) exactly when the evaluation of \(M\) terminates
successfully.
We say that the evaluation of \(M\) \emph{diverges} if this maximal reduction
sequence is infinite.
We adopt the same terminology for the calculi introduced in later sections.

\begin{example}
  We assume that \(\mam\) is extended with integers and arithmetic
  operators.\footnote{Henceforth, we freely use this extension when we give
    examples.}
  Then,
  \(\eval{\mam}{\app{(\abs{f}{\app{\force{f}}{2}})}{\thunk{\abs{x}{x + 1}}}} =
  3\).
  In particular, the reduction proceeds as follows:
  \begin{align*}
    \app{(\abs{f}{\app{\force{f}}{2}})}{\thunk{\abs{x}{x + 1}}} &\arrM \app{\force{\thunk{\abs{x}{x + 1}}}}{2} \\
                                                                &\arrM \app{(\abs{x}{x + 1})}{2} \\
                                                                &\arrM 2 + 1 \arrM 3.
  \end{align*}
\end{example}

\subsection{Macro-expressibility}

To compare the expressive power of one-shot control operators, we use the notion
of \emph{macro-expressibility}, which was originally proposed by
Felleisen~\cite{felleisen1991expressive}, and has been adopted in a number of
studies on control operators.
This notion serves as a finer basis for comparison than computability, and is
especially useful when comparing the expressive power of extensions of a
Turing-complete language, such as \(\mam\).

Before defining the notion of macro-expressibility formally, let us clarify the
target of comparison.
An extension of \(\mam\) considered in this paper should contain the syntactic
categories of values and computations, and may have finitely many additional
syntactic categories.
It may also have additional constructors.
Following Felleisen~\cite{felleisen1991expressive}, we regard a binding
constructor as a family of constructors indexed by the bound variable, e.g.,
\(\letin{x}{-}{-}\) for each \(x \in \syntacticset{V}\), and a variable as a
nullary constructor.
For each extension, the set of programs and the evaluation function should be
defined explicitly.\footnote{Our treatment of extensions of \(\mam\) is
  deliberately made somewhat informal.
  We could have defined a general notion of languages and macro-expressibility
  over them, but it would involve complex arguments about binding structures,
  possibly using the theory of abstract syntax~\cite{fiore1999abstractsyntax}.
  Such a rigorous treatment is necessary when developing a general theory of
  macro-expressibility; however, in this paper, we are interested in the
  expressive power of several fixed calculi.
}

\begin{definition}\label{def:macro-expressibility}
  Let $\mathscr{L}_1$ and $\mathscr{L}_2$ be extensions of $\mam$.
  A partial function $\phi$ from $\mathscr{L}_1$-terms to $\mathscr{L}_2$-terms
  is a (strong) \emph{macro-translation} if and only if the following conditions
  are satisfied:
  \begin{enumerate}
  \item If $M$ is an $\mathscr{L}_1$-program, then $\phi(M)$ is defined and an $\mathscr{L}_2$-program;
  \item If $F$ is an $n$-ary constructor of $\mam$, then for any $\mathscr{L}_1$-terms $M_1, \ldots, M_n$,
    \[
      \phi(F(M_1, \ldots, M_n)) = F(\phi(M_1), \ldots, \phi(M_n));
    \]
  \item For each $n$-ary constructor $F \in \mathscr{L}_1\setminus\mam$, there
    is an $n$-hole syntactic abstraction\footnote{An $n$-hole syntactic
      abstraction in $\mathscr{L}$ is an $\mathscr{L}$-term with $n$ holes.}
    $\mathcal{A}$ in $\mathscr{L}_2$ such that
    \[
      \phi(F(M_1, \ldots, M_n)) = \mathcal{A}[\phi(M_1), \ldots, \phi(M_n)]
    \]
    for any $\mathscr{L}_1$-terms $M_1, \ldots, M_n$;
  \item\label{item:4} For every \(\mathscr{L}_1\) program \(M\),
    $\eval{\mathscr{L}_1}{M}$ is defined if and only if
    $\eval{\mathscr{L}_2}{\phi(M)}$ is defined.
  \end{enumerate}
  We say $\mathscr{L}_1$ is (strongly) \emph{macro-expressible} in $\mathscr{L}_2$ if a macro-translation from $\mathscr{L}_1$
  to $\mathscr{L}_2$ exists.
  We define a \emph{weak} macro-translation by replacing ``if and only if'' in
  (\ref{item:4}) with ``only if''.
\end{definition}

\begin{example}
  An example of macro-expressibility is the sequencing \(M; N\), which evaluates
  \(M\) first, discards the result, and then evaluates \(N\).
  Consider extending \(\mam\) with this sequencing constructor.
  Then we define a macro-translation $\phi$ from this extension to \(\mam\) such
  that $\phi({M}; {N}) = (\letin{\underscore}{\phi(M)}{\phi(N)})$ holds.
  Other constructors, such as the abstraction \(\abs{x}{M}\), are translated
  homomorphically: \(\phi(\abs{x}{M}) = \abs{x}{\phi(M)}\).
\end{example}

We remark that macro-expressibility is reflexive and transitive: an identity
translation is a macro-translation, and the composition of macro-translations is
also a macro-translation.

%% file: sec3-deltoac.tex
\section{One-shot delimited continuations as asymmetric coroutines}\label{sec:DELoneToAC}

This section investigates the macro-expressibility of one-shot delimited
continuations in terms of asymmetric coroutines.
Although this expressibility had been considered trivially true, establishing it
turns out to be unexpectedly complicated.
We will first introduce two calculi: $\del$ for one-shot delimited continuations
and $\ac$ for asymmetric coroutines.
We will then present a simple translation from $\del$ to $\ac$, and explain why
this translation fails.
Finally, we will give a refined translation which is proved to be a correct
macro-translation.

\subsection{The calculus for one-shot delimited continuations}
\label{sec:DELoneToAC:del}

We present the calculus $\del$, which incorporates a one-shot version of the
control operator \emph{shift0/dollar}.
The \emph{shift0/dollar} operator
\cite{materzok2012ADI,DBLP:conf/tlca/KiselyovS07} is a variant of
\emph{shift0/reset0}~\cite{danvy1989functional} and has the same expressive
power.
We adopt the calculus $\mathbf{DEL}$ defined by Forster et
al.~\cite{forster2019expressive}, restricting the use of each captured
continuation to one-shot.

\begin{figure}[tb]
  {\setlength{\parindent}{0cm}
  \setlength{\arraycolsep}{4pt}
  \setlength{\tabcolsep}{4pt}
  \renewcommand{\arraystretch}{0.9}
    \begin{tabular}[t]{rcl@{\hspace{0.1\linewidth}}l}
      \(V,W\) & \(::=\) & \(\ldots\) & \textsf{value} \\
      & \textbar & \(l\in\syntacticset{L}_{\mathbf{D}}\) & continuation label \\
      \(M,N\) & \(::=\) & \(\ldots\) & \textsf{computation} \\
      & \textbar & \(\shift{k}{M}\) & shift0 \\
      & \textbar & \(\dollar{M}{x}{N}\) & dollar \\
      & \textbar & \(\throw{V}{W}\) & throw
    \end{tabular}\\[1em]
    \begin{tabular}{rrcl}
      \textsf{pure frame} & $\mathcal{P}$ & ::= & \ldots \\
      \textsf{computational frame} & $\mathcal{F}$ & ::= & $\ldots \mid \dollar{\hole}{x}{N}$ \\
      \textsf{pure context} & $\mathcal{H}$ & ::= & \ldots \\
      \textsf{evaluation context} & $\mathcal{C}$ & ::= & \ldots \\
    \end{tabular}}
\caption{Syntax of $\del$}
\label{fig:del_syntax}
\end{figure}

Figure~\ref{fig:del_syntax} shows the syntax of $\del$ as an extension of
$\mam$, where the ellipses (\ldots) indicate the syntax from $\mam$.
A continuation label \(l\) is a value ranging over a countable set
$\syntacticset{L}_{\mathbf{D}}$.
It is used to refer to a captured continuation, as described below.
The computation \(\dollar{M}{x}{N}\) represents a \emph{generalized reset0},
called \emph{dollar}.
When evaluated, this term installs a delimiter for continuations captured in
$M$.
Unlike reset0, it has an additional part $x.N$ where $x$ is bound in $N$.
When $M$ evaluates to a value, its result is bound to $x$, and $N$ is evaluated.
The computation \(\shift{k}{M}\) represents the control operator \emph{shift0}.
When this term is evaluated, a continuation delimited by the nearest dollar term
is captured and assigned a fresh continuation label; $k$ is bound to this label,
and the body $M$ is evaluated.
If there is no surrounding dollar term, the evaluation gets stuck.
The computation \(\throw{V}{W}\) represents the invocation of the captured
continuation associated with \(V\) with the argument \(W\).
We define $\del$ programs as the set of computations containing no continuation
labels.

In $\del$, continuations can be invoked at most once.
Consider the example:
\[
  \dollar{
    \begin{aligned}
      \shift{k}{
      &\letin{a}{\throw{k}{1}}{\\&\letin{b}{\throw{k}{2}}{\return{\vpair{a}{b}}}}
      }
    \end{aligned}       
  }{x}{\return{x}}
\]
The term attempts to invoke the captured continuation $k$ twice, first to
compute $a$ and then to compute $b$, which violates the one-shotness constraint.
Since whether a continuation has already been invoked can only be determined at
runtime, we detect such violations dynamically, using a \textit{store} to record
the content and status of continuations.

A store is a partial function
$\theta :
\syntacticset{L}_{\mathbf{D}}\;\rightharpoonup\;\mathsf{computation}\sqcup\{\nil\}$.
$\dom{\theta}$ is the set of continuation labels $l$ such that $\theta(l)$ is
defined.
Note that $\theta(l) = \nil$ does \emph{not} mean that \(\theta(l)\) is
undefined.
Instead, $\theta(l)=\nil$ means that the continuation labeled $l$ is
\emph{invalid}---that is, this continuation has already been invoked.
If \(\theta(l) = M\) for some computation \(M\), we say that the continuation
labeled \(l\) is \emph{valid}.
The empty set $\emptyset$ represents the store whose domain is empty.
We write \(\theta[l := M]\) for the usual store update.

\begin{figure}[tb]
  \centering
  \begin{tabular}{cl}
    $(\mam)$ & $\inferrule{\redbetaM{M}{M'}}{\redbetaDp{M}{\theta}{M'}{\theta}}$ \\
    $(\mathrm{ret})$ & $\redbetaDp{\dollar{\return{V}}{x}{M}}{\theta}{M[V/x]}{\theta}$ \\
    $(\mathrm{shift})$ &
    $\inferrule
    {
      l \notin \dom{\theta}
    }
    { 
      \redbetaDp{\dollar{\plug{H}{\shift{k}{M}}}{x}{N}}{\theta}{M[l/k]}{\theta[l := \abs{y}{\dollar{\plug{H}{\return{y}}}{x}{N}}]}
    }$ \\
    $(\mathrm{throw})$ &
    $\inferrule
    {
      \theta(l) = \abs{y}{\dollar{\plug{H}{\return{y}}}{x}{N}}
    }
    {
      \redbetaDp{\throw{l}{V}}{\theta}{\dollar{\plug{H}{\return{V}}}{x}{N}}{\theta[l := \nil]}
    }$ \\
    $(\mathrm{fail})$ &
    $\inferrule
    {
      \theta(l) = \nil
    }
    {
      \config{\throw{l}{V}}{\theta} \rightarrow_{\mathbf{D}}^{\beta} \bot
    }$
  \end{tabular}
  \caption{Beta reduction rules of $\del$}
  \label{fig:del_beta}
\end{figure}

We introduce \textit{configurations} to describe the runtime states of programs.
A configuration $C$ is either a pair of a $\del$ computation $M$ and a store
$\theta$, written as \(\config{M}{\theta}\), or the error state $\bot$.
The beta reduction rules $\arrbetaD$ on configurations are defined in
Figure~\ref{fig:del_beta}.\footnote{Here \(\arrbetaM\) is understood as the
  relation on \(\del\) computations obtained by reading the beta-reduction rules
  of \(\mam\) in Figure~\ref{fig:mam_beta} as rules over \(\del\) computations.
  We adopt the same reading for the calculi introduced in later sections.}
In $(\mathrm{shift})$, we assume that $y$ does not freely occur in the context
$\context{H}$.
Throughout this paper, we assume similar freshness conditions for the variables
introduced in reduction rules and translations.

The reduction rules of $\del$ are defined as follows:
\[
  \inferrule{\redbetaDp{M}{\theta}{M'}{\theta'}}{\redDp{\plug{C}{M}}{\theta}{\plug{C}{M'}}{\theta'}}
  \hspace {4em}
  \inferrule{\redbetaD{\config{M}{\theta}}{\bot}}{\redD{\config{\plug{C}{M}}{\theta}}{\bot}}
\]
Evaluation of a $\del$ program is defined by: $\eval{\del}{M} \defeq V$ if
there exists a store $\theta$ such that
$\redDclos{\config{M}{\emptyset}}{\config{\return{V}}{\theta}}$.
$\Eval{\del}$ is a well-defined partial function, since $\arrD{}$ is
deterministic.
\footnote{Strictly speaking, there is an ambiguity in the choice of a fresh
  label, which implies the nondeterminism of the reduction.
  To avoid this, we fix a linear order on \(\syntacticset{L}_{\mathbf{D}}\) and
  take the least label not in \(\dom{\theta}\) as a fresh label.
  We adopt the same convention for the calculi introduced in later sections.
}

\begin{example}
  The evaluation of the program \(M \defeq \dollar{(\shift{k}{((\throw{k}{2}) + 3)})}{x}{x+4}\) proceeds as follows.
  \begin{align*}
    \config{M}{\emptyset}
    &\arrD \config{(\throw{l}{2}) + 3}{\; \emptyset[l := \abs{y}{\dollar{\return{y}}{x}{x+4}}]} \\
    &\arrD \config{\dollar{\return{2}}{x}{x+4} + 3}{\; \emptyset[l := \nil]} \\
    &\arrD \config{(2 + 4) + 3}{\; \emptyset[l := \nil]} \\
    &\arrD^{+} \config{9}{\; \emptyset[l := \nil]}.
  \end{align*}
  Thus, \(\eval{\del}{M} = 9\).
\end{example}

\subsection{The calculus for asymmetric coroutines}
\label{sec:DELoneToAC:AC}

A coroutine is a computation that can suspend its execution and later resume
from the point of suspension, in contrast to an ordinary subroutine.
De Moura and Ierusalimschy studied several variations of coroutines found in
programming languages and introduced two calculi, symmetric coroutines and
asymmetric coroutines~\cite{moura2009revisiting}.
We choose asymmetric coroutines because they are the most expressive among
several variations of coroutines and ``much simpler to manage and understand''
from the programmer's perspective~\cite{moura2009revisiting}.
They are also prevalent in modern programming languages, such as Lua and Ruby,
while symmetric coroutines are adopted in ``classical'' languages, such as
Simula and Modula-2.

\begin{figure}[tb]
  {\setlength{\parindent}{0cm}
  \setlength{\arraycolsep}{4pt}
  \setlength{\tabcolsep}{4pt}
  \renewcommand{\arraystretch}{0.9}
    \begin{tabular}[t]{rcl@{\hspace{0.1\linewidth}}l}
      $V,W$ & $::=$ & $\ldots$ & \textsf{value} \\
            & \textbar\ & $l\in\syntacticset{L}_{\ac}$ & coroutine label \\
      \(M,N\) & \(::=\) & \(\ldots\) & \textsf{computation} \\
            & \textbar\ & $\yield{V}$ & yield \\
            & \textbar\ & $\create{V}$ & create \\
            & \textbar\ & $\resume{V}{W}$ & resume \\
            & \textbar\ & $\labeledc{l}{M}$ & labeled computation
    \end{tabular}\\[1em]
    \begin{tabular}{rrcl}
      \textsf{pure frame} & $\mathcal{P}$ & ::= & \ldots \\
      \textsf{computational frame} & $\mathcal{F}$ & ::= & $\ldots \mid \labeledc{l}{\hole}$ \\
      \textsf{pure context} & $\mathcal{H}$ & ::= & \ldots \\
      \textsf{evaluation context} & $\mathcal{C}$ & ::= & \ldots \\
    \end{tabular}}

\caption{Syntax of $\ac$}
\label{fig:ac_syntax}
\end{figure}

We define the calculus $\ac$ for asymmetric coroutines, employing the basic
constructs of de Moura and Ierusalimschy's calculus, and adapting these
constructs to the call-by-push-value style.
Figure~\ref{fig:ac_syntax} presents the syntax of $\ac$.
Besides the constructs of $\mam$, $\ac$ has coroutine labels as values, and four
constructs to manipulate coroutines as computations: labeled computation,
\textbf{create}, \textbf{resume}, and \textbf{yield}.
$\syntacticset{L}_{\ac}$ is a countable set of coroutine labels.
The labeled computation $\labeledc{l}{M}$ represents the coroutine $l$ executing
the computation $M$.
The computation $\create{V}$ produces a new coroutine whose body is $V$,
creating a fresh label for it.
The computation $\resume{V}{W}$ invokes a coroutine labeled $V$, with a
parameter $W$.
The computation $\yield{V}$ suspends the current coroutine, yielding $V$ to its
caller.
We define $\ac$ programs as the set of computations containing no coroutine
labels.

Similarly to \(\del\), we use a store to track the content and status of
coroutines.
A store $\theta$ is a partial function that maps labels to \textit{values} or
\(\nil\).
Note the difference from the store in $\del$, which maps labels to
\textit{computations} or \(\nil\).
A configuration in $\ac$ is either a pair $\config{M}{\theta}$ of a computation
$M$ and a store $\theta$, or the error state $\bot$.

\begin{figure}[tb]
  {\renewcommand{\arraystretch}{1.3}
    \begin{tabular}{cl}
    $(\mam)$ & $\inferrule{\redbetaM{M}{M'}}{\redbetaACp{M}{\theta}{M'}{\theta}}$ \\

    $(\mathrm{create})$ &
    $\inferrule
    {
      l \notin \dom{\theta}
    }
    {
      \redbetaACp{\create{V}}{\theta}{\return{l}}{\theta[l := V]}
    }$ \\

    $(\mathrm{resume})$ &
    $\inferrule
    {
      l \in \dom{\theta} \\ \theta(l) \neq \nil
    }
    {
      \redbetaACp{\resume{l}{V}}{\theta}{\labeledc{l}{(\app{\force{\theta(l)}}{V})}}{\theta[l := \nil]}
    }$ \\

    $(\mathrm{fail})$ &
    $\inferrule
    {
      \theta(l) = \nil
    }
    {
      \redbetaAC{\config{\resume{l}{V}}{\theta}}{\bot}
    }$ \\

    $(\mathrm{ret})$ & $\redbetaACp{\labeledc{l}{\return{V}}}{\theta}{\return{V}}{\theta}$ \\
    $(\mathrm{yield})$ & $\redbetaACp{\labeledc{l}{\context{H}[\yield{V}]}}{\theta}{\return{V}}{\theta[l := \thunk{\abs{y}{\context{H}[\return{y}]}}]}$
  \end{tabular}}
  \[
  \inferrule{\redbetaACp{M}{\theta}{M'}{\theta'}}{\redACp{\plug{C}{M}}{\theta}{\plug{C}{M'}}{\theta'}} \hspace {4em} \inferrule{\redbetaAC{\config{M}{\theta}}{\bot}}{\redAC{\config{\plug{C}{M}}{\theta}}{\bot}}
  \]
  \caption{Semantics of $\ac$}
  \label{fig:ac_beta}
\end{figure}

Figure~\ref{fig:ac_beta} defines the operational semantics of $\ac$, which
includes the beta reduction rules $\arrAC^{\beta}$ and the reduction rules
$\arrAC$.
Note that coroutines are inherently one-shot: although a coroutine may be called
multiple times by suspending and resuming it, it cannot be duplicated or reused.

Evaluation of an $\ac$-program is defined by: $\eval{\ac}{M} \defeq V$ if
there exists a store $\theta$ such that
$\redACclos{\config{M}{\emptyset}}{\config{\return{V}}{\theta}}$.
Since $\arrAC{}$ is deterministic, $\Eval{\ac}$ is a well-defined partial
function.

\begin{example}
  The evaluation of the program \(M\), defined by
  \[
    M \defeq
    \begin{pmatrix*}[l]
      \letin{f}{\create{\thunk{
      \begin{array}{@{}l@{}}
        \lambda i.\;\letin{j}{\yield{i}}
        {\\\phantom{\lambda i.\;}{j + 1}}
        \end{array}}}}
      {\\\letin{a}{\resume{f}{2}}
      {\\\letin{b}{\resume{f}{3}}
      {\\\return{\vpair{a}{b}}}}}
    \end{pmatrix*},
  \]
  proceeds as follows.
  \begin{align*}
    \config{M}{\emptyset}
    &\arrAC \config{\begin{pmatrix*}[l]
      \letin{f}{\return{l}}
      {\\\letin{a}{\resume{f}{2}}
      {\\\letin{b}{\resume{f}{3}}
      {\\\return{\vpair{a}{b}}}}}
    \end{pmatrix*}}{\; \emptyset\left[l := \thunk{\begin{array}{@{}l@{}}
        \lambda i.\;\letin{j}{\yield{i}}
        {\\\phantom{\lambda i.\;}{j + 1}}
        \end{array}}\right]} \\
    &\arrAC \config{\begin{pmatrix*}[l]
      \letin{a}{\resume{l}{2}}
      {\\\letin{b}{\resume{l}{3}}
      {\\\return{\vpair{a}{b}}}}
    \end{pmatrix*}}{\; \emptyset\left[l := \thunk{\begin{array}{@{}l@{}}
        \lambda i.\;\letin{j}{\yield{i}}
        {\\\phantom{\lambda i.\;}{j + 1}}
        \end{array}}\right]} \\
    &\arrAC^+ \config{\begin{pmatrix*}[l]
      \letin{a}{\labeledc{l}{\app{\begin{pmatrix*}[l]
        \lambda i.\;\letin{j}{\yield{i}}
        {\\\phantom{\lambda i.\;}{j + 1}}
        \end{pmatrix*}}{2}}}
      {\\\letin{b}{\resume{l}{3}}
      {\\\return{\vpair{a}{b}}}}
    \end{pmatrix*}}{\; \emptyset\left[l := \nil\right]} \\
    &\arrAC \config{\begin{pmatrix*}[l]
      \letin{a}{\labeledc{l}{\begin{pmatrix*}[l]
        \letin{j}{\yield{2}}
        {\\{j + 1}}
        \end{pmatrix*}}}
      {\\\letin{b}{\resume{l}{3}}
      {\\\return{\vpair{a}{b}}}}
    \end{pmatrix*}}{\; \emptyset\left[l := \nil\right]} \\
    &\arrAC \config{\begin{pmatrix*}[l]
      \letin{a}{\return{2}}
      {\\\letin{b}{\resume{l}{3}}
      {\\\return{\vpair{a}{b}}}}
    \end{pmatrix*}}{\; \emptyset\left[l := \thunk{\begin{array}{@{}l@{}}
        \lambda y.\;\letin{j}{\return{y}}
        {\\\phantom{\lambda y.\;}{j + 1}}
        \end{array}}\right]} \\
    &\arrAC \config{\begin{pmatrix*}[l]
      \letin{b}{\resume{l}{3}}
      {\\\return{\vpair{2}{b}}}
    \end{pmatrix*}}{\; \emptyset\left[l := \thunk{\begin{array}{@{}l@{}}
        \lambda y.\;\letin{j}{\return{y}}
        {\\\phantom{\lambda y.\;}{j + 1}}
        \end{array}}\right]} \\
    &\arrAC^+ \config{\begin{pmatrix*}[l]
      \letin{b}{\labeledc{l}{\app{\begin{pmatrix*}[l]
        \lambda y.\;\letin{j}{\return{y}}
        {\\\phantom{\lambda y.\;}{j + 1}}
        \end{pmatrix*}}{3}}}
      {\\\return{\vpair{2}{b}}}
    \end{pmatrix*}}{\; \emptyset\left[l := \nil\right]} \\
    &\arrAC \config{\begin{pmatrix*}[l]
      \letin{b}{\labeledc{l}{\begin{pmatrix*}[l]
        \letin{j}{\return{3}}
        {\\ j + 1}
      \end{pmatrix*}}}
      {\\\return{\vpair{2}{b}}}
    \end{pmatrix*}}{\; \emptyset\left[l := \nil\right]} \\
    &\arrAC^+ \config{\begin{pmatrix*}[l]
      \letin{b}{\labeledc{l}{\return{4}}}
      {\\\return{\vpair{2}{b}}}
    \end{pmatrix*}}{\; \emptyset\left[l := \nil\right]} \\
    &\arrAC^+ \config{\return{\vpair{2}{4}}}{\; \emptyset\left[l := \nil\right]}
  \end{align*}
  Thus, \(\eval{\ac}{M} = \vpair{2}{4}\).
\end{example}

\subsection{Na\"{i}ve translation and its failure}
\label{sec:DELoneToAC:naive-translation}

\begin{figure}[tb]
  \centering
    $\begin{array}{rll}
    \mt{\dollar{M}{x}{N}} & \defeq &
                                       \left(\begin{array}{l}
                                         \letin{z}{\create{\thunk{\abs{\underscore}{\letin{x}{\mt{M}}{\return{\thunk{\abs{\underscore}{\mt{N}}}}}}}}}
                                         {\\\letin{res}{\resume{z}{\unit}}
                                         {\\\app{\force{\var{res}}}{z}}}
                                       \end{array}\right) \\
    \mt{\shift{k}{L}} & \defeq & \yield{\thunk{\abs{k}{\mt{L}}}} \\
    \mt{\throw{V}{W}} & \defeq &
                                   \left(\begin{array}{l}
                                     \letin{res}{\resume{\mt{V}}{\mt{W}}}
                                     {\\\app{\force{\var{res}}}{\mt{V}}}
                                   \end{array}\right)
  \end{array}$
  \caption{Na\"{i}ve translation from $\del$ to $\ac$}
  \label{fig:deltoac:naive}
\end{figure}

Given the similarity between delimited continuations and coroutines, one may think that it is straightforward to macro-translate $\del$ to $\ac$ with the following correspondence: a dollar term in $\del$ is translated to a create term in $\ac$, a shift0 term to a yield term, and a throw term to a resume term.
Based on this intuition, we can define the na\"{i}ve translation from $\del$ to $\ac$ in Figure~\ref{fig:deltoac:naive}.
Note that on the right-hand side of $\mt{\dollar{M}{x}{N}}$, $x$ corresponds to the binding occurrence of $x$ in the original term $\dollar{M}{x}{N}$.
Figure~\ref{fig:deltoac:naive} only shows nontrivial cases.
Other cases are translated homomorphically except for continuation labels, whose translations are not defined.
This is not problematic, since a macro-translation only has to translate all $\del$ programs, which do not contain continuation labels.

The na\"{i}ve translation works for simple expressions.
However, it fails to preserve semantics, as shown below.
Consider the following $\del$ program $M$:
\begin{align*}
  M &\defeq
  \left\langle M' \;\middle\vert\;\var{i}.\,\return{i} \right\rangle \\
  M' &\defeq \begin{pmatrix*}[l]
      \mathbf{let}\;\var{j} = (\shift{k_{\mathrm{1}}}{\mathbf{let}\;\var{r_{\mathrm{1}}} = \throw{k_1}{10}\;\mathbf{in}} \\
      \phantom{\mathbf{let}\;\var{j} = (\mathbf{S_0} \var{k_{\mathrm{1}}}.\;}\mathbf{let}\;\mathrlap{\var{r_{\mathrm{2}}}}\phantom{\var{r_{\mathrm{1}}}} = \throw{k_1}{20}\;\mathbf{in}\;\return{\var{r_{\mathrm{1}}}})\;\mathbf{in} \\
      \shift{k_{\mathrm{2}}}{\return{30}}
    \end{pmatrix*}  
\end{align*}
We shall explain how the evaluation of \(M\) proceeds in \(\del\).
(For brevity, the full content of the store is omitted from the evaluation
trace.
Only the changes relevant to the main computation are mentioned.)
First, the shift0 term $\shift{k_{\mathrm{1}}}{\cdots}$ is invoked, and the pure
context surrounding it is captured and stored under a fresh label $l_1$:
\begin{flalign}
  M &\arrD
      \begin{pmatrix*}[l]
        \mathbf{let}\;\var{r_{\mathrm{1}}} = \throw{l_1}{10}\;\mathbf{in}\;
        \\\mathbf{let}\;\mathrlap{\var{r_{\mathrm{2}}}}\phantom{\var{r_{\mathrm{1}}}} = \throw{l_1}{20}\;\mathbf{in}\;\return{\var{r_{\mathrm{1}}}}
      \end{pmatrix*}\label{eq:evalM1}  \\
  &\left[\text{where \(l_1 \mapsto \abs{y}{\dollar{\letin{j}{\return{y}}{\shift{k_{\mathrm{2}}}{\return{30}}}}{i}{\return{i}}}\)} \right] \notag \\
\shortintertext{Next, $\throw{l_1}{10}$ invokes the captured continuation labeled $l_1$,
invalidating the continuation:}
    \cdots &\arrD^+
    \paren{\begin{array}{@{}l@{}}
    \mathbf{let}\;\var{r_{\mathrm{1}}} = \left\langle
      \begin{array}{@{}l@{}}
      \mathbf{let}\;\var{j} = \return{10}\;\mathbf{in}\;\\
      \shift{k_{\mathrm{2}}}{\return{30}}
      \end{array}
    \;\middle\vert\;\var{i}.\,\return{i} \right\rangle\;\mathbf{in}\\
    \mathbf{let}\;\mathrlap{\var{r_{\mathrm{2}}}}\phantom{\var{r_{\mathrm{1}}}} = \throw{l_1}{20}\;\mathbf{in}\;\return{\var{r_{\mathrm{1}}}}      
    \end{array}}\label{eq:evalM2} \\
  &\left[\text{where \(l_1 \mapsto \nil\)}\right] \\
  &\arrD\paren{\begin{array}{@{}l@{}}
    \mathbf{let}\;\var{r_{\mathrm{1}}} = \left\langle
      \shift{k_{\mathrm{2}}}{\return{30}}
    \;\middle\vert\;\var{i}.\,\return{i} \right\rangle\;\mathbf{in}\\
    \mathbf{let}\;\mathrlap{\var{r_{\mathrm{2}}}}\phantom{\var{r_{\mathrm{1}}}} = \throw{l_1}{20}\;\mathbf{in}\;\return{\var{r_{\mathrm{1}}}}      
    \end{array}}\label{eq:evalM3} \\
\shortintertext{Then, the shift0 term $\shift{k_{\mathrm{2}}}{\cdots}$ is invoked.
This captures the context and stores it under another fresh label $l_2$.
The evaluation then continues as follows (note that the continuation labeled $l_2$ is not used):}
    \cdots
     &\arrD^+\begin{pmatrix*}[l]
      \mathbf{let}\;\var{r_{\mathrm{1}}} = \return{30}\;\mathbf{in}\;\\
      \mathbf{let}\;\mathrlap{\var{r_{\mathrm{2}}}}\phantom{\var{r_{\mathrm{1}}}} = \throw{l_1}{20}\;\mathbf{in}\;\return{\var{r_{\mathrm{1}}}}
    \end{pmatrix*}\label{eq:evalM4} \\
     &\left[\text{where \(l_2 \mapsto \abs{y}{\dollar{\return{y}}{i}{\return{i}}}\)}\right] \\
    &\arrD\begin{array}{@{}l@{}}
      \mathbf{let}\;\mathrlap{\var{r_{\mathrm{2}}}}\phantom{\var{r_{\mathrm{1}}}} = \throw{l_1}{20}\;\mathbf{in}\;\return{30}      
    \end{array}\label{eq:evalM5}
\end{flalign}
Finally, $\throw{l_1}{20}$ is invoked, but since the continuation labeled $l_1$ has already been consumed, the evaluation fails.

Therefore, in $\ac$, the evaluation of $\mt{M}$ must not successfully terminate
since macro-translations preserve semantics; however, this is not the case.
To understand the reason, consider the evaluation trace of \(\mt{M}\).
First, after applying the translation, a coroutine corresponding to the body of
the source dollar term is created, labeled as \(m\), and invoked
(\(\context{P}\) is a pure context defined as
\(\seq{i}{\hole{}}{\return{\thunk{\abs{\underscore}{\return{i}}}}}\)):
\begin{flalign}
  \mt{M}
    &\equiv
    \begin{pmatrix*}[l]
      \letin{z}{\create{\thunk{\abs{\underscore}{\plug{P}{\mt{M'}}}}}}
      {\\\letin{res}{\resume{z}{\unit}}
      {\\\app{\force{\var{res}}}{\var{z}}}}
    \end{pmatrix*} \nonumber
  \\
  &\arrAC^{+} 
    \begin{pmatrix*}[l]
      \letin{res}{\\\quad\labeledc{m}{\plug{P}{\begin{array}{@{}l@{}}
        \mathbf{let}\;\var{j} = \\
        \quad \yield{\thunk{\abs{k_{\mathrm{1}}}{
        \begin{pmatrix*}[l]
          \seq{r_{\mathrm{1}}}{\app{\app{\mt{\mathbf{throw}}}{k_1}}{10}}
          {\\\seq{r_{\mathrm{2}}}{\app{\app{\mt{\mathbf{throw}}}{k_1}}{20}}
          {\\\return{r_{\mathrm{1}}}}}
        \end{pmatrix*}}}}
        \;\mathbf{in} \\
        \yield{\thunk{\abs{k_{\mathrm{2}}}{\return{30}}}}
      \end{array}}}}
      {\\\app{\force{\var{res}}}{\var{m}}}
    \end{pmatrix*}
  \\
    &\left[\text{where \(m \mapsto \nil\)}\right] \\
  \shortintertext{%
  Then, the current coroutine \(m\) is suspended, yielding the thunk
  \(\thunk{\abs{k_{\mathrm{1}}}{\cdots}}\).%
  }
  \cdots &\arrAC^{+} 
      \begin{pmatrix*}[l]
        \letin{res}{\return{\thunk{\abs{k_{\mathrm{1}}}{
          \begin{pmatrix*}[l]
            \seq{r_{\mathrm{1}}}{\app{\app{\mt{\mathbf{throw}}}{k_1}}{10}}
            {\\\seq{r_{\mathrm{2}}}{\app{\app{\mt{\mathbf{throw}}}{k_1}}{20}}
            {\\\return{r_{\mathrm{1}}}}}
          \end{pmatrix*}}}}}
        {\\\app{\force{\var{res}}}{\var{m}}}
      \end{pmatrix*}
  \\
  &\left[\begin{array}{@{}l@{}}
                              \text{where} \\
                              m\mapsto\thunk{\begin{array}{@{}l@{}}
                                \abs{y}{\plug{P}{
                                \begin{array}{@{}l@{}}
                                  \seq{j}{\return{y}}
                                  {\\\yield{\thunk{\abs{k_{\mathrm{2}}}{\return{30}}}}}
                                \end{array}
                                }}
                              \end{array}} \\
                            \end{array}\right] \\
 &\arrAC^{+}  \paren{\begin{array}{@{}l@{}}
      \seq{r_{\mathrm{1}}}{\app{\app{\mt{\mathbf{throw}}}{m}}{10}}
      {\\\seq{r_{\mathrm{2}}}{\app{\app{\mt{\mathbf{throw}}}{m}}{20}}
                            {\\\return{r_{\mathrm{1}}}}}
 \end{array}} \nonumber \\
  \shortintertext{%
  This term corresponds to \eqref{eq:evalM1}, where the continuation captured by
  $\shift{k_{\mathrm{1}}}{\cdots}$ in $M$, labeled \(l_1\), corresponds to the
  coroutine labeled $m$.
  Next, $\app{\app{\mt{\mathbf{throw}}}{m}}{{10}}$ is evaluated.
  This resumes the coroutine labeled $m$, mapping \(m\) to
  $\nil$ in the store, and the evaluation reaches the following term corresponding
  to \(\eqref{eq:evalM2}\):}
  \cdots  & \arrAC^{+} \paren{\begin{array}{@{}l@{}}
    \seq{r_{\mathrm{1}}}{\paren{
      \begin{array}{@{}l@{}}
        \seq{res}{\labeledc{m}{\paren{
        \plug{P}{\begin{array}{@{}l@{}}
          \seq{j}{\return{10}}
          {\\\yield{\thunk{\abs{k_{\mathrm{2}}}{\return{30}}}}}
        \end{array}}
        }} }
        {\\\app{\force{\var{res}}}{m}}
      \end{array}
      }}
      {\\\seq{r_{\mathrm{2}}}{\app{\app{\mt{\mathbf{throw}}}{m}}{20}}
      {\\\return{r_{\mathrm{1}}}}}
  \end{array}} \label{eq:evalmtM1}\\
    &\left[\text{where \(m \mapsto \nil\)}\right] \\  
  \shortintertext{%
  Then, the value \(10\) is discarded, and the term
  $\yield{\thunk{\abs{k_{\mathrm{2}}}{\return{30}}}} (=
  \mt{\shift{k_{\mathrm{2}}}{\return{30}}})$ is invoked:
  }
  \cdots &\arrAC^{+} \paren{\begin{array}{@{}l@{}}
    \seq{r_{\mathrm{1}}}{\app{\force{\thunk{\abs{k_{\mathrm{2}}}{\return{30}}}}}{m}}
    {\\\seq{r_{\mathrm{2}}}{\app{\app{\mt{\mathbf{throw}}}{m}}{20}}
    {\\\return{r_{\mathrm{1}}}}}
                            \end{array}}\;
                            \left[\begin{array}{@{}l@{}}
                              \text{where} \\
                              m \mapsto \thunk{\begin{array}{@{}l@{}}
                                \abs{y}{\plug{P}{\return{y}}}
                              \end{array}} \\
                            \end{array}\right] \label{eq:evalmtM2} \\
  \cdots &\arrAC^{+} \paren{\begin{array}{@{}l@{}}
    \seq{r_{\mathrm{1}}}{\return{30}}
    {\\\seq{r_{\mathrm{2}}}{\app{\app{\mt{\mathbf{throw}}}{m}}{20}}
    {\\\return{r_{\mathrm{1}}}}}
  \end{array}} \\
  \shortintertext{%
  In this term, which corresponds to \eqref{eq:evalM4}, the label \(m\)
  substituted for \(k_{\mathrm{2}}\) is discarded without being consumed.
  Thus, the invocation of \(\mt{\mathbf{throw}}\,m\,20\) (corresponding to that
  of \(\throw{l_1}{20}\) in \eqref{eq:evalM5}) succeeds, although it
  should not, and the evaluation terminates successfully:
  }
  \cdots &\arrAC^{+} \paren{\begin{array}{@{}l@{}}
    \seq{r_{\mathrm{2}}}{\app{\app{\mt{\mathbf{throw}}}{m}}{20}}
    {\\\return{30}}
                            \end{array}}\,
                            \left[\begin{array}{@{}l@{}}
                              \text{where} \\
                              m \mapsto \thunk{\begin{array}{@{}l@{}}
                                \abs{y}{\plug{P}{\return{y}}}
                              \end{array}} \\
                            \end{array}\right] \label{eq:evalmtM3} \\
    &\arrAC^{+}
      \paren{\begin{array}{@{}l@{}}
      \seq{r_{\mathrm{2}}}{\begin{pmatrix*}[l]
        \seq{res}{\labeledc{m}{\paren{
        \seq{i}{\return{20}}{\return{\thunk{\abs{\underscore}{\return{i}}}}}
        }}}
        {\\\app{\force{\var{res}}}{m}}
      \end{pmatrix*}}
                            {\\\return{30}}
      \end{array}}\\
    &\left[\text{where \(m \mapsto \nil\)}\right] \\
    &\arrAC^{+} \paren{\begin{array}{@{}l@{}}
      \seq{r_{\mathrm{2}}}{\app{\force{\thunk{\abs{\underscore}{\return{20}}}}}{\var{m}}}
      {\\\return{30}}
    \end{array}} \\
  &\arrAC^{+} \return{30}
\end{flalign}

In short, the na\"{i}ve translation fails because it ignores the difference
between \(\del\) and \(\ac\) in how the validity of labels is managed.
In \(\del\), validity is \emph{one-way}: once a shift0 invocation creates a
fresh label \(l\), this label remains valid---i.e., \(\theta(l)\) holds the
captured continuation---until the continuation is invoked.
After that, \(\theta(l)\) becomes \(\nil\) permanently.
For example, once the continuation labeled \(l_1\) is invoked in
\eqref{eq:evalM2}, it is kept invalid throughout the evaluation.
In \(\ac\), by contrast, validity is \emph{alternating}: each time coroutine
\(m\) is suspended and later resumed, \(\theta(m)\) alternates between a value
(the coroutine \(m\) is suspended, hence valid) and \(\nil\) (the coroutine is
running, hence invalid).

Consequently, since both \(l_1\) and \(l_2\) are mapped to \(m\), the generation
of \(l_2\) corresponds to the suspension of \(m\), and the second invocation of
\(l_1\) corresponds to its resumption.
This resumption is legitimate in \(\ac\), as the coroutine \(m\) is suspended.
Interpreted back in \(\del\)'s terms, however, this means that \(l_1\) is
illegitimately \emph{reactivated} by the creation of \(l_2\).
This is the reason why the na\"{i}ve translation fails to preserve the
semantics.

We remark that it would be harmless if the captured continuation labeled \(l_2\)
were consumed (or, equivalently, the thunk returned by \textbf{yield} in
\eqref{eq:evalmtM1} eventually resumed $m$), as that would invalidate $m$ again.
In our counterexample, however, this is not the case.

This phenomenon had been overlooked for years.
In fact, it was folklore that a simple macro-translation from $\del$ to $\ac$
should exist.
For example, Kawahara and Kameyama proposed such a translation from one-shot
effect handlers to asymmetric coroutines, which has been implemented in Lua, Go,
and several other languages~\cite{kawahara2020one}.
Our analysis, however, reveals that these simple translations fail to preserve
semantics and therefore cannot be macro-translations.
The problem becomes apparent when we examine a dollar term that contains more
than one occurrence of shift0.\footnote{It can be shown that Kawahara and
  Kameyama's translation fails to preserve semantics by considering a similar
  counterexample.}

\subsection{Refined translation from \texorpdfstring{$\del$}{DELone} to
  \texorpdfstring{$\ac$}{AC}}
\label{sec:DELoneToAC:refined-translation}

Behind the failure analyzed in the previous section lie two gaps to be
addressed:
\begin{enumerate}
\item Each shift0 invocation creates a fresh continuation label, whereas a
  coroutine label, once created, is reused across resumptions.
  A single coroutine label therefore has to serve all the continuations
  captured within one dollar term.
\item The validity of a continuation label is one-way, whereas that of a
  coroutine label is alternating.
\end{enumerate}
The na\"{i}ve translation addresses neither.

Our key idea to fill the gaps is to map each continuation label not directly to
a coroutine label, but to a \emph{reference cell} that holds both the
continuation label's validity and a coroutine label.
Specifically, a reference cell holds either \(\valid{m}\), meaning that the
continuation is still valid and is realized by the coroutine labeled \(m\), or
\(\invalid\), meaning that it has already been invoked.
A continuation is then represented by the label of a cell, while the
continuation's validity is recorded in the content of the cell.
Thus, invoking a continuation through one copy of its label invalidates all the
other copies as well.

Reference cells are not primitive in \(\ac\), but they can be implemented using
coroutines.
See the following encoding (this encoding can be extended into a
macro-translation from the calculus for reference cells to \(\ac\); see
Theorem~\ref{thm:reftoac}):
\begin{align*}
  \var{ref} &\defeq 
  \begin{array}{l}
    \thunk{\lambda \var{v}.\;\create{\refcell{\var{v}}}}
  \end{array} \\
  \refcell{v} &\defeq
  \thunk{
  \begin{array}{@{}l@{}}
    \abs{y}{\letin{q'}{\return{\var{y}}}
    {\\\phantom{\lambda \var{y}.\;}\app{\force{\var{loop}}}{\app{\var{v}}{\var{q'}}}}}
  \end{array}} \\
  \var{loop} &\defeq \thunk{\app{\force{\thunk{\abs{x}{\app{\force{\var{th}}}{\thunk{\app{\force{\var{x}}}{\var{x}}}}}}}}{\thunk{\abs{x}{\app{\force{\var{th}}}{\thunk{\app{\force{\var{x}}}{\var{x}}}}}}}} \\
  \var{th} &\defeq
  \left\{\begin{array}{@{}l@{}}
    \lambda \var{f}. \lambda \var{s}. \lambda \var{q}.\;\mathbf{case}\;\var{q}\;\mathbf{of}\;\{\\
    \quad(\inj{Set}{\var{v}}) \mapsto \letin{q'}{\yield{\unit}}{\app{\app{\force{\var{f}}}{\var{v}}}{\var{q'}}} \\
    \quad(\inj{Get}{\underscore}) \mapsto \letin{q'}{\yield{\var{s}}}{\app{\app{\force{\var{f}}}{\var{s}}}{\var{q'}}}\}
  \end{array}\right\} \\
    \var{get} &\defeq \thunk{\abs{c}{\resume{\var{c}}{(\inj{Get}{\unit})}}} \\
    \var{set} &\defeq \thunk{\abs{c}{\abs{v}{\resume{\var{c}}{(\inj{Set}{\var{v}})}}}}
\end{align*}
Given a value \(v\), \(\var{ref}\) creates a coroutine \(\refcell{v}\) that
holds \(v\) as its content and runs an infinite loop waiting for a query \(q\).
When it is passed a query of the form \(\inj{Get}{\unit}\), it returns the
current state, continuing the loop; when passed a query of the form
\(\inj{Set}{v}\), it updates the state to \(v\), returning the unit value
\(\unit\).

\begin{figure}[tb]
  \noindent
  \begin{minipage}[t]{0.50\linewidth}
      $\begin{array}{l}
        \mt{\shift{k}{M}} \defeq \yield{\thunk{\abs{x}{\letin{k}{\force{x}}{\mt{M}}}}} \\
        \mt{\dollar{M}{x}{N}} \defeq \\
        \left(\begin{array}{l}
          \letin{z}{\mathbf{create}\\\;\;\thunk{
          \begin{array}{@{}l@{}}
            \lambda \var{\underscore}.\;\mathbf{let}\;\var{x} = \mt{M}\;\mathbf{in} \\
            \phantom{\lambda \var{\underscore}.\;}\return{\thunk{\abs{\underscore}{\mt{N}}}}
          \end{array}
          }}
          {\\\letin{zc}{\return{\thunk{\app{\force{\var{ref}}}{(\inj{Valid}{z})}}}}
          {\\\letin{res}{\resume{z}{\unit}}
          {\\\app{\force{\var{res}}}{\var{zc}}}}}
        \end{array}\right) \\
        \mt{\throw{V}{W}} \defeq \\
        \left(\begin{array}{l}
          \mathbf{let}\;\var{zc} = \app{\force{\var{get}}}{\mt{V}}\;\mathbf{in} \\
          \mathbf{case}\;\var{zc}\;\mathbf{of}\;\{\\
          \;\;\paren{\inj{Valid}{z}} \mapsto\\
          \;\;\;\;\mathbf{let}\;\underscore = \app{\app{\force{\var{set}}}{\mt{V}}}{(\inj{Invalid}{\unit})}\;\mathbf{in}\\
          \;\;\;\;\mathbf{let}\;\var{zc'} = {\return{\thunk{\app{\force{\var{ref}}}{(\inj{Valid}{z})}}}}\,\mathbf{in}\\
          \;\;\;\;\mathbf{let}\;\var{res} = \app{\app{\mathbf{resume}}{\var{z}}}{\mt{W}}\;\mathbf{in}\\
          \;\;\;\;\app{\force{\var{res}}}{zc'} \\
          \;\;\paren{\inj{Invalid}{\underscore}} \mapsto \force{\var{fail}}\\
          \} \\
        \end{array} \right) \\
      \end{array}$
    \end{minipage}
  \begin{minipage}[t]{0.46\linewidth}
   \noindent$\begin{array}{l}
        \text{where} \\
        \var{fail} \defeq \\
        \left\{ \begin{array}{l}
          \letin{z}{\create{\thunk{\abs{\underscore}{\return{\unit}}}}}
          {\\\letin{\underscore}{\resume{z}{\unit}}
          {\\\resume{z}{\unit}}}
        \end{array} \right\}\\
       \var{ref} \defeq 
                               \begin{array}{l}
                                 \thunk{\lambda \var{v}.\;\create{\refcell{\var{v}}}}
                               \end{array}\\
        \refcell{v} \defeq \\
        \thunk{
        \begin{array}{@{}l@{}}
          \abs{y}{\letin{q'}{\return{\var{y}}}
          {\\\phantom{\lambda \var{y}.\;}\app{\force{\var{loop}}}{\app{\var{v}}{\var{q'}}}}}
        \end{array}} \\
     \var{loop} \defeq \thunk{\app{\force{\thunk{\abs{x}{\app{\force{\var{th}}}{\thunk{\app{\force{\var{x}}}{\var{x}}}}}}}}{\thunk{\abs{x}{\app{\force{\var{th}}}{\thunk{\app{\force{\var{x}}}{\var{x}}}}}}}} \\
        \var{th} \defeq \\
                              \left\{\begin{array}{@{}l@{}}
                                \lambda \var{f}. \lambda \var{s}. \lambda \var{q}.\;\mathbf{case}\;\var{q}\;\mathbf{of}\;\{\\
                                \quad(\inj{Set}{\var{v}}) \mapsto \\\qquad\letin{q'}{\yield{\unit}}{\app{\app{\force{\var{f}}}{\var{v}}}{\var{q'}}} \\
                                \quad(\inj{Get}{\underscore}) \mapsto \\\qquad\letin{q'}{\yield{\var{s}}}{\app{\app{\force{\var{f}}}{\var{s}}}{\var{q'}}}\}
                              \end{array}\right\}\\
        \var{get} \defeq \thunk{\abs{c}{\resume{\var{c}}{(\inj{Get}{\unit})}}} \\
        \var{set} \defeq \thunk{\abs{c}{\abs{v}{\resume{\var{c}}{(\inj{Set}{\var{v}})}}}}
      \end{array}$
  \end{minipage}
  \caption{Refined translation from $\del$ to $\ac$}
  \label{fig:deltoac}
\end{figure}

Using reference cells, we give the refined translation in
Figure~\ref{fig:deltoac}.
In the translation of a dollar term, we introduce a thunk \(zc\) that creates a
reference cell holding the coroutine \(z\) tagged with \(\mathrm{Valid}\) and
pass it in place of the coroutine in the na\"{i}ve translation.
This thunk is forced in the body of the translation of a shift0 term to obtain
the reference cell.
In the translation of a throw term \(\mt{\throw{V}{W}}\), we pattern-match on
the content of the reference cell labeled as \(\mt{V}\).
If its content is \(\valid{m}\), then we update the content of the reference
cell to \(\invalid\), which makes every later use of \(\mt{V}\) fail; we then
resume the coroutine labeled \(m\) with \(\mt{W}\), passing it a thunk \(zc'\)
built in the same way as \(zc\).
Otherwise, we let the computation fail.

We address the first gap by building \(zc\) in the dollar term and \(zc'\) in
the throw term, each of which creates a new reference cell when forced.
Thus, the continuations captured by successive shift0 invocations are
represented by distinct reference cells, even though they all refer to the same
coroutine \(z\).
We address the second gap by never turning the content of a reference cell from
\(\invalid\) back to \(\valid{m}\).
Thus, once a reference cell drops the \(\mathrm{Valid}\) tag, it never regains
it, whereas the coroutine \(m\) keeps alternating between valid and invalid.

Now, the evaluation of \(M\) in Section~\ref{sec:DELoneToAC:naive-translation}
under the refined translation proceeds as follows:
\begin{align}
  \mt{M}
  &\equiv
    \begin{pmatrix*}[l]
      \letin{z}{\create{\thunk{\abs{\underscore}{\plug{P}{\mt{M'}}}}}}
      {\\\letin{zc}{\return{\thunk{\app{\force{\var{ref}}}{(\inj{Valid}{z})}}}}
      {\\\letin{res}{\resume{z}{\unit}}
      {\\\app{\force{\var{res}}}{\var{zc}}}}}
    \end{pmatrix*} \nonumber
  \\
  &\arrAC^{+} 
    \begin{pmatrix*}[l]
      \letin{res}{\\\quad\labeledc{m}{\plug{P}{\begin{array}{@{}l@{}}
        \mathbf{let}\;\var{j} = \\
        \quad \yield{\thunk{\abs{x_{\mathrm{1}}}{
        \begin{pmatrix*}[l]
          \seq{k_{\mathrm{1}}}{\force{\var{x_{\mathrm{1}}}}}
          {\\\seq{r_{\mathrm{1}}}{\app{\app{\mt{\mathbf{throw}}}{k_1}}{10}}
          {\\\seq{r_{\mathrm{2}}}{\app{\app{\mt{\mathbf{throw}}}{k_1}}{20}}
          {\\\return{r_{\mathrm{1}}}}}}
        \end{pmatrix*}}}}
        \;\mathbf{in} \\
        \yield{\thunk{\abs{x_{\mathrm{2}}}{\seq{k_{\mathrm{2}}}{\force{\var{x_{\mathrm{2}}}}}{\return{30}}}}}
      \end{array}}}}
      {\\\app{\force{\var{res}}}{\thunk{\app{\force{\var{ref}}}{\inj{Valid}{m}}}}}
    \end{pmatrix*}
  \\
    &\left[\text{where \(m \mapsto \nil\)}\right] \\
    &\arrAC^{+}  \app{\force{\thunk{\abs{x_{\mathrm{1}}}{
            \begin{pmatrix*}[l]
              \seq{k_{\mathrm{1}}}{\force{\var{x_{\mathrm{1}}}}}
          {\\\seq{r_{\mathrm{1}}}{\app{\app{\mt{\mathbf{throw}}}{k_1}}{10}}
          {\\\seq{r_{\mathrm{2}}}{\app{\app{\mt{\mathbf{throw}}}{k_1}}{20}}
          {\\\return{r_{\mathrm{1}}}}}}
            \end{pmatrix*}}}}}{\thunk{\app{\force{\var{ref}}}{\inj{Valid}{m}}}} \\
  &\left[\begin{array}{@{}l@{}}
                              \text{where} \\
                              m\mapsto\thunk{\begin{array}{@{}l@{}}
                                \abs{y}{\plug{P}{
                                \begin{array}{@{}l@{}}
                                  \seq{j}{\return{y}}
                                  {\\\yield{\thunk{\abs{x_{\mathrm{2}}}{\seq{k_{\mathrm{2}}}{\force{\var{x_{\mathrm{2}}}}}{\return{30}}}}}}
                                \end{array}
                                }}
                              \end{array}} \\
                            \end{array}\right] \\
 &\arrAC^{+}  \paren{\begin{array}{@{}l@{}}
      \seq{r_{\mathrm{1}}}{\app{\app{\mt{\mathbf{throw}}}{m_1}}{10}}
      {\\\seq{r_{\mathrm{2}}}{\app{\app{\mt{\mathbf{throw}}}{m_1}}{20}}
                            {\\\return{r_{\mathrm{1}}}}}
                            \end{array}} \label{eq:evalrmtM1} \\
  &\left[\text{where \(m_1 \mapsto \refcell{\inj{Valid}{m}}\)}\right] \\
    & \arrAC^{+} \paren{\begin{array}{@{}l@{}}
      \seq{r_{\mathrm{1}}}{\paren{
      \begin{array}{@{}l@{}}
        \seq{res}{\labeledc{m}{\paren{
        \plug{P}{\begin{array}{@{}l@{}}
          \seq{j}{\return{10}}
          {\\\yield{\thunk{\abs{x_{\mathrm{2}}}{\seq{k_{\mathrm{2}}}{\force{\var{x_{\mathrm{2}}}}}{\return{30}}}}}}
        \end{array}}
        }} }
        {\\\app{\force{\var{res}}}{\thunk{\app{\force{\var{ref}}}{\inj{Valid}{m}}}}}
      \end{array}
      }}
      {\\\seq{r_{\mathrm{2}}}{\app{\app{\mt{\mathbf{throw}}}{m_1}}{20}}
      {\\\return{r_{\mathrm{1}}}}}
    \end{array}}\\
    &\left[\text{where \(m \mapsto \nil, m_1 \mapsto \refcell{\inj{Invalid}{\unit}}\)}\right] \\  
    &\arrAC^{+} \paren{\begin{array}{@{}l@{}}
      \seq{r_{\mathrm{1}}}{\app{\force{\thunk{\abs{x_{\mathrm{2}}}{\seq{k_{\mathrm{2}}}{\force{\var{x_{\mathrm{2}}}}}{\return{30}}}}}}{\thunk{\app{\force{\var{ref}}}{\inj{Valid}{m}}}}}
      {\\\seq{r_{\mathrm{2}}}{\app{\app{\mt{\mathbf{throw}}}{m_1}}{20}}
                            {\\\return{r_{\mathrm{1}}}}}
    \end{array}}\\
  &\left[\text{where \(m \mapsto \thunk{\abs{y}{\plug{P}{\return{y}}}}\)} \right] \\
    &\arrAC^{+} \paren{\begin{array}{@{}l@{}}
      \seq{r_{\mathrm{1}}}{\left( \seq{k_{\mathrm{2}}}{\force{\thunk{\app{\force{\var{ref}}}{\inj{Valid}{m}}}}}{\return{30}}\right)}
      {\\\seq{r_{\mathrm{2}}}{\app{\app{\mt{\mathbf{throw}}}{m_1}}{20}}
      {\\\return{30}}}
                            \end{array}}\\
    &\arrAC^{+} \paren{\begin{array}{@{}l@{}}
      \seq{r_{\mathrm{1}}}{\left(\seq{k_{\mathrm{2}}}{\return{m_2}}{\return{30}}\right)}
      {\\\seq{r_{\mathrm{2}}}{\app{\app{\mt{\mathbf{throw}}}{m_1}}{20}}
      {\\\return{30}}}
                            \end{array}}  \label{eq:evalrmtM2}\\
    &\left[ \text{where \(m_2 \mapsto \refcell{\inj{Valid}{m}}\)} \right] \\
    &\arrAC^{+} \paren{\begin{array}{@{}l@{}}
      \seq{r_{\mathrm{1}}}{\return{30}}
      {\\\seq{r_{\mathrm{2}}}{\app{\app{\mt{\mathbf{throw}}}{m_1}}{20}}
      {\\\return{30}}}
    \end{array}} \\
    &\arrAC^{+}
      \paren{\begin{array}{@{}l@{}}
      \seq{r_{\mathrm{2}}}{\app{\app{\mt{\mathbf{throw}}}{m_1}}{20}}
        {\\\return{30}}
      \end{array}} \label{eq:evalrmtM3}\\
  &\arrAC^{+} \bot
\end{align}
Thus, the evaluation of \(\mt{M}\) results in the error state \(\bot\), as it
should.
There are two points worth noting.
First, the two continuations captured in \(\mt{M}\) are represented by distinct
cells \(m_1\) and \(m_2\), even though both are backed by the same coroutine
\(m\).
Therefore, the capture of \(m_2\) no longer reactivates \(m_1\), as it did under
the na\"{i}ve translation.
Second, the invocation of the first \(\mt{\mathbf{throw}}\) in
\eqref{eq:evalrmtM1} makes the cell \(m_1\) hold \(\invalid\), and nothing
restores it afterwards, so the second throw fails in \eqref{eq:evalrmtM3}.

\subsection{Simulation relation}
\label{sec:DELoneToAC:simulation}

To prove the correctness of the translation in Figure~\ref{fig:deltoac}, we
use \emph{simulation}.
For two transition systems $\mathscr{L}$ and $\mathscr{L}'$, a binary relation
on these systems is a simulation relation if the following condition holds:
\begin{quote}
  If $M \rightarrow_{\mathscr{L}} M'$ and $M \sim N$ hold, then there exists an
  $N'$ such that $N \rightarrow^{*}_{\mathscr{L}'} N'$ and $M' \sim N'$ hold.
\end{quote}
If such a simulation relation exists, we say $\mathscr{L}'$ \emph{simulates}
$\mathscr{L}$.
In our case, we shall construct a simulation relation $\sim$ on $\del$
configurations and $\ac$ configurations that respects the translation in
Figure~\ref{fig:deltoac}.
We then use this to show that the translation preserves the semantics.

Constructing such a relation is the main obstacle to the correctness proof.
In the refined translation, the validity of a continuation is recorded in the
corresponding reference cell, not in the coroutine that realizes it.
Therefore, a target term alone does not tell us which source configuration it
represents: our simulation relation must relate stores as well as computations.

We do so with \emph{label maps}, which associate a continuation label (e.g.,
\(l_1\) in our running example) with a reference cell that represents it (e.g.,
\(m_1\)) and a coroutine label that realizes it (e.g., \(m\)).
We also enforce some \emph{invariant conditions}: several continuation labels
may share one coroutine (e.g., \(l_1\) and \(l_2\) share \(m\)), but at most one
of them is valid, and each reference cell's content follows the validity of its
label.
Setting up reference cells and coroutines in \(\ac\) needs several steps that
have no counterpart in \(\del\).
We therefore let the \emph{core} relation relate a source configuration to a
\emph{canonical} target representative, and close it under these
\emph{administrative} reductions.

\paragraph{Label maps and invariant conditions}

First, we define the \emph{label maps}: A label map is a partial function
\( \eta : \syntacticset{L}_{\mathbf{D}} \rightharpoonup \syntacticset{L}_{\ac}
\times \syntacticset{L}_{\ac}\).
If \(\eta(l) = (mc, m)\), then \(mc\) is the reference cell that represents
\(l\), and \(m\) is the coroutine label that \(mc\) may hold as \(\valid{m}\).
We then parameterize and extend the translation \(\mtempty\) with $\eta$, by
translating the source continuation label \(l\) to the label \(mc\) of the
reference cell.
We call this extension a \emph{runtime translation} and write it as
\(\mtempty_{\eta}\).
Similarly, we extend runtime translations to translations on contexts by mapping
a hole to a hole.

Using label maps and the extended translation, we define coherence conditions,
which constrain a label map against a source store, and invariant conditions,
which relate a source store to a target one.
\begin{definition}[coherence conditions] For a \(\del\) store \(\theta\) and a
  label map \(\eta\), we define \(\Coh(\theta, \eta)\) to hold if and only if
  both of the following conditions are satisfied:
  \begin{enumerate}
  \item[(C1)] \(\dom{\eta} = \dom{\theta}\);
  \item[(C2)] \(\eta\) is injective in its first component: if
    \(\eta(l) = (mc, m)\) and \(\eta(l') = (mc, m')\), then \(l = l'\).
  \end{enumerate}
\end{definition}

\begin{definition}[invariant conditions]
  For a \(\del\) store \(\theta\), an \(\ac\) store \(\tau\), and a label
  map \(\eta\), we define \(\Inv(\theta, \tau, \eta)\) to hold if and only
  if the following conditions are satisfied.
  \begin{enumerate}
  \item[(IC1)] If \(\theta(l) = \nil\) and \(\eta(l) = (mc, m)\), then
    \(\tau(mc) = \refcell{\invalid}\);
  \item[(IC2)] If \(\theta(l) = \abs{y}{\dollar{\plug{H}{\return{y}}}{x}{M}}\)
    and \(\eta(l) = (mc, m)\), then
    \begin{enumerate}
    \item \(\tau(mc) = \refcell{\valid{m}}\),
    \item \(\tau(m) = \Cont_{\eta}(H, x, M)\), and
    \item for any \(l' \in \dom{\theta}\setminus\{l\}\) and
      \(mc' \in \syntacticset{L}_{{\ac}}\), if \(\eta(l') = (mc', m)\), then
      \(\theta(l') = \nil\),
    \end{enumerate}
  \end{enumerate}
  where
  \[
    \Cont_{\eta}(H, x, M) \defeq
    \thunk{\abs{y}{\letin{x}{\mte{\plug{H}{\return{y}}}{\eta}}{\return{\thunk{\abs{\underscore}{\mte{M}{\eta}}}}}}}.
  \]
\end{definition}

Condition (IC1) ensures that a reference cell representing an invalid
continuation has \(\invalid\).
Condition (IC2) states that a valid continuation corresponds to a valid
reference cell containing a coroutine label that realizes the continuation.
Moreover, condition (IC2-c) says that a coroutine label \(m\) realizes only one
valid source continuation label.
In other words, since the second component of \(\eta\) is not necessarily
injective, the label \(m\) may have more than one corresponding source
continuation labels, but condition (IC2-c) ensures that only one of them is
valid.

\paragraph{Core relation}

We then define the \emph{core relation} that specifies the correspondence
between source and target configurations.

\begin{definition}[core relation]
  For a label map \(\eta\), we define a relation \(\rsimue{\eta}\) on \(\del\)
  configurations and \(\ac\) configurations inductively on the structure of
  \(\del\) computations as follows.

  For every constructor \(E\) of \(\del\) among
  \begin{mathpar}
    \pcase{V}{x_1}{x_2}{M},
    \and
    \scase{V}{L_{\mathnormal{i}}}{x_i}{M_i},
    \and
    \force{V},\and \return{V},\\
    \abs{x}{M},
    \and
    \cpair{M_1}{M_2},
    \and
    \shift{k}{M},
    \and
    \throw{V}{W},
  \end{mathpar}
  we have the base clause
  \[
    \simue{\config{E}{\theta}}{\config{\mte{E}{\eta}}{\tau}}{\eta},
  \]
  for any \(\theta\) and \(\tau\).

  For other constructs of \(\del\), we have
  \begin{mathpar}
    \inferrule
    {
      \simue{\config{M_1}{\theta}}{\config{N_1}{\tau}}{\eta}
    }
    {
      \simue
      {\config{\letin{x}{M_1}{M_2}}{\theta}}
      {\config{\letin{x}{N_1}{\mte{M_2}{\eta}}}{\tau}}
      {\eta}
    },
    \\
    \inferrule
    {
      \simue{\config{M}{\theta}}{\config{N}{\tau}}{\eta}
    }
    {
      \simue
      {\config{\app{M}{V}}{\theta}}
      {\config{\app{N}{\mte{V}{\eta}}}{\tau}}
      {\eta}
    },
    \and
    \inferrule
    {
      \simue{\config{M}{\theta}}{\config{N}{\tau}}{\eta}
    }
    {
      \simue
      {\config{\prj{i}{M}}{\theta}}
      {\config{\prj{i}{N}}{\tau}}
      {\eta}
    },
  \end{mathpar}
  and
  \[
     \inferrule
    {
      \simue{\config{M_1}{\theta}}{\config{N_1}{\tau}}{\eta} \\
      \tau(m) = \nil \\
      \forall l,mc.\;\eta(l) = (mc, m) \Rightarrow \theta(l) = \nil
    }
    {
      \simue
      {\config{\dollar{M_1}{x}{M_2}}{\theta}}
      {\config{\Act_{\eta}(N_1, x, M_2, m)}{\tau}}
      {\eta}
    },
  \]
  where
  \[
    \Act_{\eta}(N, x, M, m) \defeq
    \begin{pmatrix*}[l]
      \letin{\var{res}}{\labeledc{m}{\left( \letin{x}{N}{\return{\thunk{\abs{\underscore}{\mte{M}{\eta}}}}} \right)}}{\\\app{\force{\var{res}}}{\thunk{\app{\force{\var{ref}}}{(\valid{m})}}}}
    \end{pmatrix*}.
  \]
\end{definition}

We defined the core relation on configurations, not on computations.
This is because, in the clause for the dollar term, the target configuration has
an active coroutine \(m\) in it, which requires two additional conditions in the
premise.
These premises are complementary to the invariant conditions.
The invariant conditions are concerned with continuations to be invoked---thus
suspended coroutines.
On the other hand, these two premises are concerned with the current
continuation to be captured---thus the active coroutine \(m\).

We extend the core relation onto contexts:
\(\simuec{\config{\context{C}}{\theta}}{\config{\context{D}}{\tau}}{\eta}\),
where \(\context{C}\) is a \(\del\) evaluation context, \(\context{D}\) is an
\(\ac\) evaluation context, \(\theta\) is a \(\del\) store, and \(\tau\) is an
\(\ac\) store.
The definition of
\(\simuec{\config{\context{C}}{\theta}}{\config{\context{D}}{\tau}}{\eta}\) has
clauses similar to the last four in the definition of \(\rsimue{\eta}\), plus a
clause for a hole:
\(\simuec{\config{\hole}{\theta}}{\config{\hole}{\tau}}{\eta}\).

Then, we put several side conditions together into a single relation:
\begin{definition}
  For \(\del\) and \(\ac\) configurations \(C = \config{M}{\theta}\) and
  \(D = \config{N}{\tau}\), we write \(\simues{C}{D}{\eta}\) if and only if all
  of the following hold:
  \begin{enumerate}
  \item \(\WF_{\del}(C)\), which means that all continuation labels occurring in
    \(M\) or in \(\theta\) are in \(\dom{\theta}\);\footnote{See \version{the
        appendix of the full version}{Appendix~\ref{app:sec:well-formedness}}
      for the precise definition.
      We remark that well-formedness is preserved by reduction.
    }
  \item \(\WF_{\ac}(D)\), which is defined similarly, and which moreover requires
    that \(\tau(m) = \nil\) for every coroutine \(m\) running in \(N\), that is,
    for every \(m\) such that \(N\) has a subterm of the form
    \(\labeledc{m}{N'}\);
  \item \(\Coh(\theta, \eta)\);
  \item \(\Inv(\theta, \tau, \eta)\);
  \item \(\simue{C}{D}{\eta}\).
  \end{enumerate}
\end{definition}

There are two points worth noting about this definition.
First, \(\simues{C}{D}{\eta}\) is parameterized by \(\eta\), which is necessary
to record the correspondence between source continuations and target coroutines
that are dynamically generated during evaluation.
The final simulation relation will be given by existentially quantifying
\(\eta\).
Second, the relation mentions only \emph{canonical} \(\ac\) configuration \(D\).
For example, the base clause for the dollar term,
\(\simue{\config{\dollar{M_1}{x}{M_2}}{\theta}}
{\config{\mte{\dollar{M_1}{x}{M_2}}{\eta}}{\tau}}{\eta}\), is not included.
This is because, by reducing the configuration on the right-hand side, we obtain
\[ {\config{\mte{\dollar{M_1}{x}{M_2}}{\eta}}{\tau}} \arrAC^+
  {\config{\Act_{\eta}(N_1, x, M_2, m)}{\tau'}},
\]
for some \(m\) and \(\tau'\), and then we can prove that
\[
  \simue
  {\config{\dollar{M_1}{x}{M_2}}{\theta}}
  {\config{Act_{\eta}(N_1, x, M_2, m)}{\tau'}}
  {\eta}.
\]
Put differently, we consider the target configurations \emph{modulo
  administrative reductions}, which leads to the following definition of the
simulation relation:

\begin{definition}[simulation relation]\label{def:deltoac-simulation}
  For \(\del\) and \(\ac\) configurations \(C\) and \(D\), define \(C \sim D\)
  if and only if either \(C = D = \bot\), or there exist a label map \(\eta\)
  and an \(\ac\) configuration \(D'\) such that
  \[
    D \arrAC^{*} D' \quad\text{and}\quad \simues{C}{D'}{\eta}.
  \]
\end{definition}

We remark that $\simu{C}{D}$ has the following properties, the latter of which
ensures that we can \emph{canonicalize} the target configuration.

\begin{lemma}\label{lem:deltoac:properties}
  \begin{enumerate}
  \item For every \(\del\) program \(M\),
    \(\simu{\config{M}{\emptyset}}{\config{\mt{M}}{\emptyset}}\) holds.
    
  \item Assume \(\WF_{\del}(\config{M}{\theta})\),
    \(\WF_{\ac}(\config{\mte{M}{\eta}}{\tau})\), \(\Coh(\theta, \eta)\), and
    \(\Inv(\theta, \tau, \eta)\).
    Then, there exist an \(\ac\) computation \(N\) and an \(\ac\) store
    \(\tau'\) such that
    \( \config{\mte{M}{\eta}}{\tau} \arrAC^{*} \config{N}{\tau'} \) and
    \(\simues{\config{M}{\theta}}{\config{N}{\tau'}}{\eta}\).
  \end{enumerate}
\end{lemma}

\subsection{Correctness of the refined translation}
\label{sec:DELoneToAC:correctness}

We prove that our refined translation and the simulation relation form a
simulation in three steps.
First, we prove the simulation of beta-reduction
(Proposition~\ref{prop:beta-sim}), then its lift onto the general reduction
(Proposition~\ref{prop:beta-sim-lift}), and finally the simulation of the
general reduction (Theorem~\ref{thm:sim}).
The preservation of semantics follows from it as a corollary.
We shall give all the key propositions and proof sketches for important cases.
The full proof appears in \version{the appendix of the full
  version}{Appendix~\ref{app:sec:DELoneToAC:proof}}.

\begin{proposition}\label{prop:beta-sim}
  Assume \(\simues{\config{M}{\theta}}{\config{N}{\tau}}{\eta}\).
  Then:
  \begin{enumerate}
  \item If \(\redbetaD{\config{M}{\theta}}{\bot}\), then
    \[
      \redACplus{\config{N}{\tau}}{\bot}.
    \]
    \label{prop:beta-sim:1}

  \item If \(\redbetaD{\config{M}{\theta}}{\config{M'}{\theta'}}\) then there exist \(N'\), \(\tau'\), \(\eta'\) such that
    \[
      \redACplus{\config{N}{\tau}}{\config{N'}{\tau'}}
      \quad\text{and}\quad
      \simues{\config{M'}{\theta'}}{\config{N'}{\tau'}}{\eta'}.
    \]
    \label{prop:beta-sim:2}
  \end{enumerate}
\end{proposition}
\begin{proof}
  By case analysis on the beta-reduction rule used in the reduction of \(\config{M}{\theta}\).
  We write out the cases for \((\mathrm{throw})\) and \((\mathrm{fail})\).

  \noindent \textbf{Case} $(\mathrm{throw})$:
  \begin{caseindent}
    Suppose \(M = \throw{l}{V}\) and
    \(\theta(l) = \abs{y}{\dollar{\plug{H}{\return{y}}}{x}{M_0}}\).
    Then
    \[
      \config{M}{\theta} \arrD
      \config{\dollar{\plug{H}{\return{V}}}{x}{M_0}}{\theta'}, \qquad \theta'
      \defeq \theta[l := \nil].
    \]
    The last rule used to derive
    \(\simues{\config{M}{\theta}}{\config{N}{\tau}}{\eta}\) is the base
    clause for \(\mathbf{throw}\).
    Hence
    \(N \equiv \mte{\throw{l}{V}}{\eta}\).
    By (C1), there are target labels \(mc\) and \(m\) such that
    \(\eta(l)=(mc,m)\).
    By (IC2),
    \[
      \tau(mc)=\refcell{\valid{m}},
      \qquad
      \tau(m)=\Cont_{\eta}(H, x, M_0),
    \]
    and for any continuation label \(l'\) and coroutine label \(\var{mc'}\), if \(l'\neq l\) and
    \(\eta(l')=(mc',m)\), then \(\theta(l')=\nil\).
    Put
    \[
      \tau' \defeq \tau[mc := \refcell{\invalid},\; m := \nil].
    \]
    By calculation, we have
    \[
      \config{\mte{\throw{l}{V}}{\eta}}{\tau} \arrAC^{+}
      \config{\Act_{\eta}(\mte{\plug{H}{\return{V}}}{\eta},x,M_0,m)}{\tau'}.
    \]

    We next check the hypotheses of Lemma~\ref{lem:deltoac:properties}(2) for
    \(\plug{H}{\return{V}}\) as \(M\) with \(\theta'\), \(\tau'\) and \(\eta\).
    \(\WF_{\del}(\ldots)\) and \(\WF_{\ac}(\ldots)\) are immediate.
    \(\Coh(\theta', \eta)\) follows from \(\Coh(\theta, \eta)\) because
    \(\dom{\theta'}=\dom{\theta}\).
    To prove \(\Inv(\theta',\tau',\eta)\), we show (IC1) as follows:
    The label \(l\) is invalid in \(\theta'\) and \(mc\) has the invalid tag in
    \(\tau'\).
    All other invalid continuation labels in \(\theta'\) are invalid in
    \(\theta\), and mapped to invalid cells by \(\eta\) and (IC1) for
    \(\theta\).
    Thus (IC1) holds.
    We then show (IC2) as follows:
    If \(l'\) is valid in \(\theta'\), then \(l' \neq l\), and by (C2) of
    \(\Coh(\theta', \eta)\), the first component of \(\eta(l')\) is not
    \(\var{mc}\).
    Moreover, the coroutine realizing \(l'\) cannot be \(m\): if so, (IC2-c) of
    \(\Inv(\theta, \tau, \eta)\) would imply \(\theta(l')=\nil\).
    Thus the store contents relevant to \(l'\) are unchanged, and (IC2) follows.

    Hence, we can apply Lemma~\ref{lem:deltoac:properties}(2) for
    \(M = \plug{H}{\return{V}}\) with \(\theta'\), \(\tau'\), and \(\eta\), to
    obtain
    \[
      \config{\mte{\plug{H}{\return{V}}}{\eta}}{\tau'} \arrAC^{*} \config{N'}{\tau''}
      \quad\text{and}\quad
      \simues{\config{\plug{H}{\return{V}}}{\theta'}}{\config{N'}{\tau''}}{\eta},
    \]
    for some \(\ac\) computation \(N'\) and \(\ac\) store \(\tau''\).
    This reduction takes place in the body of the active coroutine \(m\), so it
    lifts to
    \[
      \config{\Act_{\eta}(\mte{\plug{H}{\return{V}}}{\eta}, x, M_0, m)}{\tau'}
      \arrAC^{*}
      \config{\Act_{\eta}(N', x, M_0, m)}{\tau''},
    \]
    and hence
    \(\redACplus{\config{N}{\tau}}{\config{\Act_{\eta}(N', x, M_0, m)}{\tau''}}\).

    It remains to show
    \(\simues{\config{\dollar{\plug{H}{\return{V}}}{x}{M_0}}{\theta'}}{\config{\Act_{\eta}(N',
        x, M_0, m)}{\tau''}}{\eta}\).
    The two well-formedness conditions, \(\Coh(\theta',\eta)\), and
    \(\Inv(\theta',\tau'',\eta)\) are immediately derived from the relation
    \(\rsimues{\eta}\) obtained above.
    For the core relation, we apply the dollar-term rule of the core relation to
    \(\simue{\config{\plug{H}{\return{V}}}{\theta'}}{\config{N'}{\tau''}}{\eta}\),
    so it suffices to check its two side conditions
    \[
      \tau''(m) = \nil
      \quad\text{and}\quad
      \forall l',mc'.\;\eta(l') = (mc', m) \Rightarrow \theta'(l') = \nil.
    \]
    The former holds because the coroutine \(m\) keeps running throughout the
    lifted reduction above, and \(\WF_{\ac}\), which is preserved by reduction,
    forces \(\tau''(m)\) to be \(\nil\) (see \version{the appendix of the full
      version}{the proof of Proposition~\ref{app:prop:beta-sim} in the
      appendix} for the details).
    As for the latter, (IC2-c) of \(\Inv(\theta, \tau, \eta)\) tells us that the
    label \(l\) is the only valid continuation label defined in \(\theta\) that
    is associated with the coroutine \(m\).
    This label is no longer valid in \(\theta'\), which implies the condition
    above.
  \end{caseindent}

  \noindent \textbf{Case} $(\mathrm{fail})$:
  \begin{caseindent}
    Suppose \(M=\throw{l}{V}\) and \(\theta(l)=\nil\).
    Then \(\redbetaD{\config{M}{\theta}}{\bot}\).
    The last rule used to derive
    \(\simues{\config{M}{\theta}}{\config{N}{\tau}}{\eta}\) is the base clause
    for \(\mathbf{throw}\), so \(N\equiv\mte{\throw{l}{V}}{\eta}\).
    By (C1), \(\eta(l)=(mc,m)\) is defined.
    Since \(\theta(l)=\nil\), (IC1) gives
    \[
      \tau(mc)=\refcell{\invalid}.
    \]
    Therefore, by calculation, we obtain
    \[
      \config{\mte{\throw{l}{V}}{\eta}}{\tau}
      \arrAC^{+}
      \config{\force{\var{fail}}}{\tau}
      \arrAC^{+}
      \bot.
    \]
  \end{caseindent}
  This concludes the proof.
\end{proof}

\begin{proposition}\label{prop:beta-sim-lift}
  Suppose
  \(\simues{\config{M}{\theta}}{\config{N}{\tau}}{\eta}\),
  \(\simuec{\config{\context{C}}{\theta}}{\config{\context{D}}{\tau}}{\eta}\),
  \(\WF_{\del}(\config{\plug{C}{M}}{\theta})\),
  and
  \(\WF_{\ac}(\config{\plug{D}{N}}{\tau})\).
  \begin{enumerate}
  \item If \(\redbetaD{\config{M}{\theta}}{\bot}\), then
    \[
      \redACplus{\config{\plug{D}{N}}{\tau}}{\bot}.
    \]
    \label{prop:beta-sim-lift:1}

  \item If \(\redbetaD{\config{M}{\theta}}{\config{M'}{\theta'}}\) then there exist \(N'\), \(\tau'\), \(\eta'\) such that
    \begin{mathpar}
      \redACplus{\config{N}{\tau}}{\config{N'}{\tau'}},
      \and
      \simues{\config{\plug{C}{M'}}{\theta'}}{\config{\plug{D}{N'}}{\tau'}}{\eta'},
      \and
      \simuec{\config{\context{C}}{\theta'}}{\config{\context{D}}{\tau'}}{\eta'}.
    \end{mathpar}
    \label{prop:beta-sim-lift:2}
  \end{enumerate}
\end{proposition}
\begin{proof}
  By induction on the derivation of
  \(\simuec{\config{\context{C}}{\theta}}{\config{\context{D}}{\tau}}{\eta}\).
  The base case is Proposition~\ref{prop:beta-sim}.
  For the inductive step, the cases for \(\mam\) frames trivially hold.
  The only nontrivial case is that for a dollar frame, where we have to
  carefully check that the side conditions involving an active coroutine are
  preserved through reduction (\version{See the full version for the
    details}{See the proof of Proposition~\ref{app:prop:beta-sim-lift} in the
    appendix}).
\end{proof}

\begin{theorem}[simulation]\label{thm:sim}
  If \(\simu{C}{D}\) and \(\redD{C}{C'}\), then there exists an \(\ac\) configuration \(D'\) such that
  \[
    \redACplus{D}{D'} \quad\text{and}\quad \simu{C'}{D'}.
  \]
\end{theorem}
\begin{proof}
  Since \(\redD{C}{C'}\), the configuration \(C\) is not \(\bot\).
  By the definition of \(\sim\), there exist an \(\ac\) computation \(P\), an \(\ac\) store \(\tau\), and a label map \(\eta\) such that
  \[
    D \arrAC^{*} \config{P}{\tau}
    \quad\text{and}\quad
    \simues{C}{\config{P}{\tau}}{\eta}.
  \]
  Write \(C=\config{M_0}{\theta}\).
  By the definition of \(\arrD\), the source step decomposes into an evaluation
  context and a beta-reduction step.
  Thus, there exist a \(\del\) evaluation context \(\context{C}_0\) and a
  \(\del\) computation \(M\) such that \( M_0 \equiv \context{C}_0[M] \) and
  either
  \[
    \redbetaD{\config{M}{\theta}}{\bot}
    \quad\text{and}\quad
    C'=\bot,
  \]
  or there exist a \(\del\) computation \(M'\) and a \(\del\) store \(\theta'\) such that
  \[
    \redbetaD{\config{M}{\theta}}{\config{M'}{\theta'}}
    \quad\text{and}\quad
    C'=\config{\context{C}_0[M']}{\theta'}.
  \]

  From \(\simues{\config{\context{C}_0[M]}{\theta}}{\config{P}{\tau}}{\eta}\),
  we obtain in particular
  \(\simue{\config{\context{C}_0[M]}{\theta}}{\config{P}{\tau}}{\eta}\).
  By induction on \(\context{C}_0\), we can check that there exist an \(\ac\)
  evaluation context \(\context{D}_0\) and an \(\ac\) computation \(N\) such
  that \(P \equiv \context{D}_0[N]\),
  \(\simuec{\config{\context{C}_0}{\theta}}{\config{\context{D}_0}{\tau}}{\eta}\),
  and \(\simue{\config{M}{\theta}}{\config{N}{\tau}}{\eta}\).
  From the side conditions of
  \(\simues{\config{\context{C}_0[M]}{\theta}}{\config{\context{D}_0[N]}{\tau}}{\eta}\),
  we obtain \(\simues{\config{M}{\theta}}{\config{N}{\tau}}{\eta}\).

  We now apply Proposition~\ref{prop:beta-sim-lift} to
  \(\simues{\config{M}{\theta}}{\config{N}{\tau}}{\eta}\) and
  \(\simuec{\config{\context{C}_0}{\theta}}{\config{\context{D}_0}{\tau}}{\eta}
  \).
  First, suppose \( \redbetaD{\config{M}{\theta}}{\bot} \).
  Proposition~\ref{prop:beta-sim-lift} gives
  \( \config{P}{\tau} = {\config{\context{D}_0[N]}{\tau}} \arrAC^{+} {\bot}
  \).
  Combining this with \(D\arrAC^{*}\config{P}{\tau}\) yields
  \( \redACplus{D}{\bot} \).
  As \(C'=\bot\), we have \(\simu{C'}{\bot}\) by the definition of the
  simulation relation.

  Next suppose \( \redbetaD{\config{M}{\theta}}{\config{M'}{\theta'}} \).
  Proposition~\ref{prop:beta-sim-lift} yields an \(\ac\) computation \(N'\), an
  \(\ac\) store \(\tau'\), and a label map \(\eta'\) such that
  \( \redACplus{\config{N}{\tau}}{\config{N'}{\tau'}} \) and
  \( \simues {\config{\context{C}_0[M']}{\theta'}}
  {\config{\context{D}_0[N']}{\tau'}} {\eta'} \).
  By the definition of \(\arrAC\), we obtain
  \(
  \redACplus{\config{\context{D}_0[N]}{\tau}}{\config{\context{D}_0[N']}{\tau'}}
  \).
  Since \(P\equiv\context{D}_0[N]\), composing with
  \(D\arrAC^{*}\config{P}{\tau}\) gives
  \( \redACplus{D}{\config{\context{D}_0[N']}{\tau'}} \).
  Also, because \( C'=\config{\context{C}_0[M']}{\theta'} \), we obtain
  \( \simu{C'}{\config{\context{D}_0[N']}{\tau'}} \) by the definition of the
  simulation relation.
\end{proof}

Using this theorem, we prove the preservation of semantics and the strong
macro-expressibility.

\begin{theorem}[preservation of semantics]\label{thm:prev-sem}
  For every \(\del\) program \(M\), \(\eval{\del}{M}\) is defined if and only if
  \(\eval{\ac}{\mt{M}}\) is defined.
\end{theorem}

\begin{proof}
  We prove the two directions separately.

  First suppose that \(\eval{\del}{M}\) is defined.
  Then there exist a \(\del\) value \(V\) and a \(\del\) store \(\theta\) such
  that \( \redDclos{\config{M}{\emptyset}}{\config{\return{V}}{\theta}} \).
  By Lemma~\ref{lem:deltoac:properties},
  \( \simu{\config{M}{\emptyset}}{\config{\mt{M}}{\emptyset}} \).
  Repeatedly applying Theorem~\ref{thm:sim} along the source reduction sequence,
  we obtain an \(\ac\) configuration \(D\) such that
  \[
    \redACclos{\config{\mt{M}}{\emptyset}}{D}
    \quad\text{and}\quad
    \simu{\config{\return{V}}{\theta}}{D}.
  \]
  By the definition of \(\sim\), there exist an \(\ac\) computation \(N\),
  an \(\ac\) store \(\tau\), and a label map \(\eta\) such that
  \[
    D\arrAC^{*}\config{N}{\tau}
    \quad\text{and}\quad
    \simues{\config{\return{V}}{\theta}}{\config{N}{\tau}}{\eta},
  \]
  the latter of which yields
  \( \simue{\config{\return{V}}{\theta}}{\config{N}{\tau}}{\eta} \).
  By the definition of \(\rsimue{\eta}\), \(N\) equals \(\return{\mte{V}{\eta}}\).
  Hence
  \( \redACclos{\config{\mt{M}}{\emptyset}}{\config{\return{\mte{V}{\eta}}}{\tau}} \),
  and therefore \(\eval{\ac}{\mt{M}}\) is defined.

  We prove the converse direction by contraposition.
  Assume that \(\eval{\del}{M}\) is undefined.
  Since \(\arrD\) is deterministic, exactly one of the following alternatives
  holds.
  \begin{enumerate}
  \item The evaluation diverges.

  \item The evaluation reaches the error state \(\bot\):
    \( \redDplus{\config{M}{\emptyset}}{\bot} \).

  \item The evaluation reaches a stuck configuration: there exist a
    \(\del\) computation \(M'\) and a \(\del\) store \(\theta\) such that
    \( \redDclos{\config{M}{\emptyset}}{\config{M'}{\theta}} \), the
    configuration \(\config{M'}{\theta}\) is stuck, and \(M'\) is not of the
    form \(\return{V}\).
  \end{enumerate}
  We can prove that, in each case, the evaluation of \(\mt{M}\) diverges,
  reaches the error state, or gets stuck, respectively (See \version{the full
    version}{Appendix~\ref{app:sec:DELoneToAC:stuckness}} for the details).
  Therefore, in all cases, \(\eval{\ac}{\mt{M}}\) is undefined.
\end{proof}

\begin{theorem}[strong macro-expressibility]\label{thm:strong-macro} The
  translation \(M \mapsto \mt{M}\) of Figure~\ref{fig:deltoac} is a strong
  macro-translation from \(\del\) to \(\ac\).
\end{theorem}

\begin{proof}
  The translation maps every \(\del\) program to an \(\ac\) program and is
  homomorphic on the \(\mam\) constructors by definition.
  Every constructor peculiar to \(\del\) is translated by a fixed syntactic
  abstraction, as shown in Figure~\ref{fig:deltoac}.
  The preservation of semantics is shown by Theorem~\ref{thm:prev-sem}.
\end{proof}

%% file: sec4-efftodel.tex
\section{One-shot effect handlers as asymmetric coroutines}\label{sec:EFFoneToDELone}

Since Plotkin and Pretnar's proposal~\cite{plotkin2009handlers}, effect handlers
have been actively studied in recent years.
This section extends our results to effect handlers; namely, we establish the
macro-expressibility of one-shot effect handlers in terms of asymmetric
coroutines.
Since macro-expressibility is transitive, it suffices to show that one-shot
effect handlers are macro-expressible by one-shot delimited-control operators.

\subsection{The calculus for one-shot effect handlers}
\label{subsec:eff}

Figure~\ref{fig:eff_syntax} presents the syntax of $\eff$, which is a one-shot
variant of the calculus \(\mathbf{EFF}\) for effect handlers defined by Forster
et al.~\cite{forster2019expressive}.
A continuation label \(l\) (drawn from a countable set
\(\syntacticset{L}_{\mathbf{E}}\)) is a runtime representation of a one-shot
continuation, as in \(\del\).
The computation $\opcall{op}{V}$ is an operation invocation with an argument
$V$, where the operation $\op{op}$ ranges over a finite set $\syntacticset{O}$.
A handler
\(\handler{x}{M_{\mathrm{ret}}}{(\app{\op{op}_i}{\app{p_i}{k_i}} \mapsto
  M_i)_i}\) consists of the \emph{return clause}
\(\return{x} \mapsto M_{\mathrm{ret}}\), and for each operation symbol
\(\op{op}_i\), the \emph{operational clause}
\(\app{\op{op}_i}{\app{p_i}{k_i}} \mapsto M_i\).
The computation $\handle{H}{M}$ installs a handler $H$ and evaluates $M$ under
it.
The computation $\throw{V}{W}$, when evaluated, invokes the continuation
associated with $V$ with the argument \(W\) if \(V\) is a continuation label.
We define $\eff$ programs as the set of computations containing no continuation
labels.

\begin{figure}[tb]
  \setlength{\parindent}{0cm}%
  \setlength{\arraycolsep}{4pt}%
  \renewcommand{\arraystretch}{0.9}%
\begin{tabular}{rcll}
  $V,W$ & ::= & $\ldots$ &  \textsf{value}\\
& \textbar\  & $l\quad (\in\syntacticset{L}_{\mathbf{E}})$ & continuation label \\
  $M,N$ & ::= & $\ldots$ &  \textsf{computation}\\
& \textbar\  & $\opcall{op}{V}\quad(\op{op} \in \syntacticset{O})$ & operation call \\
& \textbar\  & $\handle{H}{M}$ & handle \\
& \textbar\  & $\throw{V}{W}$ & throw \\
  $H$ & ::= & $\handler{x}{M_{\mathrm{ret}}}{(\app{\op{op}_i}{\app{p_i}{k_i}} \mapsto M_i)_i}$ & \textsf{handler}\footnotemark \\
\end{tabular}\\[1em]
\begin{tabular}{rrcl}
  \textsf{pure frame} & $\mathcal{P}$ & ::= & \ldots \\
  \textsf{computational frame} & $\mathcal{F}$ & ::= & $\ldots \mid \handle{H}{\hole}$ \\
  \textsf{pure context} & $\mathcal{H}$ & ::= & \ldots \\
  \textsf{evaluation context} & $\mathcal{C}$ & ::= & \ldots \\
\end{tabular}
\caption{Syntax of $\eff$}
\label{fig:eff_syntax}
\end{figure}
\footnotetext{Note that $\var{x}$ is bound in $M_{\mathrm{ret}}$, and
  $\var{p_i}$ and $\var{k_i}$ are bound in their respective $M_i$.}

In our formulation, every handler must handle all operations in
\(\syntacticset{O}\).
This is not an essential restriction, as any handler that handles only some of
the operations in \(\syntacticset{O}\) can be rewritten\footnote{Given such a
  handler \(H\), we add a clause
  $\opcall{op}{\app{p}{k}} \mapsto \seq{r}{\opcall{op}{p}}{\throw{k}{\var{r}}}$
  to \(H\) for each unhandled operation $\op{op}$.}
to satisfy this condition.

Similarly to $\del$, captured continuations can be invoked at most once during
evaluation.
Hence, the evaluation of \(\handle{H}{(\opcall{error}{10})}\) where
\[
   H \equiv \left\{
    \begin{array}{@{}l@{}}
      \return{x} \mapsto \return{x}, \\
      \opcall{error}{\app{\var{p}}{\var{k}}} \mapsto \letin{a}{\throw{k}{1}}{\letin{b}{\throw{k}{2}}{\return{\vpair{a}{b}}}}
    \end{array}
  \right\}
\]
must result in an error.
$\eff$ also uses a store, which is a partial function
$\theta :
\syntacticset{L}_{\mathbf{E}}\;\rightharpoonup\;\mathsf{computation}\sqcup\{\nil\}$,
to enforce one-shotness.

A configuration $C$ is either a pair of an $\eff$ computation $M$ and a store $\theta$,
or the error state $\bot$.
The beta reduction rules $\rightarrow_{\mathbf{E}}^{\beta}$ on configurations
are defined in Figure~\ref{fig:eff_beta}.
\begin{figure}[tb]
  \centering
  \setlength{\extrarowheight}{10pt}
  \begin{tabular}{cl}
    $(\mam)$ & $\inferrule{\redbetaM{M}{M'}}{\redbetaEp{M}{\theta}{M'}{\theta}}$ \\
    $(\mathrm{ret})$ &
    $\inferrule
       {H \equiv \handler{x}{M_{\mathrm{ret}}}{\ldots}}
       {\redbetaEp{\handle{H}{\paren{\return{V}}}}{\theta}{M_{\mathrm{ret}}[V/x]}{\theta}}$ \\
    $(\mathrm{op})$ &
    $\inferrule
    {
      H \equiv \handler{x}{M_{\mathrm{ret}}}{(\app{\op{op}_i}{\app{p_i}{k_i}} \mapsto M_i)_i} \\
      l \notin \dom{\theta}
    }
    { 
      \config{\handle{H}{\paren{\plug{H}{\app{\op{op}_j}{V}}}}}{\theta} \\\\
      \rightarrow_{\mathbf{E}}^{\beta} \config{M_j[V/p_j, l/k_j]}{\theta[l := \abs{x}{\handle{H}{\plug{H}{\return{x}}}}]}
      }$ \\
    $(\mathrm{throw})$ &
    $\inferrule
    {
      \theta(l) = \abs{x}{\handle{H}{\plug{H}{\return{x}}}}
    }
    {
      \redbetaEp{\throw{l}{V}}{\theta}{\handle{H}{\plug{H}{\return{V}}}}{\theta[l := \nil]}
    }$ \\
    $(\mathrm{fail})$ &
    $\inferrule
    {
      \theta(l) = \nil
    }
    {
      \config{\throw{l}{V}}{\theta} \rightarrow^{\beta}_{\mathbf{E}} \bot
    }$
  \end{tabular}
  \caption{Beta reduction rules of $\eff$}
  \label{fig:eff_beta}
\end{figure}
The reduction of $\eff$ is defined as follows:
\[
  \inferrule{\redbetaEp{M}{\theta}{M'}{\theta'}}{\redEp{\plug{C}{M}}{\theta}{\plug{C}{M'}}{\theta'}} \hspace {4em} \inferrule{\redbetaE{\config{M}{\theta}}{\bot}}{\redE{\config{\plug{C}{M}}{\theta}}{\bot}}
\]
Evaluation of an $\eff$ program is defined by: $\eval{\eff}{M} \defeq V$ if
there exists a store $\theta$ such that
$\redEclos{\config{M}{\emptyset}}{\config{\return{V}}{\theta}}$.
$\Eval{\eff}$ is a well-defined partial function, since $\arrE$ is
deterministic.

\subsection{Macro-translation from \texorpdfstring{$\eff$}{EFFone} to \texorpdfstring{$\del$}{DELone}}
\label{subsec:efftodel_macro_translation}

\begin{figure}[tb]
  \centering
  \[
  \begin{array}{rll}
    \mt{\handle{H}{M}} & \defeq & \dollar{\app{\dollart{\mt{M}}{H^{\mathrm{ret}}}}{\thunk{H^{\mathrm{ops}}}}}{z}{\return{z}} \\
                       & \mathrm{where} & (\handler{x}{M_{\mathrm{ret}}}{\ldots})^{\mathrm{ret}} = x. \abs{\underscore}{\mt{M_{\mathrm{ret}}}} \\
                       & & (\handler{x}{M_{\mathrm{ret}}}{(\app{\op{op}_i}{\app{p_i}{k_i}} \mapsto M_i)_i})^{\mathrm{ops}} \\
                       & &\quad= \abs{c}{\scase{c}{\op{op}_i}{\vpair{p_i}{k_i}}{\mt{M_i}}} \\
    \mt{\opcall{op}{V}} & \defeq & \mathbf{S_0} \var{k}.\;\lambda \var{h}.\;(\mathbf{let}\;\var{y} = (\shift{k'}{\app{\force{h}}{\inj{\op{op}}{\vpair{\mt{V}}{\var{k'}}}}})\;\mathbf{in}\\
                       & & \phantom{(\mathbf{S_0} \var{k}.\;\lambda \var{h}.\;}
                           \app{\throw{k}{y}}{h}) \\
    
    \mt{\throw{V}{W}} & \defeq & \throw{\mt{V}}{\mt{W}}
  \end{array}
  \]
  \caption{Translation from $\eff$ to $\del$}
  \label{fig:efftodel}
\end{figure}

We present a translation from $\eff$ to $\del$ in Figure~\ref{fig:efftodel},
which is inspired by Forster et al.'s translation \cite{forster2019expressive}.
A handle computation \(\handle{H}{M}\) is translated using two nested dollar
terms: an outer one that delimits the whole translated computation, and an inner
one, whose additional part \(H^{\mathrm{ret}}\) is built from the return clause
of \(H\), that delimits \(\mt{M}\).
The inner dollar term, once evaluated, is applied to the thunked
\emph{dispatcher} \(\thunk{H^{\mathrm{ops}}}\) built from the operational
clauses of \(H\).
An operation invocation \(\opcall{op}{V}\) is translated to a computation that
performs two shift0's in succession.
The first captures the continuation up to the inner delimiter, binding \(k\);
the resulting abstraction is then applied to the dispatcher, binding \(h\).
The second captures the continuation \(\letin{y}{\hole}{\app{\throw{k}{y}}{h}}\)
up to the outer delimiter, binding \(k'\), and calls \(h\) with the pair, tagged
with \(\op{op}\), of the argument \(\mt{V}\) and \(k'\)---the label of this
second captured continuation.
A source throw term is translated to the target throw term homomorphically.
Specifically, when an operational clause invokes its continuation variable,
which is bound to \(k'\), the target throw term invokes \(k'\), which in turn
immediately invokes \(k\), reinstalling the dispatcher \(h\).

The use of two successive shift0's is where our translation departs from that of
Forster et al.
Adapting their translation to \(\eff\) constructs, we obtain the following:
\begin{align*}
  \mt{\handle{H}{M}} & \defeq \app{\dollart{\mt{M}}{H^{\mathrm{ret}}}}{\thunk{H^{\mathrm{ops}}}}, \\
  \mt{\opcall{op}{V}} &\defeq
                        \shift{k}{\abs{h}{\app{\force{h}}{(\inj{\op{op}}{\vpair{\mt{V}}{\thunk{\abs{y}{\app{\throw{k}{y}}{h}}}}})}}}, \\
  \mt{\throw{V}{W}} &\defeq \app{\force{\mt{V}}}{\mt{W}},
\end{align*}
where a source continuation label is represented by the thunk
\(\thunk{\abs{y}{\app{\throw{k}{y}}{h}}}\), which is a value that can be
consumed by the \((U)\) rule of \(\mam\).
This translation would still be a weak macro-translation, but not a strong one
in our setting, since a source program that forces a captured continuation gets
stuck in \(\eff\) whereas its translation does not.\footnote{For instance, take
  \(H \equiv \handler{x}{\return{x}}{\opcall{op}{\app{p}{k}} \mapsto
    \app{\force{\var{k}}}{\unit}}\).
  Then, \(\handle{H}{\opcall{op}{\unit}}\) gets stuck in \(\eff\), because the
  continuation label substituted for \(\var{k}\) is not a thunk, whereas its
  translation evaluates to \(\unit\).}
Therefore, we represent an \(\eff\) continuation label by a genuine \(\del\)
continuation label instead of a thunk.
This is the reason why we perform two successive shift0's: the second one
captures the computation that Forster et al.'s translation encapsulates as a
thunk above.

\paragraph{Simulation relation}

As in Section~\ref{sec:DELoneToAC:simulation}, we prove the correctness of the
translation by constructing a simulation relation between \(\eff\)
configurations and \(\del\) configurations.
While Forster et al.\ defined a simulation relation using their translation
directly, we must take into account labels and stores that are necessary to
capture one-shotness.
Therefore, our simulation relation must relate the stores as well as the
computations, so that it records which captured continuations in \(\del\)
correspond to a given \(\eff\) captured continuation.

As in Section~\ref{sec:DELoneToAC:simulation}, we use \emph{label maps}: here a
label map is a partial function
\(\eta : \syntacticset{L}_{\mathbf{E}} \rightharpoonup
\syntacticset{L}_{\mathbf{D}} \times \syntacticset{L}_{\mathbf{D}}\) that sends
an \(\eff\) continuation label \(l\) to a pair
\(\eta(l) = (\eta(l)_1, \eta(l)_2)\) of \(\del\) continuation labels.
Note that, in our translation, an \(\eff\) continuation label \(l\) is
represented by the label passed to the dispatcher, namely the one bound to
\(k'\).
We thus let \(\eta(l)_1\) be that label, generated by the second shift0, and
\(\eta(l)_2\) be the one bound to \(k\), generated by the first shift0.
We then parameterize and extend the translation with \(\eta\) by
\(\mte{l}{\eta} \defeq \eta(l)_1\), provided \(\eta(l)\) is defined.
As in Section~\ref{sec:DELoneToAC:simulation}, we call this extension a
\emph{runtime translation} and write it as \(\mtempty_{\eta}\).
We also extend runtime translations to translations on contexts by mapping a
hole to a hole.

Using the label map and the runtime translation, we define \emph{coherence
  conditions} and \emph{invariant conditions}.

\begin{definition}[coherence conditions]\label{def:efftodel:Coh}
  For an \(\eff\) store \(\theta\) and a label map \(\eta\), we define
  \(\Coh(\theta, \eta)\) to hold if and only if the following conditions are
  satisfied.
  \begin{enumerate}
  \item[(C1)] \(\dom{\theta} = \dom{\eta}\);
  \item[(C2)] for all \(l, l' \in \dom{\theta}\) and \(i, j \in \{1, 2\}\), if
    \(\eta(l)_i = \eta(l')_j\), then \(l = l'\) and \(i = j\).
  \end{enumerate}
\end{definition}

\begin{definition}[invariant conditions]\label{def:efftodel:IC}
  For an \(\eff\) store \(\theta\), a \(\del\) store \(\tau\), and a label map
  \(\eta\), we define \(\Inv(\theta, \tau, \eta)\) to hold if and only if the
  following conditions are satisfied for every \(l \in \dom{\theta}\).
  \begin{enumerate}
  \item[(IC1)] \(\eta(l)_1 \in \dom{\tau}\) and \(\eta(l)_2 \in \dom{\tau}\);
  \item[(IC2)] if \(\theta(l) = \nil\), then \(\tau(\eta(l)_1) = \nil\);
  \item[(IC3)] if \(\theta(l) = \abs{y}{\handle{H}{\plug{H}{\return{y}}}}\) for
    an \(\eff\) handler \(H\) and an \(\eff\) pure context \(\context{H}\), then
    \begin{align*}
      \tau(\eta(l)_1) &= \abs{x}{\dollar{\letin{y}{\return{x}}{\app{\throw{\eta(l)_2}{y}}{\thunk{H^{\mathrm{ops}}_{\eta}}}}}{z}{\return{z}}}, \\
      \tau(\eta(l)_2) &= \abs{y}{\dollart{\mte{\plug{H}{\return{y}}}{\eta}}{H^{\mathrm{ret}}_{\eta}}},
    \end{align*}
  \end{enumerate}
  where \(H^{\mathrm{ops}}_{\eta}\) (resp.\ \(H^{\mathrm{ret}}_{\eta}\)) denotes
  the translation of the operational clauses (resp.\ return clause) of \(H\)
  with respect to \(\eta\).
\end{definition}

Condition (IC2) ensures that an invalid source continuation corresponds to an
invalid target continuation, and condition (IC3) states that a valid source
continuation is realized by a pair of target continuations.

We can now define the simulation relation.

\begin{definition}[simulation relation]\label{def:efftodel:simulation}
  For an \(\eff\) configuration \(C\) and a \(\del\) configuration \(D\), we
  write \(\simu{C}{D}\) if and only if either \(C = D = \bot\), or there exist
  an \(\eff\) computation \(M\), a \(\del\) computation \(N\), an \(\eff\)
  store \(\theta\), a \(\del\) store \(\tau\), and a label map \(\eta\) such
  that
  \begin{enumerate}
  \item \(C = \config{M}{\theta}\) and \(D = \config{N}{\tau}\);
  \item \(N \equiv \mte{M}{\eta}\);
  \item \(\WF_{\eff}(C)\), which means that every continuation label occurring
    in \(M\) or in \(\theta\) must be in
    \(\dom{\theta}\);\footnote{See \version{the appendix of the full
        version}{Appendix~\ref{app:sec:efftodel:well-formedness}} for the
      precise definition.
      We remark that well-formedness is preserved by reduction.
    }
  \item \(\Coh(\theta, \eta)\) and \(\Inv(\theta, \tau, \eta)\).
  \end{enumerate}
\end{definition}

In contrast to Section~\ref{sec:DELoneToAC:simulation}, we relate a source
computation to its runtime translation directly; in other words, we need not
consider administrative reductions.
The essential difference is that, in the translation from \(\del\) to \(\ac\), a
dollar term administratively creates a coroutine, which forces us to relate only
\emph{canonical} target representatives, whereas here every \(\del\)
continuation label is created by a shift0 term that originates from an operation
invocation in the source.

\paragraph{Correctness of the translation}

We shall prove the correctness of the translation in Figure~\ref{fig:efftodel}.
We first prove that \(\simu{C}{D}\) forms a simulation.

\begin{proposition}[simulation]\label{prop:eff_to_del_sim}
  The following statements hold.
  \begin{enumerate}
  \item If $\simu{C}{D}$ and $\redbetaE{C}{C'}$, then there exists a
    $\del$ configuration $D'$ such that $\redDplus{D}{D'}$ and $\simu{C'}{D'}$.

  \item (1) can be lifted onto the general reduction: if $\simu{C}{D}$ and
    $\redE{C}{C'}$, then there exists a $\del$ configuration $D'$ such that
    $\redDplus{D}{D'}$ and $\simu{C'}{D'}$.
  \end{enumerate}
\end{proposition}

\begin{proof}
  We prove (1) by case analysis on $\redbetaE{C}{C'}$, giving only nontrivial
  cases.
  The other cases follow by routine arguments, and (2) follows from (1) together
  with the fact that an \(\eff\) configuration not in normal form decomposes
  into an evaluation context and a redex, and that this decomposition is
  preserved by the translation.
  \version{The details are given in the appendix of the full version}{The
    details are given in the proofs of Lemma~\ref{lem:eff_to_del_beta_sim} and
    Proposition~\ref{prop:EFFoneToDELone:sim} in the appendix}.

  \noindent \textbf{Case} $(\mathrm{op})$:
  \begin{caseindent}
    Suppose
    $C = \config{\handle{H}{\paren{\plug{H}{\opcall{op}{V}}}}}{\theta}$ and
    \(C \arrE C' = \config{M_{\op{op}}[V/p, l/k]}{\theta'}\), where
    \begin{align*}
      H &= \handler{x}{M_{\mathrm{ret}}}{\ldots (\opcall{op}{\app{p}{k}}
          \mapsto M_{\op{op}}) \ldots}, \\
      \theta' &\defeq \theta[l := \abs{x}{\handle{H}{\plug{H}{\return{x}}}}],
    \end{align*}
    and $l$ is a fresh label.
    By the definition of $\simu{C}{D}$, there exist \(\tau\) and \(\eta\) such
    that
  \[
    D = \config{\dollar{\app{\dollart{N}{H^{\mathrm{ret}}_{\eta}}}{\thunk{H^{\mathrm{ops}}_{\eta}}}}{z}{\return{z}}}{\;\tau},
  \]
  where
  \[
    N \equiv \plug{\mte{\context{H}}{\eta}}{
              \begin{array}{@{}l@{}}
                \mathbf{S_0} \var{k}.\;\lambda \var{h}.\;(\mathbf{let}\;\var{y} = (\shift{k'}{\app{\force{h}}{\inj{\op{op}}{\vpair{\mte{V}{\eta}}{\var{k'}}}}})\;\mathbf{in}\\
                \phantom{(\mathbf{S_0} \var{k}.\;\lambda \var{h}.\;}
                \app{\throw{k}{y}}{h})
              \end{array}
            }.
  \]
  This configuration evaluates to
  \[
    D' \defeq \config{\mte{M_{\op{op}}}{\eta}\left[\mte{V}{\eta}/p, m/k\right]}
    {\tau'}
  \]
  where $m$ and \(n\) are fresh, and
  \[
    \tau' \defeq \tau\left[
      \begin{array}{@{}l@{}}
        m := \abs{x}{\dollar{\letin{y}{\return{x}}{\app{\throw{n}{y}}{\thunk{H^{\mathrm{ops}}_{\eta}}}}}{z}{\return{z}}}, \\
        n := \abs{y}{\dollart{\plug{\mte{\context{H}}{\eta}}{\return{y}}}{H^{\mathrm{ret}}_{\eta}}}
      \end{array} \right].
  \]
  Let $\eta'$ be $\eta[ l := (m, n) ]$.
  By the substitution lemma (\version{proved in the appendix of the full
    version}{Lemma~\ref{lem:eff_to_del_subst} in the appendix}), it follows that
  \[
    \mte{M_{\op{op}}}{\eta}[\mte{V}{\eta}/p, m/k] \equiv \mte{M_{\op{op}}[V/p, l/k]}{\eta'}.
  \]
  Note that
  $\mte{M_{\op{op}}}{\eta} \equiv \mte{M_{\op{op}}}{\eta'}$,
  $\mte{\context{H}}{\eta} \equiv \mte{\context{H}}{\eta'}$,
  $H^{\mathrm{ret}}_{\eta} \equiv H^{\mathrm{ret}}_{\eta'}$, and
  $H^{\mathrm{ops}}_{\eta} \equiv H^{\mathrm{ops}}_{\eta'}$ hold
  since $l$ does not appear in any of $M_{\op{op}}$, $\context{H}$,
  $H^{\mathrm{ret}}$, or $H^{\mathrm{ops}}$.
  Finally, we conclude
  \[
    C' \sim D',
  \]
  witnessed by \(\eta'\).
  It is straightforward to check that \(\Coh(\theta', \eta')\) and
  \(\Inv(\theta', \tau', \eta')\) hold.
  Here, the freshness of \(l\), \(m\), and \(n\) gives (C1) and (C2).
  \end{caseindent}

  \noindent \textbf{Case} $(\mathrm{throw})$:
  \begin{caseindent}
    Suppose $C = \config{\throw{l}{V}}{\theta}$ and
    $C' = \config{\handle{H}{\plug{H}{\return{V}}}}{\theta'}$, where
    $\theta(l) = \abs{x}{\handle{H}{\plug{H}{\return{x}}}}$ and
    $\theta' \defeq \theta[l := \nil]$.
    By the definition of $\simu{C}{D}$, there exist \(\tau\) and \(\eta\) such that
    \(D = \config{\throw{\eta(l)_1}{\mte{V}{\eta}}}{\tau}\),
    \[
      \tau(\eta(l)_1) = \abs{x}{\dollar{\letin{y}{\return{x}}{\app{\throw{\eta(l)_2}{y}}{\thunk{H^{\mathrm{ops}}_{\eta}}}}}{z}{\return{z}}},
    \]
    and
    \[
      \tau(\eta(l)_2) = \abs{y}{\dollart{\mte{\plug{H}{\return{y}}}{\eta}}{H^{\mathrm{ret}}_{\eta}}}.
    \]
    Thus, $D$ evaluates to
    \[
      D' = \config{\dollar{\app{\paren{\dollart{\mte{\plug{H}{\return{y}}}{\eta}}{H^{\mathrm{ret}}_{\eta}}\left[\mte{V}{\eta}/y\right]}}{\thunk{H^{\mathrm{ops}}_{\eta}}}}{z}{\return{z}}}{\tau'},
    \]
    where $\tau' \defeq \tau[\eta(l)_1 := \nil, \eta(l)_2 := \nil]$.
    Then, from the substitution lemma, we obtain
    \[
      \dollart{\mte{\plug{H}{\return{y}}}{\eta}}{H^{\mathrm{ret}}_{\eta}}[\mte{V}{\eta}/y] \equiv \dollart{\mte{\plug{H}{\return{V}}}{\eta}}{H^{\mathrm{ret}}_{\eta}}.
    \]
    Therefore, we conclude that
    \[
    C' \sim D',
    \]
    witnessed by \(\eta\), since
    \[
      \mte{\handle{H}{\plug{H}{\return{V}}}}{\eta} \equiv \dollar{\app{\dollart{\mte{\plug{H}{\return{V}}}{\eta}}{H^{\mathrm{ret}}_{\eta}}}{\thunk{H^{\mathrm{ops}}_{\eta}}}}{z}{\return{z}}.
    \]
    Note that \(\Coh(\theta', \eta)\) holds since
    \(\dom{\theta'} = \dom{\theta}\), and that \(\Inv(\theta', \tau', \eta)\)
    also holds.
  \end{caseindent}

  \noindent \textbf{Case} $(\mathrm{fail})$:
  \begin{caseindent}
    Suppose $C = \config{\throw{l}{V}}{\theta}$, $\theta(l) = \nil$, and $\redbetaE{C}{\bot}$.
    By the definition of $\simu{C}{D}$, there exist \(\tau\) and \(\eta\) such that
    \(D = \config{\throw{\eta(l)_1}{\mte{V}{\eta}}}{\tau}\) and $\tau(\eta(l)_1) = \nil$.
    Hence, the invocation of $\eta(l)_1$ fails and $D$ evaluates to $\bot$, which
    concludes the current case.
  \end{caseindent}
  This concludes the proof.
\end{proof}

\begin{proposition}[preservation of semantics]\label{prop:eff_to_del_preservation_of_sem}
  For every \(\eff\) program \(M\), $\eval{\eff}{M}$ is defined if and only if
  $\eval{\del}{\mt{M}}$ is defined.
\end{proposition}
\begin{proof}
  We first show the ``only if'' direction.
  Suppose that $\redEclos{\config{M}{\emptyset}}{\config{\return{V}}{\theta}}$
  for some $\theta$.
  By induction on \(M\), we obtain
  $\simu{\config{M}{\emptyset}}{\config{\mt{M}}{\emptyset}}$.
  By applying Proposition~\ref{prop:eff_to_del_sim} iteratively, there exist
  \(\tau\) and $\eta$ such that
  $\redDclos{\config{\mt{M}}{\emptyset}}{\config{\return{\mte{V}{\eta}}}{\tau}}$.
  This implies that $\eval{\del}{\mt{M}}$ is defined.
  We then show the ``if'' direction.
  We prove the contrapositive of the proposition: if $\eval{\eff}{M}$ is
  undefined, $\eval{\del}{\mt{M}}$ is also undefined.
  There are three cases where $\eval{\eff}{M}$ is undefined: the evaluation of
  \(M\) (1) diverges, (2) gets stuck, or (3) reaches $\bot$.
  We can check that in each case the evaluation of \(\mt{M}\) diverges, gets
  stuck, or reaches $\bot$, respectively\version{}{
    (Lemmas~\ref{lem:efftodel:diverge}, \ref{lem:efftodel:stuck}, and
    \ref{lem:efftodel:bot} in the appendix)}.
  Thus, $\eval{\del}{\mt{M}}$ is undefined, which completes the proof.
\end{proof}

\begin{theorem}\label{thm:efftodel}
  $\eff$ is macro-expressible in $\del$.
\end{theorem}

\begin{proof}
  The translation in Figure~\ref{fig:efftodel} preserves the semantics by
  Proposition~\ref{prop:eff_to_del_preservation_of_sem}, so all we have to do is to check the
  syntactic conditions.
  First, this translation maps $\eff$ programs to $\del$ programs since it
  introduces no continuation labels.
  The translation does not alter $\mam$ constructors, so it acts homomorphically
  on them.
  Also, Figure~\ref{fig:efftodel} shows clearly that the program constructs
  peculiar to $\eff$ are expressed by syntactic abstractions of $\del$.
  Therefore, the translation is a macro-translation.
\end{proof}

From Theorem~\ref{thm:efftodel} together with Theorem~\ref{thm:strong-macro} and
the transitivity of macro-expressibility, we get the following result.

\begin{theorem}\label{thm:efftoac}
$\eff$ is macro-expressible in $\ac$.
\end{theorem}

\subsection{Macro-translation
  from \texorpdfstring{$\del$}{DELone}
  to \texorpdfstring{$\eff$}{EFFone}
}
We can also show the macro-expressibility of $\del$ in $\eff$, using Forster et
al.'s translation for the multi-shot $\mathit{shift}_0/\mathit{dollar}$.

\begin{theorem}\label{thm:deltoeff}
  Assume that $\op{shift0}$ is contained in the set of operations
  $\syntacticset{O}$.
  Then, the following translation is a macro-translation from $\del$ to $\eff$.%
  \[
  \begin{array}{rll}
    \mt{\shift{k}{M}} & \defeq & \opcall{shift0}{\thunk{\abs{k}{\mt{M}}}} \\
    \mt{\throw{V}{W}} & \defeq & \throw{\mt{V}}{\mt{W}} \\
    \mt{\dollar{M}{x}{N}} & \defeq & \handle{\handler{x}{\mt{N}}{\opcall{shift0}{\app{p}{k}} \mapsto \app{\force{p}}{k}}}{\mt{M}}
  \end{array}
  \]  
\end{theorem}

We omit the proof, since the argument is quite similar to that in the previous
subsection.
See \version{the appendix of the full
  version}{Appendix~\ref{app:DELoneToEFFone}} for the full proof.

%% file: sec5-nonexistence.tex
\section{Inexpressibility results}\label{sec:nonexistent}

Now that asymmetric coroutines can macro-express one-shot effect handlers, it is
natural to ask whether the converse direction also holds.
In this section, we prove that the converse direction does not hold, which is
witnessed by \emph{reference cells} (Section~\ref{sec:nonexistent:reftodel}).
We then show that reference cells, in turn, cannot macro-express one-shot
delimited-control operators (Section~\ref{sec:nonexistent:deltoref}), which
completes the comparison of the four calculi studied in this paper.

First, we introduce $\reflang$, a calculus with ML-style reference cells.
Its syntax is defined in Figure~\ref{fig:ref_syntax}; the definition of frames
and contexts is omitted since they are the same as those of $\mam$.
\begin{figure}[tb]
  \begin{minipage}[t]{0.40\columnwidth}
    \begin{tabular}[t]{@{}l@{\hspace{0.5em}}l@{\hspace{1em}}l@{}}
      $V,W$ & $::= \ldots$ & \textsf{value} \\
            & \textbar\ $l\in\syntacticset{L}_{\mathbf{R}}$ & reference cell
    \end{tabular}
  \end{minipage}%
    \begin{minipage}[t]{0.38\columnwidth}
    \begin{tabular}[t]{@{}l@{\hspace{0.5em}}l@{\hspace{1em}}l@{}}
      $M,N$ & $::= \ldots$ & \textsf{computation} \\
            & \textbar\ $\create{V}$ & create
    \end{tabular}%
  \end{minipage}%
    \begin{minipage}[t]{0.2\columnwidth}
    \begin{tabular}[t]{@{}l@{\hspace{1em}}l@{}}
      \textbar\ $\set{V}{W}$ & set \\
      \textbar\ $\get{V}$ & get
    \end{tabular}%
  \end{minipage}%
\caption{Syntax of $\reflang$}
\label{fig:ref_syntax}
\end{figure}
$\reflang$ programs are the computations with no reference cells.
Similarly to the previous calculi, we use stores and configurations: a store is
a partial function that maps reference cells to \textit{values}, and a
configuration is a pair of a $\reflang$ computation and a store.
Figure~\ref{fig:ref_beta} presents the semantics of $\reflang$.

\begin{figure}[tb]
  \begin{minipage}[c]{0.60\columnwidth}
    \begin{tabular}{cl}
    $(\mam)$ & $\inferrule{\redbetaM{M}{M'}}{\redbetaRp{M}{\theta}{M'}{\theta}}$ \\
    
    $(\mathrm{create})$ &
    $\inferrule
    {
      l \notin \dom{\theta}
    }
    {
      \redbetaRp{\create{V}}{\theta}{\return{l}}{\theta[l := V]}
    }$ \\
    
    $(\mathrm{set})$ &
    $\inferrule
    {
      l \in \dom{\theta}
    }
    {
      \redbetaRp{\set{l}{V}}{\theta}{\return{\unit}}{\theta[l := V]}
    }$ \\

    $(\mathrm{get})$ &
    $\inferrule
    {
      \theta(l) = V
    }
    {
      \redbetaRp{\get{l}}{\theta}{\return{V}}{\theta}}$
  \end{tabular}
  \end{minipage}%
  \begin{minipage}[c]{0.35\columnwidth}
    \fbox{$\inferrule
      {
        \redbetaRp{M}{\theta}{M'}{\theta'}
      }
      {
        \redRp{\plug{C}{M}}{\theta}{\plug{C}{M'}}{\theta'}
      }$}
  \end{minipage}%
  \caption{Semantics of $\reflang$}
  \label{fig:ref_beta}
\end{figure}

\subsection{Non-macro-expressibility of \texorpdfstring{\(\reflang\)}{REF} in
  terms of \texorpdfstring{\(\del\)}{DELone}}
\label{sec:nonexistent:reftodel}

First, we prove that $\del$ cannot weakly macro-express $\reflang$:

\begin{theorem}\label{thm:reftoeff_nonexist}
  There is no weak macro-translation from $\reflang$ to $\del$.
\end{theorem}

To prove the theorem, we need the following lemma:
\begin{lemma}\label{lem:reset-free-decomp}
  If \(\config{\seq{x}{M}{N}}{\theta} \arrD^* \config{\return{V}}{\theta''}\),
  then there exist a value \(V'\) and a store \(\theta'\) such that
  \(\config{M}{\theta} \arrD^* \config{\return{V'}}{\theta'}\) and
  \(\config{N[V'/x]}{\theta'} \arrD^* \config{\return{V}}{\theta''}\).
\end{lemma}

\begin{proof}
  By induction on the length of the reduction sequence.
  Since the frame \(\seq{x}{\hole}{N}\) contains no dollar frame, redexes for
  the \((\mathrm{shift})\) rule can occur only within \(M\).
  The frame is preserved until \(M\) becomes \(\return{V'}\), as the only rule
  that can consume the frame is \((F)\).
\end{proof}

\begin{proof}[Proof of Theorem~\ref{thm:reftoeff_nonexist}]
  Let
  \( \Omega \defeq
  \app{(\abs{x}{\app{\force{x}}{x}})}{\thunk{\abs{x}{\app{\force{x}}{x}}}} \),
  and define
  \[
    \Delta_{\mathrm{A}, \mathrm{B}}(x, y) \defeq
    \begin{pmatrix*}[l]
      \mathbf{case}\;(x, y)\;\mathbf{of}\; \{\\
      \quad (\inj{A}{\underscore}, \inj{B}{\underscore}) \mapsto \return{\unit} \\
      \quad \underscore \mapsto \Omega \}      
    \end{pmatrix*}.
  \]
  Consider the \(\reflang\) program
  \[
    M_{\neq} \defeq 
    \begin{pmatrix*}[l]
      \letin{r}{\create{\inj{A}{\unit}}}
      {\\\letin{i}{\get{r}}
      {\\\letin{\underscore}{\set{r}{\inj{B}{\unit}}}
      {\\\letin{j}{\get{r}}
      {\\\Delta_{\mathrm{A}, \mathrm{B}}(i, j)}}}}
    \end{pmatrix*}.
  \]
  In \(\reflang\), the first \(\mathbf{get}\) returns \(\inj{A}{\unit}\) and the
  second \(\mathbf{get}\) returns \(\inj{B}{\unit}\), and therefore
  \(\eval{\reflang}{M_{\neq}}\) is defined.

  Suppose that there exists a weak macro-translation \(\mtempty\) from
  \(\reflang\) to \(\del\), and let \(\ccreate\), \(\cget\), \(\cset\) be the
  syntactic abstractions corresponding to \textbf{create}, \textbf{get}, and
  \textbf{set}, respectively.
  Then, since \(\mtempty\) acts homomorphically on \(\mam\) constructors, the
  translation \(\mt{M_{\neq}}\) has the shape
  \[
    \begin{pmatrix*}[l]
      \letin{r}{\ccreate\left[\inj{A}{\unit}\right]}
      {\\\letin{i}{\cget\left[r\right]}
      {\\\letin{\underscore}{\cset\left[r, \inj{B}{\unit}\right]}
      {\\\letin{j}{\cget[r]}
      {\\\Delta_{\mathrm{A}, \mathrm{B}}(i, j)}}}}
    \end{pmatrix*},
  \]
  and contains no continuation label.
  Since \(\eval{\reflang}{M_{\neq}}\) is defined, so is
  \(\eval{\del}{\mt{M_{\neq}}}\):
  there exists a store \(\theta\) such that
  \[
    \config{\mt{M_{\neq}}}{\emptyset} \arrD^* \config{\return{\unit}}{\theta}.
  \]

  Applying Lemma~\ref{lem:reset-free-decomp} repeatedly to this reduction
  sequence (note that there is no dollar frame enclosing \(\mt{M_{\neq}}\))
  yields values \(X, V_1, W, V_2\) and stores
  \(\theta_0, \theta_1, \theta_2, \theta_3\) such that
  \begin{align}
    \config{\ccreate\left[\inj{A}{\unit}\right]}{\emptyset} &\arrD^* \config{\return{X}}{\theta_0}, \label{eq:1}\\
    \config{\cget\left[X\right]}{\theta_0} &\arrD^* \config{\return{V_1}}{\theta_1}, \label{eq:2}\\
    \config{\cset\left[X, \inj{B}{\unit}\right]}{\theta_1} &\arrD^* \config{\return{W}}{\theta_2}, \label{eq:3}\\
    \config{\cget[X]}{\theta_2} &\arrD^* \config{\return{V_2}}{\theta_3}, \label{eq:4}
  \end{align}
  and moreover the evaluation of \(\Delta_{\mathrm{A}, \mathrm{B}}(V_1, V_2)\)
  terminates successfully.
  Thus, \(V_1 \equiv \inj{A}{V_1'}\) and \(V_2 \equiv \inj{B}{V_2'}\) for some
  values \(V_1'\) and \(V_2'\).

  As Figure~\ref{fig:del_beta} shows, a reduction step of \(\del\) may store a
  continuation under a fresh label or invalidate an existing one.
  Unlike the \((\mathrm{set})\) rule of \(\reflang\), it never replaces the
  content of a valid entry.
  Hence, when restricted to the labels \emph{reachable} from
  \(\cget\left[X\right]\), \(\theta_2\) differs from \(\theta_0\) only in that
  some of the entries are invalidated.
  Since the reduction sequence in \eqref{eq:4} results in a value \(V_2\), this
  implies that no step of that sequence depends on the difference between
  \(\theta_0\) and \(\theta_2\).
  Thus, from the determinism of the reduction, we obtain
  \(V_1 = V_2\).\footnote{To be precise, \(V_1\) is identical to \(V_2\) modulo
    freshly generated labels.
    Rigorously proving this part requires cumbersome analysis of stores; see
    \version{the appendix of the full
      version}{Appendix~\ref{app:subsec:nonexistent:reftodel}} for the details.
  } However, this contradicts \(V_1 \equiv \inj{A}{V_1'}\) and
  \(V_2 \equiv \inj{B}{V_2'}\).
  Therefore, there is no weak macro-translation from \(\reflang\) to \(\del\).
\end{proof}

We remark that this proof does not depend on the one-shotness of continuations.
Thus, by a similar argument, we can prove that multi-shot delimited-control
operators (thus multi-shot effect handlers) also cannot macro-express reference
cells.
This proves, in the setting of macro-expressibility, the conjecture of Kammar
and Pretnar that reference cells cannot be implemented by effect handlers
alone~\cite{DBLP:journals/jfp/KammarP17}.

On the other hand, $\ac$ can macro-express $\reflang$, using the encoding
introduced in Section~\ref{sec:DELoneToAC:refined-translation}:
\[
  \begin{array}{rll}
    \mt{\create{V}} & \defeq & \create{\refcell{\mt{V}}} \\
    \mt{\set{V}{W}} & \defeq & \resume{\mt{V}}{\paren{\inj{Set}{\mt{W}}}} \\
    \mt{\get{V}} & \defeq & \resume{\mt{V}}{\paren{\inj{Get}{\unit}}} %
  \end{array}
\]
The simulation relation, \(\simu{\config{M}{\theta}}{\config{N}{\tau}}\), holds
if and only if there exists a partial function
\(\eta : \syntacticset{L}_{\mathbf R} \rightharpoonup \syntacticset{L}_{\ac}\)
such that:
\begin{enumerate}
\item $N \equiv \mte{M}{\eta}$
\item $\eta$ is injective
\item $\eta(\dom{\theta}) \subseteq \dom{\tau}$
\item\label{item:3} For any $l \in \dom{\theta}$, if $\theta(l) = V$, then\footnotemark
  \[
    \tau(\eta(l)) = \thunk{\abs{x}{\letin{y}{\return{x}}{\app{\app{\force{\var{loop}}}{\mte{V}{\eta}}}{y}}}}.
  \]
  \footnotetext{\(\var{loop}\) is defined in Figure~\ref{fig:deltoac}.}
\end{enumerate}
Note that (\ref{item:3}) specifies the correspondence between source reference
cells and target coroutines.

\begin{theorem}\label{thm:reftoac}
  $\reflang$ is strongly macro-expressible in $\ac$.
\end{theorem}
\begin{proof}
  We can show that \(M \mapsto \mt{M}\) is a strong macro-translation from
  $\reflang$ to $\ac$.
  The proof proceeds similarly to other correctness proofs (see \version{the
    full version}{Appendix~\ref{subsec:reftoac}} for the details).
\end{proof}

\begin{theorem}\label{thm:nonexistent}
  $\del$ cannot weakly macro-express $\ac$.
  Moreover, $\eff$ also cannot weakly macro-express $\ac$.
\end{theorem}
\begin{proof}
  Suppose that a weak macro-translation from $\ac$ to $\del$ exists.
  Then, by composing it with a (strong) macro-translation from \(\reflang\) to
  \(\ac\) (Theorem~\ref{thm:reftoac}), we obtain a weak macro-translation from
  $\reflang$ to $\del$, which contradicts Theorem~\ref{thm:reftoeff_nonexist}.
  The latter statement follows from Theorem~\ref{thm:efftodel}.
\end{proof}

One may find Theorem~\ref{thm:nonexistent} counter-intuitive, since effect
handlers (and delimited-control operators) are considered a universal tool for
expressing various computational effects, such as coroutines.
This gap arises from the strict condition which macro-expressibility imposes on
translations: a macro-translation should transform programs \emph{locally} and
\emph{homomorphically}.
This prevents a translation from introducing a \emph{global} handler---thus
effect handlers cannot macro-express global effects such as reference cells.

\subsection{Non-macro-expressibility of \(\del\) in terms of \(\reflang\)}
\label{sec:nonexistent:deltoref}

The results obtained so far are that while \(\del\) and \(\eff\) are
equi-expressive, \(\ac\) is strictly more expressive than both; the difference
is witnessed by the non-macro-expressibility \(\reflang \nrightarrow \del\).
In this section, to complete the picture, we consider the converse direction and
prove the following theorem:

\begin{theorem}\label{thm:no-del-to-ref}
  There is no weak macro-translation from \(\del\) to \(\reflang\).
\end{theorem}

Since macro-expressibility is transitive, this theorem implies the following
result.

\begin{theorem}
  There is no weak macro-translation from \(\eff\) to \(\reflang\), and there is
  no weak macro-translation from \(\ac\) to \(\reflang\).
\end{theorem}

The gist of the proof of Theorem~\ref{thm:no-del-to-ref} is that \(\reflang\)
lacks \emph{control effects}.
The proof proceeds as follows:
\begin{enumerate}
\item\label{item:1} In \(\reflang\), for any evaluation context \(\context{E}\), the following
  computation cannot terminate successfully:
  \[
    \plug{E}{\letin{\underscore}{M}{\Omega}},
  \]
  where \(\Omega\) is the diverging computation
  \((\abs{x}{\app{\force{x}}{x}}) \thunk{\abs{x}{\app{\force{x}}{x}}}\).
  To successfully terminate the computation, we need to escape from the
  \textbf{let} frame when evaluating \(M\), but \(\reflang\) does not have any
  facility for this, in contrast to \(\del\).
  
\item\label{item:2} Therefore, for an arbitrary \(\reflang\) context
  \(\context{C}\) that expects a value and arbitrary computations \(M_i\),
  \(\reflang\) cannot distinguish the computations
  \(\context{C}[\thunk{\letin{\underscore}{M_i}{\Omega}}]\) from one another.
  
\item On the other hand, in \(\del\), the computations of the form
  \[
    \dollar{\force{\thunk{\letin{\underscore}{\shift{k}{\return{(\inj{L_{\mathnormal{i}}}{\unit})}}}{\Omega}}}}{x}{\return{x}},
  \]
  parameterized by variant labels \(\mathrm{L}_i\), can be distinguished.
  
\item Thus, if there exists a weak macro-translation, \(\reflang\) must be able
  to distinguish the computations of the form
  \[
    \cdollar{x}\left[{\force{\thunk{\letin{\underscore}{\cshift{k}\left[
                {\return{(\inj{\mathit{L_i}}{\unit})}}\right]}{\Omega}}}},
      {\return{x}}\right],
  \]
  where \(\cdollar{x}\) is the syntactic abstraction corresponding to a source
  dollar binding \(x\) in its additional part, and \(\cshift{k}\) is the one
  corresponding to a source shift0 binding \(k\) in its body.
  However, this contradicts (\ref{item:2}).
\end{enumerate}

To avoid clutter, we only sketch the proof of (\ref{item:2}) and omit the
proofs of the other statements, which are straightforward to check; see
\version{the full version}{Appendix~\ref{app:sec:nonexistent:deltoref}} for the
details.

\begin{proposition}\label{prop:obs-equiv}
  Let \(\context{L} \defeq (\letin{\underscore}{\hole}{\Omega})\).
  For any \(\reflang\) computations \(M\) and \(N\), the thunks
  \(\thunk{\plug{L}{M}}\) and \(\thunk{\plug{L}{N}}\) are observationally
  equivalent; that is, for every context \(\context{C}\) that expects a value,
  \(\eval{\reflang}{\plug{C}{\thunk{\plug{L}{M}}}}\) is defined if and only if
  \(\eval{\reflang}{\plug{C}{\thunk{\plug{L}{N}}}}\) is defined.
\end{proposition}

\begin{proof}
  It suffices to show the ``only if'' direction, since the converse direction
  follows by symmetry.
  Suppose that \(\eval{\reflang}{\plug{C}{\thunk{\plug{L}{M}}}}\) is defined and
  let \(\sim\) be the congruence relation on \(\reflang\) configurations
  generated by \(\thunk{\plug{L}{M'}} \sim \thunk{\plug{L}{N'}}\) for any
  \(\reflang\) computations \(M'\) and \(N'\).
  Then we can prove that \(\sim\) forms a lock-step simulation, except for a
  case where a thunk of the form \(\thunk{\plug{L}{M'}}\) is forced; such a step
  results in a configuration \(\config{\plug{E}{\plug{L}{M'}}}{\theta}\) for
  some evaluation context \(\context{E}\) and store \(\theta\), whose evaluation
  never terminates successfully by (\ref{item:1}), but this contradicts the
  assumption that \(\eval{\reflang}{\plug{C}{\thunk{\plug{L}{M}}}}\) is defined.
  Therefore, by replacing each occurrence of \(\thunk{\plug{L}{M'}}\) with the
  corresponding \(\thunk{\plug{L}{N'}}\) in each configuration in the reduction
  sequence of \(\plug{C}{\thunk{\plug{L}{M}}}\), we obtain a reduction sequence
  of \(\plug{C}{\thunk{\plug{L}{N}}}\) that terminates successfully.
  Thus \(\eval{\reflang}{\plug{C}{\thunk{\plug{L}{N}}}}\) is defined.
\end{proof}

%% file: sec6-related-work.tex
\section{Related work}\label{sec:related-work}

\paragraph{Expressive power of control operators}

As a finer notion than computability, Felleisen proposed the notion of
\emph{macro-expressibility} to rigorously compare the expressive power of two
programming languages when one is an extension of the
other~\cite{felleisen1991expressive}.
This notion was later adjusted to compare two languages when both are extensions
of a common language.
Riecke and Thielecke compared first-class continuations (\emph{call/cc}) and
exceptions in a typed language and proved that they cannot macro-express each
other~\cite{DBLP:conf/icalp/RieckeT99}.
Shan relaxed macro-expressibility into \emph{top-macro-expressibility} and
compared \emph{shift/reset} with dynamic delimited-control operators such as
\emph{control/prompt}~\cite{DBLP:journals/lisp/Shan07}.

More recently, Forster et al.\ compared three abstractions for user-defined
effects: effect handlers, monads, and delimited-control operators (in
particular, \emph{shift0/dollar})~\cite{DBLP:journals/pacmpl/0002KLP17,
  forster2019expressive}.
They proved that, in an untyped setting, the three are all equi-expressive.
They also compared the expressive power in a typed setting, and proved that
typed monads can macro-express typed delimited-control operators but not typed
effect handlers.
Following Forster et al.'s approach, Pir\'{o}g et al.\ established typed
equivalence between effect handlers and delimited-control operators in a
polymorphic type system~\cite{DBLP:conf/rta/PirogPS19}.
Ikemori et al.\ extended this result to the equivalence of labeled effect
handlers and labeled delimited-control
operators~\cite{DBLP:conf/ppdp/IkemoriCM23}.

However, all of these previous studies concentrated on multi-shot control
operators, and to our knowledge, the present paper is the first study to conduct
a systematic comparison of the relative expressive power of one-shot control
operators.
James and Sabry claimed that the one-shot yield operator can be seen as a
one-shot variant of delimited-control operators~\cite{yield2011}, but they did
not provide a formal justification.
Meanwhile, earlier work by de Moura and Ierusalimschy examined the expressive
power of symmetric coroutines, asymmetric coroutines, and one-shot
subcontinuations in the presence of mutable state~\cite{moura2009revisiting}.
In contrast, we study the extensions of a base calculus without mutable states,
which isolates the expressive power of each control operator itself: it is
exactly the ability to macro-express global effects such as reference cells that
separates coroutines from delimited-control operators, since macro-translations
into \(\del\) are local and hence cannot introduce a dollar enclosing the entire
program (Section~\ref{sec:nonexistent:reftodel}).

\paragraph{One-shot continuations and control operators}

Historically, the notion of one-shot continuations was formalized in the Scheme
community by Bruggeman et al.\ \cite{bruggeman1996oneshot}, who introduced them
as an efficient yet less expressive variant of general (multi-shot) first-class
continuations, motivated by the desire to avoid the stack-copying overhead upon invocation.
Their implementation of one-shot continuations, named \emph{call/1cc}, is
available in Chez Scheme~\cite{DBLP:conf/icfp/Dybvig06}.

Recently, one-shot continuations have been attracting attention as a means of
integrating control abstractions into mainstream programming languages.
Since version 5.0, OCaml has implemented one-shot effect handlers because they
allow smooth integration into OCaml's runtime
system~\cite{DBLP:conf/pldi/Sivaramakrishnan21}.
Phipps-Costin et al.\ proposed WasmFX, which extends WebAssembly with one-shot
effect handlers, enabling type-safe non-local control features such as
async/await~\cite{phipps2023continuing}.
Although WasmFX gives types to continuations, it enforces one-shotness
dynamically by aborting execution on a second invocation; OCaml likewise raises
a runtime error.
Our calculi model this discipline directly, and we give a first formal account
of what such one-shot control abstractions can and cannot macro-express.

The one-shot usage of continuations has also been studied from a type-theoretic
and semantic perspective.
Van Rooij and Krebbers presented Affect, an affine type-and-effect system equipped
with effect handlers and reference cells~\cite{DBLP:journals/pacmpl/RooijK25}.
In Affect, one-shot and multi-shot continuations coexist, and their usage is
statically tracked by the type system, in contrast to our untyped setting, where
one-shotness is enforced dynamically.
From a semantic perspective, Berdine et al.\ demonstrated that various control
behaviors, such as exceptions, labeled jumping, and coroutines, can be
CPS-transformed into a target calculus with linear
types~\cite{berdine2002linear}.
Notably, they pointed out that, in these transforms, continuations are
\emph{used} linearly.
Based on this observation, Hasegawa~\cite{DBLP:conf/flops/Hasegawa02} showed
that Berdine et al.'s linear CPS is a special instance of the monadic transforms
of the computational lambda calculus~\cite{DBLP:conf/lics/Moggi89} into the
linear lambda calculus, and proved its full completeness.
It would be interesting to analyze the expressive power of one-shot control
operators from this semantic perspective.

%% file: sec7-conclusion.tex
\section{Conclusion}\label{sec:conclusion}

\begin{figure}[t]
  \centering
  \begin{tikzcd}
    & & \ac & & \\
    & &     & & \\
    \eff \arrow[rruu] \arrow[rr, shift left=1] & & \del \arrow[ll, shift left=1] \arrow[uu]& & \reflang \arrow[lluu]
  \end{tikzcd}
  \caption{Macro-expressibility obtained in this work.
    An arrow indicates the existence of a strong macro-translation in that
    direction, while the absence of an arrow indicates the non-existence of a
    weak macro-translation.}
  \label{fig:relations-conclusion}
\end{figure}
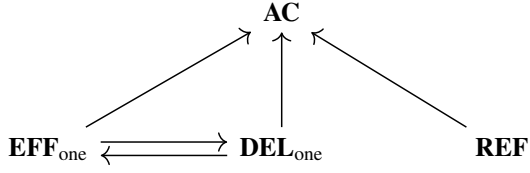

This study investigated the expressive power of one-shot control operators,
summarized in Figure~\ref{fig:relations-conclusion}, resulting in the following
two key findings:
\begin{enumerate}
\item One-shot delimited-control operators and one-shot effect handlers
  can be macro-expressed by asymmetric coroutines.
\item The converse direction does not hold, and this non-macro-expressibility is
  witnessed by reference cells.
\end{enumerate}

Previously, the former result was considered trivial; however, we found a gap in
the literature and fixed the problem by introducing reference cells, which can
be simulated using coroutines, and using them to manage the validity of
continuations.
As this result demonstrates, establishing such a connection requires rigorous
analysis.

The latter result depends on the non-macro-expressibility of reference cells in
terms of effect handlers.
Kammar and Pretnar conjectured that effect handlers alone cannot implement
reference cells, while establishing that they can macro-express dynamically
scoped state~\cite{DBLP:journals/jfp/KammarP17}.
We proved their conjecture in the macro-expressibility setting, and the
proof exploited the \emph{locality} of macro-translations.

This non-macro-expressibility result suggests that macro-expressibility is too
strict as a measure for comparing expressive power.
We conjecture that, under a suitable relaxation of the locality of
macro-translations, $\eff$ can macro-express $\reflang$ and $\ac$.
One possible relaxation is to allow macro-translations to insert a fixed context
at the root of each translated term, as Shan did for comparing the expressive
power of delimited-control operators~\cite{DBLP:journals/lisp/Shan07}.
We leave the exploration of this idea for future work.

Comparing the expressive power in a typed setting is an important direction to
explore.
For effect handlers and delimited-control operators, there are a number of
proposals for their type systems, including linear/affine
ones~\cite{DBLP:journals/pacmpl/TangHLM24,DBLP:journals/pacmpl/RooijK25}.
Such type systems are particularly interesting here, since they would enforce
one-shotness statically, whereas our calculi detect its violation dynamically.
As Forster et al.'s results indicate, moving to a typed setting can change the
relationships obtained in the untyped setting~\cite{forster2019expressive};
whether the macro-expressibility results in
Figure~\ref{fig:relations-conclusion} carry over remains open.

Another important direction is to design type systems for one-shot control
operators, in the spirit of Berdine et al.'s linearly used
continuations~\cite{berdine2002linear}, whose examples already include a
restricted form of coroutines.
The asymmetric coroutines considered here, however, are much more general, and
the design space for their type systems is correspondingly broad---thus
requiring a systematic approach.
An instance of such an approach is Anton and Thiemann's type system for
coroutines~\cite{DBLP:conf/aplas/AntonT10}, which is derived via a translation
to a calculus with the \emph{shift/reset} operators and reference cells.
Whether a type system can be derived in the same spirit from our
macro-translations remains to be investigated.

\subsection*{Acknowledgments}
The authors were supported by JSPS KAKENHI Grant Number JP23K24819, Japan.
The first author is supported by JST BOOST, Grant Number JPMJBS2414, Japan.

%% file: appA-deltoac.tex
\section{Supplementary proofs for Section~\ref{sec:DELoneToAC}}
\label{app:sec:deltoac}

\subsection{Simulation relation}
\label{app:sec:deltoac:simulation}

\subsubsection{Well-formedness}\label{app:sec:well-formedness}

\begin{definition}
  For a \(\del\) computation \(M\), let \(\lb{M}\) be the set of continuation labels occurring in \(M\).
  For a \(\del\) store \(\theta\), define
  \[
    \lb{\theta} \defeq \bigcup \{\lb{\theta(l)} \mid l \in \dom{\theta}, \theta(l) \neq \nil \}.
  \]
  For a \(\del\) configuration \(C = \config{M}{\theta}\), define
  \[
    \lb{C} \defeq \lb{M} \cup \lb{\theta}.
  \]
  For \(\ac\) computations, stores, and configurations, define these functions accordingly.
\end{definition}

\begin{definition}
  A \(\del\) configuration \(\config{M}{\theta}\) is \emph{well-formed}, written \(\WF_{\del}(\config{M}{\theta})\), if and only if
  \[
    \lb{\config{M}{\theta}} \subseteq \dom{\theta}.
  \]
\end{definition}

\begin{proposition}[preservation of \(\del\) well-formedness]\label{app:prop:del-wellformedness}
  If \(\WF_{\del}(\config{M}{\theta})\) and \(\redD{\config{M}{\theta}}{\config{M'}{\theta'}}\), then \(\WF_{\del}(\config{M'}{\theta'})\).
\end{proposition}
\begin{proof}
  We prove this by case analysis on the \(\del\) beta-reduction rule used in the source step.  The surrounding context
  itself is unchanged by the reduction, so it suffices to check that the reduct \(M'\) contains only labels already
  defined in the updated store.

  In the \(\mam\) cases, no continuation label is generated and the store is unchanged.  Every continuation label
  occurring in the reduct already occurs in the redex, hence is contained in \(\dom{\theta}\) by the premise.

  In the \((\mathrm{ret})\) case,
  \[
    \config{\dollar{\return{V}}{x}{N}}{\theta}
    \arrD
    \config{N[V/x]}{\theta},
  \]
  and every label in \(N[V/x]\) already occurs in \(N\) or in \(V\).  Hence it already occurs in the premise and belongs
  to \(\dom{\theta}\).

  In the \((\mathrm{shift})\) case,
  \[
    \config{\dollar{\plug{H}{\shift{k}{M}}}{x}{N}}{\theta}
    \arrD
    \config{M[l/k]}{\theta[l:=\abs{y}{\dollar{\plug{H}{\return{y}}}{x}{N}}]},
  \]
  where \(l\) is fresh.  All old labels remain in \(\dom{\theta}\), and the only new label that can occur in either the
  current computation or the new store entry is \(l\), which is added to the store domain by the reduction.

  In the \((\mathrm{throw})\) case,
  \[
    \config{\throw{l}{V}}{\theta}
    \arrD
    \config{\dollar{\plug{H}{\return{V}}}{x}{N}}{\theta[l:=\nil]},
  \]
  with \(\theta(l)=\abs{y}{\dollar{\plug{H}{\return{y}}}{x}{N}}\).  The store domain is unchanged, and every label in
  \(H\), \(V\), or \(N\) already occurs either in the current computation or in the store entry for \(\theta(l)\) of the
  premise.  Hence all labels in the reduct lie in \(\dom{\theta[l:=\nil]}\).
  
  The \((\mathrm{fail})\) case does not occur.
\end{proof}

\begin{definition}
  For an \(\ac\) term \(E\), let \(\activelabels{E}\) be the set of active labels, i.e., labels occurring in a
  subterm of the form \(\labeledc{l}{M}\).  An \(\ac\) configuration \(\config{N}{\tau}\) is \emph{well-formed}, written
  \(\WF_{\ac}(\config{N}{\tau})\), if and only if
  \begin{enumerate}
  \item \(\wellformed{N}\) holds;
  \item \(\activelabels{V}=\emptyset\) for every value \(V\in\im{\tau}\);
  \item for every \(l\in\activelabels{N}\), \(\tau(l)=\nil\).
  \end{enumerate}
  Here, \(\wellformed{N}\) is a predicate inductively defined by the following rules
  \begin{mathpar}
    \inferrule
    {\activelabels{V} \cup \activelabels{M} = \emptyset}
    {\wellformed{\pcase{V}{x_1}{x_2}{M}}}
  \and
    \inferrule
    {\activelabels{V} \cup (\cup_i \activelabels{M_i}) = \emptyset}
    {\wellformed{\scase{V}{L_i}{x_i}{M_i}}}
  \\
    \inferrule
    {\activelabels{V} = \emptyset}
    {\wellformed{\force{V}}}
  \and
    \inferrule
    {\activelabels{V} = \emptyset}
    {\wellformed{\return{V}}}
  \and
    \inferrule
    {\wellformed{M} \\ \activelabels{N} = \emptyset}
    {\wellformed{\seq{x}{M}{N}}}
  \\
    \inferrule
    {\activelabels{M} = \emptyset}
    {\wellformed{\abs{x}{M}}}
  \and
    \inferrule
    {\wellformed{M} \\ \activelabels{N} = \emptyset}
    {\wellformed{\app{M}{V}}}
  \and
    \inferrule
    {\activelabels{M} \cup \activelabels{N} = \emptyset}
    {\wellformed{\cpair{M}{N}}}
  \\
    \inferrule
    {\wellformed{M}}
    {\wellformed{\prj{i}{M}}}
  \and
    \inferrule
    {\wellformed{M} \\ l \notin \activelabels{M}}
    {\wellformed{\labeledc{l}{M}}}
  \and
    \inferrule
    {\activelabels{V} = \emptyset}
    {\wellformed{\create{V}}}
  \and
  \inferrule
    {\activelabels{V} \cup \activelabels{W} = \emptyset}
    {\wellformed{\resume{V}{W}}}
  \and
  \inferrule
    {\activelabels{V} = \emptyset}
    {\wellformed{\yield{V}}}
  \end{mathpar}
\end{definition}

\begin{proposition}[preservation of \(\ac\) well-formedness]\label{prop:ac-wellformedness}
  Suppose that $\config{\plug{C}{M}}{\theta}$ is well-formed and $\redbetaAC{\config{M}{\theta}}{\config{M'}{\theta'}}$.
  Then, $\config{\plug{C}{M'}}{\theta'}$ is also well-formed.
\end{proposition}
\begin{proof}
  We prove the proposition by case analysis on $\redbetaAC{\config{M}{\theta}}{\config{M'}{\theta'}}$.

  \paragraph{Case \((\mam)\).}
  We prove one representative case.
  Suppose that $\redbetaAC{\config{M}{\theta}}{\config{M'}{\theta'}}$ is
  \[
    \redbetaAC{\config{\seq{x}{\return{V}}{M}}{\theta}}{\config{M[V/x]}{\theta'}}.
  \]
  Since $\seq{x}{\return{V}}{M}$ is well-formed, $V$ and $M$ do not contain any active labels, implying $\activelabels{M[V/x]} = \emptyset$.
  By straightforward induction on $\context{C}$, we conclude that $\config{\plug{C}{M[V/x]}}{\theta}$ is well-formed.  

  \paragraph{Case \((\mathrm{create})\).}
  Suppose that $\redbetaAC{\config{M}{\theta}}{\config{M'}{\theta'}}$ is
  \[
    \redbetaACp{\create{V}}{\theta}{\return{l}}{\theta[l := V]},
  \]
  where \(l\) is a fresh label.
  Since $l$ is fresh, $l \notin \activelabels{\context{C}}$.
  The well-formedness of $\create{V}$ ensures $\activelabels{V} = \emptyset$.
  By induction on $\context{C}$, we obtain that $\config{\plug{C}{\return{l}}}{{\theta[l := V]}}$ is well-formed.

  \paragraph{Case \((\mathrm{resume})\).}
  Suppose that $\redbetaAC{\config{M}{\theta}}{\config{M'}{\theta'}}$ is
  \[
    \inferrule
    {
      l \in \dom{\theta} \\ \theta(l) \neq \nil
    }
    {
      \redbetaACp{\resume{l}{V}}{\theta}{\labeledc{l}{(\app{\force{\theta(l)}}{V})}}{\theta[l := \nil]}
    }.
  \]
  If $l$ is active in a well-formed configuration $\config{\plug{C}{\resume{l}{V}}}{\theta}$, it follows that $\theta(l) = \nil$, but this is a contradiction.
  Therefore, $l$ occurs only once as an active label in $\config{\plug{C}{\labeledc{l}{(\app{\force{\theta(l)}}{V})}}}{\theta[l := \nil]}$ and is mapped to $\nil$ in $\theta[l := \nil]$.
  From the well-formedness of $\config{\plug{C}{\resume{l}{V}}}{\theta}$, we have that $\activelabels{\theta(l)} = \emptyset$ and $\activelabels{V} = \emptyset$.
  Then, it is easy to check that this configuration is well-formed by induction on $\context{C}$.

  \paragraph{Case \((\mathrm{fail})\).}
  This case does not occur, as $\config{\plug{C}{M}}{\theta}$ reduces to $\config{\plug{C}{M'}}{\theta'}$, not $\bot$.

  \paragraph{Case \((\mathrm{ret})\).}
  This case is immediate from the definition of well-formedness.

  \paragraph{Case \((\mathrm{yield})\).}

  Suppose $\redbetaAC{\config{M}{\theta}}{\config{M'}{\theta'}}$ is
  \[
    \redbetaACp{\labeledc{l}{\context{H}[\yield{V}]}}{\theta}{\return{V}}{\theta[l := \thunk{\abs{y}{\context{H}[\return{y}]}}]}.
  \]
  The label $l$ is no longer active in
  $\config{\plug{C}{\return{V}}}{\theta[l :=
    \thunk{\abs{y}{\context{H}[\return{y}]}}]}$, while other active labels
  remain distinct and are mapped to $\nil$ in
  $\theta[l := \thunk{\abs{y}{\context{H}[\return{y}]}}]$.
  Moreover, from the well-formedness of $\plug{H}{\yield{V}}$, $\context{H}$ and
  $V$ do not contain any active labels.
  Therefore, by induction on $\context{C}$, we see that
  $\config{\plug{C}{\return{V}}}{\theta[l :=
    \thunk{\abs{y}{\context{H}[\return{y}]}}]}$ is well-formed.
\end{proof}

\begin{lemma}\label{lem:ac-wellformedness-subterm}
  For an \(\ac\) computation \(N\) and its subcomputation \(N'\), if \(\WF_{\ac}(N)\), then \(\WF_{\ac}(N')\).
  Moreover, for a store \(\tau\), \(\WF_{\ac}(\config{N}{\tau})\) implies \(\WF_{\ac}(\config{N'}{\tau})\).
\end{lemma}

\begin{proof}
  The former statement can be proved by routine induction on the structure of \(N\).
  Since \(\activelabels{N'} \subseteq \activelabels{N}\), it is straightforward to check the other side conditions.
\end{proof}

\subsubsection{Invariants}

\begin{definition}
  A label map is a partial function
  \[
    \eta : \syntacticset{L}_{\mathbf{D}} \rightharpoonup \syntacticset{L}_{\mathbf{A}} \times \syntacticset{L}_{\mathbf{A}}.
  \]
  If \(\eta(l) = (mc, m)\), then \(mc\) is the reference cell that represents \(l\), and \(m\) is the coroutine label that \(mc\) may hold as \(\valid{m}\).
  We parameterize and extend the translation with $\eta$, by translating \(l\) to \(mc\). We call this extension a \emph{runtime translation} and denote it by \(\mtempty_{\eta}\).
  Similarly, we can extend the runtime translation to a translation on contexts by mapping a hole to a hole.
\end{definition}

\begin{definition}[coherence conditions]
  For a \(\del\) store \(\theta\) and a label map \(\eta\), we define \(\Coh(\theta, \eta)\) to hold if and only if the following conditions are satisfied:
  \begin{enumerate}
  \item[(C1)] \(\dom{\eta} = \dom{\theta}\);
  \item[(C2)] the first projection of \(\eta\) is injective: if \(\eta(l) = (mc, m)\) and \(\eta(l') = (mc, m')\), then \(l = l'\).
  \end{enumerate}
\end{definition}

\begin{definition}
  For a pure \(\del\) context \(H\) and a computation \(M\) with a free variable \(x\), define
  \[
    \Cont_{\eta}(H, x, M) \defeq \thunk{\abs{y}{\letin{x}{\mte{\plug{H}{\return{y}}}{\eta}}{\return{\thunk{\abs{\underscore}{\mte{M}{\eta}}}}}}}.
  \]
\end{definition}

\begin{definition}[invariant conditions]
  For a \(\del\) store \(\theta\), an \(\ac\) store \(\tau\), and a runtime
  label map \(\eta\), we define \(\Inv(\theta, \tau, \eta)\) to hold if and only
  if the following conditions are satisfied.
  \begin{enumerate}
  \item[(IC1)] If \(\theta(l) = \nil\) and \(\eta(l) = (mc, m)\), then
    \(\tau(mc) = \refcell{\invalid}\);
  \item[(IC2)] If \(\theta(l) = \abs{y}{\dollar{\plug{H}{\return{y}}}{x}{M}}\)
    and \(\eta(l) = (mc, m)\), then
    \begin{enumerate}
    \item \(\tau(mc) = \refcell{\valid{m}}\),
    \item \(\tau(m) = \Cont_{\eta}(H, x, M)\), and
    \item for any \(l' \in \dom{\theta}\setminus\{l\}\) and
      \(mc' \in \syntacticset{L}_{{\ac}}\), if \(\eta(l') = (mc', m)\), then
      \(\theta(l') = \nil\),
    \end{enumerate}
  \end{enumerate}
  where
  \[
    \Cont_{\eta}(H, x, M) \defeq
    \thunk{\abs{y}{\letin{x}{\mte{\plug{H}{\return{y}}}{\eta}}{\return{\thunk{\abs{\underscore}{\mte{M}{\eta}}}}}}}.
  \]
\end{definition}

  Condition (IC1) ensures that a reference cell representing an invalid
  continuation has \(\invalid\).
  Condition (IC2) states that a valid continuation corresponds to a valid
  reference cell containing a coroutine label that realizes the continuation.
  Moreover, condition (IC2-c) says that a coroutine label \(m\) realizes only one
  valid source continuation label.
  In other words, since the second component of \(\eta\) is not necessarily
  injective, the label \(m\) may have more than one corresponding source
  continuations labels, but condition (IC2-c) ensures that only one of them is
  valid.

\subsubsection{Simulation relation}

\begin{definition}
  For coroutine label \(m\), an \(\ac\) computation \(N\), and a \(\del\) computation \(M\) with a free variable \(x\), define
  \[
    \Act_{\eta}(N, x, M, m) \defeq
    \begin{pmatrix*}[l]
      \letin{\var{res}}{\labeledc{m}{\left( \letin{x}{N}{\return{\thunk{\abs{\underscore}{\mte{M}{\eta}}}}} \right)}}{\\\app{\force{\var{res}}}{\thunk{\app{\force{\var{ref}}}{\valid{m}}}}}
    \end{pmatrix*}
    ,
  \]
  and for an \(\ac\) evaluation context \(\context{D}\),
  \[
    \ActCtx_{\eta}(\context{D}, x, M, m) \defeq
    \begin{pmatrix*}[l]
      \letin{\var{res}}{\labeledc{m}{\left( \letin{x}{\context{D}}{\return{\thunk{\abs{\underscore}{\mte{M}{\eta}}}}} \right)}}{\\\app{\force{\var{res}}}{\thunk{\app{\force{\var{ref}}}{\valid{m}}}}}
    \end{pmatrix*}
  \]
\end{definition}

\begin{definition}[core relation]
  For a label map \(\eta\), define a relation \(\rsimue{\eta}\, \subseteq \mathrm{Conf}_{\del} \times \mathrm{Conf}_{\ac}\) inductively as follows.

  For every constructor \(E\) among
  \begin{mathpar}
    \pcase{V}{x_1}{x_2}{M},
    \and
    \scase{V}{L_{\mathnormal{i}}}{x_i}{M_i},
    \and
    \force{V},\and \return{V},\\
    \abs{x}{M},
    \and
    \cpair{M_1}{M_2},
    \and
    \shift{k}{M},
    \and
    \throw{V}{W},
  \end{mathpar}
  we have the structural clause
  \[
    \simue{\config{E}{\theta}}{\config{\mte{E}{\eta}}{\tau}}{\eta},
  \]
  for any \(\theta\) and \(\tau\).

  In addition,%
  \begin{mathpar}
    \inferrule
    {
      \simue{\config{M_1}{\theta}}{\config{N_1}{\tau}}{\eta}
    }
    {
      \simue
      {\config{\letin{x}{M_1}{M_2}}{\theta}}
      {\config{\letin{x}{N_1}{\mte{M_2}{\eta}}}{\tau}}
      {\eta}
    },
    \\
    \inferrule
    {
      \simue{\config{M}{\theta}}{\config{N}{\tau}}{\eta}
    }
    {
      \simue
      {\config{\app{M}{V}}{\theta}}
      {\config{\app{N}{\mte{V}{\eta}}}{\tau}}
      {\eta}
    },
    \and
    \inferrule
    {
      \simue{\config{M}{\theta}}{\config{N}{\tau}}{\eta}
    }
    {
      \simue
      {\config{\prj{i}{M}}{\theta}}
      {\config{\prj{i}{N}}{\tau}}
      {\eta}
    },
  \end{mathpar}
  and the dollar clause:
  \[
     \inferrule
    {
      \simue{\config{M_1}{\theta}}{\config{N_1}{\tau}}{\eta} \\
      \tau(m) = \nil \\
      \forall l,mc.\;\eta(l) = (mc, m) \Rightarrow \theta(l) = \nil
    }
    {
      \simue
      {\config{\dollar{M_1}{x}{M_2}}{\theta}}
      {\config{\Act_{\eta}(N_1, x, M_2, m)}{\tau}}
      {\eta}
    }.
  \]
\end{definition}

\begin{definition}[core relation on contexts]
  For a label map \(\eta\), define a relation \(\rsimuec{\eta} \subseteq \mathrm{CConf}_{\del} \times \mathrm{CConf}_{\ac}\) inductively as follows.
  \begin{mathpar}
    \simuec{\config{\hole}{\theta}}{\config{\hole}{\tau}}{\eta},
    \and
    \inferrule
    {
      \simuec{\config{\context{C}}{\theta}}{\config{\context{D}}{\tau}}{\eta}
    }
    {
      \simuec
      {\config{\letin{x}{\context{C}}{M_2}}{\theta}}
      {\config{\letin{x}{\context{D}}{\mte{M_2}{\eta}}}{\tau}}
      {\eta}
    },
    \and
    \inferrule
    {
      \simuec{\config{\context{C}}{\theta}}{\config{\context{D}}{\tau}}{\eta}
    }
    {
      \simuec
      {\config{\app{\context{C}}{V}}{\theta}}
      {\config{\app{\context{D}}{\mte{V}{\eta}}}{\tau}}
      {\eta}
    },
    \and
    \inferrule
    {
      \simuec{\config{\context{C}}{\theta}}{\config{\context{D}}{\tau}}{\eta}
    }
    {
      \simuec
      {\config{\prj{i}{\context{C}}}{\theta}}
      {\config{\prj{i}{\context{D}}}{\tau}}
      {\eta}
    },
    \and
    \inferrule
    {
      \simuec{\config{\context{C}}{\theta}}{\config{\context{D}}{\tau}}{\eta} \\
      \tau(m) = \nil \\
      \forall l,mc.\;\eta(l) = (mc, m) \Rightarrow \theta(l) = \nil
    }
    {
      \simuec
      {\config{\dollar{\context{C}}{x}{M_2}}{\theta}}
      {\config{\ActCtx_{\eta}(\context{D}, x, M_2, m)}{\tau}}
      {\eta}
    }.
  \end{mathpar}
\end{definition}

\begin{definition}
  For \(\del\) and \(\ac\) configurations \(C = \config{M}{\theta}\) and \(D = \config{N}{\tau}\), we write \(\simues{C}{D}{\eta}\)
  if and only if all of the following hold:
  \begin{enumerate}
  \item \(\WF_{\del}(C)\);
  \item \(\WF_{\ac}(D)\);
  \item \(\Coh(\theta, \eta)\);
  \item \(\Inv(\theta, \tau, \eta)\);
  \item \(\simue{C}{D}{\eta}\).
  \end{enumerate}
\end{definition}

\begin{definition}[simulation relation]
  For \(\del\) and \(\ac\) configurations \(C\) and \(D\), define \(C \sim D\) if and only if either \(C = D = \bot\), or there exist a label map \(\eta\) and an \(\ac\) configuration \(D'\) such that
  \[
    D \arrAC^{*} D' \quad\text{and}\quad \simues{C}{D'}{\eta}.
  \]
\end{definition}

\subsection{Proof}
\label{app:sec:DELoneToAC:proof}

\subsubsection{Lemmas}

\begin{lemma}[substitution]\label{lem:subst}
  For any \(\del\) term \(E\) with free variables \(x_1, \ldots, x_n\), \(\del\)
  values \(V_1, \ldots, V_n\), and label map \(\eta\),
  \[
    \mte{E}{\eta}[\mte{V_1}{\eta}/x_1, \ldots,\mte{V_n}{\eta}/x_n] \equiv \mte{E[V_1/x_1, \ldots, V_n/x_n]}{\eta}.
  \]

  For any \(\del\) pure context \(H\) and \(\del\) computation \(M\), and label map \(\eta\),
  \[
    \mte{\context{H}}{\eta}[\mte{M}{\eta}] \equiv \mte{\plug{H}{M}}{\eta}.
  \]
\end{lemma}
\begin{proof}
  Both are proved by induction on \(E\) and \(\context{H}\), respectively.
\end{proof}

\begin{lemma}\label{lem:calc}
The following AC reductions hold.
\begin{enumerate}
\item
  \[
    \config{\app{\force{\var{ref}}}{V}}{\tau}
    \arrAC^{+}
    \config{\return{l}}{\tau[l:=\refcell{V}]},
  \]
  where \(l\) is a fresh label.
\item If $\tau(l)=\refcell{V}$, then
  \[
    \config{\app{\app{\force{\var{get}}}{l}}{\unit}}{\tau}
    \arrAC^{+}
    \config{\return{V}}{\tau}.
  \]
\item If $\tau(l)=\refcell{V}$, then
  \[
    \config{\app{\app{\force{\var{set}}}{l}}{W}}{\tau}
    \arrAC^{+}
    \config{\return{\unit}}{\tau[l:=\refcell{W}]}.
  \]
\item $\config{\force{\var{fail}}}{\tau} \arrAC^{+} \bot$.
\end{enumerate}
\end{lemma}

\begin{proof}
  By routine calculation.
\end{proof}

\begin{lemma}[pure contexts]\label{lem:pure}
If \(\context{H}\) is a pure \(\del\) context, then for any \(\ac\) context \(\context{D}\),
\[
  \simuec{\config{\context{H}}{\theta}}{\config{\context{D}}{\tau}}{\eta}
  \qquad\text{iff}\qquad
  \context{D} \equiv \mte{H}{\eta}.
\]
\end{lemma}

\begin{proof}
By induction on \(\context{H}\).
\end{proof}

\begin{lemma}\label{lem:plug}
  If \(\simuec{\config{\context{C}}{\theta}}{\config{\context{D}}{\tau}}{\eta}\) and \(\simue{\config{M}{\theta}}{\config{N}{\tau}}{\eta}\), then
  \[
    \simue{\config{\plug{C}{M}}{\theta}}{\config{\plug{D}{N}}{\tau}}{\eta}.
  \]
\end{lemma}

\begin{proof}
  We prove the lemma by induction on the derivation of
  \[
    \simuec{\config{\context{C}}{\theta}}{\config{\context{D}}{\tau}}{\eta}.
  \]
  If both contexts are holes, then the conclusion is exactly the second premise.

  Suppose the last rule introduces a let-frame.  Then
  \[
    \context{C}=(\letin{x}{\context{C}_0}{M_2}),
    \qquad
    \context{D}=(\letin{x}{\context{D}_0}{\mte{M_2}{\eta}}),
  \]
  and
  \(
    \simuec{\config{\context{C}_0}{\theta}}{\config{\context{D}_0}{\tau}}{\eta}
  \).
  By the induction hypothesis,
  \[
    \simue{\config{\context{C}_0[M]}{\theta}}{\config{\context{D}_0[N]}{\tau}}{\eta}.
  \]
  Applying the let-clause of the core relation gives the required conclusion.
  The application-frame and projection-frame cases are identical, using the
  corresponding clauses of the core relation.

  Finally suppose the last rule introduces a dollar frame.  Then
  \[
    \context{C}=\dollar{\context{C}_0}{x}{M_2},
    \qquad
    \context{D}=\ActCtx_{\eta}(\context{D}_0, x, M_2,m),
  \]
  with
  \[
    \simuec{\config{\context{C}_0}{\theta}}{\config{\context{D}_0}{\tau}}{\eta},
    \qquad
    \tau(m)=\nil,
  \]
  and
  \[
    \forall l,mc.\; \eta(l)=(mc,m)\Rightarrow \theta(l)=\nil.
  \]
  By the induction hypothesis,
  \[
    \simue{\config{\context{C}_0[M]}{\theta}}{\config{\context{D}_0[N]}{\tau}}{\eta}.
  \]
  The dollar clause of the core relation, together with the two side conditions above, yields
  \[
    \simue
    {\config{\dollar{\context{C}_0[M]}{x}{M_2}}{\theta}}
    {\config{\Act_{\eta}(\context{D}_0[N], x, M_2,m)}{\tau}}
    {\eta}.
  \]
  Since
  \( \Act_{\eta}(\context{D}_0[N], x, M_2,m) \equiv
  (\ActCtx_{\eta}(\context{D}_0, x, M_2,m))[N], \) this is the required
  relation.
\end{proof}

\begin{lemma}\label{lem:decomp}
If
\[
  \simue{\config{\plug{C}{M}}{\theta}}{\config{P}{\tau}}{\eta},
\]
then there exist a target evaluation context \(\context{D}\) and a target computation \(N\)
such that
\[
  P \equiv \plug{D}{N},
  \qquad
  \simuec{\config{\context{C}}{\theta}}{\config{\context{D}}{\tau}}{\eta},
  \qquad
  \simue{\config{M}{\theta}}{\config{N}{\tau}}{\eta}.
\]
\end{lemma}

\begin{proof}
  We prove the lemma by induction on the structure of \(\context{C}\).

  If \(\context{C}=\hole\), take \(\context{D}=\hole\) and \(N=P\).  The contextual
  relation for holes is immediate, and the third conclusion is the premise.

  Suppose \(\context{C}=(\letin{x}{\context{C}_0}{M_2})\).  Then the last rule used in the derivation of
  \( \simue{\config{\letin{x}{\context{C}_0[M]}{M_2}}{\theta}}{\config{P}{\tau}}{\eta} \) must be the let-clause of the
  core relation.  Therefore
  \[
    P\equiv \letin{x}{P_0}{\mte{M_2}{\eta}}
  \]
  for some \(P_0\), and
  \[
    \simue{\config{\context{C}_0[M]}{\theta}}{\config{P_0}{\tau}}{\eta}.
  \]
  Applying the induction hypothesis to this gives \(\context{D}_0\) and
  \(N\) such that \(P_0 \equiv \context{D}_0[N]\),
  \(
    \simuec{\config{\context{C_0}}{\theta}}{\config{\context{D_0}}{\tau}}{\eta}
  \), and
  \(
    \simue{\config{M}{\theta}}{\config{N}{\tau}}{\eta}
  \).
  Take
  \[
    \context{D}=\letin{x}{\context{D}_0}{\mte{M_2}{\eta}}.
  \]
  The application and projection cases are the same.

  Suppose \(\context{C}=\dollar{\context{C}_0}{x}{M_2}\).  Then the last rule of
  the core relation derivation must be the dollar clause.  Hence for some
  \(m\) and \(P_0\),
  \[
    P\equiv \Act_{\eta}(P_0, x, M_2,m),
    \quad
    \simue{\config{\context{C}_0[M]}{\theta}}{\config{P_0}{\tau}}{\eta},
  \]
  and the side conditions
  \[
    \tau(m)=\nil,
    \qquad
    \forall l,mc.\;\eta(l)=(mc,m)\Rightarrow\theta(l)=\nil
    \tag{\(*\)}
  \]
  hold.  Applying the induction hypothesis to the subconfiguration involving \(P_0\)
  gives \(\context{D}_0\) and \(N\) such that
  \[
    P_0 \equiv \context{D}_0[N],
    \qquad
    \simuec{\config{\context{C}_0}{\theta}}{\config{\context{D}_0}{\tau}}{\eta},
    \qquad
    \simue{\config{M}{\theta}}{\config{N}{\tau}}{\eta}.
  \]
  Take
  \[
    \context{D}=\ActCtx_{\eta}(\context{D}_0, x, M_2,m).
  \]
  Applying the contextual dollar clause to \(\simuec{\config{\context{C}_0}{\theta}}{\config{\context{D}_0}{\tau}}{\eta}\) is justified by \((*)\), and the
  equality \(P\equiv\plug{D}{N}\) follows from the definition of \(\ActCtx\).
\end{proof}

\begin{lemma}\label{lem:shape}
  Let \(V\) be a \(\del\)-value.
  \begin{enumerate}
  \item If \(\mte{V}{\eta}\equiv\unit\), then \(V=\unit\).
  \item If \(\mte{V}{\eta}\equiv(W_1,W_2)\), then \(V=(V_1,V_2)\) for some
    \(V_1,V_2\), with \(W_i\equiv\mte{V_i}{\eta}\).
  \item If \(\mte{V}{\eta}\equiv\inj{L}{W}\), then \(V=\inj{L}{V'}\) for some \(V'\),
    with \(W\equiv\mte{V'}{\eta}\).
  \item If \(\mte{V}{\eta}\equiv\thunk{P}\), then \(V=\thunk{M}\) for some \(M\), with
    \(P\equiv\mte{M}{\eta}\).
  \item If \(\mte{V}{\eta}\equiv l\) for some coroutine label \(l\), then \(V=l'\) for
    some continuation label \(l'\) and \(\eta(l')=(l,m)\) for some \(m\).
  \end{enumerate}
\end{lemma}

\begin{proof}
  By case analysis on the syntax of \(V\).
  The translation is homomorphic on all ordinary constructors, and the only source values mapped to target labels are source continuation labels.
\end{proof}

\begin{lemma}\label{lem:return-shape}
  If \(\simue{\config{M}{\theta}}{\config{\return{W}}{\tau}}{\eta}\), then there exists a \(\del\) value \(V\) such that \(M=\return{V}\) and \(W\equiv\mte{V}{\eta}\).
\end{lemma}

\begin{proof}
  By case analysis on the derivation of \(\simue{\config{M}{\theta}}{\config{\return{W}}{\tau}}{\eta}\).
  The only rule where the outermost constructor of the target \(\ac\) computation is \(\return\) is the structural clause for \(\mathbf{return}\).
\end{proof}

\begin{lemma}\label{lem:yield-shape}
Suppose \(\context{E}\) is a pure \(\ac\) context and
\[
  \simue{\config{M}{\theta}}{\config{\plug{E}{\yield{\thunk{P}}}}{\tau}}{\eta}.
\]
Then there exist a pure \(\del\) context \(\context{H}\), a variable \(k\), and a source
computation \(M_1\) such that
\[
  M=\plug{H}{\shift{k}{M_1}},
  \qquad
  \context{E} \equiv \mte{\context{H}}{\eta},
  \qquad
  P \equiv \abs{x}{\letin{k}{\force{x}}{\mte{M_1}{\eta}}}.
\]
\end{lemma}

\begin{proof}
  We prove the lemma by induction on the derivation of
  \(\simue{\config{M}{\theta}}{\config{\plug{E}{\yield{\thunk{P}}}}{\tau}}{\eta}\).  If the last rule is the structural
  clause for \(\shift{k}{M_1}\), then take \(\context{H}=\hole\).  The target computation is exactly
  \(\yield{\thunk{\abs{x}{\letin{k}{\force{x}}{\mte{M_1}{\eta}}}}}\).

  The structural clauses for product matching, variant matching, force, return, abstraction, lazy pairing, and throw are
  impossible because their target computations cannot be written as a pure-context plugged with a yield term.

  In the let case, the target has the form \(\letin{x}{N_1}{\mte{M_2}{\eta}}\).  Since the whole target computation is
  \(\plug{E}{\yield{\thunk{P}}}\) with \(E\) pure, the occurrence of \(\yield{}\) must lie in \(N_1\).  Apply the
  induction hypothesis to the premise relating the distinguished source subcomputation to \(N_1\), and then extend the
  source pure context and the target pure context by the same let-frame.  The application and projection cases are
  identical.

  The dollar clause is impossible.  Its target has the form \(\Act_{\eta}(N_1,x,M_2,m)\), which contains an active labeled
  computation and therefore cannot be written as a pure context plugged with a yield term.
\end{proof}

\subsubsection{Canonicalization}
\label{subsec:canonicalization}

\begin{proposition}[canonicalization]\label{prop:canonicalization}
  Assume \(\WF_{\del}(\config{M}{\theta})\), \(\WF_{\ac}(\config{\mte{M}{\eta}}{\tau})\), \(\Coh(\theta, \eta)\), and \(\Inv(\theta, \tau, \eta)\).
  Then, there exist an \(\ac\) computation \(N\) and an \(\ac\) store \(\tau'\) such that
  \[
    \config{\mte{M}{\eta}}{\tau} \arrAC^{*} \config{N}{\tau'}
  \]
  and
  \begin{mathpar}
    \simue{\config{M}{\theta}}{\config{N}{\tau'}}{\eta},
    \and
    \WF_{\ac}(\config{N}{\tau'}),
    \and
    \Inv(\theta, \tau', \eta),
  \end{mathpar}
  i.e., \(\simues{\config{M}{\theta}}{\config{N}{\tau'}}{\eta}\).
\end{proposition}

\begin{proof}
  By induction on \(M\).
  In the cases where \(M\) has one of the forms covered by the structural clauses of \(\rsimue{\eta}\), take \(N \equiv \mte{M}{\eta}\) and \(\tau' = \tau\).
  If \(M = \letin{x}{M_1}{M_2}\), \(M = \app{M_1}{V}\), or \(M = \prj{i}{M_1}\), apply the induction hypothesis to \(M_1\) and then rebuild the constructor by the corresponding structural clause of \(\rsimue{\eta}\).

  Finally, suppose \(M = \dollar{M_1}{x}{M_2}\).
  The translation \(\mte{\dollar{M_1}{x}{M_2}}{\eta}\) evaluates to
  \[
    N \defeq
    \begin{pmatrix*}[l]
      \letin{\var{res}}{\labeledc{m}{\left( \letin{x}{\mte{M_1}{\eta}}{\return{\thunk{\abs{\underscore}{\mte{M_2}{\eta}}}}} \right)}}{\\\app{\force{\var{res}}}{\thunk{\app{\force{\var{ref}}}{\valid{m}}}}}
    \end{pmatrix*},
  \]
  where \(m\) is a fresh coroutine label mapped to \(\nil\) in the store.
  Let \(\tau'\) be \(\tau[m := \nil]\).
  It is straightforward to check that \(\Inv(\theta, \tau', \eta)\) holds.
  Apply the induction hypothesis to \(M_1\).
  Then, we obtain
  \[
    \config{\mte{M_1}{\eta}}{\tau'} \arrAC^{*} \config{N_1}{\tau''}
  \]
  and
  \begin{mathpar}
    \simue{\config{M_1}{\theta}}{\config{N_1}{\tau''}}{\eta},
    \and
    \WF_{\ac}(\config{N_1}{\tau''}),
    \and
    \Inv(\theta, \tau'', \eta),
  \end{mathpar}
  for some \(\ac\) computation \(N_1\) and \(\ac\) store \(\tau''\).

  To obtain
  \[
    \simues
    {\config{M}{\theta}}
    {\config{\begin{pmatrix*}[l]
      \letin{\var{res}}{\labeledc{m}{\left( \letin{x}{N_1}{\return{\thunk{\abs{\underscore}{\mte{M_2}{\eta}}}}} \right)}}{\\\app{\force{\var{res}}}{\thunk{\app{\force{\var{ref}}}{\valid{m}}}}}
    \end{pmatrix*}}{\tau''}}
    {\eta},
  \]
  it suffices to check the two side conditions of the dollar clause.
  The condition \(\tau''(m) = \nil\) holds because \(m\) is active in every
  configuration occurring in the reduction of \(\mte{M_1}{\eta}\) under the frame
  \(\labeledc{m}{\hole}\), because that reduction preserves well-formedness by
  Proposition~\ref{prop:ac-wellformedness}, and because the third clause of
  \(\ac\) well-formedness forces the content of an active label to be \(\nil\).
  The condition
  \(\forall l,mc.\;\eta(l) = (mc, m) \Rightarrow \theta(l) = \nil\) holds
  vacuously, since \(m\) is fresh and thus occurs in no value of \(\eta\).
\end{proof}

\begin{corollary}[initial simulation]\label{cor:initial-sim}
  For every \(\del\) program \(M\),
  \[
    \simu{\config{M}{\emptyset}}{\config{\mt{M}}{\emptyset}}.
  \]
\end{corollary}

\begin{proof}
  \(\config{M}{\emptyset}\) is trivially well-formed, and the empty label map satisfies \(\Coh\) and \(\Inv\).
  Checking the well-formedness of \(\config{\mt{M}}{\emptyset}\) is straightforward since the computation of the form \(\labeledc{l}{M}\) is not in the image of the translation and so \(\activelabels{\mt{M}} = \emptyset\).
  Then, apply Proposition~\ref{prop:canonicalization}.
\end{proof}

\subsubsection{Simulation}
\label{subsec:simulation}

\begin{proposition}[simulation of beta-reduction]\label{app:prop:beta-sim}
  Assume \(\simues{\config{M}{\theta}}{\config{N}{\tau}}{\eta}\).
  Then:
  \begin{enumerate}
  \item If \(\redbetaD{\config{M}{\theta}}{\bot}\), then
    \[
      \redACplus{\config{N}{\tau}}{\bot}.
    \]

  \item If \(\redbetaD{\config{M}{\theta}}{\config{M'}{\theta'}}\) then there exist \(N'\), \(\tau'\), \(\eta'\) such that
    \[
      \redACplus{\config{N}{\tau}}{\config{N'}{\tau'}}
      \quad\text{and}\quad
      \simues{\config{M'}{\theta'}}{\config{N'}{\tau'}}{\eta'}.
    \]
  \end{enumerate}
\end{proposition}

\begin{proof}
  We prove by case analysis on the beta-reduction rule used in the source.

  \paragraph{Case (\(\mam\)).}
  These cases are routine.
  We spell out one representative case.
  If \(M \equiv (\seq{x}{\return{V}}{M_0})\), then the target must be
  \[
    N \equiv (\seq{x}{\return{\mte{V}{\eta}}}{\mte{M_0}{\eta}}).
  \]
  Reducing this computation under \(\tau\) yields
  \[
    \config{\mte{M_0}{\eta}[\mte{V}{\eta}/x]}{\tau} = \config{\mte{M_0[V/x]}{\eta}}{\tau}
  \]
  by Lemma~\ref{lem:subst}.
  Then, canonicalize this by Proposition~\ref{prop:canonicalization} to obtain the desired conclusion.

  \paragraph{Case \((\mathrm{ret})\).} Similar to the \(\mam\) case.

  \paragraph{Case \((\mathrm{throw})\).}
  Suppose \(M = \throw{l}{V}\) and
  \(\theta(l) = \abs{y}{\dollar{\plug{H}{\return{y}}}{x}{M_0}}\).
  Then
  \[
    \config{M}{\theta}
    \arrD
    \config{\dollar{\plug{H}{\return{V}}}{x}{M_0}}{\theta'},
    \qquad
    \theta' \defeq \theta[l := \nil].
  \]
  Since \(\simues{\config{M}{\theta}}{\config{N}{\tau}}{\eta}\), the last rule
  used to derive the core relation is the structural clause for \(\mathbf{throw}\).
  Hence
  \(N \equiv \mte{\throw{l}{V}}{\eta}\).
  By (C1), there are target labels \(mc\) and \(m\) such that
  \(\eta(l)=(mc,m)\).
  By (IC2),
  \[
    \tau(mc)=\refcell{\valid{m}},
    \qquad
    \tau(m)=\Cont_{\eta}(H, x, M_0),
  \]
  and, if \(l'\neq l\) and \(\eta(l')=(mc',m)\), then \(\theta(l')=\nil\).
  Put
  \[
    \tau' \defeq \tau[mc := \refcell{\invalid},\; m := \nil].
  \]
  By the definition of the translation of \(\mathbf{throw}\), Lemma~\ref{lem:calc},
  and the \(\ac\) rule for \(\mathbf{resume}\), we have
  \[
    \config{\mte{\throw{l}{V}}{\eta}}{\tau}
    \arrAC^{+}
    \config{\Act_{\eta}(\mte{\plug{H}{\return{V}}}{\eta},x,M_0,m)}{\tau'}.
  \]
  In the last step, the resumed coroutine \(m\) runs
  \(\app{\force{\Cont_{\eta}(H, x, M_0)}}{\mte{V}{\eta}}\), and Lemma~\ref{lem:subst} identifies
  the body of the coroutine with
  \[
    \letin{x}{\mte{\plug{H}{\return{V}}}{\eta}}
         {\return{\thunk{\abs{\underscore}{\mte{M_0}{\eta}}}}}.
  \]

  We next check the hypotheses of Proposition~\ref{prop:canonicalization} for the
  source reduct \(\plug{H}{\return{V}}\), the stores \(\theta'\) and \(\tau'\),
  and the label map \(\eta\); the well-formedness conditions are immediate.
  First, \(\Coh(\theta',\eta)\) holds because \(\dom{\theta'}=\dom{\theta}\), so
  (C1) is preserved and (C2) is unaffected.
  Next, \(\Inv(\theta',\tau',\eta)\) holds as follows.
  The label \(l\) is now invalid and its reference cell \(mc\) has the invalid
  tag, since \(\theta'(l)=\nil\) and \(\tau'(mc)=\refcell{\invalid}\).
  All other invalid labels in \(\theta'\) keep their old invalid cells.
  Thus (IC1) for \(\Inv(\theta',\tau',\eta)\) holds.
  If \(l'\) is valid in \(\theta'\), then \(l'\neq l\), and by (C2), the
  first component of \(\eta(l')\) is not \(\var{mc}\).
  Moreover its associated coroutine cannot be \(m\); if so, (IC2-c) for the old valid label \(l\)
  would imply \(\theta(l')=\nil\).
  Thus the store contents relevant to \(l'\) are unchanged from \(\tau\), and
  (IC2) follows from \(\Inv(\theta,\tau,\eta)\).

  By applying Proposition~\ref{prop:canonicalization} with \(M\) taken as
  \(\plug{H}{\return{V}}\), we obtain an \(\ac\) computation \(N'\) and an \(\ac\)
  store \(\tau''\) such that
  \[
    \config{\mte{\plug{H}{\return{V}}}{\eta}}{\tau'} \arrAC^{*} \config{N'}{\tau''}
    \quad\text{and}\quad
    \simues{\config{\plug{H}{\return{V}}}{\theta'}}{\config{N'}{\tau''}}{\eta}.
  \]
  This reduction takes place in the body of the active coroutine \(m\), so it
  lifts to
  \[
    \config{\Act_{\eta}(\mte{\plug{H}{\return{V}}}{\eta}, x, M_0, m)}{\tau'}
    \arrAC^{*}
    \config{\Act_{\eta}(N', x, M_0, m)}{\tau''},
  \]
  and hence
  \(\redACplus{\config{N}{\tau}}{\config{\Act_{\eta}(N', x, M_0, m)}{\tau''}}\).

  It remains to show
  \(\simues{\config{\dollar{\plug{H}{\return{V}}}{x}{M_0}}{\theta'}}{\config{\Act_{\eta}(N', x, M_0, m)}{\tau''}}{\eta}\).
  The two well-formedness conditions, \(\Coh(\theta',\eta)\), and
  \(\Inv(\theta',\tau'',\eta)\) are all supplied by
  \(\simues{\config{\plug{H}{\return{V}}}{\theta'}}{\config{N'}{\tau''}}{\eta}\)
  obtained above; note that it is \(\tau''\), not \(\tau'\), for which the
  invariant conditions are thereby established.
  For the core relation, we apply the dollar clause to
  \(\simue{\config{\plug{H}{\return{V}}}{\theta'}}{\config{N'}{\tau''}}{\eta}\),
  so it suffices to check its two side conditions
  \[
    \tau''(m) = \nil
    \quad\text{and}\quad
    \forall l',mc'.\;\eta(l') = (mc', m) \Rightarrow \theta'(l') = \nil.
  \]
  The former holds as follows.
  The label \(m\) is active in every configuration occurring in the lifted
  reduction above, since each step of that reduction takes place inside the frame
  \(\labeledc{m}{\hole}\).
  Moreover, the first of these configurations is well-formed, because
  \(\tau'(m)=\nil\) by the definition of \(\tau'\), and hence so are all the
  others by Proposition~\ref{prop:ac-wellformedness}.
  The third clause of \(\ac\) well-formedness thus gives \(\tau''(m)=\nil\).
  As for the latter, (IC2-c) of \(\Inv(\theta, \tau, \eta)\) tells us that the label \(l\) is the only valid continuation label defined in \(\theta\) that is associated with the coroutine \(m\).
  This label is no longer valid in \(\theta'\), which implies that the condition above is satisfied.

  \paragraph{Case \((\mathrm{shift})\).}
    Suppose
  \[
    M \equiv \dollar{\plug{H}{\shift{k}{M_1}}}{x}{M_2},
  \]
  where \(H\) is a \(\del\) pure context.
  Let \(l\) be the fresh continuation label generated by the source reduction,
  and put
  \[
    \theta' \defeq
    \theta[l := \abs{y}{\dollar{\plug{H}{\return{y}}}{x}{M_2}}].
  \]
  Then
  \[
    \redbetaD
    {\config{\dollar{\plug{H}{\shift{k}{M_1}}}{x}{M_2}}{\theta}}
    {\config{M_1[l/k]}{\theta'}}.
  \]

  The rule used to derive the core relation
  \(\simues{\config{M}{\theta}}{\config{N}{\tau}}{\eta}\) is the clause for the
  dollar term.
  Hence there exist an \(\ac\) computation \(N_0\) and a coroutine label \(m\)
  such that
  \[
    N \equiv \Act_{\eta}(N_0,x,M_2,m),
  \]
  \[
    \simue{\config{\plug{H}{\shift{k}{M_1}}}{\theta}}{\config{N_0}{\tau}}{\eta},
    \qquad
    \tau(m)=\nil,
  \]
  and
  \[
    \forall l_0,mc_0.\;\eta(l_0)=(mc_0,m)\Rightarrow \theta(l_0)=\nil.
  \]
  By Lemma~\ref{lem:decomp} applied to the pure context \(H\), together with
  Lemma~\ref{lem:pure}, we may write
  \[
    N_0
    \equiv
    \mte{\context{H}}{\eta}
    \left[
      \yield{
        \thunk{\abs{u}{\letin{k}{\force{u}}{\mte{M_1}{\eta}}}}
      }
    \right],
  \]
  where \(u\) is chosen fresh.

  We now compute the corresponding target reduction.
  In the active coroutine \(m\), \((\mathrm{yield})\) occurs under the pure
  target context
  \[
    \letin{x}{\mte{\context{H}}{\eta}}
    {\return{\thunk{\abs{\underscore}{\mte{M_2}{\eta}}}}}.
  \]
  Therefore,
  \[
    \config{N}{\tau}
    \arrAC^{+}
    \config{
      \app{
        \force{
          \thunk{\abs{u}{\letin{k}{\force{u}}{\mte{M_1}{\eta}}}}
        }}
        {\thunk{\app{\force{\var{ref}}}{\valid{m}}}}
    }{
      \tau[m:=\Cont_{\eta}(H, x, M_2)]
    }.
  \]
  Reducing the force of the yielded thunk gives
  \[
    \config{
      \letin{k}{\force{\thunk{\app{\force{\var{ref}}}{\valid{m}}}}}
      {\mte{M_1}{\eta}}
    }{
      \tau[m:=\Cont_{\eta}(H, x, M_2)]
    }.
  \]
  By Lemma~\ref{lem:calc}, for a fresh coroutine label \(mc\), we have
  \begin{align*}
    &\config{\force{\thunk{\app{\force{\var{ref}}}{\valid{m}}}}}
    {\tau[m:=\Cont_{\eta}(H, x, M_2)]} \\
    &\arrAC^{+}
    \config{\return{mc}}
    {\tau[m:=\Cont_{\eta}(H, x, M_2),\; mc:=\refcell{\valid{m}}]}.    
  \end{align*}
  Hence the whole target computation reduces to
  \[
    \config{\mte{M_1}{\eta}[mc/k]}
    {\tau[m:=\Cont_{\eta}(H, x, M_2),\; mc:=\refcell{\valid{m}}]}.
  \]

  Define
  \[
    \eta' \defeq \eta[l := (mc,m)]
  \]
  and
  \[
    \tau' \defeq \tau[m:=\Cont_{\eta}(H, x, M_2),\; mc:=\refcell{\valid{m}}].
  \]
  Since \(l\) is fresh, it does not occur in \(M_1\), \(H\), or \(M_2\).
  Therefore Lemma~\ref{lem:subst} gives
  \[
    \mte{M_1}{\eta}[mc/k]
    \equiv
    \mte{M_1[l/k]}{\eta'}.
  \]
  Moreover,
  \[
    \Cont_{\eta}(H, x, M_2) \equiv \Cont_{\eta'}(H, x, M_2),
  \]
  again because \(l\) is fresh for \(H\) and \(M_2\).

  We verify the side conditions for
  \[
    \simues{\config{M_1[l/k]}{\theta'}}{\config{\mte{M_1[l/k]}{\eta'}}{\tau'}}{\eta'}.
  \]
  First, \(\WF_{\del}(\config{M_1[l/k]}{\theta'})\) follows from the freshness of
  \(l\) and from preservation of \(\del\)-well-formedness.
  Next, \(\Coh(\theta',\eta')\) holds: (C1) because both \(\theta\) and \(\eta\)
  are extended by the same fresh label \(l\), and (C2) because the new first
  component \(mc\) is fresh.
  We now check \(\Inv(\theta',\tau',\eta')\).
  For the new label \(l\), we have
  \[
    \theta'(l)=\abs{y}{\dollar{\plug{H}{\return{y}}}{x}{M_2}},
    \qquad
    \eta'(l)=(mc,m),
  \]
  and by definition of \(\tau'\),
  \[
    \tau'(mc)=\refcell{\valid{m}},
    \qquad
    \tau'(m)=\Cont_{\eta}(H, x, M_2)=\Cont_{\eta'}(H, x, M_2).
  \]
  Thus (IC2-a) and (IC2-b) hold for \(l\).
  If \(l_0\neq l\) and \(\eta'(l_0)=(mc_0,m)\), then \(l_0\in\dom{\eta}\).
  The side condition of the outer dollar clause gives \(\theta(l_0)=\nil\),
  hence \(\theta'(l_0)=\nil\).
  Therefore (IC2-c) also holds for the new valid label \(l\).

  For an old label \(l_0\in\dom{\eta}\), the only store entries changed from
  \(\tau\) to \(\tau'\) are those at \(m\) and at the fresh label \(mc\).
  If \(l_0\) is invalid, then (IC1) remains true because its first component is
  not \(mc\), by freshness, and its old cell is unchanged.
  If \(l_0\) is valid, then its associated coroutine cannot be \(m\): otherwise the
  side condition
  \[
    \forall l_0,mc_0.\;\eta(l_0)=(mc_0,m)\Rightarrow\theta(l_0)=\nil
  \]
  would force \(\theta(l_0)=\nil\), a contradiction.
  Hence the target entries relevant to such an old valid label are unchanged, and
  (IC2) follows from \(\Inv(\theta,\tau,\eta)\).
  Thus \(\Inv(\theta',\tau',\eta')\) holds.

  Finally, apply Proposition~\ref{prop:canonicalization} to
  \(M_1[l/k]\), \(\theta'\), \(\tau'\), and \(\eta'\).
  We obtain an \(\ac\) computation \(N'\) and an \(\ac\) store \(\tau''\) such that
  \[
    \config{\mte{M_1[l/k]}{\eta'}}{\tau'}
    \arrAC^{*}
    \config{N'}{\tau''}
  \]
  and
  \[
    \simues{\config{M_1[l/k]}{\theta'}}{\config{N'}{\tau''}}{\eta'}.
  \]
  Combining the reductions above gives
  \[
    \config{N}{\tau}
    \arrAC^{+}
    \config{N'}{\tau''}.
  \]

    \paragraph{Case \((\mathrm{fail})\).}
  Suppose \(M=\throw{l}{V}\) and \(\theta(l)=\nil\).
  Then \(\redbetaD{\config{M}{\theta}}{\bot}\).
  As in the previous case, the core relation must be the structural clause for
  \(\mathbf{throw}\), so
  \(N\equiv\mte{\throw{l}{V}}{\eta}\).
  By (C1), write \(\eta(l)=(mc,m)\).
  Since \(\theta(l)=\nil\), (IC1) gives
  \[
    \tau(mc)=\refcell{\invalid}.
  \]
  Therefore, by the definition of the translation and Lemma~\ref{lem:calc},
  \[
    \config{\mte{\throw{l}{V}}{\eta}}{\tau}
    \arrAC^{+}
    \config{\force{\var{fail}}}{\tau}
    \arrAC^{+}
    \bot.
  \]
\end{proof}

\begin{proposition}\label{app:prop:beta-sim-lift}
  Assume
  \begin{mathpar}
    \simues{\config{M}{\theta}}{\config{N}{\tau}}{\eta},
    \and
    \simuec{\config{\context{C}}{\theta}}{\config{\context{D}}{\tau}}{\eta},
    \and
    \WF_{\del}(\config{\plug{C}{M}}{\theta}),
    \and
    \WF_{\ac}(\config{\plug{D}{N}}{\tau}).
  \end{mathpar}
  Then,
  \begin{enumerate}
  \item If \(\redbetaD{\config{M}{\theta}}{\bot}\), then
    \[
      \redACplus{\config{\plug{D}{N}}{\tau}}{\bot}.
    \]

  \item If \(\redbetaD{\config{M}{\theta}}{\config{M'}{\theta'}}\) then there exist \(N'\), \(\tau'\), \(\eta'\) such that
    \begin{mathpar}
      \redACplus{\config{N}{\tau}}{\config{N'}{\tau'}},
      \and
      \simues{\config{\plug{C}{M'}}{\theta'}}{\config{\plug{D}{N'}}{\tau'}}{\eta'},
      \and
      \simuec{\config{\context{C}}{\theta'}}{\config{\context{D}}{\tau'}}{\eta'}.
    \end{mathpar}
  \end{enumerate}
\end{proposition}

\begin{proof}
  We prove the proposition by induction on the derivation of
  \[
    \simuec{\config{\context{C}}{\theta}}{\config{\context{D}}{\tau}}{\eta}.
  \]
  If \(\context{C}=\hole\) and \(\context{D}=\hole\), the result is exactly Proposition~\ref{prop:beta-sim}.
  Indeed, in the non-error case, Proposition~\ref{prop:beta-sim} gives
  \[
    \redACplus{\config{N}{\tau}}{\config{N'}{\tau'}}
    \qquad\text{and}\qquad
    \simues{\config{M'}{\theta'}}{\config{N'}{\tau'}}{\eta'}.
  \]
  The contextual relation
  \[
    \simuec{\config{\hole}{\theta'}}{\config{\hole}{\tau'}}{\eta'}
  \]
  follows from the empty-context clause.

  If the last rule of the derivation of the contextual relation is a pure-frame rule, the induction hypothesis gives the desired target reduction and the desired relation for the immediate subcontext.
  We then rebuild the same pure frame by the corresponding clause of the contextual core relation and by the corresponding clause of the core relation on computations.
  The source well-formedness of the rebuilt configuration follows from preservation of \(\del\)-well-formedness, and the target well-formedness follows from preservation of \(\ac\)-well-formedness along the target reduction.
  The conditions \(\Coh\) and \(\Inv\) are supplied by the induction hypothesis and do not depend on the outer pure frame.
  This proves the pure-frame cases.

  It remains to consider the dollar-frame case.  Thus the last rule is of the form
  \[
    \inferrule
    {
      \simuec{\config{\context{C}_0}{\theta}}{\config{\context{D}_0}{\tau}}{\eta}\\
      \tau(m)=\nil\\
      \forall l,mc.\; \eta(l)=(mc,m) \Rightarrow \theta(l)=\nil
    }
    {
      \simuec
      {\config{\dollar{\context{C}_0}{x}{M_2}}{\theta}}
      {\config{\ActCtx_{\eta}(\context{D}_0,x,M_2,m)}{\tau}}
      {\eta}
    }.
  \]
  Therefore
  \[
    \context{C}=\dollar{\context{C}_0}{x}{M_2},
    \qquad
    \context{D}=\ActCtx_{\eta}(\context{D}_0,x,M_2,m).
  \]

  First suppose that
  \[
    \redbetaD{\config{M}{\theta}}{\bot}.
  \]
  By the induction hypothesis applied to the premise
  \[
    \simuec{\config{\context{C}_0}{\theta}}{\config{\context{D}_0}{\tau}}{\eta},
  \]
  we get
  \[
    \redACplus{\config{\context{D}_0[N]}{\tau}}{\bot}.
  \]
  Since
  \[
    \plug{D}{N}
    =
    (\ActCtx_{\eta}(\context{D}_0,x,M_2,m))[N]
    =
    \Act_{\eta}(\context{D}_0[N],x,M_2,m),
  \]
  contextual closure of \(\arrAC\) yields
  \[
    \redACplus{\config{\plug{D}{N}}{\tau}}{\bot}.
  \]
  This proves the error part in the dollar-frame case.

  Now suppose that
  \[
    \redbetaD{\config{M}{\theta}}{\config{M'}{\theta'}}.
  \]
  Applying the induction hypothesis to the premise
  \[
    \simuec{\config{\context{C}_0}{\theta}}{\config{\context{D}_0}{\tau}}{\eta}
  \]
  gives an \(\ac\) computation \(N'\), an \(\ac\) store \(\tau'\), and a label map \(\eta'\) such that
  \[
    \redACplus{\config{N}{\tau}}{\config{N'}{\tau'}},
  \]
  \[
    \simues
    {\config{\context{C}_0[M']}{\theta'}}
    {\config{\context{D}_0[N']}{\tau'}}
    {\eta'},
  \]
  and
  \[
    \simuec{\config{\context{C}_0}{\theta'}}{\config{\context{D}_0}{\tau'}}{\eta'}.
  \]
  We have to rebuild the outer dollar frame.  For this, it suffices to prove
  that the two side conditions of the dollar clause are still valid:
  \[
    \tau'(m)=\nil
    \tag{\(*\)}
  \]
  and
  \[
    \forall l,mc.\; \eta'(l)=(mc,m) \Rightarrow \theta'(l)=\nil.
    \tag{\(**\)}
  \]

  We verify \((*)\) and \((**)\) by inspecting the source beta reduction rule used in
  \[
    \redbetaD{\config{M}{\theta}}{\config{M'}{\theta'}}.
  \]
  For the \(\mam\) cases and for the \((\mathrm{ret})\) case, the label map is
  not extended with any source continuation label associated with the outer
  coroutine \(m\), and no store entry associated with \(m\) is modified.  Hence
  both \((*)\) and \((**)\) follow immediately from the premises of the outer
  dollar frame.

  Consider the \((\mathrm{throw})\) case.  Let the valid label used by the throw be \(l_0\), and write
  \[
    \eta(l_0)=(mc_0,m_0).
  \]
  In the target reduction, \(mc_0\) is invalidated, and the coroutine \(m_0\) is set to \(\nil\).
  We claim that \(m_0\neq m\).
  If \(m_0=m\), then the side condition of the outer dollar frame would imply
  \[
    \theta(l_0)=\nil,
  \]
  contradicting the assumption of the \((\mathrm{throw})\) rule.
  Hence the throw does not change the store entry at \(m\), and it does not create any new label mapped to \(m\).
  Therefore \((*)\) and \((**)\) are preserved.

  Finally consider the \((\mathrm{shift})\) case.  In this case the local redex
  has the form
  \[
    \dollar{\plug{H}{\shift{k}{M_1}}}{y}{M_0},
  \]
  and the proof of Proposition~\ref{prop:beta-sim} introduces a fresh label
  \(l_{\mathrm{new}}\) and extends the label map by
  \[
    \eta'(l_{\mathrm{new}})=(mc_{\mathrm{new}},m_0),
  \]
  where \(m_0\) is the coroutine label used by the inner dollar clause for this
  redex.  We claim that \(m_0\neq m\).  Indeed, if \(m_0=m\), then the target
  configuration \(\config{\plug{D}{N}}{\tau}\) would contain two active
  occurrences of the same coroutine label \(m\): one from the outer dollar frame
  and one from the inner dollar frame.
  This contradicts the assumed well-formedness
  \[
    \WF_{\ac}(\config{\plug{D}{N}}{\tau}).
  \]
  Thus \(m_0\neq m\).
  Since the newly added source label is associated with \(m_0\), not with \(m\), condition \((**)\) is not affected by the fresh extension of \(\eta\).
  Moreover, the store update performed by the shift changes the entry at \(m_0\), not at \(m\).
  Hence \((*)\) is also preserved.

  Therefore, we can now rebuild the contextual relation for the outer dollar frame:
  \[
    \simuec
    {\config{\dollar{\context{C}_0}{x}{M_2}}{\theta'}}
    {\config{\ActCtx_{\eta'}(\context{D}_0,x,M_2,m)}{\tau'}}
    {\eta'}.
  \]
  Since the extension from \(\eta\) to \(\eta'\) is on labels fresh for the outer context and for \(M_2\), we have
  \[
    \mte{M_2}{\eta'} \equiv \mte{M_2}{\eta}.
  \]
  Hence the target context above is the same target context \(\context{D}\), now considered with the updated label map and store.

  From the induction hypothesis we also have
  \[
    \simues
    {\config{\context{C}_0[M']}{\theta'}}
    {\config{\context{D}_0[N']}{\tau'}}
    {\eta'}.
  \]
  In particular,
  \[
    \simue
    {\config{\context{C}_0[M']}{\theta'}}
    {\config{\context{D}_0[N']}{\tau'}}
    {\eta'}.
  \]
  Using the dollar clause of the core relation together with \((*)\) and
  \((**)\), we obtain
  \[
    \simue
    {\config{\dollar{\context{C}_0[M']}{x}{M_2}}{\theta'}}
    {\config{\Act_{\eta'}(\context{D}_0[N'],x,M_2,m)}{\tau'}}
    {\eta'}.
  \]
  This is exactly
  \[
    \simue
    {\config{\plug{C}{M'}}{\theta'}}
    {\config{\plug{D}{N'}}{\tau'}}
    {\eta'}.
  \]

  Finally, the well-formedness conditions required for
  \[
    \simues
    {\config{\plug{C}{M'}}{\theta'}}
    {\config{\plug{D}{N'}}{\tau'}}
    {\eta'}
  \]
  follow from preservation of source and target well-formedness, while \(\Coh(\theta',\eta')\) and \(\Inv(\theta',\tau',\eta')\) are inherited from the induction hypothesis applied to
    \[
    \simuec
    {\config{\dollar{\context{C}_0}{x}{M_2}}{\theta'}}
    {\config{\ActCtx_{\eta'}(\context{D}_0,x,M_2,m)}{\tau'}}
    {\eta'}.
  \]
  Therefore,
  \[
    \simues
    {\config{\plug{C}{M'}}{\theta'}}
    {\config{\plug{D}{N'}}{\tau'}}
    {\eta'}
  \]
  holds, as required.
\end{proof}

\begin{theorem}[simulation]\label{app:thm:sim}
  If \(\simu{C}{D}\) and \(\redD{C}{C'}\), then there exists an \(\ac\) configuration \(D'\) such that
  \[
    \redACplus{D}{D'} \quad\text{and}\quad \simu{C'}{D'}.
  \]
\end{theorem}

\begin{proof}
  Since \(\redD{C}{C'}\), the configuration \(C\) is not \(\bot\).
  By the definition of \(\sim\), there exist an \(\ac\) computation \(P\), an \(\ac\) store \(\tau\), and a label map \(\eta\) such that
  \[
    D \arrAC^{*} \config{P}{\tau}
    \quad\text{and}\quad
    \simues{C}{\config{P}{\tau}}{\eta}.
  \]
  Write \(C=\config{M_0}{\theta}\).
  By the definition of \(\arrD\), the source step decomposes into an evaluation context and a beta-reduction step.
  Thus, there exist a \(\del\) evaluation context \(\context{C}_0\) and a \(\del\) computation \(M\) such that
  \[
    M_0 \equiv \context{C}_0[M]
  \]
  and either
  \[
    \redbetaD{\config{M}{\theta}}{\bot}
    \quad\text{and}\quad
    C'=\bot,
  \]
  or there exist a \(\del\) computation \(M'\) and a \(\del\) store \(\theta'\) such that
  \[
    \redbetaD{\config{M}{\theta}}{\config{M'}{\theta'}}
    \quad\text{and}\quad
    C'=\config{\context{C}_0[M']}{\theta'}.
  \]

  From
  \[
    \simues{\config{\context{C}_0[M]}{\theta}}{\config{P}{\tau}}{\eta}
  \]
  we obtain in particular
  \[
    \simue{\config{\context{C}_0[M]}{\theta}}{\config{P}{\tau}}{\eta}.
  \]
  By Lemma~\ref{lem:decomp}, there exist an \(\ac\) evaluation context \(\context{D}_0\) and an \(\ac\) computation \(N\) such that
  \[
    P \equiv \context{D}_0[N],
  \]
  \[
    \simuec{\config{\context{C}_0}{\theta}}{\config{\context{D}_0}{\tau}}{\eta},
    \qquad
    \simue{\config{M}{\theta}}{\config{N}{\tau}}{\eta}.
  \]
  The side conditions contained in \(\config{\context{C}_0[M]}{\theta} \sim_{\sharp} \config{\context{D}_0[N]}{\tau}\) restrict to the subconfiguration; thus,
  \[
    \simues{\config{M}{\theta}}{\config{N}{\tau}}{\eta}.
  \]
  Moreover, the well-formedness assumptions required by Proposition~\ref{prop:beta-sim-lift} follow from \(\config{\context{C}_0[M]}{\theta} \sim_{\sharp} \config{\context{D}_0[N]}{\tau}\).
  We now apply Proposition~\ref{prop:beta-sim-lift} to
  \[
    \simues{\config{M}{\theta}}{\config{N}{\tau}}{\eta}
    \quad\text{and}\quad
    \simuec{\config{\context{C}_0}{\theta}}{\config{\context{D}_0}{\tau}}{\eta}.
  \]

  First suppose
  \[
    \redbetaD{\config{M}{\theta}}{\bot}.
  \]
  Proposition~\ref{prop:beta-sim-lift}(1) gives
  \[
    \config{P}{\tau} = {\config{\context{D}_0[N]}{\tau}} \arrAC^{+} {\bot}.
  \]
  Combining this with \(D\arrAC^{*}\config{P}{\tau}\) yields
  \[
    \redACplus{D}{\bot}.
  \]
  As \(C'=\bot\), we have \(\simu{C'}{\bot}\) by the definition of the simulation relation.
  Thus we may take \(D'=\bot\).

  Next suppose
  \[
    \redbetaD{\config{M}{\theta}}{\config{M'}{\theta'}}.
  \]
  Proposition~\ref{prop:beta-sim-lift}(2) yields an \(\ac\) computation \(N'\), an \(\ac\) store \(\tau'\), and a label map \(\eta'\) such that
  \[
    \redACplus{\config{N}{\tau}}{\config{N'}{\tau'}},
  \]
  and
  \[
    \simues
    {\config{\context{C}_0[M']}{\theta'}}
    {\config{\context{D}_0[N']}{\tau'}}
    {\eta'}.
  \]
  By the definition of \(\arrAC\), we obtain
  \[
    \redACplus{\config{\context{D}_0[N]}{\tau}}{\config{\context{D}_0[N']}{\tau'}}.
  \]
  Since \(P\equiv\context{D}_0[N]\), composing with \(D\arrAC^{*}\config{P}{\tau}\) gives
  \[
    \redACplus{D}{\config{\context{D}_0[N']}{\tau'}}.
  \]
  Also, because
  \[
    C'=\config{\context{C}_0[M']}{\theta'},
  \]
  we obtain
  \[
    \simu{C'}{\config{\context{D}_0[N']}{\tau'}}
  \]
  by the definition of the simulation relation.
\end{proof}

\subsubsection{Preservation of non-termination}
\label{app:sec:DELoneToAC:stuckness}

The next definition collects the shapes that an evaluation frame can consume.

\begin{definition}\label{def:headmatch}
  For a \(\del\) computation \(M\) and an \(\ac\) computation \(N\), we write
  \(\headmatch{M}{N}\) if all of the following hold:
  \begin{enumerate}
  \item[(H1)] if \(N \equiv \return{W}\), then \(M \equiv \return{V}\) for some
    \(\del\) value \(V\);
  \item[(H2)] if \(N \equiv \abs{x}{N_0}\), then \(M \equiv \abs{x}{M_0}\) for
    some \(\del\) computation \(M_0\);
  \item[(H3)] if \(N \equiv \cpair{N_1}{N_2}\), then
    \(M \equiv \cpair{M_1}{M_2}\) for some \(M_1\) and \(M_2\);
  \item[(H4)] if \(N \equiv \plug{E}{\yield{W}}\) for a pure \(\ac\) context
    \(\context{E}\), then \(M \equiv \plug{H}{\shift{k}{M_0}}\) for a pure
    \(\del\) context \(\context{H}\), a variable \(k\), and a computation
    \(M_0\).
  \end{enumerate}
\end{definition}

We say a \(\del\) or an \(\ac\) configuration \(C\) is \emph{stuck} if \(C\) is
neither \(\bot\) nor a return configuration, and no reduction rule is applicable
to \(C\).
We also say \(C\) is \emph{irreducible} if no reduction rule is applicable to
\(C\); thus a stuck configuration is irreducible, while a return configuration
is irreducible but not stuck.

\begin{lemma}\label{lem:stuck-shape}
  Assume \(\simues{\config{M}{\theta}}{\config{N}{\tau}}{\eta}\) and that
  \(\config{M}{\theta}\) is irreducible.
  Then there exist an \(\ac\) computation \(N'\) and an \(\ac\) store \(\tau'\)
  such that
  \[
    \config{N}{\tau} \arrAC^{*} \config{N'}{\tau'},
  \]
  the configuration \(\config{N'}{\tau'}\) is irreducible, and
  \(\headmatch{M}{N'}\).
\end{lemma}

\begin{proof}
  We prove the statement by induction on the derivation of
  \[
    \simue{\config{M}{\theta}}{\config{N}{\tau}}{\eta}.
  \]
  Since \(\simues{\config{M}{\theta}}{\config{N}{\tau}}{\eta}\) holds,
  \(\WF_{\del}(\config{M}{\theta})\) and \(\Inv(\theta, \tau, \eta)\) hold.

  We consider the clauses where the target term is exactly the runtime
  translation of the source term.
  In all of them the store is unchanged, so we take \(\tau' = \tau\); and we
  take \(N' \equiv N\) except for the throw term, where the target does take
  administrative steps.

  Suppose first that
  \[
    M \equiv \pcase{V}{x_1}{x_2}{M_0}.
  \]
  Then
  \[
    N \equiv \pcase{\mte{V}{\eta}}{x_1}{x_2}{\mte{M_0}{\eta}}.
  \]
  Since \(\config{M}{\theta}\) is irreducible, the value \(V\) is not a pair.
  By Lemma~\ref{lem:shape}, \(\mte{V}{\eta}\) is not a pair either.
  Therefore the target product-matching term is not reducible.
  It also has no distinguished subcomputation to reduce, so \(\config{N}{\tau}\)
  is irreducible.
  As the outermost constructor of \(N\) is a product matching, none of (H1) to
  (H4) has its premise satisfied, so \(\headmatch{M}{N}\) holds vacuously.
  The variant-matching case and the force case are similar.

  In the cases
  \[
    M\equiv\return{V},\qquad
    M\equiv\abs{x}{M_0},\qquad
    M\equiv\cpair{M_1}{M_2},
  \]
  \(N\) is the runtime translation of \(M\) and so has the same outermost
  constructor.
  These computations have no beta-reduction rule and contain no reducible
  subcomputation, so their translations are irreducible as well.
  For \(\headmatch{M}{N}\), the premise of exactly one of (H1), (H2) and (H3) is
  satisfied, and its conclusion holds because \(M\) has that same outermost
  constructor; (H4) cannot happen.

  If
  \[
    M\equiv\shift{k}{M_0},
  \]
  then, by the translation in Figure~\ref{fig:deltoac},
  \[
    N\equiv
    \yield{\thunk{\abs{u}{\letin{k}{\force{u}}{\mte{M_0}{\eta}}}}}
  \]
  for a fresh variable \(u\).  A yield-redex exists only inside an active labeled computation.  This term is not under
  an active label, so the target configuration is irreducible.
  Only the premise of (H4) is satisfied, with \(\context{E} = \hole\), and its
  conclusion holds with \(\context{H} = \hole\).

  Finally suppose \(M\equiv\throw{V}{W}\) (and so
  \(N \equiv \mte{\throw{V}{W}}{\eta}\)).
  Since \(\config{M}{\theta}\) is irreducible and well-formed, \(V\) cannot be a
  continuation label.
  Indeed, if \(V=l\), then well-formedness gives \(l\in\dom{\theta}\), and hence
  either \(\theta(l)=\nil\), in which case the source \((\mathrm{fail})\) rule
  applies, or \(\theta(l)\neq\nil\), in which case the source
  \((\mathrm{throw})\) rule applies.
  Both contradict irreducibility.

  By Lemma~\ref{lem:shape}, \(\mte{V}{\eta}\) is not a coroutine label.
  In the reduction of \(N\), some administrative redexes are executed first;
  after that, the computation has the form
  \[
    N' \equiv \letin{zc}{\resume{\mte{V}{\eta}}{(\inj{Get}{\unit})}}{N_0},
  \]
  for some computation \(N_0\).
  In these steps, the store is leave unchanged.
  Since \(\mte{V}{\eta}\) is not a coroutine label, no \(\ac\) rule can be
  applied to the distinguished subcomputation, and the \((\mathrm{let})\) rule
  cannot be applied either, so \(\config{N'}{\tau}\) is irreducible.
  Thus, \(N\) reduces in finitely many steps to an irreducible target
  configuration.
  It is straightforward to check that \(\headmatch{M}{N'}\) holds vacuously.

  Then, we show the inductive step by case analysis on the last rule used to
  derive \(\simue{\config{M}{\theta}}{\config{N}{\tau}}{\eta}\).
  \paragraph{Let case.}
  Suppose the last rule is the let clause:
  \[
    \inferrule
    {
      \simue{\config{M_1}{\theta}}{\config{N_1}{\tau}}{\eta}
    }
    {
      \simue
      {\config{\letin{x}{M_1}{M_2}}{\theta}}
      {\config{\letin{x}{N_1}{\mte{M_2}{\eta}}}{\tau}}
      {\eta}
    }.
  \]
  Since \(\config{\letin{x}{M_1}{M_2}}{\theta}\) is irreducible, the subconfiguration \(\config{M_1}{\theta}\) is irreducible and \(M_1\) is not of the form \(\return{V}\).
  By the induction hypothesis, there are \(N_1'\) and \(\tau'\) such that
  \[
    \config{N_1}{\tau}\arrAC^{*}\config{N_1'}{\tau'},
  \]
  the configuration \(\config{N_1'}{\tau'}\) is irreducible, and
  \(\headmatch{M_1}{N_1'}\).
  By (H1), the computation \(N_1'\) is not of the form \(\return{W}\).
  Hence
  \[
    \config{\letin{x}{N_1}{\mte{M_2}{\eta}}}{\tau}
    \arrAC^{*}
    \config{\letin{x}{N_1'}{\mte{M_2}{\eta}}}{\tau'}.
  \]
  The target configuration is irreducible: the subcomputation \(N_1'\) is not reducible, and the \((\mathrm{let})\) rule cannot be applied since \(N_1'\) is not a return computation.

  It remains to check
  \(\headmatch{\letin{x}{M_1}{M_2}}{\letin{x}{N_1'}{\mte{M_2}{\eta}}}\).
  The premises of (H1), (H2) and (H3) are not satisfied.
  For (H4), suppose the target is of the form \(\plug{E}{\yield{W}}\) for a pure
  \(\ac\) context \(\context{E}\).
  Then there exists a pure context \(\context{E}_1\) such that
  \(\context{E} = (\letin{x}{\context{E}_1}{\mte{M_2}{\eta}})\) and
  \(N_1' \equiv \context{E}_1[\yield{W}]\).
  By (H4) for \(\headmatch{M_1}{N_1'}\), we get
  \(M_1 \equiv \context{H}_1[\shift{k}{M_0}]\) for a pure \(\del\) context
  \(\context{H}_1\), and therefore
  \(\letin{x}{M_1}{M_2} \equiv \plug{H}{\shift{k}{M_0}}\) for the pure context
  \(\context{H} = (\letin{x}{\context{H}_1}{M_2})\), as required.

  \paragraph{Application case.}
  Suppose the last rule is the application clause, so
  \[
    M\equiv\app{M_1}{V},
    \qquad
    N\equiv\app{N_1}{\mte{V}{\eta}},
  \]
  with
  \[
    \simue{\config{M_1}{\theta}}{\config{N_1}{\tau}}{\eta}.
  \]
  Since \(\config{\app{M_1}{V}}{\theta}\) is irreducible, the subconfiguration
  \(\config{M_1}{\theta}\) is irreducible and \(M_1\) is not an abstraction.  By the
  induction hypothesis, \(\config{N_1}{\tau}\) reduces to an irreducible target
  configuration \(\config{N_1'}{\tau'}\) with \(\headmatch{M_1}{N_1'}\).  By (H2),
  this target computation \(N_1'\) is not an abstraction either.  Therefore
  \[
    \config{\app{N_1}{\mte{V}{\eta}}}{\tau}
    \arrAC^{*}
    \config{\app{N_1'}{\mte{V}{\eta}}}{\tau'},
  \]
  and the target configuration is irreducible.
  The check of \(\headmatch{\app{M_1}{V}}{\app{N_1'}{\mte{V}{\eta}}}\) is as in
  the let case, with the application frame in place of the let frame.

  \paragraph{Projection case.}
  Suppose the last rule is the projection clause.  Then
  \[
    M\equiv\prj{i}{M_1},
    \qquad
    N\equiv\prj{i}{N_1},
  \]
  with
  \[
    \simue{\config{M_1}{\theta}}{\config{N_1}{\tau}}{\eta}.
  \]
  Since the source projection is irreducible, \(\config{M_1}{\theta}\) is irreducible and
  \(M_1\) is not a lazy pair.
  By the induction hypothesis, \(\config{N_1}{\tau}\) reduces to an irreducible target
  configuration \(\config{N_1'}{\tau'}\) with \(\headmatch{M_1}{N_1'}\), and by (H3)
  the computation \(N_1'\) is not a lazy pair either.
  Hence
  \[
    \config{\prj{i}{N_1}}{\tau}
    \arrAC^{*}
    \config{\prj{i}{N_1'}}{\tau'},
  \]
  and the target configuration is irreducible.
  The check of \(\headmatch{\prj{i}{M_1}}{\prj{i}{N_1'}}\) is again as in the let
  case.

  \paragraph{Dollar case.}
  Suppose the last rule is the dollar clause.  Thus
  \[
    M\equiv\dollar{M_1}{x}{M_2},
    \qquad
    N\equiv\Act_{\eta}(N_1,x,M_2,m),
  \]
  and
  \[
    \simue{\config{M_1}{\theta}}{\config{N_1}{\tau}}{\eta},
    \qquad
    \tau(m)=\nil,
  \]
  together with the side condition that every source label mapped to the
  coroutine \(m\) is invalid.

  Since \(\config{\dollar{M_1}{x}{M_2}}{\theta}\) is irreducible, \(M_1\) cannot
  be reduced, cannot be a return computation, and cannot be of the form
  \(\plug{H}{\shift{k}{M_0}}\) for any pure context \(H\).
  These are exactly the three ways in which a dollar term can be reduced: by
  reducing its body, by the ret rule, or by the shift rule.

  By the induction hypothesis, \(\config{N_1}{\tau}\) reduces to an irreducible
  target configuration \(\config{N_1'}{\tau'}\) with \(\headmatch{M_1}{N_1'}\).
  By (H1), the computation \(N_1'\) is not a return computation; by (H4), it is
  not of the form \(\plug{E}{\yield{W}}\) for a pure target context
  \(\context{E}\).

  Therefore
  \[
    \config{\Act_{\eta}(N_1,x,M_2,m)}{\tau}
    \arrAC^{*}
    \config{\Act_{\eta}(N_1',x,M_2,m)}{\tau'}.
  \]
  The target configuration is irreducible: the subcomputation \(N_1'\) cannot be
  reduced, the \((F)\) rule for the let frame cannot be applied since \(N_1'\)
  is not a return computation, and the \((\mathrm{yield})\) rule also cannot be
  applied since \(N_1'\) is not a pure-context occurrence of \(\mathbf{yield}\).
  Finally, \(\headmatch{M}{\Act_{\eta}(N_1',x,M_2,m)}\) holds vacuously: the
  outermost constructor of \(\Act_{\eta}(N_1',x,M_2,m)\) is a let, and it has as
  a subcomputation the labeled computation \(\labeledc{m}{\ldots}\), which
  cannot be represented as a pure-context occurrence of \(\mathbf{yield}\).
\end{proof}

\begin{proposition}\label{prop:stuck-refl}
  Suppose \(\simues{\config{M}{\theta}}{\config{N}{\tau}}{\eta}\).
  If \(\config{M}{\theta}\) is stuck, then \(\config{N}{\tau}\) reduces to a stuck configuration, i.e., there exist an \(\ac\)
  computation \(N'\) and an \(\ac\) store \(\tau'\) such that \(\config{N}{\tau} \arrAC^{*} \config{N'}{\tau'}\) and
  \(\config{N'}{\tau'}\) is stuck.
\end{proposition}

\begin{proof}
  Since \(\config{M}{\theta}\) is stuck, it is irreducible and \(M\) is not of
  the form \(\return{V}\).
  Lemma~\ref{lem:stuck-shape} therefore yields an \(\ac\) computation \(N'\) and
  an \(\ac\) store \(\tau'\) such that
  \(\config{N}{\tau} \arrAC^{*} \config{N'}{\tau'}\), the configuration
  \(\config{N'}{\tau'}\) is irreducible, and \(\headmatch{M}{N'}\).
  By (H1), the computation \(N'\) is not of the form \(\return{W}\).
  Hence \(\config{N'}{\tau'}\) is neither \(\bot\) nor a return configuration,
  and no reduction rule applies to it; that is, it is stuck.
\end{proof}

\begin{corollary}\label{cor:stuck-refl-sim}
  If \(\simu{C}{D}\) and \(C\) is stuck, then there exists an \(\ac\) configuration
  \(D'\) such that
  \[
    D\arrAC^{*}D'
  \]
  and \(D'\) is stuck.
\end{corollary}

\begin{proof}
  Since \(C\) is stuck, it is not \(\bot\).
  Hence, we write \(C=\config{M}{\theta}\).
  By the definition of \(\sim\), there exist an \(\ac\) configuration \(D_0\)
  and a label map \(\eta\) such that
  \[
    D\arrAC^{*}D_0
    \qquad\text{and}\qquad
    \simues{\config{M}{\theta}}{D_0}{\eta}.
  \]
  Write \(D_0=\config{N}{\tau}\).  Proposition~\ref{prop:stuck-refl} gives
  \(N'\) and \(\tau'\) such that
  \[
    \config{N}{\tau}\arrAC^{*}\config{N'}{\tau'}
  \]
  and \(\config{N'}{\tau'}\) is stuck.  Composing this reduction with
  \(D\arrAC^{*}D_0\), we obtain the desired stuck target configuration.
\end{proof}

\begin{lemma}\label{lem:divergence-preserv}
  If \(\eval{\del}{M}\) diverges, then \(\eval{\ac}{\mt{M}}\) also diverges.
\end{lemma}

\begin{proof}
  By Corollary~\ref{cor:initial-sim} and repeated use of Theorem~\ref{app:thm:sim}.
  Each source step induces at least one target step.
\end{proof}

\begin{lemma}\label{lem:error-preserv}
  If \(\redDplus{\config{M}{\emptyset}}{\bot}\), then \(\redACplus{\config{\mt{M}}{\emptyset}}{\bot}\).
\end{lemma}
\begin{proof}
  By Corollary~\ref{cor:initial-sim} and repeated use of Theorem~\ref{app:thm:sim}.
\end{proof}

\subsubsection{Strong macro-expressibility}
\label{subsec:strong-macro}

\begin{theorem}[preservation of semantics]\label{app:thm:prev-sem}
  For every \(\del\) program \(M\), \(\eval{\del}{M}\) is defined if and only if \(\eval{\ac}{\mt{M}}\) is defined.
\end{theorem}

\begin{proof}
  We prove the two directions separately.

  First suppose that \(\eval{\del}{M}\) is defined.  Then there exist a
  \(\del\)-value \(V\) and a \(\del\)-store \(\theta\) such that
  \[
    \redDclos{\config{M}{\emptyset}}{\config{\return{V}}{\theta}}.
  \]
  By Corollary~\ref{cor:initial-sim},
  \[
    \simu{\config{M}{\emptyset}}{\config{\mt{M}}{\emptyset}}.
  \]
  Repeatedly applying Theorem~\ref{app:thm:sim} along the source reduction, we obtain an \(\ac\)-configuration
  \(D\) such that
  \[
    \redACclos{\config{\mt{M}}{\emptyset}}{D}
    \qquad\text{and}\qquad
    \simu{\config{\return{V}}{\theta}}{D}.
  \]
  By the definition of \(\sim\), either both configurations are \(\bot\), which is
  impossible because the source configuration is \(\config{\return{V}}{\theta}\),
  or there exist an \(\ac\)-computation \(N\), an \(\ac\)-store \(\tau\), and a label
  map \(\eta\) such that
  \[
    D\arrAC^{*}\config{N}{\tau}
    \qquad\text{and}\qquad
    \simues{\config{\return{V}}{\theta}}{\config{N}{\tau}}{\eta}.
  \]
  The latter yields
  \[
    \simue{\config{\return{V}}{\theta}}{\config{N}{\tau}}{\eta}.
  \]
  By Lemma~\ref{lem:return-shape}, there exists an \(\ac\)-value \(W\) such that
  \[
    N\equiv\return{W}.
  \]
  Hence
  \[
    \redACclos{\config{\mt{M}}{\emptyset}}{\config{\return{W}}{\tau}},
  \]
  and therefore \(\eval{\ac}{\mt{M}}\) is defined.

  We prove the converse direction by contraposition.
  Assume that \(\eval{\del}{M}\) is undefined.
  Since \(\arrD\) is deterministic, exactly one of the following alternatives
  holds.
  \begin{enumerate}
  \item The source evaluation diverges.  Then \(\eval{\ac}{\mt{M}}\) diverges by
    Lemma~\ref{lem:divergence-preserv}.

  \item The source evaluation reaches the error state \(\bot\).  Thus
    \[
      \redDplus{\config{M}{\emptyset}}{\bot}.
    \]
    By Lemma~\ref{lem:error-preserv},
    \[
      \redACplus{\config{\mt{M}}{\emptyset}}{\bot}.
    \]
    Since \(\arrAC\) is deterministic and \(\bot\) is not a return configuration,
    the target cannot also terminate successfully.

  \item The source evaluation reaches a stuck configuration.  Then
    there exist a \(\del\)-computation \(M'\) and a \(\del\)-store \(\theta\) such
    that
    \[
      \redDclos{\config{M}{\emptyset}}{\config{M'}{\theta}},
    \]
    the configuration \(\config{M'}{\theta}\) is stuck; in particular, \(M'\) is
    not of the form \(\return{V}\).
    By Corollary~\ref{cor:initial-sim} and repeated applications of
    Theorem~\ref{app:thm:sim} along this source reduction sequence, there exists
    an \(\ac\)-configuration \(D\) such that
    \[
      \redACclos{\config{\mt{M}}{\emptyset}}{D}
      \qquad\text{and}\qquad
      \simu{\config{M'}{\theta}}{D}.
    \]
    By Corollary~\ref{cor:stuck-refl-sim}, there exists an \(\ac\)-configuration
    \(D'\) such that
    \[
      D\arrAC^{*}D'
    \]
    and \(D'\) is stuck.  Therefore
    \[
      \config{\mt{M}}{\emptyset}\arrAC^{*}D'.
    \]
    Since \(\arrAC\) is deterministic and a stuck configuration is not a return configuration, the target cannot also
    reduce to a return configuration. This implies that \(\eval{\ac}{\mt{M}}\) is undefined.
  \end{enumerate}

  Therefore, in all cases, \(\eval{\ac}{\mt{M}}\) is undefined.
\end{proof}

\begin{theorem}[strong macro-translation]\label{app:thm:strong-macro}
  The translation \(M \mapsto \mt{M}\) of Figure~\ref{fig:deltoac} is a strong macro-translation from \(\del\) to \(\ac\).
\end{theorem}

\begin{proof}
  The translation maps every \(\del\) program to an \(\ac\) program and is homomorphic on the \(\mam\) constructors by definition.
  Every constructor peculiar to \(\del\) is translated by a fixed syntactic abstraction, as shown in Figure~\ref{fig:deltoac}.
  The preservation of semantics is shown by Theorem~\ref{app:thm:prev-sem}.
\end{proof}

%% file: appB-efftodel.tex
\section{Supplementary proofs for Section~\ref{sec:EFFoneToDELone}}

\subsection{From \texorpdfstring{$\eff$}{EFFone} to \texorpdfstring{$\del$}{DELone}}\label{app:EFFoneToDELone}

In this section, we prove the correctness of the macro-translation defined in
Figure~\ref{app:fig:efftodel}.

\begin{figure}[tb]
  \centering
  \[
  \begin{array}{rll}
    \mt{\opcall{op}{V}} & \defeq & \mathbf{S_0} \var{k}.\;\lambda \var{h}.\;(\mathbf{let}\;\var{y} = (\shift{k'}{\app{\force{h}}{\inj{\op{op}}{\vpair{\mt{V}}{\var{k'}}}}})\;\mathbf{in}\\
                        & & \phantom{(\mathbf{S_0} \var{k}.\;\lambda \var{h}.\;}
                            \app{\throw{k}{y}}{h}) \\
    
    \mt{\throw{V_1}{V_2}} & \defeq & \throw{\mt{V_1}}{\mt{V_2}} \\
    \mt{\handle{H}{M}} & \defeq & \dollar{\app{\dollart{\mt{M}}{H^{\mathrm{ret}}}}{\thunk{H^{\mathrm{ops}}}}}{z}{\return{z}} \\
                        & \mathrm{where} & (\handler{x}{M_{\mathrm{ret}}}{\ldots})^{\mathrm{ret}} = x. \abs{\underscore}{\mt{M_{\mathrm{ret}}}} \\
                        & & (\handler{x}{M_{\mathrm{ret}}}{(\app{\op{op}_i}{\app{p_i}{k_i}} \mapsto M_i)_i})^{\mathrm{ops}} \\
                        & &\quad= \abs{c}{\scase{c}{\op{op}_i}{\vpair{p_i}{k_i}}{\mt{M_i}}}
  \end{array}
  \]
  \caption{Translation from $\eff$ to $\del$}
  \label{app:fig:efftodel}
\end{figure}

\subsubsection{Well-formedness}\label{app:sec:efftodel:well-formedness}

\begin{definition}
  For an \(\eff\) computation \(M\), let \(\lb{M}\) be the set of continuation
  labels occurring in \(M\).
  For an \(\eff\) store \(\theta\), define
  \[
    \lb{\theta} \defeq \bigcup \{\lb{\theta(l)} \mid l \in \dom{\theta}, \theta(l) \neq \nil \}.
  \]
  For an \(\eff\) configuration \(C = \config{M}{\theta}\), define
  \(\lb{C} \defeq \lb{M} \cup \lb{\theta}\).
\end{definition}

\begin{definition}\label{app:def:efftodel:WF}
  An \(\eff\) configuration \(\config{M}{\theta}\) is \emph{well-formed},
  written \(\WF_{\eff}(\config{M}{\theta})\), if and only if
  \[
    \lb{\config{M}{\theta}} \subseteq \dom{\theta}.
  \]
\end{definition}

\begin{proposition}[preservation of \(\eff\) well-formedness]\label{lem:efftodel:eff_well_formed}
  If \(\WF_{\eff}(\config{M}{\theta})\) and
  \(\redE{\config{M}{\theta}}{\config{M'}{\theta'}}\), then
  \(\WF_{\eff}(\config{M'}{\theta'})\).
\end{proposition}
\begin{proof}
  We prove this by case analysis on the \(\eff\) beta-reduction rule used in the
  source step.
  Since an evaluation context is unchanged by the reduction and
  \(\lb{\plug{C}{M}} = \lb{\context{C}} \cup \lb{M}\), it suffices to check that
  the reduct contains only labels defined in the updated store.

  In the \(\mam\) cases and the case \((\mathrm{ret})\), no continuation label is
  generated and the store is unchanged.
  Each of these rules reduces a redex to a term built from its subterms, so
  every continuation label occurring in the reduct already occurs in the redex,
  and the claim follows from the premise.

  In the case \((\mathrm{op})\), we have
  \begin{align*}
    \config{\handle{H}{\paren{\plug{H}{\opcall{op}{V}}}}}{\theta}
    &\arrE \config{M_j[V/p_j, l/k_j]}{\theta'}, \\
    \theta' &\defeq \theta[l := \abs{x}{\handle{H}{\plug{H}{\return{x}}}}]
  \end{align*}
  for a fresh label \(l \notin \dom{\theta}\).
  The reduct contains only the labels occurring in the redex together with the
  fresh label \(l\), all of which belong to
  \(\dom{\theta'} = \dom{\theta} \cup \{l\}\).
  The new store entry contains only the labels occurring in \(H\) and
  \(\context{H}\), which occur in the redex and hence belong to \(\dom{\theta}\).
  The other entries are unchanged.

  In the case \((\mathrm{throw})\), we have
  \[
    \config{\throw{l}{V}}{\theta} \arrE
    \config{\handle{H}{\plug{H}{\return{V}}}}{\theta[l := \nil]},
  \]
  where \(\theta(l) = \abs{x}{\handle{H}{\plug{H}{\return{x}}}}\).
  The reduct contains only the labels occurring in \(\theta(l)\) and in \(V\);
  the former belong to \(\lb{\theta} \subseteq \dom{\theta}\) by the premise, and
  the latter to \(\lb{M} \subseteq \dom{\theta}\).
  Since \(\dom{\theta[l := \nil]} = \dom{\theta}\) and
  \(\lb{\theta[l := \nil]} \subseteq \lb{\theta}\), the claim follows.

  In the case \((\mathrm{fail})\), the configuration reduces to \(\bot\), for
  which the claim is vacuously true.
\end{proof}

A label map \(\eta\) is a partial function that sends an \(\eff\) continuation
label \(l\) to a pair \(\eta(l) = (\eta(l)_1, \eta(l)_2)\) of \(\del\)
continuation labels.
The label map \(\eta\) keeps track of the source-to-target correspondence
between continuation labels.
We then parameterize and extend the translation with \(\eta\) by
\[
  \mte{l}{\eta} \defeq \eta(l)_1,
\]
provided \(\eta(l)\) is defined.
We call this extension a \emph{runtime translation} and write it as
\(\mtempty_{\eta}\); we extend runtime translations to translations on contexts
by mapping a hole to a hole.

Using label maps and the runtime translation, we define coherence conditions and
invariant conditions.

\begin{definition}[coherence conditions]\label{app:def:efftodel:Coh}
  For an \(\eff\) store \(\theta\) and a label map \(\eta\), we define
  \(\Coh(\theta, \eta)\) to hold if and only if the following conditions are
  satisfied.
  \begin{enumerate}
  \item[(C1)] \(\dom{\theta} = \dom{\eta}\);
  \item[(C2)] for all \(l, l' \in \dom{\theta}\) and \(i, j \in \{1, 2\}\), if
    \(\eta(l)_i = \eta(l')_j\), then \(l = l'\) and \(i = j\).
  \end{enumerate}
\end{definition}

\begin{definition}[invariant conditions]\label{app:def:efftodel:IC}
  For an \(\eff\) store \(\theta\), a \(\del\) store \(\tau\), and a label map
  \(\eta\), we define \(\Inv(\theta, \tau, \eta)\) to hold if and only if the
  following conditions are satisfied for every \(l \in \dom{\theta}\).
  \begin{enumerate}
  \item[(IC1)] \(\eta(l)_1 \in \dom{\tau}\) and \(\eta(l)_2 \in \dom{\tau}\);
  \item[(IC2)] if \(\theta(l) = \nil\), then \(\tau(\eta(l)_1) = \nil\);
  \item[(IC3)] if \(\theta(l) = \abs{y}{\handle{H}{\plug{H}{\return{y}}}}\) for
    an \(\eff\) handler \(H\) and an \(\eff\) pure context \(\context{H}\), then
    \begin{align*}
      \tau(\eta(l)_1) &= \abs{x}{\dollar{\letin{y}{\return{x}}{\app{\throw{\eta(l)_2}{y}}{\thunk{H^{\mathrm{ops}}_{\eta}}}}}{z}{\return{z}}}, \\
      \tau(\eta(l)_2) &= \abs{y}{\dollart{\mte{\plug{H}{\return{y}}}{\eta}}{H^{\mathrm{ret}}_{\eta}}},
    \end{align*}
  \end{enumerate}
  where \(H^{\mathrm{ops}}_{\eta}\) (resp.\ \(H^{\mathrm{ret}}_{\eta}\)) denotes
  the translation of the operational clauses (resp.\ return clause) of \(H\)
  with respect to \(\eta\).
\end{definition}

\begin{definition}[simulation relation]\label{app:def:efftodel:simulation}
  For an \(\eff\) configuration \(C\) and a \(\del\) configuration \(D\), we
  write \(\simu{C}{D}\) if and only if either \(C = D = \bot\), or there exist
  an \(\eff\) computation \(M\), a \(\del\) computation \(N\), an \(\eff\)
  store \(\theta\), a \(\del\) store \(\tau\), and a label map \(\eta\) such
  that
  \begin{enumerate}
  \item \(C = \config{M}{\theta}\) and \(D = \config{N}{\tau}\);
  \item \(N \equiv \mte{M}{\eta}\);
  \item \(\WF_{\eff}(C)\);
  \item \(\Coh(\theta, \eta)\) and \(\Inv(\theta, \tau, \eta)\).
  \end{enumerate}
\end{definition}

Intuitively, $\simu{\config{M}{\theta}}{\config{N}{\tau}}$ means the
correspondence between the two configurations: $M$ is translated into $N$ and
the continuations stored in $\theta$ are translated into those in $\tau$.

\begin{lemma}\label{lem:EFFoneToDELone:property}
  The following statements hold.
  \begin{enumerate}
  \item For any $\eff$ computation $M$,
    $\simu{\config{M}{\emptyset}}{\config{\mt{M}}{\emptyset}}$ holds.

  \item For any value $V$,
    $\simu{\config{\return{V}}{\theta}}{\config{N}{\tau}}$ implies that
    $N \equiv \return{\mte{V}{\eta}}$ for some $\eta$.
  \end{enumerate}
\end{lemma}
\begin{proof}
  By case analysis on the source configuration.
\end{proof}

\begin{lemma}\label{lem:eff_to_del_subst}
  If an $\eff$ computation $M$ has the free variables
  $x_1, \ldots, x_n$, then we have
  $\mte{M[V_1/x_1,\ldots, V_n/x_n]}{\eta} \equiv \mte{M}{\eta}[\mte{V_1}{\eta}/x_1,
  \ldots, \mte{V_n}{\eta}/x_n]$ for any values $V_1, \ldots, V_n$ and $\eta$.
\end{lemma}
\begin{proof}
  By straightforward induction on $M$. For example, if $M \equiv \app{M_1}{M_2}$, then
  \begin{align*}
    \mte{M[V_1/x_1,\ldots, V_n/x_n]}{\eta} &\equiv \mte{(\app{M_1}{M_2})[V_1/x_1,\ldots, V_n/x_n]}{\eta} \\
                                           &\equiv \app{\mte{M_1[V_1/x_1,\ldots, V_n/x_n]}{\eta}}{\mte{M_2[V_1/x_1,\ldots, V_n/x_n]}{\eta}} \\
                                           &\equiv \app{(\mte{M_1}{\eta}[\mte{V_1}{\eta}/x_1, \ldots, \mte{V_n}{\eta}/x_n])}{(\mte{M_2}{\eta}[\mte{V_1}{\eta}/x_1, \ldots, \mte{V_n}{\eta}/x_n])} \\
                                           &\equiv (\app{\mte{M_1}{\eta}}{\mte{M_2}{\eta}})[\mte{V_1}{\eta}/x_1, \ldots, \mte{V_n}{\eta}/x_n].
  \end{align*}
\end{proof}

\begin{lemma}\label{lem:eff_to_del_beta_sim}
  If $\simu{C}{D}$ and $\redbetaE{C}{C'}$, then there exists a $\del$ configuration $D'$ such that $\redDplus{D}{D'}$ and $\simu{C'}{D'}$.
\end{lemma}
\begin{proof}
  We prove this by case analysis on $\redbetaE{C}{C'}$, giving only nontrivial
  cases.
  The other cases follow by routine arguments.

  \noindent \textbf{Case} $(\mathrm{op})$:
  \begin{caseindent}
    Suppose $C = \config{\handle{H}{\paren{\plug{H}{\opcall{op}{V}}}}}{\theta}$ and
  \[
    C \arrE C' = \config{M_{\op{op}}[V/p, l/k]}{\theta'},
  \]
  where
  $\theta' \defeq \theta[l := \abs{x}{\handle{H}{\plug{H}{\return{x}}}}]$,
  $H \equiv \handler{x}{M_{\mathrm{ret}}}{\ldots (\opcall{op}{\app{p}{k}}
    \mapsto M_{\op{op}}) \ldots}$, and $l$ is a fresh label.
  By the definition of $\simu{C}{D}$, there exist \(\tau\) and \(\eta\) such
  that
  \[
    D = \config{\dollar{\app{\dollart{N}{H^{\mathrm{ret}}_{\eta}}}{\thunk{H^{\mathrm{ops}}_{\eta}}}}{z}{\return{z}}}{\;\tau},
  \]
  where
  \[
    N \equiv \plug{\mte{\context{H}}{\eta}}{
              \begin{array}{@{}l@{}}
                \mathbf{S_0} \var{k}.\;\lambda \var{h}.\;(\mathbf{let}\;\var{y} = (\shift{k'}{\app{\force{h}}{\inj{\op{op}}{\vpair{\mte{V}{\eta}}{\var{k'}}}}})\;\mathbf{in}\\
                \phantom{(\mathbf{S_0} \var{k}.\;\lambda \var{h}.\;}
                \app{\throw{k}{y}}{h})
              \end{array}
            }.
  \]
  The configuration $D$ evaluates to
  \[
    D' \defeq \config{\mte{M_{\op{op}}}{\eta}\left[\mte{V}{\eta}/p, m/k\right]}
    {\tau'}
  \]
  where $m$ and \(n\) are fresh, and
  \[
    \tau' \defeq \tau\left[
      \begin{array}{@{}l@{}}
        m := \abs{x}{\dollar{\letin{y}{\return{x}}{\app{\throw{n}{y}}{\thunk{H^{\mathrm{ops}}_{\eta}}}}}{z}{\return{z}}}, \\
        n := \abs{y}{\dollart{\plug{\mte{\context{H}}{\eta}}{\return{y}}}{H^{\mathrm{ret}}_{\eta}}}
      \end{array} \right].
  \]
  Let $\eta'$ be $\eta[ l := (m, n) ]$.
  By Lemma~\ref{lem:eff_to_del_subst}, it follows that
  \[
    \mte{M_{\op{op}}}{\eta}[\mte{V}{\eta}/p, m/k] \equiv \mte{M_{\op{op}}[V/p, l/k]}{\eta'}.
  \]
  Note that
  $\mte{M_{\op{op}}}{\eta} \equiv \mte{M_{\op{op}}}{\eta'}$,
  $\mte{\context{H}}{\eta} \equiv \mte{\context{H}}{\eta'}$,
  $H^{\mathrm{ret}}_{\eta} \equiv H^{\mathrm{ret}}_{\eta'}$, and
  $H^{\mathrm{ops}}_{\eta} \equiv H^{\mathrm{ops}}_{\eta'}$ hold
  since $l$ does not appear in any of $M_{\op{op}}$, $\context{H}$,
  $H^{\mathrm{ret}}$, or $H^{\mathrm{ops}}$.
  Finally, we conclude
  \[
    C' \sim D',
  \]
  witnessed by \(\eta'\).
  It is straightforward to check that \(\Coh(\theta', \eta')\) and
  \(\Inv(\theta', \tau', \eta')\) hold, where the freshness of \(l\), \(m\), and
  \(n\) gives (C1) and (C2).
  \end{caseindent}

  \noindent \textbf{Case} $(\mathrm{throw})$:
  \begin{caseindent}
    Suppose $C = \config{\throw{l}{V}}{\theta}$ and
    $C' = \config{\handle{H}{\plug{H}{\return{V}}}}{\theta'}$, where
    $\theta(l) = \abs{x}{\handle{H}{\plug{H}{\return{x}}}}$ and
    $\theta' \defeq \theta[l := \nil]$.
    By the definition of $\simu{C}{D}$, there exist \(\tau\) and \(\eta\) such that
    \(D = \config{\throw{\eta(l)_1}{\mte{V}{\eta}}}{\tau}\),
    \[
      \tau(\eta(l)_1) = \abs{x}{\dollar{\letin{y}{\return{x}}{\app{\throw{\eta(l)_2}{y}}{\thunk{H^{\mathrm{ops}}_{\eta}}}}}{z}{\return{z}}},
    \]
    and
    \[
      \tau(\eta(l)_2) = \abs{y}{\dollart{\mte{\plug{H}{\return{y}}}{\eta}}{H^{\mathrm{ret}}_{\eta}}}.
    \]
    Thus, $D$ evaluates to
    \[
      D' = \config{\dollar{\app{\paren{\dollart{\mte{\plug{H}{\return{y}}}{\eta}}{H^{\mathrm{ret}}_{\eta}}\left[\mte{V}{\eta}/y\right]}}{\thunk{H^{\mathrm{ops}}_{\eta}}}}{z}{\return{z}}}{\tau'},
    \]
    where $\tau' \defeq \tau[\eta(l)_1 := \nil, \eta(l)_2 := \nil]$.
    Then, from Lemma~\ref{lem:eff_to_del_subst}, we obtain
    \[
      \dollart{\mte{\plug{H}{\return{y}}}{\eta}}{H^{\mathrm{ret}}_{\eta}}[\mte{V}{\eta}/y] \equiv \dollart{\mte{\plug{H}{\return{V}}}{\eta}}{H^{\mathrm{ret}}_{\eta}}.
    \]
    Therefore, we conclude that
    \[
    C' \sim D',
    \]
    witnessed by \(\eta\), since
    \[
      \mte{\handle{H}{\plug{H}{\return{V}}}}{\eta} \equiv \dollar{\app{\dollart{\mte{\plug{H}{\return{V}}}{\eta}}{H^{\mathrm{ret}}_{\eta}}}{\thunk{H^{\mathrm{ops}}_{\eta}}}}{z}{\return{z}}.
    \]
    Note that \(\Coh(\theta', \eta)\) holds since
    \(\dom{\theta'} = \dom{\theta}\), and that \(\Inv(\theta', \tau', \eta)\)
    also holds.
  \end{caseindent}

  \noindent \textbf{Case} $(\mathrm{fail})$:
  \begin{caseindent}
    Suppose $C = \config{\throw{l}{V}}{\theta}$, $\theta(l) = \nil$, and $\redbetaE{C}{\bot}$.
    By the definition of $\simu{C}{D}$, there exist \(\tau\) and \(\eta\) such that
    \(D = \config{\throw{\eta(l)_1}{\mte{V}{\eta}}}{\tau}\) and $\tau(\eta(l)_1) = \nil$.
    Hence, the invocation of $\eta(l)_1$ fails and $D$ evaluates to $\bot$, which
    concludes the current case.
  \end{caseindent}
  Thus, we are done.
\end{proof}

\begin{lemma}\label{lem:eff_to_del_plug}
  For every evaluation context \(\context{K}\) and computation $M$, we have
  $\mte{\plug{K}{M}}{\eta} \equiv \mte{\context{K}}{\eta}[\mte{M}{\eta}]$.
\end{lemma}
\begin{proof}
  By straightforward induction on \(\context{K}\).
\end{proof}

\begin{lemma}\label{lem:eff_to_del_decomp}
  Suppose that $\simu{\config{M}{\theta}}{\config{N}{\tau}}$ holds, and let
  $\eta$ be a label map witnessing it, so that $N \equiv \mte{M}{\eta}$.
  If $M$ is decomposed into an evaluation context $\context{C}$ and a redex
  $M'$, then $N$ is decomposed into the evaluation context
  $\mte{\context{C}}{\eta}$ and the redex $\mte{M'}{\eta}$.
\end{lemma}
\begin{proof}
  It is easy to check this by induction on the number of frames of $\context{C}$. Note that the runtime translation maps frames to frames and redexes to redexes.
\end{proof}

\begin{proposition}[simulation]\label{prop:EFFoneToDELone:sim}
  If $\simu{C}{D}$ and $\redE{C}{C'}$, then there exists a $\del$
  configuration $D'$ such that $\redDplus{D}{D'}$ and $\simu{C'}{D'}$.
\end{proposition}
\begin{proof}
  Since $C$ is not in normal form, there exist $M$ and $\theta$ such that
  $C = \config{M}{\theta}$, and according to the semantics of $\eff$, $M$
  can be decomposed into an evaluation context $\context{C}$ and a redex $M'$:
  \(M \equiv \plug{C}{M'}\).
  By the definition of $\simu{C}{D}$, there exist $\eta$ and $\tau$
  such that $D = \config{\mte{M}{\eta}}{\tau}$ and both
  \(\Coh(\theta, \eta)\) and \(\Inv(\theta, \tau, \eta)\) hold.
  By Lemma~\ref{lem:eff_to_del_decomp}, $\mte{M}{\eta}$ also can be decomposed
  into $\mte{\context{C}}{\eta}$ and a redex $\mte{M'}{\eta}$, namely
  $\mte{M}{\eta} \equiv \mte{\context{C}}{\eta}\left[\mte{M'}{\eta}\right]$.

  If $\redbetaE{\config{M'}{\theta}}{\bot}$, then we have $\redE{C}{\bot}$.
  By Lemma~\ref{lem:eff_to_del_beta_sim}, $\config{\mte{M'}{\eta}}{\tau}$ also
  evaluates to $\bot$.
  Hence, we conclude that
  $\redDplus{\config{\mte{\context{C}}{\eta}\left[\mte{M'}{\eta}\right]}{\tau}}{\bot}$,
  that is, $\redDplus{D}{\bot}$.

  Next, suppose $\redbetaEp{M'}{\theta}{M''}{\theta'}$.
  By Lemma~\ref{lem:eff_to_del_beta_sim}, we obtain
  $\redDplus{\config{\mte{M'}{\eta}}{\tau}}{\config{\mte{M''}{\eta'}}{\tau'}}$
  and $\simu{\config{M''}{\theta'}}{\config{\mte{M''}{\eta'}}{\tau'}}$ for some
  $\theta'$, $\tau'$, and $\eta'$ with \(\Coh(\theta', \eta')\)
  and \(\Inv(\theta', \tau', \eta')\).
  Moreover, $\mte{\context{C}}{\eta} \equiv \mte{\context{C}}{\eta'}$ holds
  since \(\eta'\) is, by construction, an extension of \(\eta\), and the labels
  in $\dom{\eta'}\setminus\dom{\eta}$ are fresh and do not appear in
  $\context{C}$.
  Therefore, we conclude that
  $\simu{\config{\plug{C}{M''}}{\theta'}}{\config{\mte{\plug{C}{M''}}{\eta'}}{\tau'}}$
  by Lemma~\ref{lem:eff_to_del_plug}.
\end{proof}

To establish the strong macro-expressibility of $\eff$ in $\del$, we prove three
lemmas.

\begin{lemma}\label{lem:efftodel:diverge}
  Suppose that the evaluation of an \(\eff\) program $M$ diverges, i.e., the
  reduction sequence starting from $M$ is infinite.
  Then the evaluation of $\mt{M}$ also diverges.
\end{lemma}
\begin{proof}
  By applying Proposition~\ref{prop:EFFoneToDELone:sim} repeatedly, we also obtain
  an infinite reduction sequence of $\mt{M}$.
  Since the reduction of $\del$ is deterministic, this implies that the
  evaluation of $\mt{M}$ also diverges.
\end{proof}

\begin{lemma}\label{lem:efftodel:stuck}
  Suppose that the evaluation of an \(\eff\) program $M$ gets stuck: there
  exists a term $M'$ and a store $\theta$ such that $\redEclos{\config{M}{\emptyset}}{\config{M'}{\theta}}$,
  $M' \neq \return{V}$ for any value $V$, and there is no rule that can reduce
  $\config{M'}{\theta}$.
  Then, the evaluation of $\mt{M}$ also gets stuck.
\end{lemma}
\begin{proof}
  It suffices to show the following statement:
  \begin{quote}
    Let $M$, $N$, $\theta$, $\tau$, and $\eta$ be such that
    $N \equiv \mte{M}{\eta}$, \(\WF_{\eff}(\config{M}{\theta})\),
    \(\Coh(\theta, \eta)\), and \(\Inv(\theta, \tau, \eta)\) hold.
    If $\config{M}{\theta}$ is stuck, then $\config{N}{\tau}$ is also stuck.
  \end{quote}
  We prove it by induction on the structure of $M$, giving only five cases for
  brevity.
  The other cases follow by similar reasoning.

  \noindent \textbf{Case} (product matching):
  \begin{caseindent}
    Suppose \(M \equiv \pcase{V}{x_1}{x_2}{M'}\).
    Since $\config{M}{\theta}$ is stuck, we know that
    $V \neq \vpair{V_1}{V_2}$ for any values $V_1$ and $V_2$.
    Consequently, $\mte{V}{\eta}$ also cannot be of the form $\vpair{W_1}{W_2}$
    for any values $W_1$ and $W_2$.
    Therefore, the configuration
    $\config{\mte{\pcase{V}{x_1}{x_2}{M'}}{\eta}}{\tau}$ is stuck, and
    this concludes the current case.
  \end{caseindent}

  \noindent \textbf{Case} (variant matching):
  \begin{caseindent}
    Suppose \(M \equiv \scase{V}{L_{\mathnormal{i}}}{x_i}{M_i}\).
    Since $\config{M}{\theta}$ is stuck, we know that $V \neq \inj{L}{V'}$
    for any variable label \(L\) of \(\eff\) and \(\eff\) value $V'$.
    Then, by definition, $\mte{V}{\eta}$ is also not a variant value.
    Thus the configuration \(\config{\mte{M}{\eta}}{\tau}\) is stuck, which
    concludes the case.
  \end{caseindent}

  \noindent \textbf{Case} (throw):
  \begin{caseindent}
    Suppose \(M \equiv \throw{V}{W}\).
    It follows that $V$ is not a continuation label, since if $V = l$ for some
    $l$, then well-formedness implies
    $l \in \dom{\theta}$, which contradicts the assumption that
    $\config{\throw{V}{W}}{\theta}$ is stuck.
    Therefore, \(\mte{V}{\eta}\) is also not a continuation label, so
    \(\config{\mte{M}{\eta}}{\tau}\) is also stuck.
  \end{caseindent}

  \noindent \textbf{Case} (sequencing):
  \begin{caseindent}
    Suppose \(M \equiv \seq{x}{M_1}{M_2}\).
    Since $\config{M}{\theta}$ is stuck, $M_1$ cannot be of the form
    $\return{V}$ for any value $V$.
    By the induction hypothesis, we obtain that $\mte{M_1}{\eta}$ is stuck and that
    $\mte{M_1}{\eta} \neq \return{W}$ for any value $W$.
    Consequently, the configuration
    $\config{\seq{x}{\mte{M_1}{\eta}}{\mte{M_2}{\eta}}}{\tau}$ is also stuck, which
    completes the current case.
  \end{caseindent}

  \noindent \textbf{Case} (handle):
  \begin{caseindent}
    Suppose \(M \equiv \handle{H}{M'}\).
    By the definition of the translation, we obtain
    \[
      N \equiv \dollar{\app{\dollart{\mte{M'}{\eta}}{H^{\mathrm{ret}}_{\eta}}}{\thunk{H^{\mathrm{ops}}_{\eta}}}}{z}{\return{z}}.
    \]
    Since $\config{M}{\theta}$ is stuck, $M'$ can be neither of the form
    $\plug{H}{\app{\op{op}}{V}}$ nor of the form $\return{V}$ for any pure
    context $\context{H}$, operation symbol $\op{op}$, and value $V$.
    Then, it follows by induction on $M'$ that $\mte{M'}{\eta}$ can be
    neither of the form $\plug{P}{\shift{k}{N'}}$ nor of the form $\return{W}$,
    for any pure context $\context{P}$, computation $N'$, and value $W$.
    Therefore, the configuration $\config{N}{\tau}$ is also stuck, and this
    completes the proof.
  \end{caseindent}
\end{proof}

\begin{lemma}\label{lem:efftodel:bot}
  If the evaluation of an \(\eff\) program $M$ reaches the error state $\bot$,
  i.e., $\redEplus{\config{M}{\emptyset}}{\bot}$, then the evaluation of
  $\mt{M}$ also reaches $\bot$.
\end{lemma}
\begin{proof}
  This lemma directly follows from Proposition~\ref{prop:EFFoneToDELone:sim}.
\end{proof}

\begin{proposition}[preservation of semantics]\label{cor:eff_to_del_correctness}
  $\eval{\eff}{M}$ is defined if and only if $\eval{\del}{\mt{M}}$ is defined.
\end{proposition}
\begin{proof}
  We first show the ``only if'' direction.
  Suppose that $\redEclos{\config{M}{\emptyset}}{\config{\return{V}}{\theta}}$
  for some $\theta$.
  Lemma~\ref{lem:EFFoneToDELone:property} yields that
  $\simu{\config{M}{\emptyset}}{\config{\mt{M}}{\emptyset}}$.
  By applying Proposition~\ref{prop:EFFoneToDELone:sim} iteratively, there
  exist \(\tau\) and $\eta$ such that
  $\redDclos{\config{\mt{M}}{\emptyset}}{\config{\return{\mte{V}{\eta}}}{\tau}}$.
  This implies that $\eval{\del}{\mt{M}}$ is defined.
  We then show the ``if'' direction.
  We prove the contrapositive of the proposition: if $\eval{\eff}{M}$ is
  undefined, $\eval{\del}{\mt{M}}$ is also undefined.
  There are three cases where $\eval{\eff}{M}$ is undefined: the evaluation of
  \(M\) (1) diverges, (2) gets stuck, or (3) reaches $\bot$.
  In each case, Lemmas~\ref{lem:efftodel:diverge}, \ref{lem:efftodel:stuck}, and
  \ref{lem:efftodel:bot} imply that the evaluation of $\mt{M}$ diverges, gets
  stuck, or reaches $\bot$, respectively.
  Thus, $\eval{\del}{\mt{M}}$ is undefined, which completes the proof.
\end{proof}

\subsection{From \texorpdfstring{$\del$}{DELone} to \texorpdfstring{$\eff$}{EFFone}}\label{app:DELoneToEFFone}
\label{app:subsec:deltoeff_macro_translation}

As in Appendix~\ref{app:EFFoneToDELone}, we use label maps: here a label map
\(\eta\) is a partial function that sends a \(\del\) continuation label \(l\) to
an \(\eff\) continuation label \(\eta(l)\).
We parameterize and extend the translation with \(\eta\) by
\[
  \mte{l}{\eta} \defeq \eta(l),
\]
provided \(\eta(l)\) is defined.
As before, we call this extension a \emph{runtime translation} and write it as
\(\mtempty_{\eta}\); we extend runtime translations to translations on contexts
by mapping a hole to a hole.

Using label maps and the runtime translation, we define coherence conditions and
invariant conditions.

\begin{definition}[coherence conditions]\label{app:def:deltoeff:Coh}
  For a \(\del\) store \(\theta\) and a label map \(\eta\), we define
  \(\Coh(\theta, \eta)\) to hold if and only if the following conditions are
  satisfied.
  \begin{enumerate}
  \item[(C1)] \(\dom{\theta} = \dom{\eta}\);
  \item[(C2)] \(\eta\) is injective.
  \end{enumerate}
\end{definition}

\begin{definition}[invariant conditions]\label{app:def:deltoeff:IC}
  For a \(\del\) store \(\theta\), an \(\eff\) store \(\tau\), and a label map
  \(\eta\), we define \(\Inv(\theta, \tau, \eta)\) to hold if and only if the
  following conditions are satisfied for every \(l \in \dom{\theta}\).
  \begin{enumerate}
  \item[(IC1)] \(\eta(l) \in \dom{\tau}\);
  \item[(IC2)] if \(\theta(l) = \nil\), then \(\tau(\eta(l)) = \nil\);
  \item[(IC3)] if \(\theta(l) = \abs{y}{\dollar{\plug{H}{\return{y}}}{x}{M}}\)
    for a \(\del\) pure context \(\context{H}\) and a \(\del\) computation
    \(M\), then
    \[
      \tau(\eta(l)) = \abs{y}{\handle{\handler{x}{\mte{M}{\eta}}{\opcall{shift0}{\app{p}{k}} \mapsto \app{\force{p}}{k}}}{\mte{\plug{H}{\return{y}}}{\eta}}}.
    \]
  \end{enumerate}
\end{definition}

\begin{definition}[simulation relation]\label{app:def:deltoeff:simulation}
  For a \(\del\) configuration \(C\) and an \(\eff\) configuration \(D\), we
  write \(\simu{C}{D}\) if and only if either \(C = D = \bot\), or there exist
  a \(\del\) computation \(M\), an \(\eff\) computation \(N\), a \(\del\)
  store \(\theta\), an \(\eff\) store \(\tau\), and a label map \(\eta\) such
  that
  \begin{enumerate}
  \item \(C = \config{M}{\theta}\) and \(D = \config{N}{\tau}\);
  \item \(N \equiv \mte{M}{\eta}\);
  \item \(\WF_{\del}(C)\);\footnote{See
      Appendix~\ref{app:sec:well-formedness} for the definition.}
  \item \(\Coh(\theta, \eta)\) and \(\Inv(\theta, \tau, \eta)\).
  \end{enumerate}
\end{definition}

\begin{lemma}\label{lem:deltoeff:property}
  \phantom{}
  \begin{enumerate}
  \item
    For any $\del$ computation $M$, $\simu{\config{M}{\emptyset}}{\config{\mt{M}}{\emptyset}}$ holds.

  \item
    For any $\del$ value $V$, $\simu{\config{\return{V}}{\theta}}{\config{N}{\tau}}$ implies that $N \equiv \return{\mte{V}{\eta}}$ for some $\eta$.
  \end{enumerate}
\end{lemma}
\begin{proof}
  By case analysis on $M$.
\end{proof}

\begin{lemma}\label{lem:deltoeff:subst}
  If a $\del$ computation $M$ has the free variables
  $x_1, \ldots, x_n$, then we have
  $\mte{M[V_1/x_1,\ldots, V_n/x_n]}{\eta} \equiv \mte{M}{\eta}[\mte{V_1}{\eta}/x_1,
  \ldots, \mte{V_n}{\eta}/x_n]$ for any values $V_1, \ldots, V_n$ and $\eta$.
\end{lemma}
\begin{proof}
  By straightforward induction on $M$.
\end{proof}

\begin{lemma}\label{lem:deltoeff:beta_sim}
  If $\simu{C}{D}$ and $\redbetaD{C}{C'}$, then there exists an $\eff$
  configuration $D'$ such that $\redEplus{D}{D'}$ and $\simu{C'}{D'}$.
\end{lemma}
\begin{proof}
  We prove this by case analysis on $\redbetaD{C}{C'}$, giving only nontrivial
  cases.
  The other cases follow by routine arguments.

  \noindent \textbf{Case} $(\mathrm{shift})$:
  \begin{caseindent}
    Suppose that
    $C = \config{\dollar{\plug{H}{\shift{k}{M}}}{x}{N}}{\theta}$,
    $l$ is a fresh label in this configuration, and
    $C' = \config{M[l/k]}{\theta'}$, where
    $\theta' \defeq \theta[l :=
    \abs{y}{\dollar{\plug{H}{\return{y}}}{x}{N}}]$.
    By the definition of $\simu{C}{D}$, there exist \(\tau\) and \(\eta\) such
    that
    \[
      D = \config {\handle
        {H}
        {\plug{\mte{\context{H}}{\eta}}{\opcall{shift0}{\thunk{\abs{k}{\mte{M}{\eta}}}}}}}
      {\tau},
    \]
    where
    $H \equiv \handler{x}{\mte{N}{\eta}}{\opcall{shift0}{\app{p}{k}}
      \mapsto \app{\force{p}}{k}}$.

    $D$ evaluates to the following configuration:
    \[
      D' \defeq \config{\mte{M}{\eta}[m_H/k]}{\tau'},
    \]
    where $\tau' \defeq \tau[m_H := \abs{y}{\handle{H}{\plug{\mte{\context{H}}{\eta}}{\return{y}}}}]$ and $m_H$ is a fresh label.
    Let $\eta'$ be $\eta[l := m_H]$.
    By Lemma~\ref{lem:deltoeff:subst}, we obtain
    \[
      \mte{M}{\eta}[m_H/k] \equiv \mte{M[l/k]}{\eta'}.
    \]
    Moreover, since $m_H$ is taken as a fresh label in $D$, we know that
    \begin{align*}
      & \abs{y}{\handle{H}{\plug{\mte{\context{H}}{\eta}}{\return{y}}}} \\
      &\equiv \abs{y}{\handle{H}{\plug{\mte{\context{H}}{\eta'}}{\return{y}}}} \\
      &\equiv \abs{y}{\handle{H}{\mte{\plug{H}{\return{y}}}{\eta'}}},
    \end{align*}
    and this implies \(\Inv(\theta', \tau', \eta')\).
    The freshness of \(l\) and \(m_H\) gives \(\Coh(\theta', \eta')\).

    Therefore, we conclude that
    \[
      \simu{\config{M[l/k]}{\theta'}}{\config{\mte{M}{\eta}[m_H/k]}{\tau'}},
    \]
    which completes the current case.
  \end{caseindent}

  \noindent \textbf{Case} $(\mathrm{throw})$:
  \begin{caseindent}
    Suppose that $C = \config{\throw{l}{V}}{\theta}$,
    $\theta(l) = \abs{y}{\dollar{\plug{H}{\return{y}}}{x}{N}}$, and
    \[
      C' = \config{\dollar{\plug{H}{\return{V}}}{x}{N}}{\theta[l := \nil]}.
    \]
    By the definition of $\simu{C}{D}$, there exist \(\tau\) and \(\eta\) such
    that
    \[
      D = \config{\throw{\eta(l)}{\mte{V}{\eta}}}{\tau}
    \]
    and, by (IC3),
    \[
      \tau(\eta(l)) = \abs{y}{\handle{\handler{x}{\mte{N}{\eta}}{\opcall{shift0}{\app{p}{k}} \mapsto \app{\force{p}}{k}}}{\mte{\plug{H}{\return{y}}}{\eta}}}.
    \]

    $D$ evaluates to
    \begin{align*}
      D' &\defeq \config{\handle{\handler{x}{\mte{N}{\eta}}{\opcall{shift0}{\app{p}{k}} \mapsto \app{\force{p}}{k}}}{\plug{\mte{\context{H}}{\eta}}{\return{\mte{V}{\eta}}}}}{\tau[\eta(l) := \nil]} \\
         &= \config{\handle{\handler{x}{\mte{N}{\eta}}{\opcall{shift0}{\app{p}{k}} \mapsto \app{\force{p}}{k}}}{\mte{\plug{\context{H}}{\return{V}}}{\eta}}}{\tau[\eta(l) := \nil]}.
    \end{align*}
    Note that \(\Coh(\theta[l := \nil], \eta)\) holds since
    \(\dom{\theta[l := \nil]} = \dom{\theta}\), and that
    \(\Inv(\theta[l := \nil], \tau[\eta(l) := \nil], \eta)\) also holds by (C2).
    Therefore, we obtain $\simu{C'}{D'}$, and this completes the current case.
  \end{caseindent}
\end{proof}

\begin{lemma}\label{lem:deltoeff:plug}
  For every evaluation context \(\context{K}\) and computation $M$, we have
  $\mte{\plug{K}{M}}{\eta} \equiv \mte{\context{K}}{\eta}[\mte{M}{\eta}]$.
\end{lemma}
\begin{proof}
  By straightforward induction on \(\context{K}\).
\end{proof}

\begin{lemma}\label{lem:deltoeff:decomp}
  Suppose that $\simu{\config{M}{\theta}}{\config{N}{\tau}}$ holds, and let
  $\eta$ be a label map witnessing it, so that $N \equiv \mte{M}{\eta}$.
  If $M$ is decomposed into an evaluation context $\context{C}$ and a redex
  $M'$, then $N$ is decomposed into the evaluation context
  $\mte{\context{C}}{\eta}$ and the redex $\mte{M'}{\eta}$.
\end{lemma}
\begin{proof}
  It is easy to check this by induction on the number of frames of $\context{C}$. Note that the runtime translation maps frames to frames and redexes to redexes.
\end{proof}

\begin{proposition}[simulation]\label{prop:deltoeff:sim}
  If $\simu{C}{D}$ and $\redD{C}{C'}$, then there exists an $\eff$
  configuration $D'$ such that $\redEplus{D}{D'}$ and $\simu{C'}{D'}$.
\end{proposition}
\begin{proof}
  Since $C$ is not in normal form, there exist $M$ and $\theta$ such that $C = \config{M}{\theta}$ and according to the semantics of $\del$, $M$ can be decomposed into an evaluation context $\context{C}$ and a redex $M'$.
  By the definition of $\simu{C}{D}$, there exist $\eta$ and $\tau$ such that $D = \config{\mte{M}{\eta}}{\tau}$ and both \(\Coh(\theta, \eta)\) and \(\Inv(\theta, \tau, \eta)\) hold.
  By Lemma~\ref{lem:deltoeff:decomp}, $\mte{M}{\eta}$ can also be decomposed into $\mte{\context{C}}{\eta}$ and a redex $\mte{M'}{\eta}$, namely $\mte{M}{\eta} \equiv \mte{\context{C}}{\eta}\left[\mte{M'}{\eta}\right]$.

  If $\redbetaD{\config{M'}{\theta}}{\bot}$, then we have $\redD{C}{\bot}$.
  By Lemma~\ref{lem:deltoeff:beta_sim}, $\config{\mte{M'}{\eta}}{\tau}$ also evaluates to $\bot$.
  Hence, we conclude that
  $\redEplus{\config{\mte{\context{C}}{\eta}\left[\mte{M'}{\eta}\right]}{\tau}}{\bot}$,
  that is, $\redEplus{D}{\bot}$.

  Next, suppose $\redbetaDp{M'}{\theta}{M''}{\theta'}$.
  By Lemma~\ref{lem:deltoeff:beta_sim}, we obtain $\redEplus{\config{\mte{M'}{\eta}}{\tau}}{\config{\mte{M''}{\eta'}}{\tau'}}$ and $\simu{\config{M''}{\theta'}}{\config{\mte{M''}{\eta'}}{\tau'}}$ for some $\theta'$, $\tau'$, and $\eta'$ with \(\Coh(\theta', \eta')\) and \(\Inv(\theta', \tau', \eta')\).
  Moreover, $\mte{\context{C}}{\eta} \equiv \mte{\context{C}}{\eta'}$ holds since the labels in $\dom{\eta'}\setminus\dom{\eta}$ are fresh and do not appear in $\context{C}$.
  Therefore, we conclude that $\simu{\config{\plug{C}{M''}}{\theta'}}{\config{\mte{\plug{C}{M''}}{\eta'}}{\tau'}}$ by Lemma~\ref{lem:deltoeff:plug}.
\end{proof}

To establish the strong macro-expressibility of $\del$ in $\eff$, we prove three
lemmas.
Recall that \(\WF_{\del}\) is preserved by reduction
(Proposition~\ref{app:prop:del-wellformedness}).

\begin{lemma}\label{lem:deltoeff:diverge}
  Suppose that the evaluation of a \(\del\) program $M$ diverges, i.e., the
  reduction sequence of $M$ is infinite.
  Then the evaluation of $\mt{M}$ also diverges.
\end{lemma}
\begin{proof}
  By applying Proposition~\ref{prop:deltoeff:sim} repeatedly, we also obtain an infinite reduction sequence of $\mt{M}$.
  Since the reduction of $\eff$ is deterministic, this implies that the
  evaluation of $\mt{M}$ also diverges.
\end{proof}

\begin{lemma}\label{lem:deltoeff:stuck}
  Suppose that the evaluation of a \(\del\) program $M$ gets stuck: there
  exists a term $M'$ and a store $\theta$ such that
  $\redDclos{\config{M}{\emptyset}}{\config{M'}{\theta}}$,
  $M' \neq \return{V}$ for any value $V$, and there is no rule that
  can reduce $\config{M'}{\theta}$.
  Then, the evaluation of $\mt{M}$ also gets stuck.
\end{lemma}
\begin{proof}
  It suffices to show the following statement:
  \begin{quote}
    Let $M$, $N$, $\theta$, $\tau$, and $\eta$ be such that
    $N \equiv \mte{M}{\eta}$, \(\WF_{\del}(\config{M}{\theta})\),
    \(\Coh(\theta, \eta)\), and \(\Inv(\theta, \tau, \eta)\) hold.
    If $\config{M}{\theta}$ is stuck, then $\config{N}{\tau}$ is
    also stuck.
  \end{quote}
  We prove it by induction on the structure of $M$, giving only four cases for
  brevity.
  The other cases follow by similar reasoning.

  \noindent \textbf{Case} (product matching):
  \begin{caseindent}
    Suppose \(M \equiv \pcase{V}{x_1}{x_2}{M'}\).
    Since $\config{M}{\theta}$ is stuck, we know that
    $V \neq \vpair{V_1}{V_2}$ for any values $V_1$ and $V_2$.
    Consequently, $\mte{V}{\eta}$ also cannot be of the form
    $\vpair{W_1}{W_2}$ for any values $W_1$ and $W_2$.
    Therefore, the configuration
    $\config{\mte{\pcase{V}{x_1}{x_2}{M'}}{\eta}}{\tau}$ is stuck,
    and this concludes the current case.
  \end{caseindent}

  \noindent \textbf{Case} (throw):
  \begin{caseindent}
    Suppose \(M \equiv \throw{V}{W}\).
    It follows that $V \neq l$ for any continuation label $l$,
    since if $V = l$ for some $l$, then by
    \(\WF_{\del}(\config{\throw{V}{W}}{\theta})\), we know that
    $l \in \dom{\theta}$. However, this contradicts the assumption that
    $\config{\throw{V}{W}}{\theta}$ is stuck.
    This implies that \(\mte{V}{\eta}\) is not an \(\eff\) continuation label
    and so $\mte{\throw{V}{W}}{\eta}$ is also stuck, which concludes the current case.
  \end{caseindent}

  \noindent \textbf{Case} (sequencing):
  \begin{caseindent}
    Suppose \(M \equiv \seq{x}{M_1}{M_2}\).
    Since $\config{M}{\theta}$ is stuck, $M_1$
    cannot be of the form $\return{V}$ for any value $V$.
    By the induction hypothesis, we obtain that $N_1$ is stuck and
    that $N_1 \neq \return{W}$ for any value $W$.
    Consequently, the configuration
    $\config{\seq{x}{N_1}{\mte{M_2}{\eta}}}{\tau}$ is also stuck,
    which completes the current case.
  \end{caseindent}

  \noindent \textbf{Case} (dollar):
  \begin{caseindent}
    Suppose \(M \equiv \dollar{M_1}{x}{M_2}\).
    By the definition of the translation, we obtain
    \[
      N \equiv \handle{\handler{x}{\mte{M_2}{\eta}}{\opcall{shift0}{\app{p}{k}} \mapsto \app{\force{p}}{k}}}{\mte{M_1}{\eta}}.
    \]
    Since $\config{M}{\theta}$ is stuck, $M_1$ cannot be of
    the form $\plug{H}{\shift{k}{L}}$ or of the form $\return{V}$ for
    any pure context $\context{H}$, computation $L$, and value $V$.
    Then, it follows by induction on $M_1$ that $\mte{M_1}{\eta}$
    cannot be of the form $\plug{P}{\opcall{op}{W}}$ or of
    the form $\return{W}$, for any pure context $\context{P}$, and
    value $W$.
    Therefore, the configuration $\config{N}{\tau}$ is also stuck, and
    this completes the current case.
\end{caseindent}
\end{proof}

\begin{lemma}\label{lem:deltoeff:bot}
  If the evaluation of a \(\del\) program $M$ reaches the error state $\bot$,
  i.e., $\redDplus{\config{M}{\emptyset}}{\bot}$, then the evaluation of
  $\mt{M}$ also reaches $\bot$.
\end{lemma}
\begin{proof}
  This lemma directly follows from Proposition~\ref{prop:deltoeff:sim}.
\end{proof}

\begin{proposition}[preservation of semantics]\label{cor:deltoeff:correctness}
  $\eval{\del}{M}$ is defined if and only if $\eval{\eff}{\mt{M}}$ is defined.
\end{proposition}
\begin{proof}
  We first show the ``only if'' direction.
  Suppose that
  $\redDplus{\config{M}{\emptyset}}{\config{\return{V}}{\theta}}$ for
  some $\theta$.
  Lemma~\ref{lem:deltoeff:property} yields that
  $\simu{\config{M}{\emptyset}}{\config{\mt{M}}{\emptyset}}$.
  By applying Proposition~\ref{prop:deltoeff:sim} iteratively,
  there exists $\eta$ and $\tau$ such that
  $\redEplus{\config{\mt{M}}{\emptyset}}{\config{\return{\mte{V}{\eta}}}{\tau}}$.
  This implies that $\eval{\eff}{\mt{M}} = V$.
  We then show the ``if'' direction.
  We prove the contrapositive of the proposition: if $\eval{\del}{M}$
  is undefined, $\eval{\eff}{\mt{M}}$ is also undefined.
  There are three cases where $\eval{\del}{M}$ is undefined: the evaluation of
  \(M\) (1) diverges, (2) gets stuck, or (3) reaches $\bot$.
  In each case, Lemmas~\ref{lem:deltoeff:diverge},
  \ref{lem:deltoeff:stuck}, and \ref{lem:deltoeff:bot} imply that the
  evaluation of $\mt{M}$ diverges, gets stuck, or reaches $\bot$,
  respectively.
  Thus, $\eval{\eff}{\mt{M}}$ is undefined, which completes the proof.
\end{proof}

%% file: appC-non-reftodel.tex
\newcommand{\reach}[2]{\mathrm{Reach}(#1, #2)}
\newcommand{\supp}[1]{\mathrm{Supp}(#1)}

\section{Supplementary proofs for Section~\ref{sec:nonexistent:reftodel}}
\label{app:sec:nonexistent:reftodel}

In this section, we give a proof of Theorem~\ref{thm:reftoeff_nonexist} and
Theorem~\ref{thm:reftoac}.

\subsection{Proof of Theorem~\ref{thm:reftoeff_nonexist}}
\label{app:subsec:nonexistent:reftodel}

In this section, we give the analysis of stores omitted from the proof in
Section~\ref{sec:nonexistent:reftodel}.

Recall that \(\lb{M}\) denotes the set of continuation labels occurring in a
\(\del\) computation \(M\) (Appendix~\ref{app:sec:well-formedness}); we use the
same notation for \(\del\) values.

We then define the labels \emph{reachable} from a term.

\begin{definition}
  For a \(\del\) term \(E\) and a store \(\theta\), define \(\reach{E}{\theta}\)
  to be the smallest subset of \(\dom{\theta}\) such that:
  \begin{itemize}
  \item if \(l \in \lb{E} \cap \dom{\theta}\), then
    \(l \in \reach{E}{\theta}\);
  \item if \(l \in \reach{E}{\theta}\) and \(\theta(l) = K \neq \nil\), then
    every label in \(\lb{K} \cap \dom{\theta}\) also belongs to
    \(\reach{E}{\theta}\).
  \end{itemize}
\end{definition}

Next, we define a relation on stores that captures invalidation.

\begin{definition}
  For a set \(S\) of labels, write
  \[
    \theta \succeq_S \theta'
  \]
  if, for every \(l \in S\), either \(\theta'(l) = \theta(l) \neq \nil\) or
  \(\theta'(l) = \nil\).
  In other words, restricted on \(S\), \(\theta'\) is obtained from \(\theta\)
  by possibly invalidating some entries.
  Note that \(\succeq_{S}\) is transitive given a fixed set \(S\).
\end{definition}

The following lemma states that a reduction sequence changes the entries of the
initial store only by invalidating reachable ones, and that the set of reachable
labels never grows.

\begin{lemma}\label{lem:reduction_and_reach}
  Suppose \(\WF_{\del}(\config{M}{\theta})\) and
  \(\config{M}{\theta} \arrD^{*} \config{M'}{\theta'}\), and let
  \(S \defeq \reach{M}{\theta}\).
  Then,
  \begin{enumerate}
  \item \(\theta \succeq_{S} \theta'\);
  \item if \(l \in \dom{\theta} \setminus S\), then \(\theta'(l) = \theta(l)\);
  \item \(\reach{M}{\theta'} \subseteq S\);
  \item \(\reach{M'}{\theta'} \cap \dom{\theta} \subseteq S\).
  \end{enumerate}
\end{lemma}

\begin{proof}
  By induction on the length of the reduction sequence.
  For the base case, there is nothing to prove.
  For the inductive step, assume
  \[
    \config{M}{\theta} \arrD^{*} \config{M_0}{\theta_0} \arrD
    \config{M'}{\theta'}.
  \]
  By the induction hypothesis, we obtain
  \begin{itemize}
  \item \(\theta \succeq_{S} \theta_0\);
  \item if \(l \in \dom{\theta} \setminus S\), then \(\theta_0(l) = \theta(l)\);
  \item \(\reach{M}{\theta_0} \subseteq S\);
  \item \(\reach{M_0}{\theta_0} \cap \dom{\theta} \subseteq S\).
  \end{itemize}
  Moreover, \(\WF_{\del}(\config{M_0}{\theta_0})\) holds, since well-formedness
  is preserved by reduction (Proposition~\ref{app:prop:del-wellformedness}).
  We conduct case analysis on the rule applied in
  \(\config{M_0}{\theta_0} \arrD \config{M'}{\theta'}\).

  \noindent\textbf{Cases} \((\mam)\) and \((\mathrm{ret})\):
  \begin{caseindent}
    These rules do not change the store, so (1)--(3) trivially hold.
    Since every label occurring in \(M'\) already occurs in \(M_0\), we have
    \(\reach{M'}{\theta'} \subseteq \reach{M_0}{\theta_0}\), and (4) follows
    from the induction hypothesis.
  \end{caseindent}

  \noindent\textbf{Case} \((\mathrm{shift})\):
  \begin{caseindent}
    In this case, a fresh label \(l\) is taken and bound in the store:
    \(\theta' = \theta_0[l := K]\) for some computation \(K\).
    Since \(l \notin \dom{\theta}\), (1) and (2) trivially hold.\footnote{Note
      that \(\succeq_S\) is transitive.}
    By \(\WF_{\del}(\config{M}{\theta})\) and
    \(\WF_{\del}(\config{M_0}{\theta_0})\), the label
    \(l \notin \dom{\theta_0}\) occurs neither in \(M\) nor in \(\theta_0\), so
    \(l\) is not reachable from \(M\).
    Thus \(l \notin \reach{M}{\theta'}\) and
    \(\reach{M}{\theta'} = \reach{M}{\theta_0} \subseteq S\).
    For (4), every label occurring in \(M'\) or in \(K\) is either \(l\) or
    occurs in \(M_0\), so
    \(\reach{M'}{\theta'} \subseteq \reach{M_0}{\theta_0} \cup \{l\}\); since
    \(l \notin \dom{\theta}\), (4) follows from the induction hypothesis.
  \end{caseindent}

  \noindent\textbf{Case} \((\mathrm{throw})\):
  \begin{caseindent}
    This rule invalidates a label \(l\) in \(\dom{\theta_0}\):
    \(\theta' = \theta_0[l := \nil]\).
    Thus, (1) trivially holds.
    Since \(l\) occurs in \(M_0\), we have \(l \in \reach{M_0}{\theta_0}\);
    hence, if \(l \in \dom{\theta}\), then \(l \in S\) by the induction
    hypothesis, and (2) holds.
    We shall prove \(\reach{M}{\theta'} \subseteq \reach{M}{\theta_0}\) by
    induction on the well-founded relation on \(\reach{M}{\theta'}\) induced
    from reachability.
    Let \(m \in \reach{M}{\theta'}\).
    Then, there are two possible cases:
    \begin{enumerate}[label=\arabic*.]
    \item \(m \in \lb{M} \cap \dom{\theta'}\), or
    \item there exists a label \(m' \in \reach{M}{\theta'}\) such that \(\theta'(m') \neq \nil\) and
      \(m \in \lb{\theta'(m')} \cap \dom{\theta'}\).
    \end{enumerate}

    For case 1, \(m\) also belongs to \(\dom{\theta_0}\) since
    \(\dom{\theta'} = \dom{\theta_0}\).
    Thus, \(m \in \reach{M}{\theta_0}\).
    For case 2, \(\theta'(m') \neq \nil\) implies \(m' \neq l\), so
    \(\theta'(m') = \theta_0(m') \neq \nil\).
    Thus \(m \in \lb{\theta_0(m')} \cap \dom{\theta_0}\), and with the
    induction hypothesis that \(m' \in \reach{M}{\theta_0}\), we obtain
    \(m \in \reach{M}{\theta_0}\).

    Therefore, \(\reach{M}{\theta'} \subseteq \reach{M}{\theta_0} \subseteq S\)
    holds.

    Finally, every label in \(\dom{\theta_0}\) occurring in \(M'\) occurs in
    \(M_0\) or in \(\theta_0(l)\), and hence belongs to
    \(\reach{M_0}{\theta_0}\), since \(l \in \reach{M_0}{\theta_0}\) and
    \(\theta_0(l) \neq \nil\).
    By the same argument as above, we obtain
    \(\reach{M'}{\theta'} \subseteq \reach{M_0}{\theta_0}\), and (4) follows
    from the induction hypothesis.
  \end{caseindent}
  This completes the proof.
\end{proof}

Since reduction generates fresh labels, two reduction sequences of the same
computation from different stores may differ in the choice of labels.
We account for this difference by \emph{finite permutations}.

\begin{definition}
  A \emph{finite permutation} \(\pi\) of continuation labels is a bijection on
  \(\syntacticset{L}_{\mathbf{D}}\) whose \emph{support}
  \(\supp{\pi} \defeq \{l \in \syntacticset{L}_{\mathbf{D}} \mid \pi(l) \neq
  l\}\) is finite.
  It acts on values and computations by renaming labels using \(\pi\).
  We write the action as \(\pi \cdot E\) where \(E\) is a value or computation.
  For a store \(\theta\), we define the action \(\pi \cdot \theta\) by
  \[
    (\pi \cdot \theta)(l) \defeq \pi \cdot (\theta(\pi^{-1}(l))),
  \]
  provided \(\pi^{-1}(l) \in \dom{\theta}\).
\end{definition}

We can now state the key lemma: a successfully terminating reduction sequence
can be also obtained from a store with fewer invalidated entries.

\begin{lemma}\label{lem:replay}
  Let \(\theta\) and \(\theta'\) be \(\del\) stores, \(S\) be a subset of
  \(\dom{\theta} \cap \dom{\theta'}\), and \(M\) be a computation such that
  \(\WF_{\del}(\config{M}{\theta'})\).
  Assume there exist a \(\del\) value \(V\) and a store \(\sigma'\) such that
  \[
    \theta \succeq_{S} \theta', \quad \reach{M}{\theta'} \subseteq S, \quad
    \text{and} \quad \config{M}{\theta'} \arrD^* \config{\return{V}}{\sigma'}.
  \]
  Then, there exist a store \(\sigma\) and a finite permutation \(\pi\) with
  \(\pi|_{S} = \mathrm{id}\) such that
  \[
    \config{M}{\theta} \arrD^* \config{\return{(\pi \cdot V)}}{\sigma}.
  \]
\end{lemma}

\begin{proof}
  Let
  \[
    \config{M}{\theta'} = C_0 \arrD C_1 \arrD \cdots \arrD C_n =
    \config{\return{V}}{\sigma'}
  \]
  be the given reduction sequence, where \(C_i = \config{M_i}{\theta_i}\).
  We say that a label is \emph{reachable in \(C_i\)} if it belongs to
  \(\reach{M_i}{\theta_i}\).
  Every \(C_i\) is well-formed by
  Proposition~\ref{app:prop:del-wellformedness}, and thus we have the following
  facts.
  \begin{enumerate}[label=(\alph*)]
  \item Every label occurring in \(M_i\), or in \(\theta_i(l)\) for some label
    \(l\) reachable in \(C_i\), is reachable in \(C_i\).
  \item Every label reachable in \(C_{i+1}\) is either reachable in \(C_i\) or
    generated by the step \(C_i \arrD C_{i+1}\)
    (by Lemma~\ref{lem:reduction_and_reach}(4)).
  \item \(S \subseteq \dom{\theta'} \subseteq \dom{\theta_i}\), since a
    reduction step never removes an entry from the store.
  \end{enumerate}

  By induction on \(i \leq n\), we construct a reduction sequence
  \[
    \config{M}{\theta} = \bar{C}_0 \arrD \bar{C}_1 \arrD \cdots \arrD \bar{C}_i,
  \]
  where \(\bar{C}_i = \config{\bar{M}_i}{\bar{\theta}_i}\), together with a
  finite permutation \(\pi_i\) with \(\pi_i|_{S} = \mathrm{id}\) such that
  \begin{enumerate}
  \item \(\bar{M}_i = \pi_i \cdot M_i\);
  \item for every label \(l\) reachable in \(C_i\), we have
    \(\pi_i(l) \in \dom{\bar{\theta}_i}\), and if \(\theta_i(l) \neq \nil\),
    then \(\bar{\theta}_i(\pi_i(l)) = \pi_i \cdot \theta_i(l)\).
  \end{enumerate}
  In other words, \(\bar{C}_i\) agrees with \(C_i\) on its reachable part, up to
  the renaming \(\pi_i\).
  The lemma then follows by taking \(\pi \defeq \pi_n\) and
  \(\sigma \defeq \bar{\theta}_n\), since
  \(\bar{M}_n = \pi_n \cdot \return{V} = \return{(\pi_n \cdot V)}\).

  For \(i = 0\), we take \(\pi_0 \defeq \mathrm{id}\), so that (1) trivially
  holds.
  For (2), every label \(l\) reachable in \(C_0\) belongs to
  \(S \subseteq \dom{\theta}\) by assumption, and if \(\theta'(l) \neq \nil\),
  then \(\theta \succeq_S \theta'\) implies \(\theta(l) = \theta'(l)\).

  For the inductive step, we proceed by case analysis on the step
  \(C_i \arrD C_{i+1}\).
  In each case, we display only the redex of \(M_i\), leaving the surrounding
  evaluation context implicit: by (1), the redex and the evaluation context of
  \(\bar{M}_i\) are the images of those of \(M_i\) under \(\pi_i\).

  \noindent\textbf{Cases} \((\mam)\) and \((\mathrm{ret})\):
  \begin{caseindent}
    The step is determined by the shape of \(M_i\) and does not change the
    store.
    Since \(\bar{M}_i = \pi_i \cdot M_i\), the same rule applies to
    \(\bar{C}_i\), yielding \(\bar{M}_{i+1} = \pi_i \cdot M_{i+1}\) with the
    store unchanged.
    We take \(\pi_{i+1} \defeq \pi_i\).
    By (b), every label reachable in \(C_{i+1}\) is reachable in \(C_i\), and
    the stores are unchanged; thus (2) for \(i\) gives (2) for \(i+1\).
  \end{caseindent}

  \noindent\textbf{Case} \((\mathrm{shift})\):
  \begin{caseindent}
    Suppose that the step reduces the redex as
    \[
      \config{\dollar{\plug{H}{\shift{k}{P}}}{x}{N}}{\theta_i} \arrD
      \config{P[l/k]}{\theta_i[l := K]},
      \qquad K \defeq \abs{y}{\dollar{\plug{H}{\return{y}}}{x}{N}},
    \]
    where \(l\) is a fresh label.
    By (1), the same rule applies to \(\bar{C}_i\), generating a fresh label
    \(l'\) and storing \(\pi_i \cdot K\) under it.
    Let \(\tau\) be the transposition that swaps \(\pi_i(l)\) and \(l'\) and
    define
    \[
      \pi_{i+1} \defeq \tau \circ \pi_i,
    \]
    so that \(\pi_{i+1}(l) = \tau(\pi_i(l)) = l'\).

    We check that \(\pi_{i+1}\) satisfies the required conditions.
    Since \(\tau\) moves only \(\pi_i(l)\) and \(l'\), it suffices to show that
    these two labels do not interfere with the correspondence established so
    far, namely that
    \begin{enumerate}[label=(\roman*)]
    \item neither of them belongs to \(S\);
    \item neither of them is \(\pi_i(m)\) for a label \(m\) reachable in
      \(C_i\).
    \end{enumerate}
    Consider \(\pi_i(l)\) first.
    Since \(l \notin \dom{\theta_i} \supseteq S\) by (c) and \(\pi_i\) is a
    bijection acting on \(S\) identically, we have \(\pi_i(l) \notin S\).
    Since a label \(m\) reachable in \(C_i\) belongs to \(\dom{\theta_i}\), we
    have \(m \neq l\), and hence \(\pi_i(m) \neq \pi_i(l)\) by injectivity.
    Consider \(l'\) next.
    We have \(l' \notin \dom{\bar{\theta}_i}\), whereas
    \(S \subseteq \dom{\bar{\theta}_i}\) by (c) and
    \(\pi_i(m) \in \dom{\bar{\theta}_i}\) by (2).

    By (i), \(\tau\) acts on \(S\) identically, and hence
    \(\pi_{i+1}|_{S} = \mathrm{id}\).
    By (ii), \(\tau\) fixes \(\pi_i(m)\) for every label \(m\) reachable in
    \(C_i\), and hence \(\pi_{i+1}(m) = \pi_i(m)\).
    By (a), every label occurring in \(M_i\), or in \(\theta_i(m)\) for such
    \(m\), is reachable in \(C_i\); therefore
    \(\pi_{i+1} \cdot M_i = \pi_i \cdot M_i\) and
    \(\pi_{i+1} \cdot \theta_i(m) = \pi_i \cdot \theta_i(m)\).
    In particular, \(\pi_{i+1} \cdot P = \pi_i \cdot P\) and
    \(\pi_{i+1} \cdot K = \pi_i \cdot K\), since every label occurring in
    \(P\) or in \(K\) occurs in \(M_i\).

    We now verify (1) and (2) for \(i + 1\).
    The step in \(\bar{C}_i\) replaces the redex by
    \((\pi_i \cdot P)[l'/k]\) and sets
    \(\bar{\theta}_{i+1} = \bar{\theta}_i[l' := \pi_i \cdot K]\).
    For (1), we have
    \[
      (\pi_i \cdot P)[l'/k]
      = (\pi_{i+1} \cdot P)[\pi_{i+1}(l)/k]
      = \pi_{i+1} \cdot (P[l/k]).
    \]
    For (2), let \(m\) be a label reachable in \(C_{i+1}\).
    By (b), either \(m = l\) or \(m\) is reachable in \(C_i\).
    If \(m = l\), then \(\pi_{i+1}(l) = l' \in \dom{\bar{\theta}_{i+1}}\) and
    \(\bar{\theta}_{i+1}(l') = \pi_i \cdot K = \pi_{i+1} \cdot
    \theta_{i+1}(l)\).
    Otherwise, \(m \neq l\) and \(\pi_{i+1}(m) = \pi_i(m) \neq l'\), so the
    entries of \(m\) in \(\theta_{i+1}\) and of \(\pi_{i+1}(m)\) in
    \(\bar{\theta}_{i+1}\) are the same as those in \(\theta_i\) and
    \(\bar{\theta}_i\), respectively.
    Thus, (2) for \(i\) gives the claim.
  \end{caseindent}

  \noindent\textbf{Case} \((\mathrm{throw})\):
  \begin{caseindent}
    Suppose that the step reduces the redex as
    \[
      \config{\throw{l}{W}}{\theta_i} \arrD
      \config{\dollar{\plug{H}{\return{W}}}{x}{N}}{\theta_i[l := \nil]},
    \]
    where \(\theta_i(l) = \abs{y}{\dollar{\plug{H}{\return{y}}}{x}{N}}\).
    Since \(l\) occurs in \(M_i\), it is reachable in \(C_i\) by (a), and (2)
    gives \(\bar{\theta}_i(\pi_i(l)) = \pi_i \cdot \theta_i(l)\).
    Hence, the same rule applies to \(\bar{C}_i\), yielding
    \(\bar{M}_{i+1} = \pi_i \cdot M_{i+1}\) and
    \(\bar{\theta}_{i+1} = \bar{\theta}_i[\pi_i(l) := \nil]\).
    We take \(\pi_{i+1} \defeq \pi_i\).
    (1) is trivial.
    For (2), let \(m\) be a label reachable in \(C_{i+1}\).
    By (b), \(m\) is reachable in \(C_i\), so
    \(\pi_i(m) \in \dom{\bar{\theta}_i} = \dom{\bar{\theta}_{i+1}}\).
    If \(m = l\), there is nothing more to check, since
    \(\theta_{i+1}(l) = \nil\).
    Otherwise, \(\pi_i(m) \neq \pi_i(l)\), so the entries of \(m\) and of
    \(\pi_i(m)\) are not changed by the step, and (2) for \(i\) gives the claim.
  \end{caseindent}

  \noindent\textbf{Case} \((\mathrm{fail})\):
  \begin{caseindent}
    This case is impossible.
  \end{caseindent}
  This concludes the proof.
\end{proof}

We also restate Lemma~\ref{lem:reset-free-decomp} of the main text.

\begin{lemma}\label{app:lem:reset-free-decomp}
  If \(\config{\seq{x}{M}{N}}{\theta} \arrD^* \config{\return{V}}{\theta''}\),
  then there exist a value \(V'\) and a store \(\theta'\) such that
  \(\config{M}{\theta} \arrD^* \config{\return{V'}}{\theta'}\) and
  \(\config{N[V'/x]}{\theta'} \arrD^* \config{\return{V}}{\theta''}\).
\end{lemma}

\begin{proof}
  By induction on the length of the reduction sequence.
  Since the frame \(\seq{x}{\hole}{N}\) contains no dollar frame, redexes for
  the \((\mathrm{shift})\) rule can occur only within \(M\).
  The frame is preserved until \(M\) becomes \(\return{V'}\), as the only rule
  that can consume the frame is \((F)\).
\end{proof}

We now prove Theorem~\ref{thm:reftoeff_nonexist}.

\begin{proof}[Proof of Theorem~\ref{thm:reftoeff_nonexist}]
  Let
  \( \Omega \defeq
  \app{(\abs{x}{\app{\force{x}}{x}})}{\thunk{\abs{x}{\app{\force{x}}{x}}}} \)
  and
  \[
    \Delta_{\mathrm{A}, \mathrm{B}}(x, y) \defeq
    \begin{pmatrix*}[l]
      \mathbf{case}\;(x, y)\;\mathbf{of}\; \{\\
      \quad (\inj{A}{\underscore}, \inj{B}{\underscore}) \mapsto \return{\unit} \\
      \quad \underscore \mapsto \Omega \}
    \end{pmatrix*}.
  \]
  Consider the \(\reflang\) program
  \[
    M_{\neq} \defeq
    \begin{pmatrix*}[l]
      \letin{r}{\create{\inj{A}{\unit}}}
      {\\\letin{i}{\get{r}}
      {\\\letin{\underscore}{\set{r}{\inj{B}{\unit}}}
      {\\\letin{j}{\get{r}}
      {\\\Delta_{\mathrm{A}, \mathrm{B}}(i, j)}}}}
    \end{pmatrix*}.
  \]
  In \(\reflang\), the first \(\mathbf{get}\) returns \(\inj{A}{\unit}\) and the
  second \(\mathbf{get}\) returns \(\inj{B}{\unit}\), and therefore
  \(\eval{\reflang}{M_{\neq}}\) is defined.

  Suppose that there exists a weak macro-translation \(\mtempty\) from
  \(\reflang\) to \(\del\), and let \(\ccreate\), \(\cget\), \(\cset\) be the
  syntactic abstractions corresponding to \textbf{create}, \textbf{get}, and
  \textbf{set}, respectively.
  Then, since \(\mtempty\) acts homomorphically on \(\mam\) constructors, the
  translation \(\mt{M_{\neq}}\) has the shape
  \[
    \begin{pmatrix*}[l]
      \letin{r}{\ccreate\left[\inj{A}{\unit}\right]}
      {\\\letin{i}{\cget\left[r\right]}
      {\\\letin{\underscore}{\cset\left[r, \inj{B}{\unit}\right]}
      {\\\letin{j}{\cget[r]}
      {\\\Delta_{\mathrm{A}, \mathrm{B}}(i, j)}}}}
    \end{pmatrix*},
  \]
  and contains no continuation label.
  Since \(\eval{\reflang}{M_{\neq}}\) is defined, so is
  \(\eval{\del}{\mt{M_{\neq}}}\):
  there exists a store \(\theta\) such that
  \[
    \config{\mt{M_{\neq}}}{\emptyset} \arrD^* \config{\return{\unit}}{\theta}.
  \]

  Applying Lemma~\ref{app:lem:reset-free-decomp} repeatedly to this reduction
  sequence (note that there is no dollar frame enclosing \(\mt{M_{\neq}}\))
  yields values \(X, V_1, W, V_2\) and stores
  \(\theta_0, \theta_1, \theta_2, \theta_3\) such that
  \begin{align}
    \config{\ccreate\left[\inj{A}{\unit}\right]}{\emptyset} &\arrD^* \config{\return{X}}{\theta_0}, \label{eq:appC:1}\\
    \config{\cget\left[X\right]}{\theta_0} &\arrD^* \config{\return{V_1}}{\theta_1}, \label{eq:appC:2}\\
    \config{\cset\left[X, \inj{B}{\unit}\right]}{\theta_1} &\arrD^* \config{\return{W}}{\theta_2}, \label{eq:appC:3}\\
    \config{\cget[X]}{\theta_2} &\arrD^* \config{\return{V_2}}{\theta_3}, \label{eq:appC:4}
  \end{align}
  and moreover the evaluation of \(\Delta_{\mathrm{A}, \mathrm{B}}(V_1, V_2)\)
  terminates successfully.
  Thus, \(V_1 \equiv \inj{A}{V_1'}\) and \(V_2 \equiv \inj{B}{V_2'}\) for some
  values \(V_1'\) and \(V_2'\).

  Note that every configuration appearing above is well-formed:
  \(\config{\mt{M_{\neq}}}{\emptyset}\) is well-formed, well-formedness is
  preserved by reduction (Proposition~\ref{app:prop:del-wellformedness}), and
  a
  subterm of the computation of a well-formed configuration again forms a
  well-formed configuration with the same store.

  Let
  \(S \defeq \reach{\cget\left[X\right]}{\theta_0} = \reach{X}{\theta_0}\) and
  \(S' \defeq \reach{\cset\left[X, \inj{B}{\unit}\right]}{\theta_1} =
  \reach{X}{\theta_1}\) (note that the syntactic abstractions \(\cget\) and
  \(\cset\) contain no continuation label).
  By applying Lemma~\ref{lem:reduction_and_reach} to~\eqref{eq:appC:2}, we
  obtain \(\theta_0 \succeq_S \theta_1\) and \(S' \subseteq S\).
  By applying it to~\eqref{eq:appC:3}, we obtain
  \(\theta_1 \succeq_{S'} \theta_2\), \(\theta_2(l) = \theta_1(l)\) for every
  \(l \in \dom{\theta_1} \setminus S'\), and
  \(\reach{X}{\theta_2} \subseteq S'\).
  The first two imply \(\theta_1 \succeq_S \theta_2\), since
  \(S \subseteq \dom{\theta_1}\).
  Hence, we have
  \[
    \theta_0 \succeq_S \theta_1 \succeq_S \theta_2
    \quad\text{and}\quad
    \reach{\cget\left[X\right]}{\theta_2} = \reach{X}{\theta_2} \subseteq S.
  \]
  Thus, by applying Lemma~\ref{lem:replay} to~\eqref{eq:appC:2}
  and~\eqref{eq:appC:4}, we obtain
  \[
    \config{\cget\left[X\right]}{\theta_0} \arrD^* \config{\pi \cdot \return{V_2}}{\sigma_0'},
  \]
  for some finite permutation \(\pi\) and store \(\sigma_0'\).
  The determinism of the reduction implies
  \[
    V_1 = \pi \cdot V_2.
  \]
  Since finite permutations act homomorphically on constructors other than
  continuation labels, we obtain
  \begin{align*}
    V_1 &= \pi \cdot V_2 \\
        &= \pi \cdot (\inj{B}{V_2'}) \\
        &= \inj{B}{(\pi \cdot V_2')},
  \end{align*}
  which contradicts \(V_1 \equiv \inj{A}{V_1'}\).
  Therefore, there is no weak macro-translation from \(\reflang\) to \(\del\).
\end{proof}

\subsection{Proof of Theorem~\ref{thm:reftoac}}
\label{subsec:reftoac}

As in Section~\ref{sec:DELoneToAC:simulation}, we use label maps: here a label
map \(\eta\) is a partial function
\(\eta : \syntacticset{L}_{\mathbf R} \rightharpoonup \syntacticset{L}_{\ac}\)
that sends a reference cell \(l\) to a coroutine label \(\eta(l)\).
We parameterize and extend the translation \(M \mapsto \mt{M}\) with \(\eta\) by
\[
  \mte{l}{\eta} \defeq \eta(l),
\]
provided \(\eta(l)\) is defined.
We call this extension a \emph{runtime translation} and write it as
\(\mtempty_{\eta}\).

Using label maps and the runtime translation, we define invariant conditions and
the simulation relation.

\begin{definition}[invariant conditions]\label{app:def:reftoac:IC}
  For a \(\reflang\) store \(\theta\), an \(\ac\) store \(\tau\), and a label
  map \(\eta\), we define \(\Inv(\theta, \tau, \eta)\) to hold if and only if
  the following conditions are satisfied.
  \begin{enumerate}
  \item[(IC1)] \(\eta\) is injective;
  \item[(IC2)] \(\eta(\dom{\theta}) \subseteq \dom{\tau}\);
  \item[(IC3)] for any \(l \in \dom{\theta}\), if \(\theta(l) = V\), then
    \[
      \tau(\eta(l)) = \thunk{\abs{x}{\letin{y}{\return{x}}{\app{\app{\force{\var{loop}}}{\mte{V}{\eta}}}{y}}}}.
    \]
  \end{enumerate}
\end{definition}

\begin{definition}[simulation relation]\label{app:def:reftoac:simulation}
  For a \(\reflang\) configuration \(C\) and an \(\ac\) configuration \(D\), we
  write \(\simu{C}{D}\) if and only if there exist a \(\reflang\) computation
  \(M\), an \(\ac\) computation \(N\), a \(\reflang\) store \(\theta\), an
  \(\ac\) store \(\tau\), and a label map \(\eta\) such that
  \begin{enumerate}
  \item \(C = \config{M}{\theta}\) and \(D = \config{N}{\tau}\);
  \item \(N \equiv \mte{M}{\eta}\);
  \item \(\Inv(\theta, \tau, \eta)\).
  \end{enumerate}
\end{definition}

We prove the following property of $\var{loop}$, where \(\var{loop}\) and
\(\var{th}\) are the thunks defined in Figure~\ref{fig:deltoac}.
\begin{lemma}\label{lem:reftoac_loop}
  For any $\ac$ store $\theta$ and value $V$,
  \[
    \redACplus{\config{\app{\force{\var{loop}}}{V}}{\theta}}{\config{\app{\app{\force{\var{th}}}{\var{loop}}}{V}}{\theta}}.
  \]
\end{lemma}
\begin{proof}
  By straightforward calculation.
\end{proof}

Next, we establish a substitution property for the translation.
\begin{lemma}\label{lem:reftoac_subst}
  Suppose that a $\reflang$ computation $M$ has the free variables
  $x_1, \ldots, x_n$.
  Then for any values $V_1, \ldots, V_n$ and label map $\eta$, the following
  equality holds:
  \[
    \mte{M[V_1/x_1,\ldots, V_n/x_n]}{\eta} \equiv \mte{M}{\eta}[\mte{V_1}{\eta}/x_1, \ldots, \mte{V_n}{\eta}/x_n].
  \]
\end{lemma}
\begin{proof}
  By induction on $M$.
\end{proof}

\begin{lemma}\label{lem:reftoac_beta_sim}
  If $\simu{C}{D}$ and $\redbetaR{C}{C'}$, then there exists an $\ac$
  configuration $D'$ such that $\redACplus{D}{D'}$ and $\simu{C'}{D'}$.
\end{lemma}
\begin{proof}
  We prove this by case analysis on $\redbetaR{C}{C'}$, giving only the cases
  for the rules specific to $\reflang$.
  The other cases follow by routine arguments.

  \noindent \textbf{Case} $(\mathrm{create})$:
  \begin{caseindent}
    Suppose that $C = \config{\create{V}}{\theta}$ and
    $C' = \config{\return{l}}{\theta'}$, where \(\theta' \defeq \theta[l := V]\)
    and $l$ is a fresh label.
    By the definition of $\simu{C}{D}$, there exists $\eta$ such that
    \[
      D = \config{\create{\thunk{\abs{x}{\letin{y}{\return{x}}{\app{\app{\force{\var{loop}}}{\mte{V}{\eta}}}{y}}}}}}{\tau}.
    \]
    The configuration $D$ evaluates to $D' = \config{\return{l'}}{\tau'}$,
    where \(l'\) is a fresh label and
    \[
      \tau' \defeq \tau\left[ l' := \thunk{\abs{x}{\letin{y}{\return{x}}{\app{\app{\force{\var{loop}}}{\mte{V}{\eta}}}{y}}}}\right].
    \]
    Let $\eta'$ be $\eta[l := l']$.
    By Lemma~\ref{lem:reftoac_subst}, we have $\return{l'} \equiv \mte{\return{l}}{\eta'}$.
    Therefore, we conclude that
    \[
      \simu{C'}{D'},
    \]
    witnessed by \(\eta'\).
    Note that \(\mte{V}{\eta} \equiv \mte{V}{\eta'}\) since the fresh label
    \(l\) does not occur in \(V\), and that \(\Inv(\theta', \tau', \eta')\)
    holds.
  \end{caseindent}
 
  \noindent \textbf{Case} $(\mathrm{set})$:
  \begin{caseindent}
      Suppose that $C = \config{\set{l}{V}}{\theta}$ and $C' = \config{\return{\unit}}{\theta[l := V]}$ for some $l \in \dom{\theta}$.
  By the definition of $\simu{C}{D}$, there exists $\eta$ such that
  \[
    D = \config{\resume{(\eta(l))}{\paren{\inj{Set}{\mte{V}{\eta}}}}}{\tau}
  \]
  and
  \[
    \tau(\eta(l)) = \thunk{\abs{x}{\letin{y}{\return{x}}{\app{\app{\force{\var{loop}}}{\mte{\theta(l)}{\eta}}}{y}}}}.
  \]
  Let \(\tau_0 \defeq \tau[\eta(l) := \nil]\).
  By Lemma~\ref{lem:reftoac_loop} and the definition of \(\var{th}\) in
  Figure~\ref{fig:deltoac}, the reduction of $D$ proceeds as follows:
  \begin{align*}
    D &= \config{\resume{(\eta(l))}{\paren{\inj{Set}{\mte{V}{\eta}}}}}{\tau} \\
      &\arrAC \config{\labeledcp{(\eta(l))}{\app{\force{\thunk{\abs{x}{\letin{y}{\return{x}}{\app{\app{\force{\var{loop}}}{\mte{\theta(l)}{\eta}}}{y}}}}}}{\paren{\inj{Set}{\mte{V}{\eta}}}}}}{\tau_0} \\
      &\arrAC^{+} \config{\labeledcp{(\eta(l))}{\app{\app{\app{\force{\var{th}}}{\var{loop}}}{\mte{\theta(l)}{\eta}}}{\paren{\inj{Set}{\mte{V}{\eta}}}}}}{\tau_0} \\
      &\arrAC^{+} \config{
          \labeledcp{(\eta(l))}{{\begin{array}{l}
            \mathbf{case}\;\paren{\inj{Set}{\mte{V}{\eta}}}\;\mathbf{of} \{\\
            \quad(\inj{Set}{\var{v}}) \mapsto \letin{q'}{\yield{\unit}}{\app{\app{\force{\var{loop}}}{\var{v}}}{\var{q'}}} \\
            \quad(\inj{Get}{\underscore}) \mapsto \letin{q'}{\yield{\mte{\theta(l)}{\eta}}}{\app{\app{\force{\var{loop}}}{\mte{\theta(l)}{\eta}}}{\var{q'}}}\}
          \end{array}}}
        }{\tau_0} \\
      &\arrAC \config{\labeledcp{(\eta(l))}{\letin{q'}{\yield{\unit}}{\app{\app{\force{\var{loop}}}{\mte{V}{\eta}}}{\var{q'}}}}}{\tau_0} \\
      &\arrAC \config{\return{\unit}}{\tau[(\eta(l)) := \thunk{\abs{x}{\letin{y}{\return{x}}{\app{\app{\force{\var{loop}}}{\mte{V}{\eta}}}{y}}}}]}
        \end{align*}        
        Therefore, we obtain
        \[
          \simu{C'}{\config{\return{\unit}}{\tau'}},
        \]
        where \(\tau' \defeq \tau[(\eta(l)) := \thunk{\abs{x}{\letin{y}{\return{x}}{\app{\app{\force{\var{loop}}}{\mte{V}{\eta}}}{y}}}}]\), since \(\Inv(\theta[l := V], \tau', \eta)\) holds.
  \end{caseindent}

  \noindent \textbf{Case} $(\mathrm{get})$:
  \begin{caseindent}

  Suppose that $C = \config{\get{l}}{\theta}$ and $C' = \config{\return{V}}{\theta}$, where $l \in \dom{\theta}$ and $\theta(l) = V$.
  By the definition of $\simu{C}{D}$, there exists $\eta$ such that
  \[
    D = \config{\resume{(\eta(l))}{\paren{\inj{Get}{\unit}}}}{\tau}
  \]
  and
  \[
    \tau(\eta(l)) = \thunk{\abs{x}{\letin{y}{\return{x}}{\app{\app{\force{\var{loop}}}{\mte{\theta(l)}{\eta}}}{y}}}}.
  \]
  Let \(\tau_0 \defeq \tau[\eta(l) := \nil]\).
  By Lemma~\ref{lem:reftoac_loop} and the definition of \(\var{th}\) in
  Figure~\ref{fig:deltoac}, the reduction of $D$ proceeds as follows:
    \begin{align*}
    D &= \config{\resume{(\eta(l))}{\paren{\inj{Get}{\unit}}}}{\tau} \\
      &\arrAC \config{\labeledcp{(\eta(l))}{\app{\force{\thunk{\abs{x}{\letin{y}{\return{x}}{\app{\app{\force{\var{loop}}}{\mte{\theta(l)}{\eta}}}{y}}}}}}{\paren{\inj{Get}{\unit}}}}}{\tau_0} \\
      &\arrAC^{+} \config{\labeledcp{(\eta(l))}{\app{\app{\app{\force{\var{th}}}{\var{loop}}}{\mte{\theta(l)}{\eta}}}{\paren{\inj{Get}{\unit}}}}}{\tau_0} \\
      &\arrAC^{+} \config{
          \labeledcp{(\eta(l))}{{\begin{array}{l}
            \mathbf{case}\;\paren{\inj{Get}{\unit}}\;\mathbf{of} \{\\
            \quad(\inj{Set}{\var{v}}) \mapsto \letin{q'}{\yield{\unit}}{\app{\app{\force{\var{loop}}}{\var{v}}}{\var{q'}}} \\
            \quad(\inj{Get}{\underscore}) \mapsto \letin{q'}{\yield{\mte{\theta(l)}{\eta}}}{\app{\app{\force{\var{loop}}}{\mte{\theta(l)}{\eta}}}{\var{q'}}}\}
          \end{array}}}
        }{\tau_0} \\
      &\arrAC \config{\labeledcp{(\eta(l))}{\letin{q'}{\yield{\mte{\theta(l)}{\eta}}}{\app{\app{\force{\var{loop}}}{\mte{\theta(l)}{\eta}}}{\var{q'}}}}}{\tau_0} \\
      &\arrAC \config{\return{\mte{\theta(l)}{\eta}}}{\tau[(\eta(l)) := \thunk{\abs{x}{\letin{y}{\return{x}}{\app{\app{\force{\var{loop}}}{\mte{\theta(l)}{\eta}}}{y}}}}]}
        \end{align*}
  
        Therefore, we obtain
        \[
          \simu{C'}{\config{\return{\mte{\theta(l)}{\eta}}}{\tau'}},
        \]
        where \(\tau' \defeq \tau[(\eta(l)) := \thunk{\abs{x}{\letin{y}{\return{x}}{\app{\app{\force{\var{loop}}}{\mte{\theta(l)}{\eta}}}{y}}}}]\), since \(\Inv(\theta, \tau', \eta)\) holds.
      \end{caseindent}
     
      \end{proof}

We also extend runtime translations to translations on contexts by mapping a hole to a hole, and write $\mte{\context{K}}{\eta}$ for the translation of a context $\context{K}$.

\begin{lemma}\label{lem:reftoac_plug}
  For every evaluation context \(\context{K}\) and computation $M$, we have
  $\mte{\plug{K}{M}}{\eta} \equiv \mte{\context{K}}{\eta}[\mte{M}{\eta}]$.
\end{lemma}
\begin{proof}
  By straightforward induction on \(\context{K}\).
\end{proof}

\begin{lemma}\label{lem:reftoac_decomp}
  Suppose that $\simu{\config{M}{\theta}}{\config{N}{\tau}}$ holds, and let
  $\eta$ be a label map witnessing it, so that $N \equiv \mte{M}{\eta}$.
  If $M$ is decomposed into an evaluation context $\context{C}$ and a redex
  $M'$, then $N$ is decomposed into the evaluation context
  $\mte{\context{C}}{\eta}$ and the redex $\mte{M'}{\eta}$.
\end{lemma}
\begin{proof}
  It is easy to check this by induction on the number of frames of $\context{C}$.
  Note that the runtime translation maps frames to frames and redexes to redexes.
\end{proof}

\begin{proposition}[simulation]\label{prop:reftoac_sim}
  If $\simu{C}{D}$ and $\redR{C}{C'}$, then there exists an $\ac$
  configuration $D'$ such that $\redACplus{D}{D'}$ and $\simu{C'}{D'}$.
\end{proposition}
\begin{proof}
  By the definition of $\simu{C}{D}$, there exist $M$, $\theta$, $\tau$, and $\eta$ such that $C = \config{M}{\theta}$, $D = \config{\mte{M}{\eta}}{\tau}$, and \(\Inv(\theta, \tau, \eta)\) holds.
  Since $C$ is not in normal form, $M$ can be decomposed into an evaluation context $\context{C}$ and a redex $M'$.
  By Lemma~\ref{lem:reftoac_decomp}, $\mte{M}{\eta}$ can also be decomposed into $\mte{\context{C}}{\eta}$ and a redex $\mte{M'}{\eta}$, namely, $\mte{M}{\eta} \equiv \mte{\context{C}}{\eta}\left[\mte{M'}{\eta}\right]$.
  
  Now, suppose $\redbetaRp{M'}{\theta}{M''}{\theta'}$.
  By Lemma~\ref{lem:reftoac_beta_sim}, there exist $\eta'$ and $\tau'$ such that $\redACplus{\config{\mte{M'}{\eta}}{\tau}}{\config{\mte{M''}{\eta'}}{\tau'}}$, $\simu{\config{M''}{\theta'}}{\config{\mte{M''}{\eta'}}{\tau'}}$, and \(\Inv(\theta', \tau', \eta')\) holds.
  Moreover, $\mte{\context{C}}{\eta} \equiv \mte{\context{C}}{\eta'}$ holds since the labels in $\dom{\eta'}\setminus\dom{\eta}$ are fresh and do not appear in $C$.
  Hence, we conclude that $\simu{\config{\plug{C}{M''}}{\theta'}}{\config{\mte{\plug{C}{M''}}{\eta'}}{\tau'}}$ by Lemma~\ref{lem:reftoac_plug}.
\end{proof}

\begin{lemma}\label{lem:reftoac_initial}
  For any $\reflang$ computation $M$, $\simu{\config{M}{\emptyset}}{\config{\mt{M}}{\emptyset}}$ holds.
\end{lemma}
\begin{proof}
  By induction on $M$.
\end{proof}

\begin{definition}\label{app:def:reftoac:WF}
  As in Appendix~\ref{app:sec:well-formedness}, for a \(\reflang\) configuration
  \(C = \config{M}{\theta}\), let \(\lb{C}\) be the set of reference cells
  occurring in \(M\) or in \(\theta(l)\) for some \(l \in \dom{\theta}\).
  The configuration \(\config{M}{\theta}\) is \emph{well-formed}, written
  \(\WF_{\reflang}(\config{M}{\theta})\), if and only if
  \[
    \lb{\config{M}{\theta}} \subseteq \dom{\theta}.
  \]
\end{definition}

\begin{proposition}[preservation of \(\reflang\) well-formedness]\label{app:prop:reftoac:WF}
  If \(\WF_{\reflang}(\config{M}{\theta})\) and
  \(\redRp{M}{\theta}{M'}{\theta'}\), then
  \(\WF_{\reflang}(\config{M'}{\theta'})\).
\end{proposition}
\begin{proof}
  We prove this by case analysis on the \(\reflang\) beta-reduction rule used in
  the step.
  Since an evaluation context is unchanged by the reduction, it suffices to
  check that the reduct and the updated store contain only reference cells
  defined in the updated store.

  In the \(\mam\) cases, the store is unchanged, and every reference cell
  occurring in the reduct already occurs in the redex.
  In the case \((\mathrm{create})\), the reduct \(\return{l}\) contains only the
  fresh reference cell \(l \in \dom{\theta[l := V]}\), and the new entry \(V\)
  occurs in the redex.
  In the case \((\mathrm{set})\), the reduct contains no reference cell, the new
  entry \(V\) occurs in the redex, and \(\dom{\theta[l := V]} = \dom{\theta}\).
  In the case \((\mathrm{get})\), the store is unchanged, and the reduct
  \(\return{\theta(l)}\) contains only reference cells in
  \(\lb{\config{M}{\theta}} \subseteq \dom{\theta}\).
\end{proof}

\begin{lemma}\label{lem:reftoac:diverge}
  Suppose that the evaluation of a \(\reflang\) program $M$ diverges, i.e., the
  reduction sequence of $M$ is infinite.
  Then the evaluation of $\mt{M}$ also diverges.
\end{lemma}
\begin{proof}
  By applying Proposition~\ref{prop:reftoac_sim} repeatedly, we also obtain an infinite reduction sequence of $\mt{M}$.
  Since the reduction of $\ac$ is deterministic, this implies that the
  evaluation of $\mt{M}$ also diverges.
\end{proof}

\begin{lemma}\label{lem:reftoac:stuck}
  Suppose that the evaluation of a \(\reflang\) program $M$ gets stuck: there
  exists a term $M'$ and a store $\theta$ such that
  $\redRclos{\config{M}{\emptyset}}{\config{M'}{\theta}}$,
  $M' \neq \return{V}$ for any value $V$, and there is no rule that
  can reduce $\config{M'}{\theta}$.
  Then, the evaluation of $\mt{M}$ also gets stuck.
\end{lemma}
\begin{proof}
  It suffices to show the following statement:
  \begin{quote}
    Let $M$, $N$, $\theta$, $\tau$, and $\eta$ be such that
    $N \equiv \mte{M}{\eta}$, \(\WF_{\reflang}(\config{M}{\theta})\), and
    \(\Inv(\theta, \tau, \eta)\) hold.
    If $\config{M}{\theta}$ is stuck, then $\config{N}{\tau}$ is also stuck.
  \end{quote}
  We prove it by induction on the structure of $M$, giving only three cases for
  brevity.
  The other cases follow by similar reasoning.

  \noindent \textbf{Case} (product matching):
  \begin{caseindent}
    Suppose \(M \equiv \pcase{V}{x_1}{x_2}{M'}\).
    Since the $\config{M}{\theta}$ is stuck, we know that
    $V \neq \vpair{V_1}{V_2}$ for any values $V_1$ and $V_2$.
    Consequently, $\mte{V}{\eta}$ also cannot be of the form
    $\vpair{W_1}{W_2}$ for any values $W_1$ and $W_2$.
    Therefore, the configuration
    $\config{\mte{\pcase{V}{x_1}{x_2}{M}}{\eta}}{\tau}$ is stuck,
    and this concludes the current case.
  \end{caseindent}

  \noindent \textbf{Case} (set):
  \begin{caseindent}
    Suppose \(M \equiv \set{V}{W}\).
    It follows that $V \neq l$ for any reference cell $l$,
    since if $V = l$ for some $l$, then by
    \(\WF_{\reflang}(\config{\set{V}{W}}{\theta})\), we know that
    $l \in \dom{\theta}$. However, this contradicts the assumption that
    $\config{\set{V}{W}}{\theta}$ is stuck.
    This implies that $\mte{\set{V}{W}}{\eta}$ is also stuck, which
    concludes the current case.
  \end{caseindent}

  \noindent \textbf{Case} (sequencing):
  \begin{caseindent}
    Suppose \(M \equiv \seq{x}{M_1}{M_2}\).
    Since $\config{M}{\theta}$ is stuck, $M_1$
    cannot be of the form $\return{V}$ for any value $V$.
    By the induction hypothesis, we obtain that $N_1$ is stuck and
    that $N_1 \neq \return{W}$ for any value $W$.
    Consequently, the configuration
    $\config{\seq{x}{N_1}{\mte{M_2}{\eta}}}{\tau}$ is also stuck,
    which completes the current case.
  \end{caseindent}
\end{proof}

We now prove Theorem~\ref{thm:reftoac}.

\begin{proof}[Proof of Theorem~\ref{thm:reftoac}]
  We prove that $\eval{\reflang}{M}$ is defined if and only if
  $\eval{\ac}{\mt{M}}$ is defined and show the ``only if'' direction
  first.
  Suppose that
  $\redRplus{\config{M}{\emptyset}}{\config{\return{V}}{\theta}}$ for
  some $\theta$.
  By Lemma~\ref{lem:reftoac_initial}, we obtain $\simu{\config{M}{\emptyset}}{\config{\mt{M}}{\emptyset}}$.
  By repeatedly applying Proposition~\ref{prop:reftoac_sim}, it follows that there exist $\eta$ and $\tau$ such that
  \[
    \redACplus{\config{\mt{M}}{\emptyset}}{\config{\return{\mte{V}{\eta}}}{\tau}}.
  \]
  Therefore, $\eval{\ac}{\mt{M}}$ is defined.

  Next, we prove the contrapositive of the ``if'' direction.
  There are two cases where $\eval{\reflang}{M}$ is undefined: the
  evaluation of \(M\) (1) diverges, or (2) gets stuck.
  In each case, Lemmas~\ref{lem:reftoac:diverge} and
  \ref{lem:reftoac:stuck} imply that the evaluation of $\mt{M}$ diverges or
  gets stuck, respectively.
  Thus, $\eval{\ac}{\mt{M}}$ is undefined, which completes the proof.
  
\end{proof}

%% file: appD-non-deltoref.tex
\section{Supplementary proofs for Section~\ref{sec:nonexistent:deltoref}}
\label{app:sec:nonexistent:deltoref}

In this section, we give the proof of Theorem~\ref{thm:no-del-to-ref}.

\begin{lemma}\label{lem:non-termination-lift}
  If the evaluation of a \(\reflang\) configuration \(\config{M}{\theta}\) diverges, then for every \(\reflang\) evaluation context \(\context{E}\), the evaluation of \(\config{\context{E}[M]}{\theta}\) also diverges.
\end{lemma}
\begin{proof}
  By induction on the structure of \(\context{E}\).
  The base case is trivial.
  Let \(\context{F}\) be the outermost frame of \(\context{E}\) and \(\context{E}'\) be the inner evaluation context.
  By the induction hypothesis, there exists the maximal reduction sequence
  \[
    \config{\context{E}'[M]}{\theta}
    = C_0 \arrR C_1 \arrR C_2 \arrR \cdots.
  \]
  By the definition of the reduction, each step
  \[
    C_i = \config{M_i}{\theta_i} \arrR C_{i+1} = \config{M_{i+1}}{\theta_{i+1}}
  \]
  induces the step
  \[
    \config{\plug{F}{M_i}}{\theta_i} \arrR \config{\plug{F}{M_{i+1}}}{\theta_{i+1}}.
  \]
  Thus, we obtain an infinite reduction sequence
  \[
    \config{\plug{E}{M}}{\theta}  = \config{\plug{F}{\context{E}'[M]}}{\theta} \arrR \config{\plug{F}{M_1}}{\theta_1} \arrR \cdots.
  \]
  Since the reduction is deterministic, the evaluation of \(\config{\plug{E}{M}}{\theta}\) diverges.
\end{proof}

\begin{lemma}\label{lem:non-termination}
  In \(\reflang\), for any evaluation context \(\context{E}\), computation \(M\), and store \(\theta\), the evaluation of \(\config{\plug{E}{\letin{\underscore}{M}{\Omega}}}{\theta}\) does not terminate successfully.
  More precisely, the evaluation diverges or gets stuck at a configuration of the form
  \(\config{\plug{E}{\letin{\underscore}{M_0}{\Omega}}}{\theta_0}\) for some
  computation \(M_0\) and store \(\theta_0\).
\end{lemma}

\begin{proof}
  Let \(\context{L} \defeq \letin{\underscore}{\hole}{\Omega}\).
  Suppose the evaluation of the given configuration terminates successfully.
  Then, there exist a value \(V\) and a store \(\theta'\) such that
  \[
    \config{\plug{E}{\plug{L}{M}}}{\theta} \arrR^{+} \config{\return{V}}{\theta'}.
  \]

  Here, we use the following lemma, which can be easily checked.

  \begin{lemma}\label{lem:let-destruct}
    Let \(\context{C}\) be an evaluation context and \(\tau\) a store, and
    suppose that
    \[
      \config{\plug{C}{\plug{L}{M}}}{\tau}
      \arrR
      \config{P}{\tau'}.
    \]
    Then one of the following holds.

    \begin{enumerate}
    \item The step is performed inside \(M\).
    \item The \textbf{let} frame \(\context{L}\) is eliminated: there exists a value \(W\) such that
      \[
        M \equiv \return{W},
      \]
      and
      \[
        \config{\plug{C}{\plug{L}{\return{W}}}}{\tau}
        \arrR
        \config{\plug{C}{\Omega}}{\tau}.
      \]
    \end{enumerate}
  \end{lemma}

  By this lemma, as long as the \textbf{let} frame \(\context{L}\) has not been eliminated, each reduction step can only occur inside \(M\).
  Since the final configuration no longer contains the \textbf{let} frame, there must be a step where it is eliminated.
  Hence, for some evaluation context \(\context{E}_0\), value \(W\), and store \(\theta_0\), the sequence is of the form
  \[
    \config{\plug{E}{\plug{L}{M}}}{\theta}
    \arrR^*
    \config{\context{E}_0[\plug{L}{\return{W}}]}{\theta_0}
    \arrR
    \config{\context{E}_0[\Omega]}{\theta_0}.
  \]

  However, the evaluation of \(\config{\context{E}_0[\Omega]}{\theta_0}\)
  diverges, and by Lemma~\ref{lem:non-termination-lift}, so does the evaluation
  of \(\config{\context{E}_0[\Omega]}{\theta_0}\).
  Since the reduction is deterministic, this implies that the unique maximal reduction sequence starting from \(\config{\plug{E}{\plug{L}{M}}}{\theta}\) is infinite.
  This contradicts the assumption that the evaluation of \(\config{\plug{E}{\plug{L}{M}}}{\theta}\) terminates successfully.

  Therefore, the evaluation does not terminate successfully.
  More precisely, the evaluation either diverges, or reaches a stuck configuration of the form \(\config{\plug{E}{\letin{\underscore}{M_0}{\Omega}}}{\theta_0}\) for some computation \(M_0\) and store \(\theta_0\).
\end{proof}

\begin{proposition}\label{prop:appD:obs-equiv}
  Let \(\context{L} \defeq \letin{\underscore}{\hole}{\Omega}\).
  For any \(\reflang\) computations \(M\) and \(N\), the thunks \(\thunk{\plug{L}{M}}\) and \(\thunk{\plug{L}{N}}\) are observationally equivalent.
  That is, for every context \(\context{C}\) that expects a value,
  \(\eval{\reflang}{\plug{C}{\thunk{\plug{L}{M}}}}\) is defined if and only if
  \(\eval{\reflang}{\plug{C}{\thunk{\plug{L}{N}}}}\) is defined.
\end{proposition}

We show this proposition by proving a lock-step simulation.
The idea is that \(\thunk{\plug{L}{M}}\) and \(\thunk{\plug{L}{N}}\) are
\emph{values}, and thus they are copied around by substitution during a
reduction; the only step that can reveal the difference between them is the one
that forces them, and such a step leads to a configuration whose evaluation
never terminates successfully by Lemma~\ref{lem:non-termination}.

\begin{definition}\label{app:def:deltoref:cong}
  We define the relation \(\sim\) on \(\reflang\) values and on \(\reflang\)
  computations as the least pair of relations closed under the following two
  rules:
  \begin{enumerate}[label=(S\arabic*), ref=S\arabic*]
  \item\label{app:def:deltoref:cong:ax}
    \(\simu{\thunk{\plug{L}{M}}}{\thunk{\plug{L}{N}}}\) for any \(\reflang\)
    computations \(M\) and \(N\);
  \item\label{app:def:deltoref:cong:cong} for every \(n\)-ary constructor \(F\)
    of \(\reflang\) and terms \(P_1, \ldots, P_n\) and \(Q_1, \ldots, Q_n\) with
    \(\simu{P_i}{Q_i}\) for each \(i\),
    \[
      \simu{F(P_1, \ldots, P_n)}{F(Q_1, \ldots, Q_n)}.
    \]
  \end{enumerate}
  We extend \(\sim\) to (general) contexts by regarding the hole as a nullary constructor,
  to stores by letting \(\simu{\theta}{\tau}\) hold if and only if
  \(\dom{\theta} = \dom{\tau}\) and \(\simu{\theta(l)}{\tau(l)}\) for every
  \(l \in \dom{\theta}\), and to configurations by letting
  \(\simu{\config{M}{\theta}}{\config{N}{\tau}}\) hold if and only if
  \(\simu{M}{N}\) and \(\simu{\theta}{\tau}\).
\end{definition}

Note that \(\sim\) is reflexive, since the nullary constructors, such as
variables, are related to themselves by (\ref{app:def:deltoref:cong:cong}), and
that \(\sim\) is symmetric, since (\ref{app:def:deltoref:cong:ax}) is imposed
for every pair of computations.

\begin{lemma}\label{app:lem:deltoref:basic}
  The following statements hold.
  \begin{enumerate}
  \item\label{app:lem:deltoref:basic:shape} If \(\simu{M}{N}\), then \(M\) and
    \(N\) have the same outermost constructor, and their corresponding immediate
    subterms are related by \(\sim\).
    The same holds for \(\simu{V}{W}\) whenever \(V\) is not a thunk; in
    particular, a value related to a reference cell \(l\) is \(l\) itself.
  \item\label{app:lem:deltoref:basic:subst} If \(\simu{M}{N}\) and
    \(\simu{V}{W}\), then \(\simu{M[V/x]}{N[W/x]}\), and likewise for values.
  \item\label{app:lem:deltoref:basic:ctx} If \(\simu{\context{C}}{\context{C}'}\)
    and \(\simu{P}{Q}\), where \(P\) and \(Q\) are terms of the sort that the
    holes expect, then \(\simu{\plug{C}{P}}{\plug{C'}{Q}}\).
  \end{enumerate}
\end{lemma}

\begin{proof}
  (\ref{app:lem:deltoref:basic:shape}) Since Rule
  (\ref{app:def:deltoref:cong:ax}) relates only thunks, any other related pair
  must be derived by (\ref{app:def:deltoref:cong:cong}).
  Thus (\ref{app:lem:deltoref:basic:shape}) holds.

  (\ref{app:lem:deltoref:basic:subst}) By induction on the derivation of
  \(\simu{M}{N}\).
  The case of (\ref{app:def:deltoref:cong:cong}) is immediate from the induction
  hypothesis.
  In the case of (\ref{app:def:deltoref:cong:ax}), since \(\Omega\) contains no
  free variables,
  \[
    \thunk{\plug{L}{M_0}}[V/x] \equiv \thunk{\plug{L}{M_0[V/x]}}
    \quad\text{and}\quad
    \thunk{\plug{L}{N_0}}[W/x] \equiv \thunk{\plug{L}{N_0[W/x]}},
  \]
  and these two are again related by (\ref{app:def:deltoref:cong:ax}).

  (\ref{app:lem:deltoref:basic:ctx}) By induction on the derivation of
  \(\simu{\context{C}}{\context{C}'}\).
  The case of (\ref{app:def:deltoref:cong:cong}) is immediate from the induction
  hypothesis.
  The base case is either the one where both contexts are the hole, which is
  immediate, or the one where (\ref{app:def:deltoref:cong:ax}) is applied.
  In the latter case, we have \(\context{C} \equiv \thunk{\plug{L}{\context{C}_0}}\)
  and \(\context{C}' \equiv \thunk{\plug{L}{\context{C}_0'}}\) for some contexts
  \(\context{C}_0\) and \(\context{C}_0'\); then \(\plug{C}{P}\) and
  \(\plug{C'}{Q}\) are again related by (\ref{app:def:deltoref:cong:ax}).
\end{proof}

\begin{lemma}\label{app:lem:deltoref:decomp}
  Suppose \(\simu{M}{N}\) and \(M \equiv \plug{E}{R}\), where \(\context{E}\) is
  an evaluation context and \(R\) is a redex.
  Then \(N \equiv \plug{E'}{R'}\) for some evaluation context \(\context{E}'\)
  and redex \(R'\) such that \(\simu{\context{E}}{\context{E}'}\),
  \(\simu{R}{R'}\), and \(R\) and \(R'\) are redexes of the same rule.
\end{lemma}

\begin{proof}
  By induction on \(\context{E}\).
  In the base case, \(M \equiv R\), and we take \(\context{E}' \defeq \hole\)
  and \(R' \defeq N\).
  That \(R'\) is a redex of the same rule as \(R\) follows from
  Lemma~\ref{app:lem:deltoref:basic}(\ref{app:lem:deltoref:basic:shape}): for
  instance, if \(R \equiv \force{V}\) with \(V \equiv \thunk{M_1}\), then
  \(R' \equiv \force{W}\) with \(\simu{V}{W}\), and \(W\) is again a thunk since
  a value related to a thunk is a thunk by either rule; if
  \(R \equiv \scase{V}{L_{\mathnormal{i}}}{x_i}{M_i}\) with
  \(V \equiv \inj{L_{\mathnormal{k}}}{V_0}\), then \(R'\) is a variant matching
  against \(\inj{L_{\mathnormal{k}}}{W_0}\) and whose branches are indexed by
  the same labels, so that the same branch is selected; and if
  \(R \equiv \get{l}\), then \(R' \equiv \get{l}\).

  In the induction step, \(\context{E} \equiv \plug{F}{\context{E}''\hole}\) for
  a frame \(\context{F}\), and
  Lemma~\ref{app:lem:deltoref:basic}(\ref{app:lem:deltoref:basic:shape}) gives
  \(N \equiv \plug{F'}{N''}\) with \(\simu{\context{F}}{\context{F}'}\) and
  \(\simu{\plug{E''}{R}}{N''}\), to which we apply the induction hypothesis.
  Thus, there exist some evaluation context \(\context{E'}\) and redex \(R'\)
  such that \(N'' \equiv \plug{E'}{R'}\), \(\simu{\context{E''}}{\context{E'}}\),
  \(\simu{R}{R'}\), and \(R\) and \(R'\) are redexes of the same rule.
  It is immediate to check
  \(\simu{\plug{F}{\context{E''} \hole}}{\plug{F'}{\context{E'} \hole}}\).
  Since the decomposition of a computation into an evaluation context and a
  redex is unique, \(\plug{F}{\plug{E'}{R'}}\) is the unique decomposition of \(N\).
\end{proof}

\begin{lemma}[simulation]\label{app:lem:deltoref:sim}
  Suppose \(\simu{\config{M}{\theta}}{\config{N}{\tau}}\) and
  \(\redRp{M}{\theta}{M'}{\theta'}\).
  Then one of the following two statements holds.
  \begin{enumerate}
  \item\label{app:lem:deltoref:sim:match} There exist a computation \(N'\) and a
    store \(\tau'\) such that \(\redRp{N}{\tau}{N'}{\tau'}\) and
    \(\simu{\config{M'}{\theta'}}{\config{N'}{\tau'}}\).
  \item\label{app:lem:deltoref:sim:bad} \(M' \equiv \plug{E}{\plug{L}{M_0}}\)
    for some evaluation context \(\context{E}\) and computation \(M_0\).
    In this case, the evaluation of \(\config{M'}{\theta'}\) does not terminate
    successfully by Lemma~\ref{lem:non-termination}.
  \end{enumerate}
\end{lemma}

\begin{proof}
  Let \(M \equiv \plug{E}{R}\) be the decomposition of \(M\) into an evaluation
  context and a redex.
  By Lemma~\ref{app:lem:deltoref:decomp}, \(N \equiv \plug{E'}{R'}\) with
  \(\simu{\context{E}}{\context{E}'}\) and \(\simu{R}{R'}\), where \(R\) and
  \(R'\) are redexes of the same rule.
  Since \(\simu{\context{E}}{\context{E}'}\), when we prove
  (\ref{app:lem:deltoref:sim:match}), it suffices by
  Lemma~\ref{app:lem:deltoref:basic}(\ref{app:lem:deltoref:basic:ctx}) to relate
  the two computations plugging the hole in order to relate the resulting
  configurations.
  We proceed by case analysis on this rule, showing the representative cases
  only.

  \noindent \textbf{Case} \((U)\):
  \begin{caseindent}
    Suppose \(R \equiv \force{\thunk{M_1}}\); then
    \(R' \equiv \force{\thunk{N_1}}\).
    If \(\simu{\thunk{M_1}}{\thunk{N_1}}\) is derived by
    (\ref{app:def:deltoref:cong:cong}), then \(\simu{M_1}{N_1}\), and hence
    \(\redRp{N}{\tau}{\plug{E'}{N_1}}{\tau}\) with
    \(\simu{\plug{E}{M_1}}{\plug{E'}{N_1}}\); thus
    (\ref{app:lem:deltoref:sim:match}) holds.
    Otherwise, it is derived by (\ref{app:def:deltoref:cong:ax}), so
    \(M_1 \equiv \plug{L}{M_0}\) for some computation \(M_0\), and therefore
    \(M' \equiv \plug{E}{\plug{L}{M_0}}\); thus (\ref{app:lem:deltoref:sim:bad})
    holds.
    This is the only case where (\ref{app:lem:deltoref:sim:bad}) arises.
  \end{caseindent}

  \noindent \textbf{Case} \((F)\):
  \begin{caseindent}
    Suppose \(R \equiv \letin{x}{\return{V}}{M_1}\).
    Then, \(R' \equiv \letin{x}{\return{W}}{N_1}\), with \(\simu{V}{W}\) and
    \(\simu{M_1}{N_1}\).
    By
    Lemma~\ref{app:lem:deltoref:basic}(\ref{app:lem:deltoref:basic:subst}), we
    have \(\simu{M_1[V/x]}{N_1[W/x]}\), and hence
    \(\simu{\plug{E}{M_1[V/x]}}{\plug{E'}{N_1[W/x]}}\).
    Since the stores are unchanged, (\ref{app:lem:deltoref:sim:match}) holds.
    The cases \((\times)\), \((+)\), \((\to)\), and \((\&)\) are analogous.
  \end{caseindent}

  \noindent \textbf{Case} \((\mathrm{create})\):
  \begin{caseindent}
    Suppose \(R \equiv \create{V}\); then \(R' \equiv \create{W}\) with
    \(\simu{V}{W}\).
    Since \(\dom{\theta} = \dom{\tau}\), the same fresh label \(l\) is chosen in
    both reductions, and \(\simu{\theta[l := V]}{\tau[l := W]}\) holds.
  \end{caseindent}

  \noindent \textbf{Case} \((\mathrm{get})\):
  \begin{caseindent}
    Suppose \(R \equiv \get{l}\); then \(R' \equiv \get{l}\) by
    Lemma~\ref{app:lem:deltoref:basic}(\ref{app:lem:deltoref:basic:shape}).
    Since \(\dom{\theta} = \dom{\tau}\), the side condition of the rule is
    satisfied on both sides, and \(\simu{\theta(l)}{\tau(l)}\) yields
    \(\simu{\plug{E}{\return{\theta(l)}}}{\plug{E'}{\return{\tau(l)}}}\).
    The case \((\mathrm{set})\) is analogous.
  \end{caseindent}
  Thus, we are done.
\end{proof}

\begin{proof}[Proof of Proposition~\ref{prop:appD:obs-equiv}]
  Here \(\context{C}\) ranges over the contexts such that both
  \(\plug{C}{\thunk{\plug{L}{M}}}\) and \(\plug{C}{\thunk{\plug{L}{N}}}\) are
  \(\reflang\) programs.
  Since \(\sim\) is symmetric, it suffices to show the ``only if'' direction.

  We first show the following statement by induction on \(n\): if
  \(\simu{\config{M_1}{\theta_1}}{\config{N_1}{\tau_1}}\) and
  \(\config{M_1}{\theta_1} \arrR^{n}
  \config{\return{V}}{\theta}\),\footnote{\(C \to^n D\) means that \(C\) is
    reduced to \(D\) in exactly \(n\)-steps.}
  then \(\config{N_1}{\tau_1} \arrR^{n} \config{\return{W}}{\tau}\) for some
  value \(W\) and store \(\tau\).
  If \(n = 0\), then \(M_1 \equiv \return{V}\), and by
  Lemma~\ref{app:lem:deltoref:basic}(\ref{app:lem:deltoref:basic:shape}), we
  have \(N_1 \equiv \return{W}\) for some value \(W\).
  If \(n > 0\), let \(\redRp{M_1}{\theta_1}{M_1'}{\theta_1'}\) be the first step
  of the given reduction sequence, and we then apply
  Lemma~\ref{app:lem:deltoref:sim}.
  Its case (\ref{app:lem:deltoref:sim:bad}) cannot occur, since the evaluation
  of \(\config{M_1'}{\theta_1'}\) does terminate successfully by the assumption
  \(\config{M_1'}{\theta_1'} \arrR^{n-1} \config{\return{V}}{\theta}\).
  Hence its case (\ref{app:lem:deltoref:sim:match}) applies, and the induction
  hypothesis yields the desired reduction sequence.

  Now suppose that \(\eval{\reflang}{\plug{C}{\thunk{\plug{L}{M}}}}\) is
  defined.
  By (\ref{app:def:deltoref:cong:ax}), the reflexivity of \(\sim\), and
  Lemma~\ref{app:lem:deltoref:basic}(\ref{app:lem:deltoref:basic:ctx}), we have
  \[
    \simu
    {\config{\plug{C}{\thunk{\plug{L}{M}}}}{\emptyset}}
    {\config{\plug{C}{\thunk{\plug{L}{N}}}}{\emptyset}},
  \]
  and therefore the above statement implies that
  \(\eval{\reflang}{\plug{C}{\thunk{\plug{L}{N}}}}\) is defined.
\end{proof}

\begin{proof}[Proof of Theorem~\ref{thm:no-del-to-ref}]
  Consider the following two \(\del\) computations:

  \[
    M_i \defeq \dollar{\force{\thunk{\plug{L}{\shift{k}{\return{(\inj{L_{\mathnormal{i}}}{\unit})}}}}}}{x}{\return{x}} \quad (i = 1, 2),
  \]
  where \(\context{L} \defeq (\letin{\underscore}{\hole}{\Omega})\).

  \begin{lemma}\label{lem:source-eval}
    \(\eval{\del}{\dollar{\force{\thunk{\letin{\underscore}{\shift{k}{\return{(\inj{L_{\mathnormal{i}}}{\unit})}}}{\Omega}}}}{x}{\return{x}}} = \inj{L_{\mathnormal{i}}}{\unit}\).
  \end{lemma}
  \begin{proof}
    By routine calculation.
  \end{proof}

  In \(\del\), these results can be distinguished by the following context \(\context{D}_i\):
  \[
    \context{D}_i \defeq
    \begin{pmatrix*}[l]
      \letin{r}{\hole}
      {\\
      \mathbf{case}\;r\;\mathbf{of}\;\{\\
      \quad\inj{L_{\mathnormal{i}}}{\underscore} \mapsto \return{\unit}\\
      \quad\underscore \mapsto \Omega\\
      \}}
    \end{pmatrix*}.
  \]
  Note that \(\context{D}_i\) consists only of \(\mam\) constructs; thus any weak macro-translation preserves it homomorphically.

  Suppose that there exists a weak macro-translation from \(\del\) to \(\reflang\).
  Then, the translation \(\mt{M_i}\) can be written as
  \[
    \cdollar{x}\left[{\force{\thunk{\plug{L}{\cshift{k}\left[{\return{(\inj{L_{\mathnormal{i}}}{\unit})}}\right]}}}},
    \return{x}\right],
  \]
  where \(\cdollar{x}\) is the syntactic abstraction corresponding to a source
  dollar binding \(x\) in its additional part, and \(\cshift{k}\) is the one
  corresponding to a source shift0 binding \(k\) in its body.
  Let \(\context{C}\) denote the context
  \(\cdollar{x}\left[{\force{\hole}}, {\return{x}}\right]\) and \(N_i\) be the
  computation
  \(\plug{L}{\cshift{k}\left[{\return{(\inj{L_{\mathnormal{i}}}{\unit})}}\right]}\).
  Then, \(\mt{M_i}\) can be written as \(\context{C}[\thunk{N_i}]\).
  By Lemma~\ref{lem:non-termination}, the evaluation of \(N_i\) under any evaluation context and any store does not terminate successfully.

  \begin{lemma}
    \(\eval{\reflang}{\mt{M_i}}\) is defined, and it is of the form \(\inj{L_{\mathnormal{i}}}{V_i}\) for some \(\reflang\) value \(V_i\).
  \end{lemma}
  \begin{proof}
    By Lemma~\ref{lem:source-eval}, both \(\eval{\del}{M_i}\) and \(\eval{\del}{\context{D}_i[M_i]}\) are defined.
    Since \(\mtempty\) is a weak macro-translation, both \(\eval{\reflang}{\mt{M_i}}\) and \(\eval{\reflang}{\mt{\context{D}_i[M_i]}}\) are defined.
    Since \(\context{D}_i\) is a \(\mam\) context, and \(\mtempty\) is homomorphic on \(\mam\) constructs,
    \[
      \mt{\context{D}_i[M_i]} \equiv \context{D}_i[\mt{M_i}].
    \]
    Thus \(\eval{\reflang}{\context{D}_i[\mt{M_i}]}\) is defined, which implies that \(\eval{\reflang}{\mt{M_i}}\) is of the form \(\inj{L_{\mathnormal{i}}}{V_i}\) for some \(\reflang\) value \(V_i\).
  \end{proof}

  Thus, the evaluation of \(\context{C}[\thunk{N_1}]\) and
  \(\context{C}[\thunk{N_2}]\) yields \(\inj{L_1}{V_1}\) and \(\inj{L_2}{V_2}\),
  respectively, so \(\thunk{N_1}\) and \(\thunk{N_2}\) are observationally
  inequivalent.
  However, this contradicts Proposition~\ref{prop:appD:obs-equiv}.
  Therefore, there is no weak macro-translation from \(\del\) to \(\reflang\).
\end{proof}

%% file: main.bbl

\begin{thebibliography}{31}

%%% the Web address.

\ifx \showCODEN    \undefined \def \showCODEN     #1{\unskip}     \fi
\ifx \showISBNx    \undefined \def \showISBNx     #1{\unskip}     \fi
\ifx \showISBNxiii \undefined \def \showISBNxiii  #1{\unskip}     \fi
\ifx \showISSN     \undefined \def \showISSN      #1{\unskip}     \fi
\ifx \showLCCN     \undefined \def \showLCCN      #1{\unskip}     \fi
\ifx \shownote     \undefined \def \shownote      #1{#1}          \fi
\ifx \showarticletitle \undefined \def \showarticletitle #1{#1}   \fi
\ifx \showURL      \undefined \def \showURL       {\relax}        \fi
% invisible to TeX
\providecommand\bibfield[2]{#2}
\providecommand\bibinfo[2]{#2}
\providecommand\natexlab[1]{#1}
\providecommand\showeprint[2][]{arXiv:#2}

\bibitem[Anton and Thiemann(2010)]%
        {DBLP:conf/aplas/AntonT10}
\bibfield{author}{\bibinfo{person}{Konrad Anton} {and} \bibinfo{person}{Peter
  Thiemann}.} \bibinfo{year}{2010}\natexlab{}.
\newblock \showarticletitle{Towards Deriving Type Systems and Implementations
  for Coroutines}. In \bibinfo{booktitle}{\emph{Programming Languages and
  Systems - 8th Asian Symposium, {APLAS} 2010, Shanghai, China, November 28 -
  December 1, 2010. Proceedings}} \emph{(\bibinfo{series}{Lecture Notes in
  Computer Science}, Vol.~\bibinfo{volume}{6461})},
  \bibfield{editor}{\bibinfo{person}{Kazunori Ueda}} (Ed.).
  \bibinfo{publisher}{Springer}, \bibinfo{pages}{63--79}.
\newblock
\showISBNx{978-3-642-17163-5}
\href{https://doi.org/10.1007/978-3-642-17164-2\_6}{doi:\nolinkurl{10.1007/978-3-642-17164-2\_6}}


\bibitem[Berdine et~al\mbox{.}(2002)]%
        {berdine2002linear}
\bibfield{author}{\bibinfo{person}{Josh Berdine}, \bibinfo{person}{Peter~W.
  O'Hearn}, \bibinfo{person}{Uday~S. Reddy}, {and} \bibinfo{person}{Hayo
  Thielecke}.} \bibinfo{year}{2002}\natexlab{}.
\newblock \showarticletitle{Linear Continuation-Passing}.
\newblock \bibinfo{journal}{\emph{High. Order Symb. Comput.}}
  \bibinfo{volume}{15}, \bibinfo{number}{2-3} (\bibinfo{year}{2002}),
  \bibinfo{pages}{181--208}.
\newblock
\href{https://doi.org/10.1023/A:1020891112409}{doi:\nolinkurl{10.1023/A:1020891112409}}


\bibitem[Bruggeman et~al\mbox{.}(1996)]%
        {bruggeman1996oneshot}
\bibfield{author}{\bibinfo{person}{Carl Bruggeman}, \bibinfo{person}{Oscar
  Waddell}, {and} \bibinfo{person}{R.~Kent Dybvig}.}
  \bibinfo{year}{1996}\natexlab{}.
\newblock \showarticletitle{Representing Control in the Presence of One-Shot
  Continuations}. In \bibinfo{booktitle}{\emph{Proceedings of the {ACM}
  SIGPLAN'96 Conference on Programming Language Design and Implementation
  (PLDI), Philadephia, Pennsylvania, USA, May 21-24, 1996}},
  \bibfield{editor}{\bibinfo{person}{Charles~N. Fischer}} (Ed.).
  \bibinfo{publisher}{{ACM}}, \bibinfo{pages}{99--107}.
\newblock
\href{https://doi.org/10.1145/231379.231395}{doi:\nolinkurl{10.1145/231379.231395}}


\bibitem[Danvy and Filinski(1989)]%
        {danvy1989functional}
\bibfield{author}{\bibinfo{person}{Olivier Danvy} {and}
  \bibinfo{person}{Andrzej Filinski}.} \bibinfo{year}{1989}\natexlab{}.
\newblock \bibinfo{title}{A functional abstraction of typed contexts}.
\newblock \bibinfo{howpublished}{DIKU Rapport 89/12}.
\newblock


\bibitem[de~Moura and Ierusalimschy(2009)]%
        {moura2009revisiting}
\bibfield{author}{\bibinfo{person}{Ana~L{\'u}cia de Moura} {and}
  \bibinfo{person}{Roberto Ierusalimschy}.} \bibinfo{year}{2009}\natexlab{}.
\newblock \showarticletitle{Revisiting coroutines}.
\newblock \bibinfo{journal}{\emph{{ACM} Trans. Program. Lang. Syst.}}
  \bibinfo{volume}{31}, \bibinfo{number}{2} (\bibinfo{year}{2009}),
  \bibinfo{pages}{6:1--6:31}.
\newblock
\href{https://doi.org/10.1145/1462166.1462167}{doi:\nolinkurl{10.1145/1462166.1462167}}


\bibitem[de~Vilhena and Pottier(2021)]%
        {vilhena2021separation}
\bibfield{author}{\bibinfo{person}{Paulo~Em{\'i}lio de Vilhena} {and}
  \bibinfo{person}{Fran{\c{c}}ois Pottier}.} \bibinfo{year}{2021}\natexlab{}.
\newblock \showarticletitle{A separation logic for effect handlers}.
\newblock \bibinfo{journal}{\emph{Proc. {ACM} Program. Lang.}}
  \bibinfo{volume}{5}, \bibinfo{number}{{POPL}} (\bibinfo{year}{2021}),
  \bibinfo{pages}{1--28}.
\newblock
\href{https://doi.org/10.1145/3434314}{doi:\nolinkurl{10.1145/3434314}}


\bibitem[Dybvig(2006)]%
        {DBLP:conf/icfp/Dybvig06}
\bibfield{author}{\bibinfo{person}{R.~Kent Dybvig}.}
  \bibinfo{year}{2006}\natexlab{}.
\newblock \showarticletitle{The development of Chez Scheme}. In
  \bibinfo{booktitle}{\emph{Proceedings of the 11th {ACM} {SIGPLAN}
  International Conference on Functional Programming, {ICFP} 2006, Portland,
  Oregon, USA, September 16-21, 2006}},
  \bibfield{editor}{\bibinfo{person}{John~H. Reppy} {and}
  \bibinfo{person}{Julia Lawall}} (Eds.). \bibinfo{publisher}{{ACM}},
  \bibinfo{pages}{1--12}.
\newblock
\showISBNx{1-59593-309-3}
\href{https://doi.org/10.1145/1159803.1159805}{doi:\nolinkurl{10.1145/1159803.1159805}}


\bibitem[Felleisen(1991)]%
        {felleisen1991expressive}
\bibfield{author}{\bibinfo{person}{Matthias Felleisen}.}
  \bibinfo{year}{1991}\natexlab{}.
\newblock \showarticletitle{On the Expressive Power of Programming Languages}.
\newblock \bibinfo{journal}{\emph{Sci. Comput. Program.}} \bibinfo{volume}{17},
  \bibinfo{number}{1-3} (\bibinfo{year}{1991}), \bibinfo{pages}{35--75}.
\newblock
\href{https://doi.org/10.1016/0167-6423(91)90036-W}{doi:\nolinkurl{10.1016/0167-6423(91)90036-W}}


\bibitem[Felleisen and Friedman(1987)]%
        {felleisen1987reduction}
\bibfield{author}{\bibinfo{person}{Matthias Felleisen} {and}
  \bibinfo{person}{Daniel~P. Friedman}.} \bibinfo{year}{1987}\natexlab{}.
\newblock \showarticletitle{A Reduction Semantics for Imperative Higher-Order
  Languages}. In \bibinfo{booktitle}{\emph{PARLE, Parallel Architectures and
  Languages Europe, Volume {II:} Parallel Languages, Eindhoven, The
  Netherlands, June 15-19, 1987, Proceedings}} \emph{(\bibinfo{series}{Lecture
  Notes in Computer Science}, Vol.~\bibinfo{volume}{259})},
  \bibfield{editor}{\bibinfo{person}{J.~W. de~Bakker}, \bibinfo{person}{A.~J.
  Nijman}, {and} \bibinfo{person}{Philip~C. Treleaven}} (Eds.).
  \bibinfo{publisher}{Springer}, \bibinfo{pages}{206--223}.
\newblock
\href{https://doi.org/10.1007/3-540-17945-3\_12}{doi:\nolinkurl{10.1007/3-540-17945-3\_12}}


\bibitem[Fiore et~al\mbox{.}(1999)]%
        {fiore1999abstractsyntax}
\bibfield{author}{\bibinfo{person}{Marcelo~P. Fiore},
  \bibinfo{person}{Gordon~D. Plotkin}, {and} \bibinfo{person}{Daniele Turi}.}
  \bibinfo{year}{1999}\natexlab{}.
\newblock \showarticletitle{Abstract Syntax and Variable Binding}. In
  \bibinfo{booktitle}{\emph{14th Annual {IEEE} Symposium on Logic in Computer
  Science, Trento, Italy, July 2-5, 1999}}. \bibinfo{publisher}{{IEEE} Computer
  Society}, \bibinfo{pages}{193--202}.
\newblock
\href{https://doi.org/10.1109/LICS.1999.782615}{doi:\nolinkurl{10.1109/LICS.1999.782615}}


\bibitem[Forster et~al\mbox{.}(2017)]%
        {DBLP:journals/pacmpl/0002KLP17}
\bibfield{author}{\bibinfo{person}{Yannick Forster}, \bibinfo{person}{Ohad
  Kammar}, \bibinfo{person}{Sam Lindley}, {and} \bibinfo{person}{Matija
  Pretnar}.} \bibinfo{year}{2017}\natexlab{}.
\newblock \showarticletitle{On the expressive power of user-defined effects:
  effect handlers, monadic reflection, delimited control}.
\newblock \bibinfo{journal}{\emph{Proc. {ACM} Program. Lang.}}
  \bibinfo{volume}{1}, \bibinfo{number}{{ICFP}} (\bibinfo{year}{2017}),
  \bibinfo{pages}{13:1--13:29}.
\newblock
\href{https://doi.org/10.1145/3110257}{doi:\nolinkurl{10.1145/3110257}}


\bibitem[Forster et~al\mbox{.}(2019)]%
        {forster2019expressive}
\bibfield{author}{\bibinfo{person}{Yannick Forster}, \bibinfo{person}{Ohad
  Kammar}, \bibinfo{person}{Sam Lindley}, {and} \bibinfo{person}{Matija
  Pretnar}.} \bibinfo{year}{2019}\natexlab{}.
\newblock \showarticletitle{On the expressive power of user-defined effects:
  Effect handlers, monadic reflection, delimited control}.
\newblock \bibinfo{journal}{\emph{J. Funct. Program.}}  \bibinfo{volume}{29}
  (\bibinfo{year}{2019}), \bibinfo{pages}{e15}.
\newblock
\href{https://doi.org/10.1017/S0956796819000121}{doi:\nolinkurl{10.1017/S0956796819000121}}


\bibitem[Hasegawa(2002)]%
        {DBLP:conf/flops/Hasegawa02}
\bibfield{author}{\bibinfo{person}{Masahito Hasegawa}.}
  \bibinfo{year}{2002}\natexlab{}.
\newblock \showarticletitle{Linearly Used Effects: Monadic and {CPS}
  Transformations into the Linear Lambda Calculus}. In
  \bibinfo{booktitle}{\emph{Functional and Logic Programming, 6th International
  Symposium, {FLOPS} 2002, Aizu, Japan, September 15-17, 2002, Proceedings}}
  \emph{(\bibinfo{series}{Lecture Notes in Computer Science},
  Vol.~\bibinfo{volume}{2441})}, \bibfield{editor}{\bibinfo{person}{Zhenjiang
  Hu} {and} \bibinfo{person}{Mario Rodr{\'{\i}}guez{-}Artalejo}} (Eds.).
  \bibinfo{publisher}{Springer}, \bibinfo{pages}{167--182}.
\newblock
\showISBNx{3-540-44233-2}
\href{https://doi.org/10.1007/3-540-45788-7\_10}{doi:\nolinkurl{10.1007/3-540-45788-7\_10}}


\bibitem[Ikemori et~al\mbox{.}(2023)]%
        {DBLP:conf/ppdp/IkemoriCM23}
\bibfield{author}{\bibinfo{person}{Kazuki Ikemori}, \bibinfo{person}{Youyou
  Cong}, {and} \bibinfo{person}{Hidehiko Masuhara}.}
  \bibinfo{year}{2023}\natexlab{}.
\newblock \showarticletitle{Typed Equivalence of Labeled Effect Handlers and
  Labeled Delimited Control Operators}. In
  \bibinfo{booktitle}{\emph{International Symposium on Principles and Practice
  of Declarative Programming, {PPDP} 2023, Lisboa, Portugal, October 22-23,
  2023}}, \bibfield{editor}{\bibinfo{person}{Santiago Escobar} {and}
  \bibinfo{person}{Vasco~T. Vasconcelos}} (Eds.). \bibinfo{publisher}{{ACM}},
  \bibinfo{pages}{4:1--4:13}.
\newblock
\href{https://doi.org/10.1145/3610612.3610616}{doi:\nolinkurl{10.1145/3610612.3610616}}


\bibitem[James and Sabry(2011)]%
        {yield2011}
\bibfield{author}{\bibinfo{person}{Roshan~P. James} {and} \bibinfo{person}{Amr
  Sabry}.} \bibinfo{year}{2011}\natexlab{}.
\newblock \showarticletitle{Yield: Mainstream Delimited Continuations}. In
  \bibinfo{booktitle}{\emph{Proceedings on the 1st International Workshop on
  Theory and Practice of Delimited Continuations ({TPDC 2011}), 2011}}.
  \bibinfo{pages}{12 pages}.
\newblock


\bibitem[Kammar and Pretnar(2017)]%
        {DBLP:journals/jfp/KammarP17}
\bibfield{author}{\bibinfo{person}{Ohad Kammar} {and} \bibinfo{person}{Matija
  Pretnar}.} \bibinfo{year}{2017}\natexlab{}.
\newblock \showarticletitle{No value restriction is needed for algebraic
  effects and handlers}.
\newblock \bibinfo{journal}{\emph{J. Funct. Program.}}  \bibinfo{volume}{27}
  (\bibinfo{year}{2017}), \bibinfo{pages}{e7}.
\newblock
\href{https://doi.org/10.1017/S0956796816000320}{doi:\nolinkurl{10.1017/S0956796816000320}}


\bibitem[Kawahara and Kameyama(2020)]%
        {kawahara2020one}
\bibfield{author}{\bibinfo{person}{Satoru Kawahara} {and}
  \bibinfo{person}{Yukiyoshi Kameyama}.} \bibinfo{year}{2020}\natexlab{}.
\newblock \showarticletitle{One-Shot Algebraic Effects as Coroutines}. In
  \bibinfo{booktitle}{\emph{Trends in Functional Programming - 21st
  International Symposium, {TFP} 2020, Krakow, Poland, February 13-14, 2020,
  Revised Selected Papers}} \emph{(\bibinfo{series}{Lecture Notes in Computer
  Science}, Vol.~\bibinfo{volume}{12222})},
  \bibfield{editor}{\bibinfo{person}{Aleksander Byrski} {and}
  \bibinfo{person}{John Hughes}} (Eds.). \bibinfo{publisher}{Springer},
  \bibinfo{pages}{159--179}.
\newblock
\href{https://doi.org/10.1007/978-3-030-57761-2\_8}{doi:\nolinkurl{10.1007/978-3-030-57761-2\_8}}


\bibitem[Kiselyov and Shan(2007)]%
        {DBLP:conf/tlca/KiselyovS07}
\bibfield{author}{\bibinfo{person}{Oleg Kiselyov} {and}
  \bibinfo{person}{Chung{-}chieh Shan}.} \bibinfo{year}{2007}\natexlab{}.
\newblock \showarticletitle{A Substructural Type System for Delimited
  Continuations}. In \bibinfo{booktitle}{\emph{Typed Lambda Calculi and
  Applications, 8th International Conference, {TLCA} 2007, Paris, France, June
  26-28, 2007, Proceedings}} \emph{(\bibinfo{series}{Lecture Notes in Computer
  Science}, Vol.~\bibinfo{volume}{4583})},
  \bibfield{editor}{\bibinfo{person}{Simona Ronchi~Della Rocca}} (Ed.).
  \bibinfo{publisher}{Springer}, \bibinfo{pages}{223--239}.
\newblock
\showISBNx{978-3-540-73227-3}
\href{https://doi.org/10.1007/978-3-540-73228-0\_17}{doi:\nolinkurl{10.1007/978-3-540-73228-0\_17}}


\bibitem[Kobayashi and Kameyama(2025)]%
        {DBLP:conf/aplas/KobayashiK25}
\bibfield{author}{\bibinfo{person}{Kentaro Kobayashi} {and}
  \bibinfo{person}{Yukiyoshi Kameyama}.} \bibinfo{year}{2025}\natexlab{}.
\newblock \showarticletitle{Expressive Power of One-Shot Control Operators and
  Coroutines}. In \bibinfo{booktitle}{\emph{Programming Languages and Systems -
  23rd Asian Symposium, {APLAS} 2025, Bengaluru, India, October 27-30, 2025,
  Proceedings}} \emph{(\bibinfo{series}{Lecture Notes in Computer Science},
  Vol.~\bibinfo{volume}{16201})}, \bibfield{editor}{\bibinfo{person}{Alex
  Potanin}} (Ed.). \bibinfo{publisher}{Springer}, \bibinfo{pages}{88--106}.
\newblock
\showISBNx{978-981-95-3584-2}
\href{https://doi.org/10.1007/978-981-95-3585-9\_5}{doi:\nolinkurl{10.1007/978-981-95-3585-9\_5}}


\bibitem[Levy(2004)]%
        {levy2004cbpv}
\bibfield{author}{\bibinfo{person}{Paul~Blain Levy}.}
  \bibinfo{year}{2004}\natexlab{}.
\newblock \bibinfo{booktitle}{\emph{Call-By-Push-Value: {A}
  Functional/Imperative Synthesis}}. \bibinfo{series}{Semantics Structures in
  Computation}, Vol.~\bibinfo{volume}{2}.
\newblock \bibinfo{publisher}{Springer}.
\newblock
\showISBNx{1-4020-1730-8}


\bibitem[Materzok and Biernacki(2012)]%
        {materzok2012ADI}
\bibfield{author}{\bibinfo{person}{Marek Materzok} {and}
  \bibinfo{person}{Dariusz Biernacki}.} \bibinfo{year}{2012}\natexlab{}.
\newblock \showarticletitle{A Dynamic Interpretation of the {CPS} Hierarchy}.
  In \bibinfo{booktitle}{\emph{Programming Languages and Systems - 10th Asian
  Symposium, {APLAS} 2012, Kyoto, Japan, December 11-13, 2012. Proceedings}}
  \emph{(\bibinfo{series}{Lecture Notes in Computer Science},
  Vol.~\bibinfo{volume}{7705})}, \bibfield{editor}{\bibinfo{person}{Ranjit
  Jhala} {and} \bibinfo{person}{Atsushi Igarashi}} (Eds.).
  \bibinfo{publisher}{Springer}, \bibinfo{pages}{296--311}.
\newblock
\href{https://doi.org/10.1007/978-3-642-35182-2\_21}{doi:\nolinkurl{10.1007/978-3-642-35182-2\_21}}


\bibitem[Moggi(1989)]%
        {DBLP:conf/lics/Moggi89}
\bibfield{author}{\bibinfo{person}{Eugenio Moggi}.}
  \bibinfo{year}{1989}\natexlab{}.
\newblock \showarticletitle{Computational Lambda-Calculus and Monads}. In
  \bibinfo{booktitle}{\emph{Proceedings of the Fourth Annual Symposium on Logic
  in Computer Science {(LICS} '89), Pacific Grove, California, USA, June 5-8,
  1989}}. \bibinfo{publisher}{{IEEE} Computer Society},
  \bibinfo{pages}{14--23}.
\newblock
\showISBNx{0-8186-1954-6}
\href{https://doi.org/10.1109/LICS.1989.39155}{doi:\nolinkurl{10.1109/LICS.1989.39155}}


\bibitem[Phipps{-}Costin et~al\mbox{.}(2023)]%
        {phipps2023continuing}
\bibfield{author}{\bibinfo{person}{Luna Phipps{-}Costin},
  \bibinfo{person}{Andreas Rossberg}, \bibinfo{person}{Arjun Guha},
  \bibinfo{person}{Daan Leijen}, \bibinfo{person}{Daniel Hillerstr{\"{o}}m},
  \bibinfo{person}{K.~C. Sivaramakrishnan}, \bibinfo{person}{Matija Pretnar},
  {and} \bibinfo{person}{Sam Lindley}.} \bibinfo{year}{2023}\natexlab{}.
\newblock \showarticletitle{Continuing WebAssembly with Effect Handlers}.
\newblock \bibinfo{journal}{\emph{Proc. {ACM} Program. Lang.}}
  \bibinfo{volume}{7}, \bibinfo{number}{{OOPSLA2}} (\bibinfo{year}{2023}),
  \bibinfo{pages}{460--485}.
\newblock
\href{https://doi.org/10.1145/3622814}{doi:\nolinkurl{10.1145/3622814}}


\bibitem[Pir{\'{o}}g et~al\mbox{.}(2019)]%
        {DBLP:conf/rta/PirogPS19}
\bibfield{author}{\bibinfo{person}{Maciej Pir{\'{o}}g}, \bibinfo{person}{Piotr
  Polesiuk}, {and} \bibinfo{person}{Filip Sieczkowski}.}
  \bibinfo{year}{2019}\natexlab{}.
\newblock \showarticletitle{Typed Equivalence of Effect Handlers and Delimited
  Control}. In \bibinfo{booktitle}{\emph{4th International Conference on Formal
  Structures for Computation and Deduction, {FSCD} 2019, Dortmund, Germany,
  June 24-30, 2019}} \emph{(\bibinfo{series}{LIPIcs},
  Vol.~\bibinfo{volume}{131})}, \bibfield{editor}{\bibinfo{person}{Herman
  Geuvers}} (Ed.). \bibinfo{publisher}{Schloss Dagstuhl - Leibniz-Zentrum
  f{\"{u}}r Informatik}, \bibinfo{pages}{30:1--30:16}.
\newblock
\showISBNx{978-3-95977-107-8}
\href{https://doi.org/10.4230/LIPICS.FSCD.2019.30}{doi:\nolinkurl{10.4230/LIPICS.FSCD.2019.30}}


\bibitem[Plotkin and Power(2003)]%
        {plotkin2003algebraic}
\bibfield{author}{\bibinfo{person}{Gordon~D. Plotkin} {and}
  \bibinfo{person}{John Power}.} \bibinfo{year}{2003}\natexlab{}.
\newblock \showarticletitle{Algebraic Operations and Generic Effects}.
\newblock \bibinfo{journal}{\emph{Appl. Categorical Struct.}}
  \bibinfo{volume}{11}, \bibinfo{number}{1} (\bibinfo{year}{2003}),
  \bibinfo{pages}{69--94}.
\newblock
\href{https://doi.org/10.1023/A:1023064908962}{doi:\nolinkurl{10.1023/A:1023064908962}}


\bibitem[Plotkin and Pretnar(2009)]%
        {plotkin2009handlers}
\bibfield{author}{\bibinfo{person}{Gordon~D. Plotkin} {and}
  \bibinfo{person}{Matija Pretnar}.} \bibinfo{year}{2009}\natexlab{}.
\newblock \showarticletitle{Handlers of Algebraic Effects}. In
  \bibinfo{booktitle}{\emph{Programming Languages and Systems, 18th European
  Symposium on Programming, {ESOP} 2009, Held as Part of the Joint European
  Conferences on Theory and Practice of Software, {ETAPS} 2009, York, UK, March
  22-29, 2009. Proceedings}} \emph{(\bibinfo{series}{Lecture Notes in Computer
  Science}, Vol.~\bibinfo{volume}{5502})},
  \bibfield{editor}{\bibinfo{person}{Giuseppe Castagna}} (Ed.).
  \bibinfo{publisher}{Springer}, \bibinfo{pages}{80--94}.
\newblock
\href{https://doi.org/10.1007/978-3-642-00590-9\_7}{doi:\nolinkurl{10.1007/978-3-642-00590-9\_7}}


\bibitem[Riecke and Thielecke(1999)]%
        {DBLP:conf/icalp/RieckeT99}
\bibfield{author}{\bibinfo{person}{Jon~G. Riecke} {and} \bibinfo{person}{Hayo
  Thielecke}.} \bibinfo{year}{1999}\natexlab{}.
\newblock \showarticletitle{Typed Exeptions and Continuations Cannot
  Macro-Express Each Other}. In \bibinfo{booktitle}{\emph{Automata, Languages
  and Programming, 26th International Colloquium, ICALP'99, Prague, Czech
  Republic, July 11-15, 1999, Proceedings}} \emph{(\bibinfo{series}{Lecture
  Notes in Computer Science}, Vol.~\bibinfo{volume}{1644})},
  \bibfield{editor}{\bibinfo{person}{Jir{\'{\i}} Wiedermann},
  \bibinfo{person}{Peter van Emde~Boas}, {and} \bibinfo{person}{Mogens
  Nielsen}} (Eds.). \bibinfo{publisher}{Springer}, \bibinfo{pages}{635--644}.
\newblock
\showISBNx{3-540-66224-3}
\href{https://doi.org/10.1007/3-540-48523-6\_60}{doi:\nolinkurl{10.1007/3-540-48523-6\_60}}


\bibitem[Shan(2007)]%
        {DBLP:journals/lisp/Shan07}
\bibfield{author}{\bibinfo{person}{Chung{-}chieh Shan}.}
  \bibinfo{year}{2007}\natexlab{}.
\newblock \showarticletitle{A static simulation of dynamic delimited control}.
\newblock \bibinfo{journal}{\emph{High. Order Symb. Comput.}}
  \bibinfo{volume}{20}, \bibinfo{number}{4} (\bibinfo{year}{2007}),
  \bibinfo{pages}{371--401}.
\newblock
\href{https://doi.org/10.1007/S10990-007-9010-4}{doi:\nolinkurl{10.1007/S10990-007-9010-4}}


\bibitem[Sivaramakrishnan et~al\mbox{.}(2021)]%
        {DBLP:conf/pldi/Sivaramakrishnan21}
\bibfield{author}{\bibinfo{person}{K.~C. Sivaramakrishnan},
  \bibinfo{person}{Stephen Dolan}, \bibinfo{person}{Leo White},
  \bibinfo{person}{Tom Kelly}, \bibinfo{person}{Sadiq Jaffer}, {and}
  \bibinfo{person}{Anil Madhavapeddy}.} \bibinfo{year}{2021}\natexlab{}.
\newblock \showarticletitle{Retrofitting effect handlers onto OCaml}. In
  \bibinfo{booktitle}{\emph{{PLDI} '21: 42nd {ACM} {SIGPLAN} International
  Conference on Programming Language Design and Implementation, Virtual Event,
  Canada, June 20-25, 2021}}, \bibfield{editor}{\bibinfo{person}{Stephen~N.
  Freund} {and} \bibinfo{person}{Eran Yahav}} (Eds.).
  \bibinfo{publisher}{{ACM}}, \bibinfo{pages}{206--221}.
\newblock
\showISBNx{978-1-4503-8391-2}
\href{https://doi.org/10.1145/3453483.3454039}{doi:\nolinkurl{10.1145/3453483.3454039}}


\bibitem[Tang et~al\mbox{.}(2024)]%
        {DBLP:journals/pacmpl/TangHLM24}
\bibfield{author}{\bibinfo{person}{Wenhao Tang}, \bibinfo{person}{Daniel
  Hillerstr{\"{o}}m}, \bibinfo{person}{Sam Lindley}, {and}
  \bibinfo{person}{J.~Garrett Morris}.} \bibinfo{year}{2024}\natexlab{}.
\newblock \showarticletitle{Soundly Handling Linearity}.
\newblock \bibinfo{journal}{\emph{Proc. {ACM} Program. Lang.}}
  \bibinfo{volume}{8}, \bibinfo{number}{{POPL}} (\bibinfo{year}{2024}),
  \bibinfo{pages}{1600--1628}.
\newblock
\href{https://doi.org/10.1145/3632896}{doi:\nolinkurl{10.1145/3632896}}


\bibitem[van Rooij and Krebbers(2025)]%
        {DBLP:journals/pacmpl/RooijK25}
\bibfield{author}{\bibinfo{person}{Orpheas van Rooij} {and}
  \bibinfo{person}{Robbert Krebbers}.} \bibinfo{year}{2025}\natexlab{}.
\newblock \showarticletitle{Affect: An Affine Type and Effect System}.
\newblock \bibinfo{journal}{\emph{Proc. {ACM} Program. Lang.}}
  \bibinfo{volume}{9}, \bibinfo{number}{{POPL}} (\bibinfo{year}{2025}),
  \bibinfo{pages}{126--154}.
\newblock
\href{https://doi.org/10.1145/3704841}{doi:\nolinkurl{10.1145/3704841}}


\end{thebibliography}
